\documentclass[11pt]{article}
\usepackage[utf8]{inputenc}
\usepackage{amsmath, amssymb, amsthm, geometry, graphicx}
\usepackage[numbers, sort & compress]{natbib}
\usepackage[ruled, vlined]{algorithm2e}
\usepackage{booktabs, rotating, threeparttable, float, multicol, multirow, accents, tikz, url, xcolor, hyperref, setspace, enumerate, enumitem}
\hypersetup{
    colorlinks=true,
    linkcolor=red,   
    citecolor=blue,    
    filecolor=red,
    urlcolor=cyan      
}

\newcommand{\pr}{\mathbb{P}}
\newcommand{\E}{\mathbb{E}}

\newcommand{\rmd}{\mathrm{d}}
\newcommand{\mc}{\mathcal}
\newcommand{\mb}{\mathbf}
\newcommand{\bigCI}{\mathrel{\text{\scalebox{1}{$\perp\mkern-10mu\perp$}}}}

\newtheorem{theorem}{Theorem}[section] 
\newtheorem{lemma}[theorem]{Lemma}

\newtheorem{proposition}[theorem]{Proposition}

\newtheorem{remark}[theorem]{Remark}
\newtheorem{assumption}[theorem]{Assumption}
\newtheorem{condition}{Condition}

\title{\bf Rate doubly robust estimation for weighted average treatment effects}

\author{
Yiming Wang$^{1}$ \and
Yi Liu$^{1,\dagger}$ \and
Shu Yang$^{1}$
}

\date{}

\begin{document}

\maketitle

\begin{center}
$^{1}$Department of Statistics, North Carolina State University, Raleigh, North Carolina, USA.\\[0.5em]
$^\dagger$Corresponding author: Yi Liu (\texttt{yliu297@ncsu.edu})
\end{center}

\bigskip

\begin{abstract}
The weighted average treatment effect (WATE) defines a versatile class of causal estimands for populations characterized by propensity score weights, including the average treatment effect (ATE), treatment effect on the treated (ATT), on controls (ATC), and for the overlap population (ATO). WATE has broad applicability in social and medical research, as many datasets from these fields align with its framework. However, the literature lacks a systematic investigation into the robustness and efficiency conditions for WATE estimation. Although doubly robust (DR) estimators are well-studied for ATE, their applicability to other WATEs remains uncertain. This paper investigates whether widely used WATEs admit DR or rate doubly robust (RDR) estimators and assesses the role of nuisance function accuracy, particularly with machine learning. Using semiparametric efficient influence function (EIF) theory and double/debiased machine learning (DML), we propose three RDR estimators under specific rate and regularity conditions and evaluate their performance via Monte Carlo simulations. Applications to NHANES data on smoking and blood lead levels, and SIPP data on 401(k) eligibility, demonstrate the methods' practical relevance in medical and social sciences.
\end{abstract}

\noindent \textbf{Keywords:} Propensity score weighting; Target population; Causal inference; Semiparametric efficient estimation; Double/debiased machine learning.

\doublespace

\section{Introduction}\label{sec:intro}

\subsection{Background}

Causal inference, or comparative effectiveness research methodology, is often used in medical and social science studies. Conventional estimands include the average treatment effect (ATE), the average treatment effect on the treated (ATT) and the average treatment effect on the controls (ATC). They are common measures for comparing mean outcomes across different treatments, exposures, or interventions for a target overall population, treated or controls, respectively. However, these estimands may not always align with the scientific question of interest in all research contexts. For example, in subgroup analyses \citep{yang2021propensity} or when covariate distributions differ systematically across treatment groups (structural violations of positivity \citep{petersen2012diagnosing}), investigators may prefer to focus on patients at intermediate risk or those with greater covariate overlap, rather than on the overall, treated, or control populations. Addressing such settings requires weight functions beyond those associated with the ATE, ATT, or ATC. One approach is trimming \citep{crump2006moving, crump2009dealing}, which shifts the estimand to a subpopulation $\mc O_\alpha(\mb X) = \{\mb X: \alpha<e(\mb X)<1-\alpha\}$, where $e(\mb X)$ is the propensity score and $\mb X$ is a vector of covariates and $\alpha\in(0,0.5)$. This corresponds to applying the weight function $\lambda\{e(\mb X)\} = I(\alpha<e(\mb X)<1-\alpha)$ to the overall population, where $I(\cdot)$ is the indicator function. Another possibility is the use of overlap weights $\lambda\{e(\mb X)\} = e(\mb X)\{1-e(\mb X)\}$ \citep{li2018balancing}, which emphasize individuals in regions of clinical equipoise ($e(\mb X)\approx 0.5$), where treatment assignment is most uncertain \citep{rizk2025and, thomas2020overlap}. Or, alternatively, one could define the target population directly by specifying the weight function as an explicit mapping of covariates, for example, $\lambda(\mb X) = I(\text{age}<65)$ \citep{tao2019doubly, parikh2025we}. In each of these settings, the ``goalpost'' of inference shifts from the overall population to a weighted (or sub-)population. The treatment effect on such a population is termed the \textit{weighted average treatment effect} (WATE) \citep{hirano2003efficient}. Estimating WATE enables investigators to focus on treatment effects within specific populations of interest, providing insights more closely aligned with the underlying scientific goals. In this way, the generalized WATE framework unifies a broad class of estimands as special cases while offering flexibility to define new ones tailored to particular research aims. 

In estimating WATEs from observational data, consistent estimation of nuisance functions, often the propensity score and outcome models, is crucial. Misspecification of these functions can lead to systematic bias \citep{kang2007demystifying}. The augmented inverse probability weighted (AIPW) estimator for the ATE \citep{glynn2010introduction, ogburn2015doubly}, also known as the {doubly robust (DR) estimator \citep{robins1994estimation, bang2005doubly, kang2007demystifying},} ensures consistency if either the propensity score or outcome model is correctly specified. This property, also termed ``model doubly robust'' \citep{ying2024geometric}, hinges on correct specification of at least one nuisance function model. However, for most WATEs beyond the ATE, no DR estimator exists \citep{mao2019propensity, tao2019doubly, matsouaka2024causal}, largely because the propensity score's role in these estimands is non-ancillary. When DR estimators are infeasible for specific WATEs, we can shift our focus to the broader concept of \textit{rate-double robustness (RDR)}, a perspective that, to our knowledge, has not been explored in existing literature. RDR ensures consistency if one or both nuisance function estimators satisfy certain rate convergence conditions. Unlike DR, RDR does not require correct model specification---that is, the nuisance functions are not required to exactly match the true, and typically unknown, functional forms; it only requires that the convergence rates satisfy certain thresholds, which can be more relaxed than parametric rates. Further details on these concepts are provided in Sections \ref{sec:EIF} and \ref{sec:DML}. By developing robust estimation methods applicable to any WATE, the framework serves as a versatile tool for addressing a wide range of practical problems.

Our paper is organized as follows. In the remainder of this section, we review prior work on various WATE estimands and the investigation of their DR estimators in Section \ref{subsec:relate}, and outline our novel contributions in exploring RDR estimators in Section \ref{subsec:ourcon}. Section \ref{sec:setup} introduces the notation and assumptions related to WATE. Section \ref{sec:EIF} derives the EIF for a general WATE estimand, followed by the EIF-based estimator and its theoretical properties. Section \ref{sec:DML} presents two DML-based estimators and discusses their asymptotic unbiasedness and efficiency. Section \ref{sec:simu} provides an extensive simulation study, comparing the finite-sample performance of several WATE estimators under diverse data-generating processes. Section \ref{sec:data} applies the proposed methods to two case studies: (i) the effects of smoking on blood lead levels using data from the 2007–2008 U.S. National Health and Nutrition Examination Survey (NHANES), and (ii) an evaluation of eligibility for enrolling in 401(k) pension plans using data from the 1991 Survey of Income and Program Participation (SIPP). Finally, Section \ref{sec:conc} concludes the paper with comments and remarks. Proofs of theorems and additional simulation results are provided in the \textit{Online Supplemental Material}.

\subsection{Related work}\label{subsec:relate}

There is a strand of literature investigating different WATEs. A seminal example was provided by Crump et al. \cite{crump2006moving} where they trimmed (excluded) units whose estimated propensity scores are out of a pre-specified range $[\alpha,1-\alpha]$, with $\alpha\in(0,0.5)$. This is because units with propensity scores that are too small or large can lead to over-weighting, making the ATE estimator inefficient and invalid under finite-sample. 
Yang and Ding \cite{yang2018asymptotic} further proposed a smooth trimming as a generalization of trimming, which stablizes the weights in the neighborhood of trimming thresholds. The smooth trimming involves a bias-variance trade-off where a small amount of bias is accepted to enhance the efficiency of estimating the treatment effect on the sub-population $\mc O_\alpha(\mb X)$.  Some literature also proposed data-driven based trimming for avoiding the ad hoc choice of threshold parameter $\alpha$. For example, Sturmer et al. \cite{sturmer2010treatment} proposed trimming based on propensity score quantiles, Ma and Wang \cite{ma2020robust}, Parikh et al. \cite{parikh2025we}, and Chaudhuri and Hill \cite{chaudhuri2014heavy} proposed methods that ensure the bias {of} ATE introduced by trimming asymptotically converges to 0. However, these data-driven trimmings alter the target population across examples, leading to unfair comparisons between study sites, and some methods fail to address the structural lack of positivity. 

To circumvent the arbitrary selection of trimming thresholds and simultaneously address the structural violation of positivity, researchers have investigated various alternative weights and their corresponding WATEs. Li and Greene \cite{li2013weighting} proposed matching weights (MW), which are analogous to one-to-one pair matching without replacement based on the propensity score with a caliper; the WATE derived from MW is termed the average treatment effect on matching (ATM). 
Li et al. \cite{li2018balancing} introduced overlap weights (OW) and average treatment effect on overlap (ATO). ATO is a remarkable member of the WATE class because of some desirable properties, such as exact balancing and minimal asymptotic variance shown by Li et al. \cite{li2018balancing}. In Table \ref{tab:WATEs} of Section \ref{sec:setup}, we summarize popular WATEs found in literature. 

As previously mentioned, investigating the existence of DR estimators for WATE is of particular interest. Several studies have explored the DR properties of the \textit{augmented estimator} for WATE, which is the natural generalization of the AIPW estimator of ATE \citep{matsouaka2024causal, moodie2018doubly}. Tao and Fu \cite{tao2019doubly} demonstrated that the augmented estimator is DR when the weight function 
$\lambda\{e(\mb X)\}$ is linear in the propensity score, i.e., 
$\lambda\{e(\mb X)\}=a+be(\mb X)$ for some constants $a,b,$ encompassing the ATE, ATT and ATC. Additionally, Matsouaka et al. \cite{matsouaka2024overlap} found that while the augmented estimators for ATO and its similar alternatives ATM and {the average treatment effect on entropy (ATEN) \cite{zhou2020propensity}} are not theoretically DR, they exhibit near-DR performance in finite samples. Table \ref{tab:WATEs} in Section \ref{sec:setup} summarizes the DR and RDR properties of several popular WATE estimators. Furthermore, this augmented estimator is based on the efficient influence function (EIF) for a general WATE, as derived by Hirano et al. \cite{hirano2003efficient}. However, their derivation treats the propensity score as known and fixed in the joint likelihood, which is an overly strong assumption in practice, as it overlooks the uncertainty associated with estimating the propensity score from data (see Remark \ref{rem:EIF} in Section \ref{sec:EIF} for further discussion). Therefore, to the best of our knowledge, the EIF-based estimator and its corresponding DR or RDR properties for a general WATE remain important topics for investigation.  

\subsection{Our contributions}\label{subsec:ourcon}

In this paper, we propose a unified framework for WATE estimation that incorporates robustness evaluation and efficiency quantification. While previous studies have assessed various estimators for different WATEs \citep{mao2019propensity, matsouaka2024overlap, li2021propensity}, detailed investigation of rate conditions under flexible nonparametric models for nuisance functions remains limited. This gap is significant, as requiring correctly specified models imposes overly strong assumptions in practice.

Building on the EIF framework \citep{tsiatis2006semiparametric} and DML theory \citep{chernozhukov2018double} aforementioned, we formalize three RDR estimators for WATEs. First, we derive the general form of the EIF for any WATE, which naturally leads to an EIF-based estimator. Our derivation for the EIF treats both the propensity score and the conditional outcome models as unknown, aligning with the realistic settings encountered in practice. The EIF-based estimator has the RDR property based on a key condition about the product-type prediction error for nusiance function estimates we found (Condition \ref{cond:EIF-nuisop} in Section \ref{sec:EIF}). However, the asymptotic efficiency of the EIF-based estimator relies on the ``Donsker condition,'' which requires the true nuisance function models to belong to some sufficiently small function class. To address the practical limitations imposed by this condition, we leverage modern machine learning techniques and a sample-splitting strategy to construct two DML-based estimators, following the framework of  Chernozhukov et al. \cite{chernozhukov2018double}. Under the conditions specified in Section \ref{sec:DML}, these two DML-based estimators satisfy the RDR property. 

Overall, our work provides a valuable guideline for practitioners utilizing WATEs, offering a clear pathway to determine whether an estimator is consistent and efficient by verifying the established conditions.

\section{Notation, Assumptions and Estimands}\label{sec:setup}

In this paper, we adopt the potential outcomes framework \citep{neyman1923applications, rubin1974estimating}. Let $A\in\{0,1\}$ denote a binary treatment, where $1$ stands for treated and $0$ stands for control. Let $Y(a)$ be the potential outcome, possibly contrary to the fact, had the individual received treatment $a$, for $a\in\{0,1\}$. 

In the observed data, we can only observe one of the potential outcomes because each participant only received one treatment. The observed outcome is denoted as $Y$, and is assumed to have a consistency relationship to the treatment assignment. That is, $Y=Y(A)=AY(1)+(1-A)Y(0)$. We also assume that {there is only one version of each treatment and }the potential outcomes for an individual are unaffected by the treatments received by other individuals or by their own potential outcomes (the standard Stable Unit Treatment Value Assumption [SUTVA]) \citep{rubin1980comment, rosenbaum1983central}.

Let $\mb X$ denote a vector of baseline covariates. The observed data is denoted as $\mc V=\{\mb V_i=(\mb X_i,A_i,Y_i)\}_{i=1}^n$ with a sample size $n>1$ are assumed to be independent and identically distributed (i.i.d.). To identify a causal effect with this data structure, we make the following additional identification assumptions. 

 

\begin{assumption}[Unconfoundedness]\label{assmp:unconfound} $\{Y(1),Y(0)\}\bigCI A\mid \mb X$.
\end{assumption}

\begin{assumption}[Positivity]\label{assmp:positivity}
There exists two constants $c_1$ and $c_2$, such that $0<c_1\leq e(\mb X)\leq c_2<1$ with probability 1.
\end{assumption}

Given the data structure, we define the conditional average treatment effect (CATE) by $\tau(\mb X)=\E[Y(1)-Y(0)\mid \mb X]$. Let $\mu_{a}(\mb X)=\E(Y\mid A=a,\mb X)$, for $a=0,1$, represent the conditional outcome mean for group $A=a$. Thus, $\tau(\mb X)=\mu_{1}(\mb X)-\mu_{0}(\mb X)$. The propensity score is defined by the probability of receiving the active treatment ($A=1$) given covariates, denoted by $e(\mb X)=P(A=1\mid \mb X)$. Assumption \ref{assmp:positivity} suggests that a deterministic treatment assignment mechanism is not permissible, ensuring all units have non-zero probabilities of receiving either treatment or control.  

To this end, we introduce a general class of estimands, the weighted average treatment effect (WATE). Assume the marginal density of $\mb X$ exists, denoted by $f(\mb X)$. For a specific target population of interest, we represent its marginal density (of $\mb X$) by $C_h^{-1}h(\mb X)f(\mb X)$, where $h(\mb X)$ is the pre-specified weight function that adjusts the distribution of $\mb X$ towards the target population, and $C_h=\E\{h(\mb X)\}$ is a normalizing constant, where the expectation is taken over density $f$. We define the WATE estimand as
\begin{align}\label{eq:WATEestimand}
    \gamma = \gamma(h) = \frac{\E\{h(\mb X)\tau(\mb X)\}}{\E\{h(\mb X)\}},
\end{align}
which represents the average treatment effect on the weighted population defined by $h$. 
Additionally, it is common to posit $h(\mb X)=\lambda\{e(\mb X)\}$ for a univariate function $\lambda(t)$, $t\in[0,1]$, leading to the following form of WATE: 
\begin{align}\label{eq:WATEestimand-PS}
	\gamma = \frac{\E[\lambda\{e(\mb X)\}\tau(\mb X)]}{\E[\lambda\{e(\mb X)\}]}.
\end{align}
Moving forward, we focus on such classes of weight functions defined by the propensity score, which also raises substantial interests in current studies \citep{matsouaka2024overlap, austin2023differences, barnard2024unified}. 
{Aside} from the more commonly considered estimands ATE, ATT and ATC  (with respect to their weight functions  $\lambda(t)=1, \lambda(t)=t$ and $\lambda(t)=1-t$, respectively), other critical WATE estimands of interests are summarized in Table \ref{tab:WATEs}. 

\begin{table}[H]
\caption{Summary of common WATEs, related literature, weight function, propensity score weights, and existences of their DR and RDR estimators}\label{tab:WATEs}
\centering
\footnotesize
\begin{tabular}{ccccccc}
    \toprule 
    Related literature  & {Weight function} $\lambda(t)$ & Target & {Propensity score weights} & DR  & RDR  \\
    \midrule 
    Rubin \cite{rubin1974estimating}  & $1$ & ATE & IPW & $\checkmark$ & $\checkmark$ \\
    Hirano et al. \cite{hirano2003efficient} & $t$ & ATT & IPW treated & $\checkmark$ & $\checkmark$ \\
    Tao and Fu \cite{tao2019doubly} & $1-t$ & ATC & IPW controls & $\checkmark$ & $\checkmark$ \\
    Crump et al. \cite{crump2006moving} & $I(\alpha<t<1-\alpha)$ & TrATE ($\alpha$)  & IPW trimming & $\times$ & $\times$ \\
    \addlinespace
    Yang and Ding \cite{yang2018asymptotic} & {$\Phi_{\epsilon}(1-\alpha-t)\times\Phi_{\epsilon}(t-\alpha)$} & TrATE ($\alpha$) & {Smooth IPW trimming} & $\times$ & $\checkmark$ \\
    \addlinespace
    Li and Greene \cite{li2013weighting} & $\min\{t,1-t\}$ & ATM & {Matching weights (MW)} & $\times$ & $\times$ \\
    \addlinespace
    Mao et al. \cite{mao2019propensity} & {$\min\{1,\min\{t,1-t\}l]$ $(l>1)$} & ATTZ & {Trapezoidal weights (TW)} & $\times$ & $\times$ \\
    \addlinespace
    Li et al. \cite{li2018balancing} & $t(1-t)$ & ATO & {Overlap weights (OW)} & $\times$ & $\checkmark$ \\
    \addlinespace
    Zhou et al. \cite{zhou2020propensity} & $\displaystyle\sum_{k\in\{t,1-t\}}-k\log k$ & ATEN & {Entropy weights (EW)} & $\times$ & $\checkmark$ \\
    \addlinespace
    Matsouaka and Zhou \cite{matsouaka2024causal} & {$t^{\nu_{1}-1}(1-t)^{\nu_{2}-1} (\nu_{1}, \nu_{2}\geq2)$} & ATB & {Beta weights (BW)} & $\times$ & $\checkmark$ \\
    \bottomrule 
\end{tabular}
\begin{tablenotes}
    \item DR: doubly robust. RDR: rate doubly robust. 
    \item TrATE ($\alpha$): Trimmed ATE by trimming threshold $\alpha$; $\Phi_\epsilon(\cdot)$ is the cumulative distribution function of normal with mean 0 and standard error $\epsilon>0$. $\Phi_{\epsilon}(t-\alpha)\Phi_{\epsilon}(1-\alpha-t)\to I(\alpha<t<1-\alpha)$ for every $t$ as $\epsilon\to 0$. 
    \item ATM, ATTZ, ATO, ATEN and ATB are abbreviations of average treatment effect defined by, respectively, MW, TW, OW, EW and BW. ATO is the special case of ATB when $\nu_1=\nu_2=2$. ATE, ATT and ATC are special cases of ATB, respectively, when $(\nu_1, \nu_2)=(1,1), (2,1)$ and $(1,2)$.
\end{tablenotes}
\end{table}

{From the definition of the WATE in \eqref{eq:WATEestimand-PS} and commonly seen WATE members in Table \ref{tab:WATEs}, we highlight the roles of the weight function $\lambda(t)$ (with $t = e(\mb X)$) and the outcome models $\mu_a(\mb X)$ ($a = 0,1$). The weight function $\lambda(t)$ determines the target population---for example, $\lambda(t) = 1$ targets the overall population, while $\lambda(t) = t(1 - t)$ (the overlap weight) emphasizes individuals with propensity scores near 0.5. These weights are applied to the CATE, $\tau(\mb X) = \mu_1(\mb X) - \mu_0(\mb X)$, where the outcome models $\mu_a(\mb X)$ represent expected outcomes under treatment $A=a$ given covariates $\mb X$. Together, $\lambda(t)$ and the outcome models define the estimand by weighting individual-level treatment effects across the population. Thus, accurate estimation of both the propensity score and outcome models is crucial for valid WATE inference, motivating our focus on estimators with the (rate) double robustness property, which ensures consistency under suitable rate convergence conditions on either or both models. }

Given all these setups and definitions, we are interested in efficient and robust inference for estimands in the WATE class. {A central component of our work is the derivation of the EIF for a general WATE estimand. We focus on the EIF for several key reasons. First, for the conventional ATE, the widely used augmented inverse probability weighting (AIPW) estimator \cite{hahn1998role}---celebrated for its double robustness---is directly motivated by the EIF of the ATE. Second, in general semiparametric estimation problems, the EIF is a foundational object that characterizes the most robust and efficient estimator within the class of regular and asymptotically linear estimators. Motivated by these considerations, we begin the next section by presenting the theoretical results for the EIF of WATE and the corresponding EIF-based estimator. }
	
\section{The Efficient Influence Function (EIF) and EIF-based Estimator}
\label{sec:EIF}

{To derive the EIF, we first introduce some general background on its underlying principles. For an estimand of interest in a general statistical problem (in our case, the WATE), denoted by $\gamma$, an estimator $\hat\gamma$ based on i.i.d. observations $\mc V = \{\mb V_i\}_{i=1}^n$ is said to be regular and asymptotically linear (RAL) if 
\begin{align*} 
\sqrt{n}(\hat\gamma - \gamma) = \frac{1}{\sqrt{n}}\sum_{i=1}^n \varphi(\mb V_i) + o_p(n^{-1/2}), \end{align*} 
where $\varphi(\cdot)$ is called an \textit{influence function}. The influence function characterizes how small perturbations in the data distribution affect the estimand $\gamma$, and it determines the asymptotic behavior of the estimator. At the same time, it expresses the estimator as an i.i.d. average, allowing the application of the central limit theorem to establish the asymptotic normality of $\hat\gamma$. 

Among all possible influence functions corresponding to different RAL estimators, the EIF is the one with the smallest asymptotic variance; thus, the RAL estimator associated with the EIF achieves the semiparametric efficiency bound under the given model. To derive the EIF based on the observed data, one first writes out the joint likelihood $f(\mb V)$. Then, to characterize the sensitivity of the estimand to changes in the data-generating distribution, one considers smooth parametric submodels $\{f_\theta(\mb V): \theta\in\mathbb{R}\}$ that pass through the true distribution at $\theta = 0$. By applying the score function and evaluating the pathwise derivative of the estimand $\gamma$ along these submodels, one identifies the EIF as the unique element in the nuisance tangent space representing this derivative \cite{kennedy2016semiparametric}. For a more detailed discussion of these technical aspects, we refer readers to Hines et al. \cite{hines2022demystifying} (particularly Sections 3.1--3.3) and Section A.1 of our \textit{Online Supplemental Material}. } 

In what follows, we derive the EIF of our WATE estimand. Denote the sample average (empirical mean) for any function $f(\cdot)$ of the observed data by $\pr_n[f(\mb V)]=n^{-1}\displaystyle\sum_{i=1}^{n}f(\mb V_i)$. The EIF-based estimator involves nuisance functions $\phi=(e,\mu_1,\mu_0)$. {Without loss of ambiguity, the propensity score $e(\mb X)$ and the two conditional outcome functions $\mu_0(\mb X)$ and $\mu_1(\mb X)$ are referred to as ``nuisance functions'' because, although they play a non-ancillary role in defining the corresponding WATE, they are not of direct inferential interest but are essential components in the estimation process.  }

{Another remark is that we decompose the joint likelihood of the observed data as $f(\mb V) = f(\mb X, A, Y) = f(\mb X) f(A \mid \mb X) f(Y \mid A, \mb X) = f(\mb X) e(\mb X)^A \{1 - e(\mb X)\}^{1-A} f(Y \mid A, \mb X)$. This decomposition explains why both the propensity score $e(\mb X)$ and the outcome models $\mu_0(\mb X)$ and $\mu_1(\mb X)$ appear in the form of the EIF derived later.  }

Denote $\hat{\phi}=(\hat{e},\hat{\mu}_1,\hat{\mu}_0)$ as the estimated nuisance functions with truth $\phi_0=(e,\mu_{0},\mu_1)$. Let $\gamma_0$ denote the truth of $\gamma$ and $\varphi(\mb V;\tilde\phi)$ denote the influence function \eqref{eq:EIF} {with the nuisance function substituted by a general value} $\tilde\phi$. Let $D(\mb V;\phi)=\lambda\{e(\mb X)\}+\dot{\lambda}\{e(\mb X)\}\{A-e(\mb X)\}$ and $N(\mb V;\phi)=\lambda\{e(\mb X)\}\psi_\tau(\mb V;\phi)+\dot{\lambda}\{e(\mb X)\}\tau(\mb X)\{A-e(\mb X)\}$, with $\psi_{\tau}(\mb V)=\dfrac{A}{e(\mb X)}\left\{ Y-\mu_{1}(\mb X)\right\} -\dfrac{1-A}{1-e(\mb X)}\left\{ Y-\mu_{0}(\mb X)\right\} - \tau(\mb X)$.  In Theorem \ref{thm:EIF}, we present the EIF of a general WATE estimand and a EIF-based estimator. 

\begin{theorem}\label{thm:EIF} For any weight function $\lambda(t)$ with first order derivative 
$\dot{\lambda}(t)=\rmd\lambda(t)/\rmd t$, the EIF of $\gamma_0=\E\{N(\mb V;\phi_0)\}/\E\{D(\mb V;\phi_0)\}$ is given by
\begin{align}\label{eq:EIF}
    \varphi(\mb V;\phi_0)=\dfrac{\lambda\{e(\mb X)\}}{\E\{D(\mb V;\phi_0)\}}\left\{\psi_{\tau}(\mb V)-\gamma_0\right\} + \dfrac{\dot{\lambda}\{e(\mb X)\}}{\E\{D(\mb V;\phi_0)\}}\{\tau(\mb X)-\gamma_0\}\{A-e(\mb X)\}. 
\end{align}
\end{theorem}

\begin{remark}\label{rem:EIF}
    Hirano et al. \cite{hirano2003efficient} also derived an EIF for a general WATE, but its form differs from ours. This distinction may seem counterintuitive, given that the EIF is theoretically unique. However, it arises because Hirano et al. \cite{hirano2003efficient} treat the propensity score as known in the joint likelihood of the observed data, thereby excluding it from the parametric submodel. As a result, the nuisance tangent space in their setting is characterized over a smaller class of models that omits variation in the propensity score. This leads to a smaller asymptotic variance for the EIF-based WATE estimator when the propensity score is known. However, the ATE is a special case in which the knowledge of the propensity score does not affect the semiparametric efficiency bound---a fact also emphasized by Hahn \citep{hahn1998role}. For other WATE estimands such as the ATO, where the propensity score is explicitly used to weight the covariate distribution, the propensity score plays a non-ancillary role that is critical for valid uncertainty quantification. Therefore, we do not adopt the model assumption of Hirano et al. \cite{hirano2003efficient} in deriving our EIF-based estimator. 
    
    However, in certain contexts where the propensity scores are known, the semiparametric variance lower bound can indeed be reduced, and the EIF of Hirano et al. \cite{hirano2003efficient} becomes applicable. In Appendix \ref{subapp:rolePS}, we restate their result (Proposition \ref{prop:hinaro}) and demonstrate that the corresponding efficiency gain (the reduction in asymptotic variance) from knowing the true propensity score is given by 
    \begin{align*}
        \E\left[\dfrac{\dot{\lambda}\{e(\mb X)\}^2}{\E\{D(\mb V;\phi_0)\}^2} \{\tau(\mb X) - \gamma_0\}^2 \{A - e(\mb X)\}^2\right]. 
    \end{align*}
    Note that for the ATE, $\dot{\lambda}\{e(\mb X)\} = 0$, which further explains why the efficiency bound is unaffected by knowledge of the propensity score in this special case. A practical example where the true propensity scores are known for all participants is a randomized study---for instance, in a 1:1 randomized clinical trial, where $e(\mb X)\equiv 0.5$. 
\end{remark}

The semiparametric efficiency bound for estimating $\gamma$ is given by $\E\{\varphi(\mb V;\phi_0)^2\}$ \citep{bickel1993efficient, tsiatis2006semiparametric}. When all propensity score $e(\mb X)$ and outcome models $\mu_a(\mb X)$ ($a=0,1$) are correctly specified, this bound can be achieved asymptotically by the following \textit{EIF-based estimator}:
\begin{equation}\label{eq:EIFestimator}
\hat{\gamma}^{\text{eif}}=\frac{\pr_{n}[\lambda\{\hat{e}(\mb X)\}\hat{\psi}_{\tau}(\mb V)+\dot{\lambda}\{\hat{e}(\mb X)\}\hat{\tau}(\mb X)\{A-\hat{e}(\mb X)\}]}{\pr_{n}[\lambda\{\hat{e}(\mb X)\}+\dot{\lambda}\{\hat{e}(\mb X)\}\{A-\hat{e}(\mb X)\}]}.
\end{equation}

To mitigate the bias arising from model misspecifications, we aim to identify an estimator for the targeted WATE that possesses the rate doubly robust (RDR) property \citep{farrell2015robust}. This means that the estimator is consistent and asymptotically normal if its nuisance functions are estimated consistently at a sufficiently fast rate (not necessarily correctly specified). Alternatively, it may allow for a slower rate of convergence for some nuisance function estimate(s), provided that other nuisance function estimates converge in a faster rate. In Section \ref{subsec:CAN-EIF} below, we show that the proposed EIF-based estimator \eqref{eq:EIFestimator} is RDR under some regularity conditions, and Theorem \ref{thm:RDR-EIF} further establishes its consistency and asymptotic normality. 

\subsection{Regularity conditions and asymptotic properties for the EIF-based estimator}\label{subsec:CAN-EIF}

We first state all regularity conditions for the EIF-based estimator. Let $C_1,C_2, C_{\min}$ and $C_{\max}$ be some finite constants with $C_1\leq C_2<\infty$ and $0<C_{\min}<C_{\max}<\infty$.

\begin{condition}\label{cond:EIF-nuisop} 
$\displaystyle\Vert\hat{e}(\mb X)-e(\mb X)\Vert_2\left\{\sum_{a=0}^{1}\Vert\hat{\mu}_{a}(\mb X)-\mu_{a}(\mb X)\Vert_2+\Vert\dot{\lambda}\{\hat{e}(\mb X)\}-\dot{\lambda}\{e(\mb X)\}\Vert_2\right\}=o_p(n^{-1/2})$
\end{condition}

Condition \ref{cond:EIF-nuisop} characterizes the convergence rate required for the nuisance function estimates to ensure the consistency and asymptotic normality of the EIF-based estimator \eqref{eq:EIFestimator}. Notably, it does not require the exactly correct nuisance model specification, but only that the estimation error term on the left-hand side of Condition \ref{cond:EIF-nuisop} is controlled at the $n^{-1/2}$ rate. 

\begin{remark}\label{rmk:ateattatc}
For ATE, ATT and ATC, Condition \ref{cond:EIF-nuisop} can be reduced to $\displaystyle\Vert\hat{e}(\mb X)-e(\mb X)\Vert_2\sum_{a=0}^{1}\Vert\hat{\mu}_{a}(\mb X)-\mu_{a}(\mb X)\Vert_2=o_p(n^{-1/2})$. This is because their $\lambda(t)$ functions are linear to $t$ with the derivations $\dot\lambda(t)$ is a constant for both $t=e(\mb X)$ and $\hat e(\mb X)$ in the condition, therefore $\dot{\lambda}\{\hat{e}(\mb X)\}-\dot{\lambda}\{e(\mb X)\}=0$. This reduced condition for ATE, ATT and ATC gives the standard product-type bias error term characterized by the biases of propensity score and outcome models. This reduced {condition} also implies that when either the propensity score or outcome models consistently estimate their counterparts, the EIF-based estimator is consistent (if we only care about the consistency). 

However, for other WATEs, it is hard to relax this condition. One important reason is that the propensity score plays a non-ancillary role in the definitions of other WATEs, while for ATE, ATT and ATC, the role of propensity score can be ancillary. Note that for ATT and ATC, they can be re-written as $\E\{Y(1)-Y(0)\mid A=1\}$ and $\E\{Y(1)-Y(0)\mid A=0\}$, respectively, which do not involve propensity score. This difference illustrates that the specification of propensity score in other WATEs should satisfy some convergence rate condition to make the EIF-based estimator consistent{, characterized by the term $\Vert\hat{e}(\mb X)-e(\mb X)\Vert_2\cdot\Vert\dot{\lambda}\{\hat{e}(\mb X)\}-\dot{\lambda}\{e(\mb X)\}\Vert_2$ in Condition \ref{cond:EIF-nuisop}}. 
\end{remark} 

\begin{condition}\label{cond:EIFbound} 
$\left\vert\E[\lambda\{e(\mb X)\}\tau(\mb X)]\right\vert<C_{\max}$, and $0<C_{\min}<\left\vert\E[\lambda\{e(\mb X)\}]\right\vert<C_{\max}$.
\end{condition}

Condition \ref{cond:EIFbound} is mild, as it holds when both the weighting function $\lambda\{e(\mb X)\}$ and the CATE $\tau(\mb X)$ are finite almost surely for all $\mb X$. This assumption is often reasonable, since one would not expect the treatment effect to be infinite for any individual, nor the weighting function of the propensity score to take arbitrarily extreme values (under Assumption \ref{assmp:positivity}). 

\begin{condition}\label{cond:EIFdonsker} For $\tilde\phi=\hat\phi$ and $\phi_0$, $N(\mb V;\tilde\phi)$ and $D(\mb V;\tilde\phi)$ belong to two Donsker classes of functions, respectively.  
\end{condition}

Condition \ref{cond:EIFdonsker}, concerning the Donsker class, imposes a mild restriction on the flexibility of the nuisance estimators, but still encompasses many complex functions, allowing $\hat\phi$ to be constructed with flexible estimators; see van der Vaart \cite{van1996weak} and Kennedy \cite{kennedy2016semiparametric} for detailed discussions.

\begin{condition}\label{cond:EIFsup}  $\displaystyle\sup_{\mb X}\left(\left\vert \frac{\lambda\{\hat{e}(\mb X)\}}{\hat{e}(\mb X)}  \right\vert\vee \left\vert \frac{\lambda\{\hat{e}(\mb X)\}}{1-\hat{e}(\mb X)}  \right\vert \vee\left\vert \tau(\mb X) \right\vert\right)\leq C_{\max}$. 
\end{condition}

Condition \ref{cond:EIFsup} imposes an upper bound on certain functions of the covariates $\mb X$. Similar to Condition \ref{cond:EIFbound}, this requirement is mild: it rules out infinite CATE and unbounded functions of the propensity score, and can therefore be readily verified.  

\begin{condition}\label{cond:EIFLipch}
$\lambda(t)$ has continuous first derivatives
$\dot{\lambda}(t)$ and $\dot{\lambda}(t)$ is Lipschitz continuous on a bounded interval $[C_1,C_2]$ where $C_1<C_2<\infty$, $C_1=\displaystyle\inf_{\mb X}\{\hat{e}(\mb X)\wedge e(\mb X)\}$, and $C_2=\displaystyle\sup_{\mb X}\{\hat{e}(\mb X)\vee e(\mb X)\}$. 
\end{condition}

Condition \ref{cond:EIFLipch} regarding Lipschitz continuity holds for commonly considered WATE estimands (such as those in Table \ref{tab:WATEs}), ensuring that if $\hat e(\mb X)$ approximates $e(\mb X)$ well, then $\dot\lambda\{e(\mb X)\}$ can also be well approximated by $\dot\lambda\{\hat e(\mb X)\}$. For example, if $\Vert\hat{e}(\mb X)-e(\mb X)\Vert_2=o_p(n^{-1/4})$, then $\Vert\hat{e}(\mb X)-e(\mb X)\Vert_2\cdot\Vert\dot{\lambda}\{\hat{e}(\mb X)\}-\dot{\lambda}\{e(\mb X)\}\Vert_2=o_p(n^{-1/2})$ holds for any Lipschitz continuous function $\lambda(t)$.  

Under these conditions, Theorem \ref{thm:RDR-EIF} establishes the consistency and asymptotic normality of $\hat\gamma^{\text{eif}}$, with its asymptotic variance characterized by the EIF, which attains the semiparametric efficiency bound. 
	
\begin{theorem}\label{thm:RDR-EIF} 

{Under Conditions \ref{cond:EIF-nuisop}--\ref{cond:EIFLipch}}, the proposed estimator $\hat{\gamma}^{\text{eif}}$ satisfies 
\begin{align}
\sqrt{n}(\hat{\gamma}^{\text{eif}}-\gamma_0)\rightarrow_d\mathcal{N}(0,\E[\varphi(\mb V;\phi_0)^2]) \label{eq:EIFasyp}. 
\end{align}
\end{theorem}

The proof of Theorem \ref{thm:RDR-EIF} is provided in the \textit{Online Supplemental Material} (Section A). By large sample theory, the variance in \eqref{eq:EIFasyp} can be approximated by substituting the analytical components with their estimated equivalents, and the expectation $\E$ with empirical average $\pr_n$. Thus, we estimate the variance of the EIF estimator using the empirical mean square of the estimated EIFs.

We identify several limitations of the EIF-based estimator. In scenarios where the true nuisance functions involve complex relationships with high-dimensional covariates, the breakdown of the Donsker property when using parametric models can lead to suboptimal performance. Additionally, simply plugging in estimated (including machine-learned) nuisance estimators into the EIF generally introduces non-negligible bias because the plug-in errors from estimating the nuisance functions can interact in a nonlinear way with the primary parameter of interest. This highlights the need for more flexible, nonparametric approaches that go beyond the traditional ``estimate-then-plug-in'' framework to better capture the true relationships between the outcome, treatment, and covariates. To address these challenges, we propose and evaluate estimation methods based on double/debiased machine learning (DML) in Section \ref{sec:DML}.

\section{Estimators Based on Double/Debiased Machine Learning (DML)}\label{sec:DML}

In this Section, we introduce two DML-based estimators (DML-1 and DML-2) that differ from the traditional ``estimate-then- plug-in'' approach. These two estimators are motivated by: (i) following the form of the EIF-based estimator in Section \ref{sec:EIF}, and (ii) extending the theory and algorithm developed by Chernozhukov et al. \cite{chernozhukov2018double} from ATE to WATE, which aim to mitigate or minimize the bias associated with nuisance parameter estimation.

Inspired by the Frisch-Waugh-Lovell theorem in linear models \cite{frisch1933partial, lovell1963seasonal}, DML uses an orthogonalization step (see Remark 4) that partials out the variation explained by the nuisance functions, effectively correcting for the bias that arises from imperfect nuisance estimation. Furthermore, sample splitting and cross-fitting are key techniques in DML that prevent overfitting by ensuring that the nuisance functions are estimated on different data from the one used to estimate the target parameter. This helps maintain independence between nuisance estimation errors and target estimation, enabling valid inference even when flexible machine learning methods are used. 

Following Chernozhukov et al. \cite{chernozhukov2018double}, we construct the two DML-based estimators for WATE in Algorithm \ref{alg:DMLs}. Following notation in Section \ref{sec:EIF}, we first write a linear score function for WATE $\gamma_0$ by $L(\mb V;\gamma_0,\phi_0)=\gamma_0 D(\mb V;\phi_0)-N(\mb V;\phi_0)$, then we have the following algorithm.  

\begin{algorithm}[H]
\caption{DML-based Estimators}\label{alg:DMLs}
\textbf{Input:} Observed data $\mc V=(\mb V_i)_{i=1}^n$, user-specified $K > 1$.

1: Randomly partition the index set $\{1,\cdots,n\}$ into $K$ folds, denoted by $(I_k)_{k=1}^K$, such that the size of each fold $I_k$ is $n/K$ (assumed to be an integer for simplicity; in practice, two folds may have a difference on size of $\pm 1$). For each $k \in \{1,\cdots,K\}$, define $I_k^c = \{1,\cdots,n\} \setminus I_k$. 

2: \textbf{Cross-fitting procedure over $K$ folds:}
\textbf{for} each $k \in \{1,\cdots,K\}$
\begin{itemize}
    \item Construct an estimator for nuisance functions using methods from \texttt{SuperLearner} \citep{van2007super}:
    $$
    \hat{\phi}_{k} = \hat{\phi}((\mb V_i)_{i \in I_k^c}).
    $$
    \item Construct estimator $\hat{\gamma}_{k}$ by solving:
    $$
    \frac{1}{n/K} \sum_{i \in I_k} L(\mb V; \hat{\gamma}_{k}, \hat{\phi}_{k}) = 0.
    $$
\end{itemize}
\textbf{end for}

3: Output the \textbf{DML-1 estimator}:
$$
\hat{\gamma}^{\text{dml-1}} = K^{-1} \sum_{k=1}^K \hat{\gamma}_{k}.
$$
4: Output the \textbf{DML-2 estimator} as the solution to:
$$
\frac{1}{n} \sum_{k=1}^K \sum_{i \in I_k} L(\mb V; \hat{\gamma}_{k}, \hat{\phi}_{k}) = 0.
$$

\end{algorithm} 

Based on Algorithm \ref{alg:DMLs}, the DML-1 and DML-2 estimators both rely on the cross-fitting procedure to mitigate overfitting and ensure robustness in high-dimensional settings. The primary difference lies in their aggregation strategies. The DML-1 estimator is a simple average of the fold-specific parameter estimates $\hat\gamma_k$, providing computational simplicity and robustness to small sample variations. In contrast, the DML-2 estimator refines the solution by solving a single global estimating equation that combines all folds, leveraging the full dataset for additional efficiency. Despite their differences, both estimators share the common goal of debiasing and account for nuisance parameter estimation through cross-fitting.

In Section \ref{subsec:CAN-DML}, we present the theoretical properties of the two DML estimators in Theorem \ref{thm:RDR-DML}, together with the required regularity conditions. 

\subsection{Regularity conditions and asymptotic properties for the DML-based estimators}\label{subsec:CAN-DML}

We first introduce Conditions \ref{cond:DMLconverge}--\ref{cond:DMLlowerB} below. The derivation of these conditions is provided in our \textit{Online Supplemental Material} and follows similar reasoning to that used in verifying the assumptions of Chernozhukov et al. \cite{chernozhukov2018double}. 

Let $C_1, C_2, C_{\max}, C_{\min}$ be some finite constants such that $C_1\leq C_2<\infty$ and $0<C_{\min}<C_{\max}<\infty$. Let $q>2$ and $K\geq 2$ be some fixed integers. Below are conditions needed for the consistency and asymptotic normality for the two DML estimators. 

\begin{condition}\label{cond:DMLconverge} Given a random subset $I$ of $\{1,\cdots,n\}$ with size $n/K$, the nuisance parameter estimator $\hat{\phi}_0=\hat{\phi}_0((\mb V_i)_{i\in I^c})$ obeys the following conditions for all $n\geq 1$. With probability $\geq 1-\Delta_n$ for a positive sequence $\{\Delta_n\}_{n\geq 1}$ converging to 0, $\Vert\hat{e}(\mb X)-e(\mb X)\Vert_2\left\{\displaystyle\sum_{a=0}^1\Vert\hat\mu_{a}(\mb X)-\mu_a(\mb X) \Vert_2 + \sup_{r\in[0,1]}\eta_r\{e(\mb X)\}\right\} \leq \delta_n n^{-1/2}$, with $\{\delta_n\}_{n\geq 1}$ is a positive sequence converging to 0 with $\delta_n\geq n^{-1/2}$ and 
$$
\eta_r\{e(\mb X)\} = \left\{\Vert\lambda^{(3)}\{e_r(\mb X)\}\Vert_2\cdot\Vert\hat{e}(\mb X)-e(\mb X)\Vert_2 + \Vert\ddot\lambda\{e_r(\mb X)\}\Vert_2\right\}\Vert\hat{e}(\mb X)-e(\mb X)\Vert_2,
$$
for $e_r(\mb X) = e(\mb X) + r\{\hat e(\mb X)-e(\mb X)\}$ and $r\in[0,1]$, where $\lambda^{(3)}(t) = \rmd^3\lambda(t)/\rmd t^3$ is the third-order derivative of function $\lambda(t)$. 
\end{condition} 

Similar to Condition \ref{cond:EIF-nuisop} in Theorem \ref{thm:RDR-EIF}, Condition \ref{cond:DMLconverge} here does not require the estimated nuisance functions to have the exact correct functional form of their true counterparts. 

\begin{remark}\label{rmk:cond2}
    Condition \ref{cond:DMLconverge} shares some similarities with Condition \ref{cond:EIF-nuisop}, but the error term includes a different element, $\eta_r\{e(\mb X)\}$, which is characterized by the second and third derivatives of the weight function. It is straightforward to verify that for ATE, ATT, and ATC, the term $\eta_r\{e(\mb X)\} = 0$, as their weight functions have identically zero second and third derivatives. Thus, for these three estimands, the bound applies to the standard product-type error term, consistent with the result in Chernozhukov et al. \cite{chernozhukov2018double}. The $\eta_r\{e(\mb X)\}$ term, therefore, captures the additional conditions on the propensity score needed for other estimands. For example, in the case of ATO, $\lambda^{(3)}\{e(\mb X)\} \equiv 0$ and $\ddot\lambda\{e(\mb X)\} \equiv -2$, so $\eta_r\{e(\mb X)\} = 2\Vert\hat e(\mb X) - e(\mb X)\Vert_2$. This requires $\Vert\hat e(\mb X) - e(\mb X)\Vert_2 \cdot \Vert\hat e(\mb X) - e(\mb X)\Vert_2$ to be bounded at the rate $\delta_n n^{-1/2}$.
\end{remark}

\begin{condition}\label{cond:DMLlambda}
	$\displaystyle\left\vert\E[\lambda\{e(\mb X)\}\tau(\mb X)]\right\vert<C_{\max}$ and $0<C_{\min}<\left\vert\E[\lambda\{e(\mb X)\}]\right\vert<C_{\max}$.
\end{condition}

\begin{condition}\label{cond:DMLbound} 
\begin{align*}
   \sup_{\mb X}\bigg( & \left\vert\frac{\lambda\{\hat e(\mb X)\}}{\hat e(\mb X)^3}\right\vert \vee
		\left\vert \frac{\lambda\{\hat e(\mb X)\}}{\{1-\hat e(\mb X)\}^3}\right\vert
        \vee \left\vert \frac{\dot{\lambda}\{\hat e(\mb X)\}}{\hat e(\mb X)^2}\right\vert \vee
		\left\vert \frac{\dot{\lambda}\{\hat e(\mb X)\}}{\{1-\hat e(\mb X)\}^2}\right\vert \vee \left\vert \frac{\ddot{\lambda}\{\hat e(\mb X)\}}{\hat e(\mb X)}\right\vert \vee 
		\left\vert \frac{\ddot{\lambda}\{\hat e(\mb X)\}}{1-\hat e(\mb X)}\right\vert \vee 
        \left\vert\lambda^{(3)}\{\hat e(\mb X)\}\right\vert\bigg) \\
        & \leq C_{\max}.
\end{align*}
\end{condition}

\begin{condition} \label{cond:DMLLips} 
	$\lambda(t),\lambda(t)/t, \lambda(t)/(1-t)$ and $\dot{\lambda}(t)$ are locally Lipschitz continuous on a bounded interval $[C_1,C_2]$, where $\displaystyle C_1=\inf_{\mb X}\hat e(\mb X)$ and $\displaystyle C_2=\sup_{\mb X}\hat e(\mb X)$. 
\end{condition}

Conditions \ref{cond:DMLlambda}--\ref{cond:DMLLips} have similar interpretations to those in Conditions \ref{cond:EIFbound}, \ref{cond:EIFsup}, and \ref{cond:EIFLipch} of Theorem \ref{thm:RDR-EIF}. Condition \ref{cond:DMLbound} involves higher-order derivatives of the weight function $\lambda(t)$, but for common WATEs listed in Table \ref{tab:WATEs}, these higher-order derivatives readily satisfy the required bounds under Assumption \ref{assmp:positivity}.

\begin{condition}\label{cond:DMLgamma}
	$\displaystyle\sup_{\mb X}\left\vert\tau(\mb X)\right\vert\leq C_{\max}$ and $\displaystyle\sum_{a=0}^1 \left\Vert\mu_a(\mb X)-\tilde{\mu}_a(\mb X)\right\Vert_q\leq C_{\max}$.
\end{condition}

\begin{condition}\label{cond:DML-U}
	Let $U=A\{Y-\mu_1(\mb X)\}+(1-A)\{Y-\mu_0(\mb X)\}$, then $\E[U^2\mid \mb X]\leq C_{\max}$ for almost surely and $\Vert U\Vert_q\leq C_{\max}$.
\end{condition}

\begin{condition}\label{cond:DMLlowerB} $\E[\varphi(\mb V;\phi_0)^2]\geq C_{\min}$. 
\end{condition}

Conditions \ref{cond:DMLgamma}--\ref{cond:DMLlowerB} further impose certain bounds on the outcome models and the second moment of the EIF, which are relatively mild and typically hold in practice. In addition, in contrast to the conditions required for Theorem \ref{thm:RDR-EIF}, our result does not rely on Donsker class assumptions, which are generally regarded as more stringent and technically challenging to verify in practice. 

Based on the above conditions, Table \ref{tab:WATEconditions} in Appendix \ref{subapp:condRDR-DML} summarizes the requirements on the estimated propensity score $\hat e(\mb X)$ for specific WATE members. We then establish the following asymptotic results for the two DML-based estimators.

\begin{theorem}\label{thm:RDR-DML} 
Under Conditions \ref{cond:DMLconverge}--\ref{cond:DMLlowerB}, it follows that 
\begin{itemize}
    \item $\hat\gamma^{\text{dml-1}}-\hat\gamma^{\text{dml-2}}=o_p(1)$; and 
    \item For $d=1,2$, 
    \begin{align*}
	\sqrt{n}\sigma^{-1}(\hat{\gamma}^{\text{dml-}d}-\gamma_0)\rightarrow_d\mc N(0,1), \text{ where }\sigma^2=J^{-2}\E[L(\mb V;\gamma_0,\phi_{0})^2].
    \end{align*} 
\end{itemize}
Additionally,
\begin{align*}
	\hat{\sigma}^2=\hat{J}^{-2}\frac{1}{K}\sum_{k=1}^{K}\frac1n\sum_{k\in I_k}L(\mb V;\hat{\gamma},\hat{\phi}_{k})^2, 
 \text{ with }\hat{J}=\frac{1}{K}\sum_{k=1}^{K}\frac1n\sum_{k\in I_k}D(\mb V;\hat{\phi}_{k})
\end{align*}
can consistently estimate the asymptotic variance $\sigma^2$, and the $100(1-\alpha)\%$ Wald-type confidence interval is valid in the sense that 
\begin{align*}
	\lim_{n\rightarrow\infty}\sup_n\left\vert P\left(\gamma_0\in [\hat{\gamma}^{\text{dml-}d}\pm \Phi^{-1}(1-\alpha/2)\hat{\sigma}/\sqrt{n})]\right) -(1-\alpha)\right\vert = 0, 
\end{align*}
for an $\alpha\in(0,1)$. 
\end{theorem}

The proof of Theorem \ref{thm:RDR-DML} is provided in our \textit{Online Supplemental Material} (Section B). For completeness, Appendix \ref{subapp:condRDR-DML} contains a more technical statement of the weak convergence (convergence in distribution) of $\hat\gamma^{\text{dml-}d}$ ($d=1,2$) above. To maintain readability in the main text, we omit these technical details here. 

At the end of this section, we provide a remark summarizing key techniques used in DML-based theory, offering interested readers insights into why DML-based estimators are proposed in their particular forms and how they work. 
\begin{remark} The DML-based methods in  Chernozhukov et al. \cite{chernozhukov2018double} rely on two key concepts: \textit{Neyman orthogonality} and \textit{Gateaux derivatives}. Neyman orthogonality ensures that small errors in estimating nuisance parameters (e.g., propensity score, conditional outcomes) minimally affect the final WATE estimator. This is achieved by constructing a \textit{Neyman orthogonal score} $\psi(\mb V;\phi)$ for each nuisance parameter, where $\mb V$ represents the data and $\phi$ the nuisance parameter vector. The estimator $\hat\phi_0$ solves the equation 
$$
\frac1n\sum_{i=1}^n\psi(\mb V_i;\hat\phi_0) = 0,
$$ subject to moment conditions formalized via Gateaux derivatives, which measure estimator sensitivity to perturbations. Neyman orthogonality ensures that this derivative is zero at the true parameter values. Sample splitting and cross-fitting mitigate overfitting and maintain orthogonality, reducing bias when using the same data for nuisance parameter estimation and WATE evaluation. For technical details, see our \textit{Online Supplemental Material} (Section B) and  Chernozhukov et al. \cite{chernozhukov2018double}
\end{remark}

So far, we have introduced our proposed methods. The next Sections \ref{sec:simu} and \ref{sec:data} present numerical experiments and empirical applications to evaluate our methods in finite-sample and assist statistical practices.

\section{Numerical Experiments}\label{sec:simu}

We conduct extensive numerical experiments using Monte Carlo simulations. Our simulation study is implemented using the R package we developed, \texttt{WATE}, which is available at \url{https://github.com/yiliu1998/WATE}. The \texttt{WATE} package integrates all methods available in the \texttt{SuperLearner} R package \citep{van2007super, polley2011super} for fitting and predicting nuisance functions.

We compare the following WATE estimators: the EIF-based, the two DML-based, and two other na\"ive estimators. The na\"ive estimators are denoted by $\hat{\gamma}_{\text{na\"ive}-1}$ and $\hat{\gamma}_{\text{na\"ive}-2}$, and are defined by
\begin{align}\label{eq:naive-est}
	\hat{\gamma}_{\text{na\"ive}-1}=\frac{\pr_n[\lambda\{\hat{e}(\mb X)\}\{\hat{\mu}_1(\mb X)-\hat{\mu}_0(\mb X)\}]}{\pr_n [\lambda\{\hat{e}(\mb X)\}]}, \text{ and }
	\hat{\gamma}_{\text{na\"ive}-2}=\frac{\pr_n[\lambda\{\hat{e}(\mb X)\}\hat\tau^{\text{ipw}}(\mb X)]}{\pr_n[\lambda\{\hat{e}(\mb X)\}]}, 
\end{align}
where $\hat\tau^{\text{ipw}}(\mb X)$ is an estimator of $\tau^{\text{ipw}}(\mb X) = \dfrac{AY}{e(\mb X)}-\dfrac{(1-A)Y}{1-e(\mb X)}$ using the estimated propensity score $\hat e(\mb X)$ for $e(\mb X)$. These two na\"ive estimators are naturally motivated from the fact that both the outcome difference $\hat{\mu}_1(\mb X)-\hat{\mu}_0(\mb X)$ and the inverse probability weighting (IPW) component $\tau^{\text{ipw}}(\mb X)$ are consistent for estimating CATE $\tau(\mb X)$ if the nuisance functions are correctly specified. They serve as competing methods for comparing efficiency and unbiasedness of our developed estimators. 

\subsection{Data generating processes}\label{subsec:DGP}
We consider the following data generation strategies that enable us to generate outcomes and treatment data through both simple and complex mechanisms. First, consider a basic outcome model, defined by
\begin{align*}
Y=\tau(\mb X) A +\mu_0(\mb X)+U, 
\end{align*}
where $\mu_0(\mb X)$ is a baseline outcome for every individual, $\tau(\mb X)=\mu_1(\mb X)-\mu_0(\mb X)$ is the CATE given the covariate vector $\mb X$, $A$ is a binary treatment generated from the Bernoulli distribution with probability $e(\mb X)$, and $U$ is a random error with $U\sim\mc N(0,1)$. The vector $\mb X\in \mathbb{R}^p=(X_1,\cdots,X_p)'$ consists of $p=10$ different variables. Each variable $X_i$ is generated from independent $\mathcal{N}(0,1)$ distributions. The baseline outcome is given by $\mu_0(\mb X)=X_1+X_5+X_4\times X_5$. 

For the treatment assignment, we consider not only the simple scenario of completely random treatment assignment in case (a), which uses a constant propensity score across all study participants, but also three additional sophisticated scenarios in cases (b) to (d).  Specifically, we generate $e(\mb X)$ using
$$e(\mb X)=\phi\left[\frac{a(\mb X)-\E\{a(\mb X)\}}{\sigma\{a(\mb X)\}}\right],$$
where $\phi(\cdot)$ is the cumulative distribution function of the standard normal distribution, with three different $a(\mb X)$'s are specified in (b) to (d) below. 
\begin{itemize}\renewcommand{\labelitemi}{\hspace{0.5em}}
	\item  (a) random: $e(\mb X)=c$ with $c=0.3$;
	\item  (b) linear: $a(\mb X)=X_2+X_3+X_5-X_8$;
	\item  (c) interaction: $a(\mb X)=\mb X'\mb b+X_2+X_5+X_3\times X_8$; and
	\item  (d) non-linear: $a(\mb X)=\mb X'\mb b+X_2+1.5\cdot\cos(X_4\times X_8)+2\cdot\sin(X_5)$.
\end{itemize}
In (c) and (d), the vector $\mb b=\mb l^{-1}$ with $\mb l=(1,2\cdots,p)'$ denotes the weights for each covariate. In addition, we choose four different CATEs $\tau(\mb X)$ with various degrees of heterogeneity below. A binary effect (case (b)) and no effect (case (b)) are with less heterogeneity, in contrast to cases (c) and (d).
\begin{itemize}\renewcommand{\labelitemi}{\hspace{0.5em}}
	\item (a) zero: $\tau(\mb X)=0$;
	\item (b) binary: $\tau(\mb X)= 2I(X_2\geq 0) -I(X_2<0)$;
	\item (c) linear: $\tau(\mb X)=X_1+X_2$; and
	\item (d) non-linear: $\tau(\mb X)=2\cdot\sin(X_1+0.5\cdot X_2+0.33\cdot X_3)+1.5\cdot \cos(X_{10})$.
\end{itemize}

We consider five complete data generating processes (DGPs) listed in Table \ref{tab:fiveDGPs} by different combinations of $e(\mb X)$ and $\tau(\mb X)$. Inherently, the nuisance functions within DGPs 2 to 5 possess sufficient smoothness to satisfy the Donsker conditions as per Theorem \ref{thm:RDR-EIF}. Conversely, DGP 1 may not meet these conditions, necessitating the use of adaptive machine learning (ML) methods as outlined in  Theorem \ref{thm:RDR-DML}. Given that the unbiasedness of a WATE estimator {hinges} on the  prediction accuracy of nuisance functions, selecting appropriate ML methods for their estimation is crucial. The \texttt{SuperLearner} R package \citep{van2007super} offers an array of viable modern ML methods, including a strategic hybrid method known as ensemble. The ensemble method allows for the combination of multiple ML approaches to estimate nuisance functions, while the \texttt{SuperLearner} algorithm subsequently assigns weights to these methods, employing data-driven strategies to minimize prediction error.

\begin{table}
    \centering
    \caption{Data generation processes (DGPs) based on different $e(\mb X)$ and $\tau(\mb X)$ models}\label{tab:fiveDGPs}
    \begin{tabular}{rcccccc}
        \toprule
        DGP  &  1 & 2 & 3 & 4 & 5 \\
        \midrule 
        Propensity score $e(\mb X)$ &  nonlinear  & linear & interaction & random & linear \\
        CATE $\tau(\mb X)$ &   nonlinear  & linear & binary & linear &  zero \\
        \bottomrule 
    \end{tabular}			
\end{table}

In addition, we consider two cases of nuisance function (conditional outcome and propensity score) model specifications: 
\begin{itemize}
    \item Simple GLM: We use only \texttt{SL.glm} for the nuisance functions in all estimators. Specifically, it models the propensity score using logistic regression and the outcomes using linear regression. This approach can be viewed as a systematic model misspecification when the true models are nonlinear and/or include higher-order covariate terms. However, this aligns with common practices in some epidemiological studies, where some investigators prioritize the interpretability, simplicity, and transparency of nuisance function models \citep{frank2015regression, shmueli2010explain}. This further enables us to evaluate whether our proposed methods exhibit greater robustness when traditional parametric models are used for nuisance functions.
    
    \item Methods ensemble: We employ an ensemble of parametric and machine learning methods to fit nuisance functions, aiming to improve prediction accuracy. These methods include the generalized linear model (\texttt{SL.glm} and \texttt{SL.glm.interaction}), Lasso (\texttt{SL.glmnet}), and gradient-boosted trees (\texttt{SL.xgboost}) provided in the \texttt{SuperLearner} R package \citep{van2007super}.  
\end{itemize}

\subsection{Weight functions, estimands and competing methods}\label{subsec:wtfuncs}

To evaluate the proposed methods, we consider various weight functions $\lambda(t)$ and their target estimands. Alongside the standard causal estimands—ATE, ATT, ATC, ATO, and ATEN—we include five ATB estimands (ATB1, $\dots$, ATB5) with parameters $(\nu_1, \nu_2)$ set as $(3,4)$, $(4,3)$, $(2,4)$, $(4,2)$, and $(4,4)$. We exclude trimming, truncation, and ATM due to the non-differentiability of their weight functions, which conflicts with the theoretical conditions in this paper. As ATM resembles ATO and ATEN, we expect similar empirical results for these estimands.

We perform $M=500$ Monte Carlo simulations, varying the sample size $n\in\{100, 400, 1000, 4000\}$ in each replication. For each run, we compute point and variance estimates using the EIF, DML-1, DML-2, and two na\"ive methods from \eqref{eq:naive-est}. For the EIF estimator, we estimate the nuisance functions on the full dataset and apply them in Equation \eqref{eq:EIFestimator}, with variance derived from the mean square of the estimated EIFs. For the DML estimators, we use 5-fold cross-fitting with 10 sample splits, estimating nuisance functions on training sets and predicting on validation sets. 

We follow the mean and median strategies from  Chernozhukov et al. \cite{chernozhukov2018double} to stabilize variability across sample splits. We refer to each DML-based estimator as ``DML-$d$-variance strategy ($d=1,2$)'', e.g., DML-1-mean for the DML-1 estimator with the mean strategy. Similar variations include DML-1-median, DML-2-mean, and DML-2-median. Inference for all EIF and DML-based estimators is done via normal approximation (Wald-type confidence interval). 

\subsection{Evaluation metrics}\label{subsec:metrics}

To evaluate the finite-sample performance of all competing methods, we use the following metrics. Let $\gamma_0$ represent the true value of the WATE estimand and $\hat{\gamma}_i$ its estimate from the $i$-th simulation, for $i = 1, \dots, 500$.
\begin{itemize}
\item To evaluate the accuracy of WATE estimation, suggested by Daniel et al. \cite{daniel2015causal} and Tao and Fu \cite{tao2019doubly}, we compute the absolute bias (ABias), standard deviation (SD) and root mean square error (RMSE) through
\begin{align*}
	\text{ABias} =\left\vert\frac{1}{M}\sum_{i=1}^{M}\hat{\gamma}_i-\gamma_0\right\vert, 
	\text{SD} = \sqrt{\frac{1}{M}\sum_{i=1}^{M}(\hat{\gamma}_i-\bar{\hat{\gamma}})^2}, \text{ and } 
	\text{RMSE} = \sqrt{\frac{1}{M}\sum_{i=1}^{M}(\hat{\gamma}_i-\gamma_0)^2}.
\end{align*}
To get the true value $\gamma_0$, we note that for DGP 5, it is straightforward that $\gamma_0=0$ for all WATEs. However, for DGPs 1 to 4, due to the heterogeneity in $\tau(\mb X)$,  we approximate $\gamma_0$ via
\begin{align*}
	\gamma_0 = \frac{1}{M}\sum_{i=1}^{M}\frac{\pr_{n_{\text{large}}}[ \lambda\{e(\mb X)\}\tau(\mb X)]}{\pr_{n_{\text{large}}}[\lambda\{e(\mb X)\}]},
\end{align*}
using $10$ independent replications of $n_{\text{large}}=10^7$ randomly generated full data of propensity score $e(\mb X)$, potential outcomes $Y(0)$, $Y(1)$, and CATE $\tau(\mb X)$ under the true DGPs 1 to 4 aforementioned. With this large sample size $10^7$ and averaging over $10$ independent replications, the uncertainty associated with each $\gamma_0$ calculation is negligible, so $\gamma_0$ can be viewed as the truth of the corresponding estimand. 
\item Coverage probability (CP\%): to evaluate the validity of the asymptotic normality by Theorem \ref{thm:RDR-EIF} and \ref{thm:RDR-DML}, we compute the CP\% of the Wald-type CI for each WATE as follows:
\begin{align*}
	\mbox{CP\%}=100\%\times\frac{1}{M}\sum_{i=1}^{M} I(\gamma_0\in[\hat{\gamma}_i \pm \Phi^{-1}(1-\alpha/2)\hat{\sigma}_i/\sqrt{n}]),
\end{align*}
where $\Phi^{-1}(\cdot)$ is the quantile function of the standard normal distribution and $[\hat{\gamma}_i\pm\Phi^{-1}(1-\alpha/2)\hat{\sigma}_i/\sqrt{n})]$ is the $100\cdot(1-\alpha)\%$ level Wald-type confidence interval constructed in the $i$th simulation. Notice that because of Monte Carlo error, a CP\% over $M=500$ simulation can be viewed as not significantly different from the 95\% nominal level if it is in the range $95\pm 1.96\sqrt{(95\times5)/M} = [93, 97]$. 
\end{itemize}

\begin{remark}\label{rmk:CP}
    We evaluate the variance estimation of the two na\"ive estimators using the mean squared values of their estimated influence functions. However, these estimators are likely invalid, as they do not fully account for the uncertainty in estimating the nuisance functions, as noted in prior works \citep{lunceford2004stratification, williamson2012variance, zou2016variance, mao2018propensity, austin2022bootstrap}. We include them in our simulations for completeness and to show the advantage of the proposed DML-based methods, that is, the validity of using this simple variance estimator. In addition, we point that the valid variance estimators for the two na\"ive methods might be obtained via (i) nonparametric bootstrap or (ii) empirical sandwich variance estimation. For empirical studies, see Matsouaka et al. \cite{matsouaka2023variance, matsouaka2024overlap} and Li et al. \cite{li2025variance}. However, we comment on their drawbacks: bootstrap can be time-consuming especially when using machine learning models or dealing with big data. The sandwich estimation is restrictive with nonparametric models, and it also lacks a uniform form across parametric models and does not fit ensemble methods like \texttt{SuperLearner} \citep{van2007super}. Fortunately, our DML-based proposal avoids using both bootstrap and sandwich methods. 
\end{remark}

\subsection{Simulation results}\label{subsec:simures}

To be succinct, in this section, we report representative results by {$n=1000$ and 4000 under DGP 1 with only estimands} ATE, ATT, ATC, ATO and ATEN but both nuisance model specifications described in Section \ref{subsec:DGP}. The complete simulation results under {all sample sizes, estimands and DGPs} are available in our \textit{Online Supplemental Material} (Section C). 

First of all, Figure \ref{fig:main-PS} displays the predicted propensity score distributions {as well as the love plots for covariate balance by absolute standardized mean difference (ASMD) \cite{austin2009balance, zhou2020psweight}} for DGP 1 (figures for other DGPs are in our \textit{Online Supplemental Material} [Section C]). Both simple GLM and methods ensemble yield similar distributions across treatment groups. The predicted propensity scores show good overlap between the two treatment groups, with a higher proportion of controls under this DGP. {Furthermore, the propensity score weights for ATO and ATEN achieve better overall covariate balance, followed by those for ATE, ATT and ATC. In contrast, weights from the beta family (ATB1–ATB5) can sometimes lead to poorer balance—even relative to the unweighted differences—due to the inherent instability of their weighting functions, which introduces greater finite-sample instability. }

\begin{figure}
    \centering
    \includegraphics[width=\textwidth]{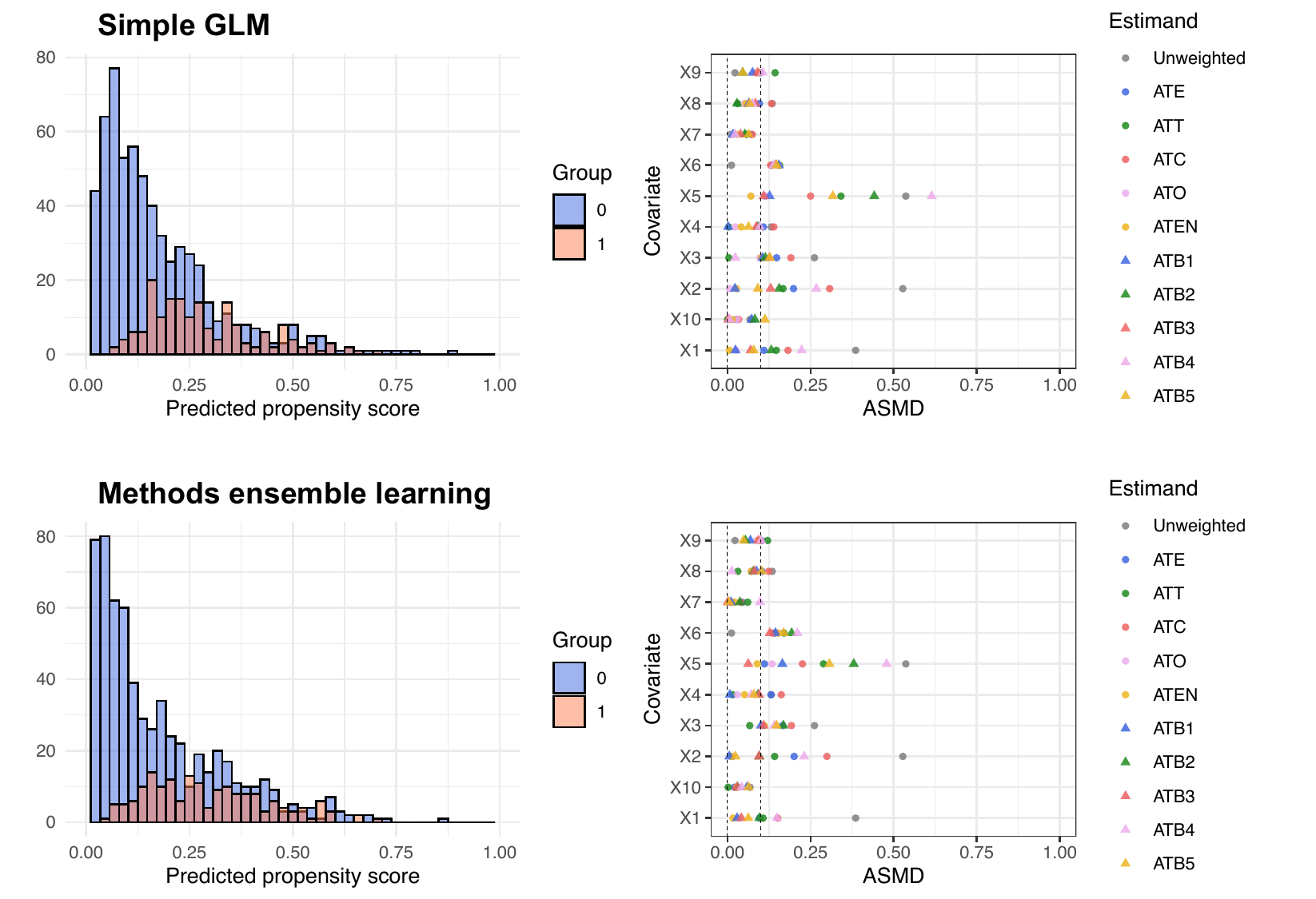}
    \caption{Predicted propensity score distributions by treatment group and covariate balance by absolute standardized mean difference (ASMD), under DGP 1 from the prediction set of a random sample ($n=4000$), where 80\% of the data are used for the training set for training the propensity score model.}
    \label{fig:main-PS}
\end{figure}

\begin{figure}
    \centering
    \includegraphics[width=\textwidth]{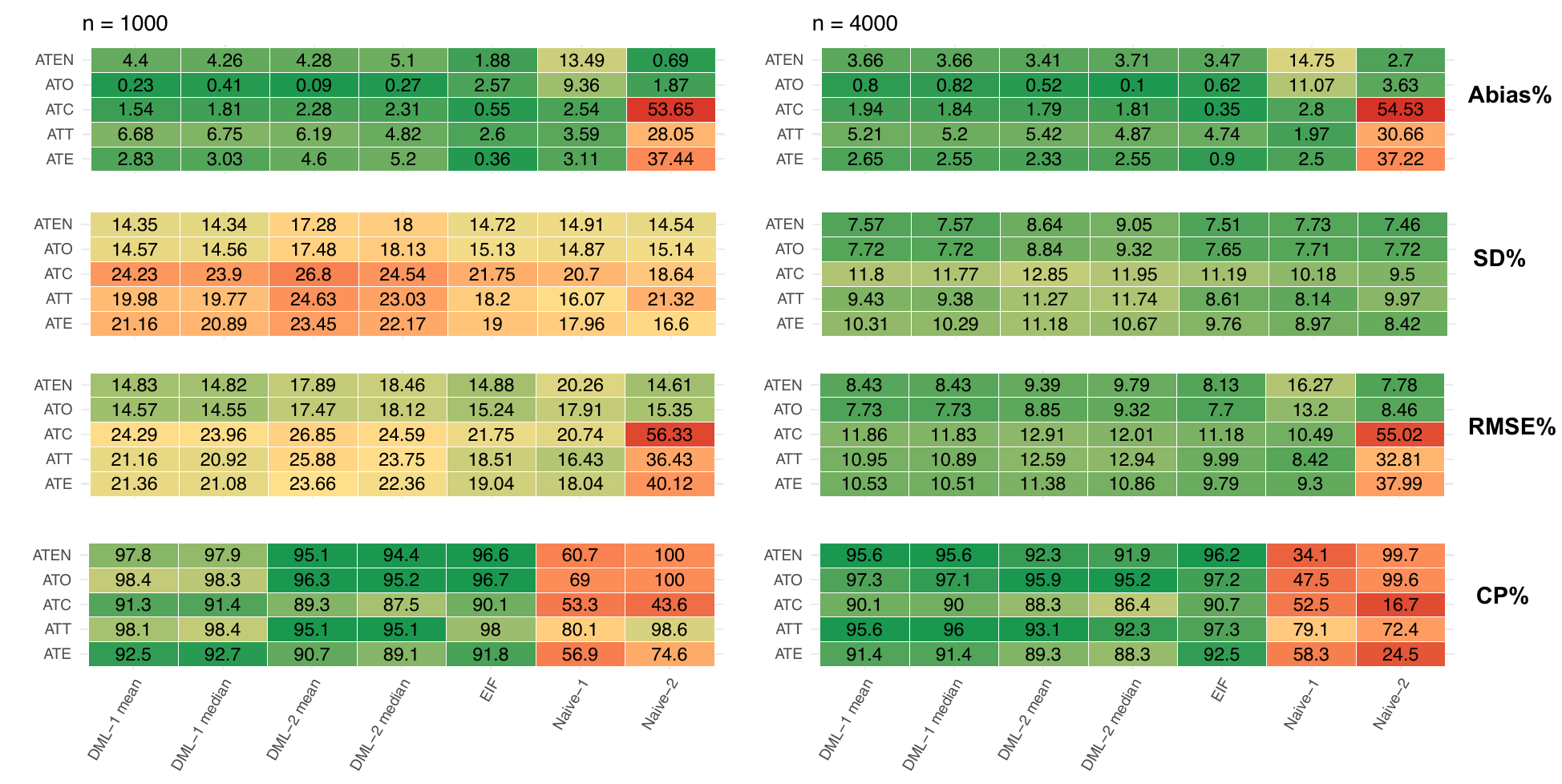}
    \caption{Simulation results by the proposed and na\"ive estimators under DGP 1 and simple GLMs for nuisance model specifications.}
    \label{fig:main-GLM}
\end{figure}

\begin{figure}
    \centering
    \includegraphics[width=\textwidth]{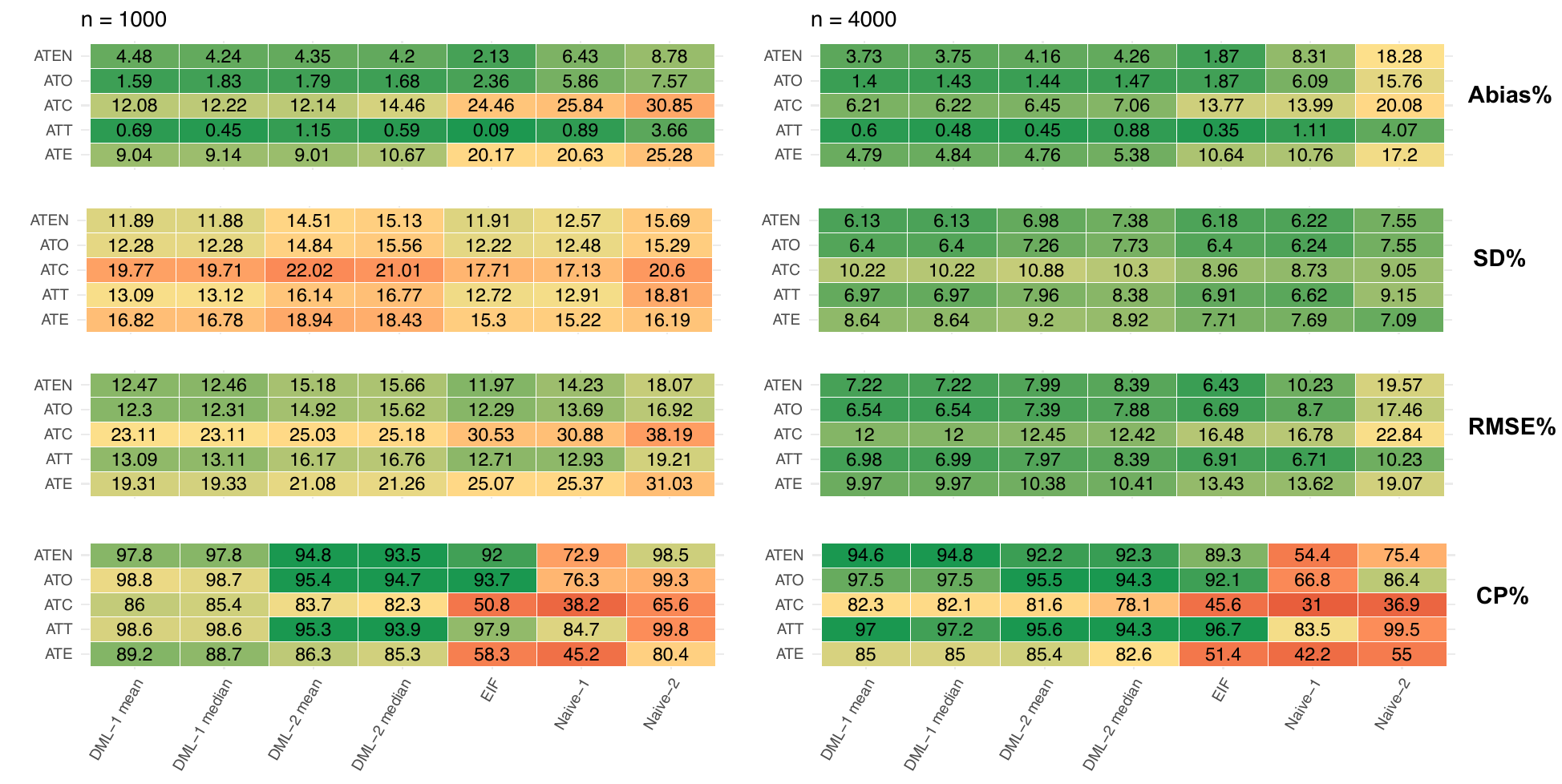}
    \caption{Simulation results by the proposed and na\"ive estimators under DGP 1 and methods ensemble for nuisance model specifications.}
    \label{fig:main-ensb}
\end{figure}

Next, we compare the different estimators. Overall, the three proposed estimators outperform the na\"ive estimators in terms of bias and RMSE. From the four ABias\% panels in Figures \ref{fig:main-GLM} and \ref{fig:main-ensb}, we can see that the largest ABias\% for the same estimand is always associated with one of the two na\"ive estimators. For example, in the leftmost panel (ABias\%, $n=1000$) of Figure \ref{fig:main-GLM}, the na\"ive-1 estimator exhibits the highest bias for ATEN and ATO, while the na\"ive-2 estimator shows the highest bias for ATC, ATT, and ATE. Similarly, RMSE results confirm that the na\"ive estimators consistently underperform, suggesting they are more sensitive to bias from nuisance models. Even with systematic model misspecification using a simple GLM (Figure \ref{fig:main-GLM}), the proposed EIF and DML-based methods effectively correct for bias and are more robust. Additionally, comparing DML-1 and DML-2 shows that DML-1 generally performs better, with lower bias and RMSE across all cases. Additionally, the DML-based estimators offer better CP\% compared to other methods. In both model specifications, the two na\"ive estimators show poor CP\% that significantly deviate from the nominal 95\% level. {This poor performance of the na\"ive estimators can be attributed to at least two factors: (i) as discussed in Remark \ref{rmk:CP}, the variance estimators applied to these methods are expected to be invalid; we include them here to demonstrate evidence of their invalidity and to highlight the advantage of our proposed estimator under the same variance estimation approach; and (ii) the poor CP\% values are also driven by the biases of the na\"ive estimators, which shift the estimates away from the truth.

We further explore the differences between the two model specification approaches, each with its advantages. First, using a simple GLM yields better CP\% for ATE and ATC in this DGP, likely due to some larger finite-sample instability in the methods ensemble, even with DML-based estimators. We suspect this is caused by extreme values of the estimated propensity scores from the methods ensemble. In such cases, propensity score weights for ATC and ATE can become large, particularly when estimated $e(\mb X)$ values are close zero in the treated group ($A=1$). This is supported by Figure \ref{fig:main-PS}, which shows the methods ensemble producing more observations with low propensity scores in the treated group, as indicated by the leftmost five coral color bars. This implies that some observations in the treated group can have large weights ($1/e(\mb X)$ for ATE and $\{1-e(\mb X)\}/e(\mb X)$ for ATC). A similar performance of the CP\% for the ATE has also been observed in a recent study by Tan et al. \cite{tan2025double}. Second, for other estimands, the methods ensemble produces similar biases and CP\%'s to the simple GLM but with smaller SDs and RMSEs, suggesting that when regularity conditions are met, flexible nonparametric modeling of nuisance functions can improve efficiency in WATE estimation.

Furthermore, we comment on the trends in the estimated variance across the two sample sizes. The SD\% panels of Figures \ref{fig:main-GLM} and \ref{fig:main-ensb} show that the SD for $n=4000$ is roughly half that for $n=1000$, indicating that the variance for $n=4000$ is about one-fourth of that for $n=1000$. This suggests that the efficiency of all assessed methods improves, with variance decreasing proportionally as the sample size increases. 

At the same time, although we did not present results for other sample sizes ($n=100$ and $400$), DGPs 2--5, and for the beta estimands (ATB1--ATB5), we briefly summarize the findings from these settings below. Overall, {under the two larger sample sizes $n=1000$ and $4000$}, results for the main estimands (ATE, ATT, ATC, ATO, ATEN) are consistent with those from DGP 1, and the key conclusions remain as aforementioned. Moreover, under the two smaller sample sizes and the main estimands, our methods continue to perform well in terms of ARBias\%, with overall smaller biases compared to the two na\"ive estimators. However, DML-1-mean and DML-2-mean often exhibit higher RMSEs when $n=100$; we suspect this is due to increased finite-sample instability from certain sample-splitting replications, as the nuisance functions are trained on small folds ($n_{\text{training}} = 80\% \times 100 = 80$). Nevertheless, DML-1-median and DML-2-median remain robust overall. When the sample size increases to $n=400$, the results from all DML methods become more stable across all main estimands and align closely with those observed at the larger sample sizes. Furthermore, under the two smaller sample sizes, comparing the two model specifications for nuisance functions, we find that the simple GLM generally results in higher RMSEs and biases---particularly when $n=100$ and for DML-1-mean and DML-2-mean estimators of ATE, ATT, and ATC.

{Next, comparing across the four sample sizes, we observe that the RMSEs of the proposed methods decrease as $n$ increases, confirming their rate convergence property. In addition, regarding CP\%, DML-1-median and DML-2-median perform best even at the smallest sample size $n=100$, providing the most accurate inferential results. When $n \geq 400$, all DML-based estimators perform well, except for some under-coverage cases for ATE and ATC under DGPs 1 and 3. However, as noted earlier, this is likely due to extreme propensity scores} (see figures in \textit{Online Supplemental Material} [Section C.1]). 

Finally, for beta weights, DML estimators using the mean procedure often exhibit higher variances {under all sample sizes and nuisance function models}, particularly with large $\nu_1$ or $\nu_2$ that makes cross-fitting unstable. Although these beta weights might be less frequently used in practice, we include them for investigating trends by different parameter setting (weight functions) in our simulation. However, we exclude them from the case studies in Section \ref{sec:data} based on the findings here. 

\section{Case Studies}\label{sec:data}

We conduct two case studies to illustrate our proposed methods in Sections \ref{sec:EIF} and \ref{sec:DML}, using data from the 2007–2008 U.S. National Health and Nutrition Examination Survey (NHANES) and the 1991 Survey of Income and Program Participation (SIPP). The first case estimates the causal effects of smoking on blood lead levels \citep{hsu2013calibrating, yang2018asymptotic}, while the second extends the work of  Chernozhukov et al. \cite{chernozhukov2018double}. By estimating WATEs in addition to ATE, we can gain more insights about the effect of the interested treatment.  

In both studies, we apply EIF, DML-1, DML-2, and the two na\"ive estimators defined in \eqref{eq:naive-est}, evaluating ATE, ATT, ATC, ATO, and ATEN estimands, with point and standard error estimates reported. Following  Chernozhukov et al. \cite{chernozhukov2018double}, we use 5-fold cross-fitting over 100 sample splits for DML-1 and DML-2, summarizing results using both mean and median strategies. For all estimators, we employ an ensemble of \texttt{SL.glm}, \texttt{SL.glm.interaction}, \texttt{SL.glmnet}, and \texttt{SL.ranger} from the \texttt{SuperLearner} R package \citep{van2007super} to predict nuisance functions. Sections \ref{subsec:smoke} and \ref{subsec:401k} provide background on the two datasets, and Table \ref{tab:datares} in Section \ref{subsec:datares} presents the analysis results.

\subsection{Study I: Effects of smoking on blood lead level}\label{subsec:smoke}

This study uses data from the 2007–2008 NHANES to evaluate the effects of smoking on blood lead levels. The dataset, available in the supporting information of Hsu and Small \cite{hsu2013calibrating}, includes 3,340 participants: 679 smokers ($A=1$) and 2,661 nonsmokers ($A=0$). The outcome $Y$ is blood lead level measured in $\mu$g/dl, with age, income, race, education, and gender included as covariates $\mb X$.

\subsection{Study II: Eligibility for participation in the 401(k) plan}\label{subsec:401k}

The second study uses data from the 1991 SIPP, available in the \texttt{DoubleML} R package, consisting of 9,915 household-level samples. The outcome $Y$ is net financial assets, and the treatment $A$ is eligibility for enrolling in a 401(k) plan. The covariates $\mb X$ include age, income, education, family size, marital status, two-earner status, pension benefits, IRA participation, and homeownership. To address the lack of random assignment, we follow the widely accepted assumption that 401(k) eligibility can be treated as exogenous after conditioning on income and related job factors, as argued by Poterba et al. \cite{poterba1992401, poterba1995Do}. Their analysis used linear regressions with limited covariates, raising concerns about insufficient control for income and the limited use of ML methods.  Chernozhukov et al. \cite{chernozhukov2018double} revisited the data and estimated ATE using the DML method with advanced ML tools. To build on this, we estimate additional WATEs using our proposed methods to further explore insights from the data. 

\subsection{Data analysis results}\label{subsec:datares}

We summarize the results of our case studies in {Figure \ref{fig:studies}} and Table \ref{tab:datares}, followed by comments and discussion. {Figure \ref{fig:studies} displays the estimated propensity score distributions by treatment group and the covariate balance assessments based on propensity score weights for all estimands considered across both studies. The results in Figure \ref{fig:studies} are based on the estimated propensity scores obtained from the ensemble learning model described earlier, using the entire sample. The estimated propensity score distributions for both studies exhibit moderate to high overlap. At the same time, in both studies, the weights corresponding to ATT, ATO, and ATEN achieve better covariate balance, with most ASMDs falling below 0.1. In contrast, for some covariates, the ASMDs under ATE and ATC exceed 0.1. However, the weights for all estimands lead to substantially better covariate balance compared to the unweighted mean differences between the two treatment groups in both studies.} 

\begin{figure}[H]
    \centering
    \includegraphics[width=\linewidth]{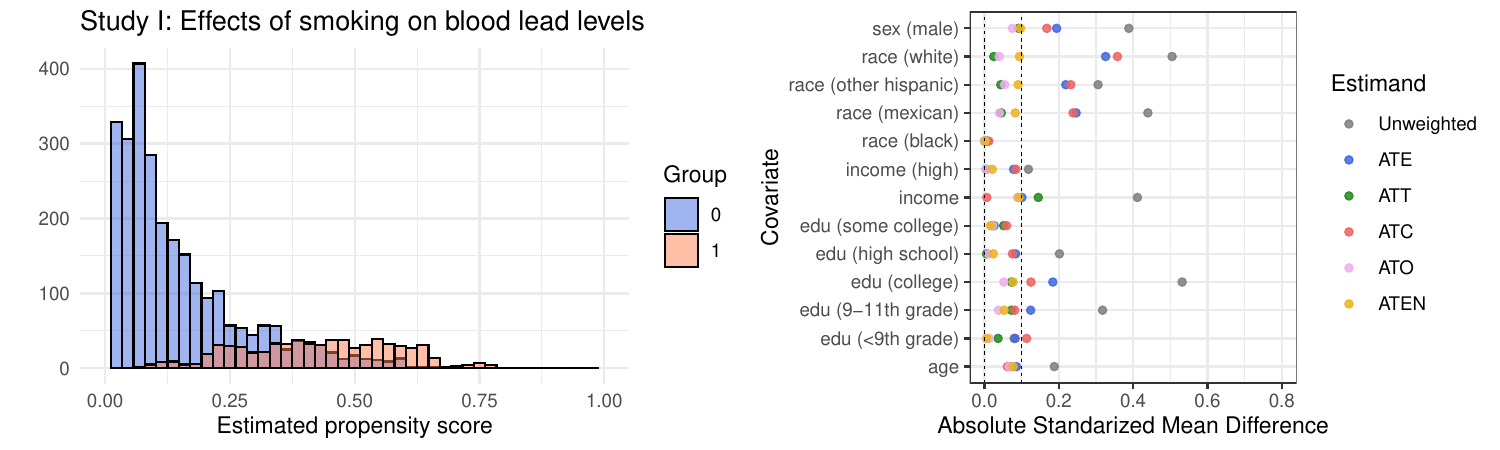}
    \includegraphics[width=\linewidth]{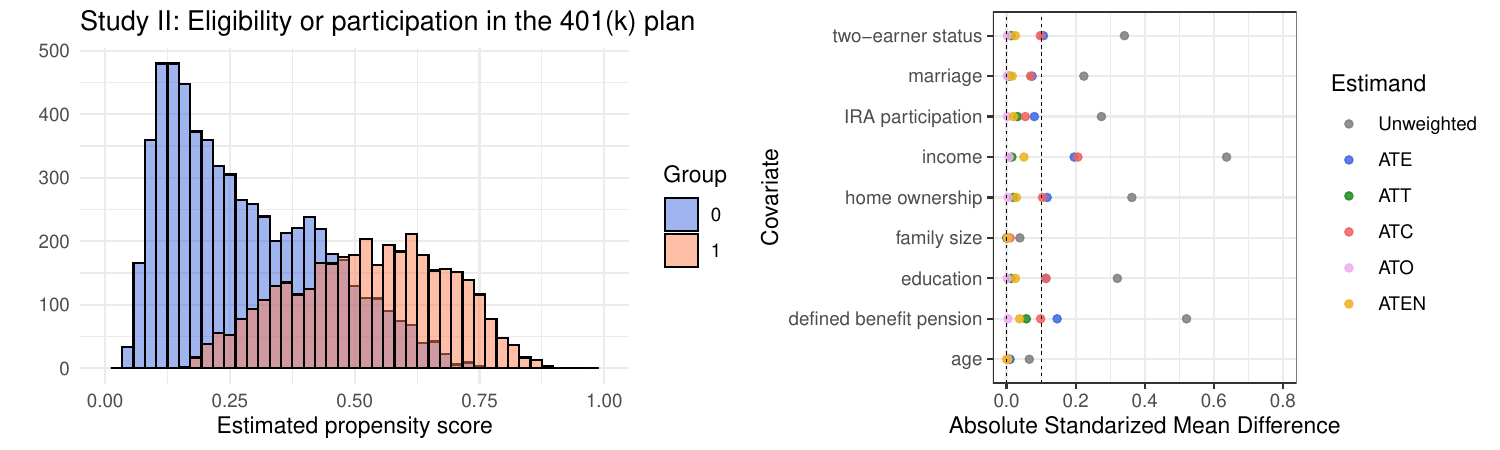}
    \caption{Estimated propensity score distributions and covariate balance in the two case studies.} 
    \label{fig:studies}
\end{figure}

\begin{table}[H]
    \centering
    \small
    \caption{Point estimates (standard error estimates) for WATEs in the two case studies.}\label{tab:datares}
    \begin{tabular}{rcccccccccccccc} 
        \toprule
        & \multicolumn{2}{c}{DML-1} & \multicolumn{2}{c}{DML-2} & \multirow{2}{*}{EIF} & \multirow{2}{*}{Na\"ive-1} & \multirow{2}{*}{Na\"ive-2} \\ 
        \cmidrule(r){2-3} \cmidrule(r){4-5}
         & Mean & Median & Mean & Median &  &  &  \\
        \midrule
        & \multicolumn{7}{c}{\textbf{Study I: Effects of smoking on blood lead levels ($\mu$g/dl)}} \\
        \addlinespace
        ATE  & 0.80 (0.10)  & 0.80 (0.10) & 0.78 (0.10) & 0.77 (0.10) & 0.85 (0.05) & 0.86 (0.01) & $-0.19$ (0.08) \\  
        ATT  & 0.76 (0.08)  & 0.76 (0.08) & 0.74 (0.10) & 0.73 (0.09) & 0.77 (0.08) & 0.80 (0.01) & 1.11 (0.13) \\
        ATC  & 0.81 (0.11)  & 0.81 (0.10) & 0.79 (0.10) & 0.77 (0.10) & 0.87 (0.04) & 0.88 (0.01) & $-0.51$ (0.08) \\
        ATO  & 0.79 (0.09)  & 0.79 (0.09) & 0.77 (0.09) & 0.73 (0.09) & 0.76 (0.08) & 0.82 (0.01) & 0.69 (0.12) \\ 
        ATEN & 0.79 (0.09)  & 0.80 (0.09) & 0.77 (0.09) & 0.75 (0.09) & 0.77 (0.07) & 0.83 (0.01) & 0.50 (0.11) \\ 
        \midrule
        & \multicolumn{7}{c}{\textbf{Study II: Financial assets difference by eligibility to enroll in the 401(k) plan (\$)}} \\
        \addlinespace
        ATE  & 7269 (1222)  & 7672 (1219) & 7105 (1226) & 6918 (1214) & 7690 (911) & 7759 (161)  & 6310 (1278) \\ 
        ATT  & 9345 (2003)  & 9364 (1994) & 9354 (2000) & 9431 (1945) & 9833 (1399) & 10199 (273) & 14046 (2041) \\ 
        ATC  & 6043 (893)   & 6029 (893)  & 5774 (905)  & 5745 (894)  & 6423 (707)  & 6423 (115)  & 1732 (909) \\ 
        ATO  & 8557 (1303)  & 8556 (1305) & 8392 (1303) & 8102 (1322) & 8696 (1062) & 8412 (176)  & 6937 (1476) \\ 
        ATEN & 8274 (1283)  & 8280 (1285) & 8110 (1285) & 7846 (1290) & 8501 (1023) & 8269 (172)  & 6805 (1429) \\ 
        \bottomrule
    \end{tabular}
\end{table}

{Furthermore, based on the results in Table \ref{tab:datares}, } we provide our interpretations on the WATE estimates, using ATE and its DML-1 (mean) estimates as example. In Study I, the ATE is 0.80 $\mu$g/dl with a standard error of 0.10 $\mu$g/dl, yielding a p-value of $<0.001$ (using normal approximation), indicating that smokers have, on average, 0.80 $\mu$g/dl higher blood lead levels than non-smokers. In Study II, the ATE is \$7,269 with a standard error of \$1,222, suggesting that those eligible for 401(k) enrollment have \$7,269 more in net financial assets (p-value $<0.001$). Other WATEs can be similarly interpreted for their respective target populations.

Our findings offer valuable insights for public health and policy decision-making. In Study I, the results suggest that across the overall (ATE), treated (ATT), control (ATC) populations, and populations where participants with propensity scores closer to 0.5 are more heavily weighted (ATO and ATEN), the estimated effects of smoking are consistently similar (approximately 0.8). This indicates that the impact of smoking on blood lead levels appears nearly constant across participants {from different target populations. In other words, smoking seems to exert a uniform marginal effect on these populations, and thus it suggests limited evidence of strong effect modification by observed covariates. However, we clarify that this statement is not intended to imply treatment effect homogeneity (nor would it be sufficient to do so), but rather reflects an intuition based on comparing results across different target populations. }

{Another remark for Study I is that, based on the results for ATE and ATC, the EIF and the two DML estimators yield notably different point estimates (e.g., for ATE, the EIF estimate is 0.85, whereas the DML estimates range from 0.77 to 0.80, with a difference of approximately one standard deviation). We suspect that this discrepancy is driven by finite-sample instability, particularly due to relatively large propensity score weights in ATE and ATC estimation, as illustrated in Figure \ref{fig:studies}. }


In Study II, a different pattern emerges. The estimated average effects on different populations show notable differences, although all effects are statistically significant at the 0.05 level. Specifically, the ATT, ATO, and ATEN estimands yield larger estimated values. This suggests that the treated population (participants eligible for enrolling in a 401(k) plan) and those with propensity scores closer to 0.5 might not need to be prioritized in policies aimed at improving financial assets. These findings highlight the importance of tailoring policy interventions to specific subpopulations based on their estimated effects.

{Moreover, }comparing results across methods reveals that: (i) DML estimates using the mean or median are nearly identical, suggesting that sensitivity to outliers is small in these data; (ii) Standard errors for DML-1 and DML-2 are generally larger than those for EIF, which is consistent with simulation findings in Section \ref{subsec:simures}; (iii) Na\"ive-2 estimators often yield substantially different point estimates, suggesting potential {biases} (e.g., in Study I, ATC is $-0.51$ $\mu$g/dl for na\"ive-2 but around $0.75$--$0.90$ $\mu$g/dl for other methods; in Study II, ATT is \$14,046 for naïve-2 compared to \$9,300--\$10,200 for other methods); and (iv) The na\"ive-1 estimator consistently produces smaller standard errors, aligning with simulation results that showed its variance estimates are highly anti-conservative.

Finally, since both Studies I and II have been previously analyzed in the literature \cite{yang2018asymptotic, chernozhukov2018double}, we briefly compare our findings with theirs. For Study I, Yang and Ding \cite{yang2018asymptotic} analyzed a trimmed sample by including participants with estimated propensity scores in $(0.05, 0.6)$, corresponding to the relatively more overlapped region between treatment groups. They considered IPW and AIPW estimators, using logistic regression with linear predictors for propensity score estimation. In terms of our framework, their analysis corresponds to a WATE defined by the weight function $\lambda\{e(\mb X)\} = I(0.05 < e(\mb X) < 0.6)$, using the naïve-2 (IPW) and EIF-based (AIPW) estimators. They reported point (standard error) estimates of 0.65 (0.14) and 0.77 (0.11), respectively. Notably, their EIF-based estimate is quite close to ours despite the difference in estimands, suggesting that the effect of smoking on different target populations is plausibly similar. In addition, for Study II, Chernozhukov et al. \cite{chernozhukov2018double} applied the DML-1 and DML-2 estimators for the ATE, reporting multiple results under different nuisance models (see their Table 2), with point estimates ranging from 6800 to 9300. Our DML-based estimates also fall within this range. The standard error estimates are similarly comparable, ranging from 1200 to 1360. The variation in point estimates suggests that some nuisance models may be misspecified or inconsistent; however, all analyses consistently indicate that the effect of 401(k) eligibility on net financial assets is statistically significant at the 0.05 level. 

\section{Concluding Remarks}\label{sec:conc}

In this paper, by leveraging the semiparametric EIF and DML tools, we propose three advanced estimators for WATEs in populations defined by propensity score weights. In this section, we provide a summary, practical recommendations, and discussions with remarks regarding our proposed EIF and DML-based estimators. 

\subsection{Summary}\label{subsec:summary}

We highlight the unique contributions of our work as follows. First, we extend the EIF theory from the conventional ATE estimand to the more general WATE class of estimands. Specifically, we derive the EIF for a general WATE, which leads to the construction of an EIF-based estimator with the RDR property. This development represents an important extension of the AIPW estimator (also known as the doubly robust estimator \cite{kang2007demystifying, bickel1993efficient}) for the conventional ATE to a broader class of causal estimands, as the AIPW estimator was also based on the EIF of ATE.

Second, we propose two DML-based estimators, which extend the cross-fitting and sample splitting strategies \cite{chernozhukov2018double, jacob2020cross} from ATE to WATE estimations for minimizing the prediction error. These strategies also enable the construction of simpler yet valid inference for WATE. The mean square of the EIFs can serve as a variance estimator, which our simulations have shown to be effective for uncertainty quantification {when using the DML-1 and DML-2 estimators}. This variance estimator is model-free, easy to compute, and does not require significantly more computational time {compared to the standard bootstrap approach}. In contrast, the traditional ``estimate and plug-in'' approach for nuisance functions (used in the EIF-based and na\"ive estimators) often requires computationally expensive bootstrapping or model-dependent sandwich variance estimators for valid inference \citep{matsouaka2023variance, sengupta2016subsampled, lunceford2004stratification}. Our proposed DML-based methods bypass these complexities, offering a more direct and efficient solution. 

Furthermore, Theorems \ref{thm:RDR-EIF} and \ref{thm:RDR-DML} show that our estimators possess the RDR property {under regularity conditions we find for both our EIF- and DML-based estimators. } Achieving RDR motivates the use of high-accuracy models like the \texttt{SuperLearner} R package \citep{van2007super}. {Therefore, we }also provide empirical assessments of the proposed methods using various nuisance model fitting strategies that practical researchers might adopt. 

Finally, a user-friendly R package \texttt{WATE} implementing our proposed methods is available at \url{https://github.com/yiliu1998/WATE} with usage instruction. In the next section, we provide some practical recommendations based on our theoretical investigations and empirical evidence from Section \ref{sec:simu}. 

\subsection{Practical recommendation}\label{subsec:pract}

To minimize bias and ensure consistency in WATE estimation, we recommend using our proposed estimators {motivated from the EIF theory} in applied research and case studies. The results of our empirical studies in Sections \ref{sec:simu} and \ref{sec:data} show that our methods (EIF, DML-1 and DML-2) are overall less biased and more robust than singly robust estimators (na\"ive-1 and na\"ive-2), while allowing flexible choices of parametric, nonparametric, or ensemble models for nuisance function modeling. Moreover, the DML-based estimators outperform the estimate-then-plug-in EIF approach overall in terms of the robustness and inference.  Thus, for researchers seeking more reliable point estimators for WATEs, we recommend prioritizing our DML-based methods over other approaches, especially in settings with complex covariate structures or potential model misspecification. 

However, we also caution against using DML-based methods when the sample size is too small, such as $n=100$ in our simulation. In this setting, the RMSEs of all EIF and DML-based estimators are not better than those of the two na\"ive methods due to larger finite-sample instability in their estimations, even if they have smaller biases. Therefore, we recommend that researchers first carefully select appropriate models for the propensity score and/or outcomes, and consider using one of the naïve estimators along with nonparametric bootstrap for variance estimation. Based on our simulation, a practical rule of thumb is that DML-based methods become reliable and robust when $n \geq 400$. We also encourage users to compare different approaches for fitting nuisance functions, whether through simple GLM models, single machine learning models, or methods ensemble, which can be easily implemented using our \texttt{WATE} R package.  

Furthermore, among the EIF, DML-1, and DML-2 estimators we proposed, the two DML-based methods show greater advantages for valid variance estimation due to their use of sample-splitting and cross-fitting. Compared to the traditional ``estimate-then-plug-in'' approach used in the EIF-based estimator, cross-fitting effectively mitigates the uncertainty in nuisance parameter estimation by separating model fitting and prediction across different data folds. This enables a simple variance estimator based on the estimated mean square of the EIFs to provide valid asymptotic variance approximation, as evidenced by our extensive simulations in Section \ref{sec:simu}. Consequently, it avoids reliance on computationally intensive algorithms such as nonparametric bootstrap or model-dependent (and less flexible) sandwich variance estimators \cite{lunceford2004stratification, matsouaka2023variance}. Therefore, for researchers seeking both robust estimation and valid uncertainty quantification, the DML-1 and DML-2 estimators are more strongly recommended. In contrast, for the EIF-based estimator, our results suggest that the same variance estimator tends to under-cover, indicating that developing more computationally efficient variance estimation methods beyond standard nonparametric bootstrap remains an important direction for future research.  

Finally, we caution against using beta family weights with large $\nu_1$ and $\nu_2$ values. As detailed in Section \ref{subsec:simures}, these beta weights are highly sensitive to model misspecifications and disproportionately target those with propensity scores near 0.5 \citep{matsouaka2024causal}. They also require selecting $(\nu_1,\nu_2)$ parameters and do not offer clear advantages over overlap or entropy weights. While they serve as candidate methods for a more comprehensive empirical comparison in this paper, their practical utility remains uncertain.

\subsection{Discussion}\label{subsec:discuss}

We first provide a remark on the \texttt{SuperLearner} toolbox \cite{van2007super}, which we employed in both our empirical studies and R package. Our paper does not claim that using \texttt{SuperLearner} guarantees the estimated nuisance models achieve the desired convergence rates. Rather, the goal and main theoretical contribution of the paper are to complement existing theory by specifying the regularity conditions required for RDR estimators. Although achieving correct model specification or ensuring convergence rates to the true functions is generally infeasible—since these properties are not testable or are difficult to verify in practice—\texttt{SuperLearner} offers a flexible and practical tool, particularly for complex data-generating processes such as those considered in our simulations. 

At the same time, we acknowledge several limitations of our method and propose potential extensions for future research. We have not fully addressed the issue of non-differentiability in the weight function $\lambda\{e(\mb X)\}$, particularly in cases of propensity score trimming \citep{crump2006moving, crump2009dealing}, truncation \citep{ju2019adaptive}, and matching weights \citep{li2013weighting}. {A possible solution is to approximate their weight functions $\lambda\{e(\mb X)\}$ with continuous, differentiable alternatives that are arbitrarily close to $\lambda\{e(\mb X)\}$. Such functions exist mathematically, ensuring that any discrepancy has a negligible impact on estimation. Yang and Ding \cite{yang2018asymptotic} addressed this in the context where $\lambda\{e(\mb X)\}$ is the indicator function used for propensity score trimming, proposing a smoothed weight function that incorporates an additional smoothing parameter. This approach introduces a bias-variance trade-off, highlighting the need for further empirical investigation to inform the selection of both trimming and smoothing parameters. } For matching weights and truncation, where $\lambda\{e(\mb X)\}$ is non-differentiable only at zero-measure points and remains continuous, we can consider smoothing the neighborhood of non-differentiable points.  

Moreover, our proposed variance estimators for EIF and DML-based methods may still show slight inconsistency in a few cases, as indicated by simulation results in Section \ref{subsec:simures}. Developing more robust variance estimators within this framework is a promising area for future research. Furthermore, as we have observed in both our simulation and data analysis, a few relatively large or extreme propensity score weights can still impact the estimation of certain members of the WATE class using our proposed estimators. Therefore, future research should explore techniques for targeting these estimands while addressing the instability caused by extreme weights, such as through weight truncation \cite{sturmer2010treatment}. Finally, we have not addressed assumption violations using sensitivity analysis, such as unmeasured confounders \citep{jacob2020cross, vanderweele2017sensitivity, fewell2007impact, lin1998assessing, faries2023real}.

Furthermore, our proposed methods have certain connections with targeted maximum likelihood estimation (TMLE) \cite{van2006targeted, schuler2017targeted}. Both our estimators and TMLE aim to achieve higher robustness and asymptotic efficiency by leveraging the EIF. TMLE (when available, e.g., for the ATE) and EIF-based estimators are asymptotically equivalent; however, TMLE adopts a targeted updating step to improve finite-sample performance, such as preserving the natural parameter range and enhancing stability. In contrast, our EIF- and DML-based estimators directly construct estimate-then-plug-in or cross-fitted estimators without an additional targeting step. Moreover, to our knowledge, existing TMLE methods have been primarily developed for ATE, ATT, and ATC, while our framework accommodates a broader class of estimands---the WATE. Extending TMLE or TMLE-type algorithms to general WATEs remains an interesting direction for future research.

Finally, other future research directions that could extend our framework include, but are not limited to, exploring data thinning as an alternative approach \citep{neufeld2024data}, extending our methods to covariate adjustment in randomized clinical trials \citep{liu2025coadvise, gao2024does}, multi-valued treatments  \citep{li2019propensity, yang2016propensity, liu2025assessing}, the weighted ATT (WATT) class of estimands \citep{liu2024average}, multi-source data \citep{yang2023elastic, gao2023integrating, han2025federated, liu2025targeted, wang2025integrating, zhuang2025assessment}, conformal inference \citep{yang2024doubly, liu2024multi}, instrumental variables \citep{levis2024nonparametric}, survival outcomes \citep{zeng2021propensity2, cao2024using}, {and estimating external population ATEs by re-weighting covariates to match the target population \citep{colnet2024causal, lee2022doubly, lee2024transporting, lee2023improving, lee2024genrct, yang2022rwd}.}

\subsection*{Acknowledgments}
The authors thank Peng Ding and Roland A. Matsouaka for their invaluable discussion and insights in this research endeavor. 

\subsection*{Funding Information}
Y.L. is supported by the National Heart, Lung, and Blood Institute (NHLBI) of the National Institutes of Health (NIH) under Award Number T32HL079896. The content is solely the responsibility of the authors and does not necessarily represent the official views of the NIH. S.Y. is partially supported by NSF (National Science Foundation) grant SES 2242776 and NIH-NIA (National Institute on Aging) grant 1R01AG066883.

\subsection*{Conflict of Interest}

Authors state no conflict of interest.

\subsection*{Author Contribution}

All authors accept responsibility for the content of this manuscript and consent to its submission. All authors reviewed the results and approved the final version. All contributed to the conceptualization and methodology. Y.W. and Y.L. designed the experiments, developed the code, and performed the simulations and data analysis. Y.L. drafted the manuscript with input from all co-authors. 

\subsection*{Data Availability Statement}

Data used in Section \ref{sec:data} are publicly available. The NHANES dataset in Section \ref{subsec:smoke} is available in the supporting information of Hsu and Small \cite{hsu2013calibrating} (\url{https://onlinelibrary.wiley.com/action/downloadSupplement?doi=10.1111%2Fbiom.12101&file=biom12101-sm-0002-SupExampleData.txt}). The SIPP dataset in Section \ref{subsec:401k} is available in the \texttt{DoubleML} R package (\url{https://docs.doubleml.org/stable/examples/R_double_ml_pension.html}). 

\bibliographystyle{vancouver} 
\bibliography{_refs}

\begin{thebibliography}{10}

\bibitem{yang2021propensity}
Yang S, Lorenzi E, Papadogeorgou G, Wojdyla DM, Li F, Thomas LE.
\newblock Propensity score weighting for causal subgroup analysis.
\newblock Statistics in Medicine. 2021;40(19):4294-309.

\bibitem{petersen2012diagnosing}
Petersen ML, Porter KE, Gruber S, Wang Y, Van Der~Laan MJ.
\newblock Diagnosing and responding to violations in the positivity assumption.
\newblock Statistical Methods in Medical Research. 2012;21(1):31-54.

\bibitem{crump2006moving}
Crump RK, Hotz VJ, Imbens G, Mitnik O.
\newblock Moving the goalposts: Addressing limited overlap in the estimation of average treatment effects by changing the estimand.
\newblock NBER Technical Working Papers. 2006.

\bibitem{crump2009dealing}
Crump RK, Hotz VJ, Imbens GW, Mitnik OA.
\newblock Dealing with limited overlap in estimation of average treatment effects.
\newblock Biometrika. 2009;96(1):187-99.

\bibitem{li2018balancing}
Li F, Morgan KL, Zaslavsky AM.
\newblock Balancing covariates via propensity score weighting.
\newblock J Am Stat Assoc. 2018;113:390-400.

\bibitem{rizk2025and}
Rizk JG.
\newblock When and Why to Use Overlap Weighting: Clarifying Its Role, Assumptions, and Estimand in Real-World Studies.
\newblock Journal of Clinical Epidemiology. 2025:111942.

\bibitem{thomas2020overlap}
Thomas LE, Li F, Pencina MJ.
\newblock Overlap weighting: a propensity score method that mimics attributes of a randomized clinical trial.
\newblock Jama. 2020;323(23):2417-8.

\bibitem{tao2019doubly}
Tao Y, Fu H.
\newblock Doubly robust estimation of the weighted average treatment effect for a target population.
\newblock Statistics in Medicine. 2019;38:315-25.

\bibitem{parikh2025we}
Parikh H, Ross RK, Stuart E, Rudolph KE.
\newblock Who Are We Missing?: A Principled Approach to Characterizing the Underrepresented Population.
\newblock Journal of the American Statistical Association. 2025;(just-accepted):1-32.

\bibitem{hirano2003efficient}
Hirano K, Imbens GW, Ridder G.
\newblock Efficient estimation of average treatment effects using the estimated propensity score.
\newblock Econometrica. 2003;71(4):1161-89.

\bibitem{kang2007demystifying}
Kang JD, Schafer JL.
\newblock Demystifying double robustness: A comparison of alternative strategies for estimating a population mean from incomplete data.
\newblock Statistical Science. 2007;22(4):523  539.

\bibitem{glynn2010introduction}
Glynn AN, Quinn KM.
\newblock An introduction to the augmented inverse propensity weighted estimator.
\newblock Political analysis. 2010;18(1):36-56.

\bibitem{ogburn2015doubly}
Ogburn EL, Rotnitzky A, Robins JM.
\newblock Doubly robust estimation of the local average treatment effect curve.
\newblock Journal of the Royal Statistical Society Series B: Statistical Methodology. 2015;77(2):373-96.

\bibitem{robins1994estimation}
Robins JM, Rotnitzky A, Zhao LP.
\newblock Estimation of regression coefficients when some regressors are not always observed.
\newblock J Am Stat Assoc. 1994;89:846-66.

\bibitem{bang2005doubly}
Bang H, Robins JM.
\newblock Doubly robust estimation in missing data and causal inference models.
\newblock Biometrics. 2005;61(4):962-73.

\bibitem{ying2024geometric}
Ying A.
\newblock A Geometric Perspective on Double Robustness by Semiparametric Theory and Information Geometry.
\newblock arXiv preprint arXiv:240413960. 2024.

\bibitem{mao2019propensity}
Mao H, Li L, Greene T.
\newblock Propensity score weighting analysis and treatment effect discovery.
\newblock Statistical Methods in Medical Research. 2019;28:2439-54.

\bibitem{matsouaka2024causal}
Matsouaka RA, Zhou Y.
\newblock Causal inference in the absence of positivity: The role of overlap weights.
\newblock Biometrical Journal. 2024;66(4):2300156.

\bibitem{yang2018asymptotic}
Yang S, Ding P.
\newblock Asymptotic inference of causal effects with observational studies trimmed by the estimated propensity scores.
\newblock Biometrika. 2018;105:487-93.

\bibitem{sturmer2010treatment}
St{\"u}rmer T, Rothman KJ, Avorn J, Glynn RJ.
\newblock Treatment effects in the presence of unmeasured confounding: dealing with observations in the tails of the propensity score distribution—a simulation study.
\newblock American Journal of Epidemiology. 2010;172(7):843-54.

\bibitem{ma2020robust}
Ma X, Wang J.
\newblock Robust inference using inverse probability weighting.
\newblock Journal of the American Statistical Association. 2020;115(532):1851-60.

\bibitem{chaudhuri2014heavy}
Chaudhuri S, Hill JB.
\newblock Heavy tail robust estimation and inference for average treatment effects.
\newblock Working paper. 2014.

\bibitem{li2013weighting}
Li L, Greene T.
\newblock A weighting analogue to pair matching in propensity score analysis.
\newblock The International Journal of Biostatistics. 2013;9:215-34.

\bibitem{moodie2018doubly}
Moodie EE, Saarela O, Stephens DA.
\newblock A doubly robust weighting estimator of the average treatment effect on the treated.
\newblock Stat. 2018;7(1):e205.

\bibitem{matsouaka2024overlap}
Matsouaka RA, Liu Y, Zhou Y.
\newblock Overlap, matching, or entropy weights: what are we weighting for?
\newblock Communications in Statistics-Simulation and Computation. 2024:1-20.

\bibitem{zhou2020propensity}
Zhou Y, Matsouaka RA, Thomas L.
\newblock Propensity score weighting under limited overlap and model misspecification.
\newblock Statistical Methods in Medical Research. 2020;29(12):3721-56.

\bibitem{li2021propensity}
Li Y, Li L.
\newblock Propensity score analysis methods with balancing constraints: A Monte Carlo study.
\newblock Statistical Methods in Medical Research. 2021;30(4):1119-42.

\bibitem{tsiatis2006semiparametric}
Tsiatis AA.
\newblock Semiparametric theory and missing data.
\newblock Springer; 2006.

\bibitem{chernozhukov2018double}
Chernozhukov V, Chetverikov D, Demirer M, Duflo E, Hansen C, Newey W, et~al.
\newblock Double/debiased machine learning for treatment and structural parameters.
\newblock The Econometrics Journal. 2018;21:1-68.

\bibitem{neyman1923applications}
Neyman J.
\newblock Sur les applications de la th{\'e}orie des probabilit{\'e}s aux experiences agricoles: Essai des principes.
\newblock Roczniki Nauk Rolniczych. 1923;10:1-51.

\bibitem{rubin1974estimating}
Rubin DB.
\newblock Estimating causal effects of treatments in randomized and nonrandomized studies.
\newblock Journal of educational Psychology. 1974;66(5):688.

\bibitem{rubin1980comment}
Rubin DB.
\newblock {Comment on "Randomization analysis of experimental data: The Fisher randomization test" by D. Basu}.
\newblock J Am Stat Assoc. 1980;75:591-3.

\bibitem{rosenbaum1983central}
Rosenbaum PR, Rubin DB.
\newblock The Central Role of the Propensity Score in Observational Studies for Causal Effects.
\newblock Biometrika. 1983a;70:41-55.

\bibitem{austin2023differences}
Austin PC.
\newblock Differences in target estimands between different propensity score-based weights.
\newblock Pharmacoepidemiology and Drug Safety. 2023;32(10):1103-12.

\bibitem{barnard2024unified}
Barnard M, Huling JD, Wolfson J.
\newblock A Unified Framework for Causal Estimand Selection.
\newblock arXiv preprint arXiv:241012093. 2024.

\bibitem{hahn1998role}
Hahn J.
\newblock On the role of the propensity score in efficient semiparametric estimation of average treatment effects.
\newblock Econometrica. 1998:315-31.

\bibitem{kennedy2016semiparametric}
Kennedy EH.
\newblock Semiparametric theory and empirical processes in causal inference.
\newblock In: Statistical Causal Inferences and Their Applications in Public Health Research. Springer; 2016. p. 141-67.

\bibitem{hines2022demystifying}
Hines O, Dukes O, Diaz-Ordaz K, Vansteelandt S.
\newblock Demystifying statistical learning based on efficient influence functions.
\newblock The American Statistician. 2022;76(3):292-304.

\bibitem{bickel1993efficient}
Bickel PJ, Klaassen C, Ritov Y, Wellner J.
\newblock {Efficient and Adaptive Inference in Semiparametric Models}.
\newblock Johns Hopkins University Press, Baltimore; 1993.

\bibitem{farrell2015robust}
Farrell MH.
\newblock Robust inference on average treatment effects with possibly more covariates than observations.
\newblock Journal of Econometrics. 2015;189:1-23.

\bibitem{van1996weak}
van~der Vaart AW, Wellner JA.
\newblock Weak Convergence and Emprical Processes: With Applications to Statistics.
\newblock New York: Springer; 1996.

\bibitem{frisch1933partial}
Frisch R, Waugh FV.
\newblock Partial time regressions as compared with individual trends.
\newblock Econometrica: Journal of the Econometric Society. 1933:387-401.

\bibitem{lovell1963seasonal}
Lovell MC.
\newblock Seasonal adjustment of economic time series and multiple regression analysis.
\newblock Journal of the American Statistical Association. 1963;58(304):993-1010.

\bibitem{van2007super}
Van~der Laan MJ, Polley EC, Hubbard AE.
\newblock Super learner.
\newblock Statistical applications in genetics and molecular biology. 2007;6(1).

\bibitem{polley2011super}
Polley EC, Rose S, Van~der Laan MJ.
\newblock Super learning.
\newblock Targeted learning. 2011:43-66.

\bibitem{frank2015regression}
Frank EH.
\newblock Regression modeling strategies with applications to linear models, logistic and ordinal regression, and survival analysis.
\newblock Spinger; 2015.

\bibitem{shmueli2010explain}
Shmueli G.
\newblock To explain or to predict?
\newblock Statistical Sciences. 2010.

\bibitem{daniel2015causal}
Daniel R, De~Stavola B, Cousens S, Vansteelandt S.
\newblock Causal mediation analysis with multiple mediators.
\newblock Biometrics. 2015;71(1):1-14.

\bibitem{lunceford2004stratification}
Lunceford JK, Davidian M.
\newblock Stratification and weighting via the propensity score in estimation of causal treatment effects: a comparative study.
\newblock Statistics in Medicine. 2004;23(19):2937-60.

\bibitem{williamson2012variance}
Williamson EJ, Morley R, Lucas A, Carpenter J.
\newblock Variance estimation for stratified propensity score estimators.
\newblock Statistics in Medicine. 2012;31(15):1617-32.

\bibitem{zou2016variance}
Zou B, Zou F, Shuster JJ, Tighe PJ, Koch GG, Zhou H.
\newblock On variance estimate for covariate adjustment by propensity score analysis.
\newblock Statistics in Medicine. 2016;35(20):3537-48.

\bibitem{mao2018propensity}
Mao H, Li L, Yang W, Shen Y.
\newblock On the propensity score weighting analysis with survival outcome: Estimands, estimation, and inference.
\newblock Statistics in Medicine. 2018;37(26):3745-63.

\bibitem{austin2022bootstrap}
Austin PC.
\newblock Bootstrap vs asymptotic variance estimation when using propensity score weighting with continuous and binary outcomes.
\newblock Statistics in Medicine. 2022;41(22):4426-43.

\bibitem{matsouaka2023variance}
Matsouaka RA, Liu Y, Zhou Y.
\newblock Variance estimation for the average treatment effects on the treated and on the controls.
\newblock Statistical Methods in Medical Research. 2023;32(2):389-403.

\bibitem{li2025variance}
Li H, Liu Y, Zhou Y, Liu J, Fu D, Matsouaka RA.
\newblock Variance estimation for weighted average treatment effects.
\newblock arXiv preprint arXiv:250808167. 2025.

\bibitem{austin2009balance}
Austin PC.
\newblock Balance diagnostics for comparing the distribution of baseline covariates between treatment groups in propensity-score matched samples.
\newblock Statistics in Medicine. 2009;28(25):3083-107.

\bibitem{zhou2020psweight}
Zhou T, Tong G, Li F, Thomas LE.
\newblock PSweight: an R package for propensity score weighting analysis.
\newblock arXiv preprint arXiv:201008893. 2020.

\bibitem{tan2025double}
Tan X, Yang S, Ye W, Faries DE, Lipkovich I, Kadziola Z.
\newblock Double machine learning methods for estimating average treatment effects: a comparative study.
\newblock Journal of Biopharmaceutical Statistics. 2025:1-20.

\bibitem{hsu2013calibrating}
Hsu JY, Small DS.
\newblock Calibrating sensitivity analyses to observed covariates in observational studies.
\newblock Biometrics. 2013;69(4):803-11.

\bibitem{poterba1992401}
Poterba JM, Venti SF, Wise DA.
\newblock 401 (K) Plans and Tax-Deferred Saving.
\newblock NBER Working Paper. 1992.

\bibitem{poterba1995Do}
Poterba JM, Venti SF, Wise DA.
\newblock Do 401(k) contributions crowd out other personal saving?
\newblock Journal of Public Economics. 1995;58(1):1-32.

\bibitem{jacob2020cross}
Jacob D.
\newblock Cross-Fitting and Averaging for Machine Learning Estimation of Heterogeneous Treatment Effects.
\newblock arXiv preprint arXiv:200702852. 2020 07.

\bibitem{sengupta2016subsampled}
Sengupta S, Volgushev S, Shao X.
\newblock A subsampled double bootstrap for massive data.
\newblock Journal of the American Statistical Association. 2016;111(515):1222-32.

\bibitem{ju2019adaptive}
Ju C, Schwab J, van~der Laan MJ.
\newblock On adaptive propensity score truncation in causal inference.
\newblock Statistical Methods in Medical Research. 2019;28(6):1741-60.

\bibitem{vanderweele2017sensitivity}
VanderWeele TJ, Ding P.
\newblock Sensitivity analysis in observational research: introducing the E-value.
\newblock Annals of internal medicine. 2017;167(4):268-74.

\bibitem{fewell2007impact}
Fewell Z, Davey~Smith G, Sterne JA.
\newblock The impact of residual and unmeasured confounding in epidemiologic studies: a simulation study.
\newblock American Journal of Epidemiology. 2007;166(6):646-55.

\bibitem{lin1998assessing}
Lin DY, Psaty BM, Kronmal RA.
\newblock Assessing the sensitivity of regression results to unmeasured confounders in observational studies.
\newblock Biometrics. 1998:948-63.

\bibitem{faries2023real}
Faries D, Gao C, Zhang X, Hazlett C, Stamey J, Yang S, et~al.
\newblock Real Effect or Bias? Best Practices for Evaluating the Robustness of Real-World Evidence through Quantitative Sensitivity Analysis for Unmeasured Confounding.
\newblock arXiv preprint arXiv:230907273. 2023.

\bibitem{van2006targeted}
Van Der~Laan MJ, Rubin D.
\newblock Targeted maximum likelihood learning.
\newblock The international journal of biostatistics. 2006;2(1).

\bibitem{schuler2017targeted}
Schuler MS, Rose S.
\newblock Targeted maximum likelihood estimation for causal inference in observational studies.
\newblock American Journal of Epidemiology. 2017;185(1):65-73.

\bibitem{neufeld2024data}
Neufeld A, Dharamshi A, Gao LL, Witten D.
\newblock Data thinning for convolution-closed distributions.
\newblock Journal of Machine Learning Research. 2024;25(57):1-35.

\bibitem{liu2025coadvise}
Liu Y, Zhu K, Han L, Yang S.
\newblock COADVISE: Covariate Adjustment with Variable Selection and Missing Data Imputation in Randomized Controlled Trials.
\newblock arXiv preprint arXiv:250108945. 2025.

\bibitem{gao2024does}
Gao Y, Liu Y, Matsouaka R.
\newblock When does adjusting covariate under randomization help? A comparative study on current practices.
\newblock BMC Medical Research Methodology. 2024;24(1):250.

\bibitem{li2019propensity}
Li F, Li F.
\newblock Propensity score weighting for causal inference with multiple treatments.
\newblock The Annals of Applied Statistics. 2019;13(4):2389-415.

\bibitem{yang2016propensity}
Yang S, Imbens GW, Cui Z, Faries DE, Kadziola Z.
\newblock Propensity score matching and subclassification in observational studies with multi-level treatments.
\newblock Biometrics. 2016;72(4):1055-65.

\bibitem{liu2025assessing}
Liu J, Liu Y, Zhou Y, Matsouaka RA.
\newblock Assessing racial disparities in healthcare expenditure using generalized propensity score weighting.
\newblock BMC Medical Research Methodology. 2025;25(1):64.

\bibitem{liu2024average}
Liu Y, Li H, Zhou Y, Matsouaka RA.
\newblock Average treatment effect on the treated, under lack of positivity.
\newblock Statistical Methods in Medical Research. 2024.

\bibitem{yang2023elastic}
Yang S, Gao C, Zeng D, Wang X.
\newblock Elastic integrative analysis of randomised trial and real-world data for treatment heterogeneity estimation.
\newblock Journal of the Royal Statistical Society Series B: Statistical Methodology. 2023;85(3):575-96.

\bibitem{gao2023integrating}
Gao C, Yang S, Shan M, Ye W, Lipkovich I, Faries D.
\newblock Integrating randomized placebo-controlled trial data with external controls: A semiparametric approach with selective borrowing.
\newblock arXiv preprint arXiv:230616642. 2023.

\bibitem{han2025federated}
Han L, Hou J, Cho K, Duan R, Cai T.
\newblock Federated adaptive causal estimation (face) of target treatment effects.
\newblock Journal of the American Statistical Association. 2025:1-14.

\bibitem{liu2025targeted}
Liu Y, Levis AW, Zhu K, Yang S, Gilbert PB, Han L.
\newblock Targeted Data Fusion for Causal Survival Analysis Under Distribution Shift.
\newblock arXiv preprint arXiv:250118798. 2025.

\bibitem{wang2025integrating}
Wang P, Hong H, Jeon K, Thomas LE.
\newblock Integrating Randomized Controlled Trial and External Control Data Using Balancing Weights: A Comparison of Estimands and Estimators.
\newblock arXiv preprint arXiv:250213871. 2025.

\bibitem{zhuang2025assessment}
Zhuang H, Wang X, George SL.
\newblock Assessment of treatment effect heterogeneity for multiregional randomized clinical trials.
\newblock Statistics in Biopharmaceutical Research. 2025;17(3):315-22.

\bibitem{yang2024doubly}
Yang Y, Kuchibhotla AK, Tchetgen~Tchetgen E.
\newblock Doubly robust calibration of prediction sets under covariate shift.
\newblock Journal of the Royal Statistical Society Series B: Statistical Methodology. 2024:qkae009.

\bibitem{liu2024multi}
Liu Y, Levis AW, Normand SL, Han L.
\newblock Multi-Source Conformal Inference Under Distribution Shift.
\newblock arXiv preprint arXiv:240509331. 2024.

\bibitem{levis2024nonparametric}
Levis AW, Kennedy EH, Keele L.
\newblock Nonparametric identification and efficient estimation of causal effects with instrumental variables.
\newblock arXiv preprint arXiv:240209332. 2024.

\bibitem{zeng2021propensity2}
Zeng S, Li F, Hu L.
\newblock Propensity score weighting analysis of survival outcomes using pseudo-observations.
\newblock arXiv preprint arXiv:210300605. 2021.

\bibitem{cao2024using}
Cao Z, Ghazi L, Mastrogiacomo C, Forastiere L, Wilson FP, Li F.
\newblock Using Overlap Weights to Address Extreme Propensity Scores in Estimating Restricted Mean Counterfactual Survival Times.
\newblock American Journal of Epidemiology. 2024:kwae416.

\bibitem{colnet2024causal}
Colnet B, Mayer I, Chen G, Dieng A, Li R, Varoquaux G, et~al.
\newblock Causal inference methods for combining randomized trials and observational studies: a review.
\newblock Statistical Sciences. 2024;39(1):165-91.

\bibitem{lee2022doubly}
Lee D, Yang S, Wang X.
\newblock Doubly robust estimators for generalizing treatment effects on survival outcomes from randomized controlled trials to a target population.
\newblock Journal of causal inference. 2022;10(1):415-40.

\bibitem{lee2024transporting}
Lee D, Gao C, Ghosh S, Yang S.
\newblock Transporting survival of an HIV clinical trial to the external target populations.
\newblock Journal of Biopharmaceutical Statistics. 2024:1-22.

\bibitem{lee2023improving}
Lee D, Yang S, Dong L, Wang X, Zeng D, Cai J.
\newblock Improving trial generalizability using observational studies.
\newblock Biometrics. 2023;79(2):1213-25.

\bibitem{lee2024genrct}
Lee D, Yang S, Berry M, Stinchcombe T, Cohen HJ, Wang X.
\newblock genRCT: a statistical analysis framework for generalizing RCT findings to real-world population.
\newblock Journal of Biopharmaceutical Statistics. 2024:1-20.

\bibitem{yang2022rwd}
Yang S, Wang X.
\newblock RWD-integrated randomized clinical trial analysis.
\newblock Biopharmaceutical Report. 2022;29(2):15.

\end{thebibliography}

\clearpage

\renewcommand{\thesection}{S}
\newcounter{Appendix}[section]
\numberwithin{equation}{subsection}
\renewcommand\theequation{\Alph{section}.\arabic{equation}}
\numberwithin{table}{section}
\numberwithin{figure}{section}

\section{Appendix: Technical Details}\label{app:techcond}

This appendix provides additional technical details and the conditions required for the theorems stated in the main text. All notation follows the definitions introduced therein. 

\subsection{Additional notation}\label{subapp:addnote}

\begin{table}[H]
    \centering
    \scriptsize
    \caption{Table of notation}\label{tab:notation}
    \begin{tabular}{rl}
    \toprule
    Notation & Definition \\
    \midrule
    $\mathbb R^d$ & Euclidean space of dimension $d$ $(0<d<\infty)$ \\
    \addlinespace
        $\mb V$ & {Observed data $(\mb X,A,Y)$, $\mb X$ is the covariate vector, $A$ is the binary treatment, and $Y$ is the outcome} \\
        \addlinespace
         $f(\mb V)$ & Nonparametric joint likelihood function based on $V$ \\
         \addlinespace
    $\theta$ & {Parameter for a regular parametric submodel $f_\theta(\mb V;\theta)$, with $\theta_0$ is the truth such that $f(\mb V)=f(\mb V;\theta_0)$} \\
    \addlinespace
    $\phi$ & Vector of nuisance parameter ${\phi}=({e},{\mu}_1, {\mu}_0)^\dagger$ \\
    $\Lambda$ & Nuisance tangent space (see Bickel et al. \cite{bickel1993efficient}) \\
    $\Lambda^\perp$ & The orthogonal complement of $\Lambda$ \\
    \addlinespace
    $\mc T$ & {A convex subset of $P$-square integrable functions containing the estimated nuisance parameter $\hat{\phi}=(\hat{e},\hat{\mu}_1,\hat{\mu}_0)$} \\
    \addlinespace
    $\rightsquigarrow$ & Weak convergence \\
    $a\preceq b$ & $a\leq Cb$ for a finite constant $C>0$\\
    $a\prec b$ & $a<Cb$ for a finite constant $C>0$ \\
    $a\vee b$ & $\max(a,b)$ \\
     $a\wedge b$ &  $\min(a,b)$ \\
     $\Vert g(\mb V)\Vert_q$ & $\left\{\int g(\mb v)^qf(\mb v) \mu(\rmd\mb v)\right\}^{1/q}$, the $L_q$-norm for function $g(\mb V)$, $q>0$\\
     \addlinespace
     $\pr[g(\mb V;\tilde\phi)]$ & $\int g(\mb v;\phi)f(\mb v)\mu(\rmd\mb v)\big\vert_{\phi=\tilde\phi}$, where $\mu(\cdot)$ is a base measure\\
    \bottomrule
    \end{tabular}
    \begin{tablenotes}
        \item $\dagger$: $e = e(\mb X)$ is the propensity score, $\mu_a=\mu_a(\mb X)$ is the model for conitional outcome $\E\{Y(a)\mid \mb X\}$ for $a=0,1$. 
    \end{tablenotes}
\end{table}

We first summarize all useful notation in Table \ref{tab:notation} that will be used in stating some technical conditions and results as well as lemmata, related theory and theorems in our \textit{Online Supplemental Material}. 

\subsection{Comparison of EIFs under different roles of the propensity score (technical details for Remark \ref{rem:EIF})}\label{subapp:rolePS}

In this section, we discuss two distinct EIFs for a general WATE. By ``distinct,'' we refer to the fact that these EIFs are derived under different assumptions about the joint distribution of the observed data $\mb V=(\mb X, A, Y)$. Specifically, we consider whether the true propensity score $e(\mb X)$ for each participant is known or unknown, and we quantify the efficiency loss due to the unknown propensity score. We follow all notation made in Section \ref{sec:EIF}. 

Consider the following decomposition of the joint likelihood: 
\begin{align}\label{eq:likeli}
f(\mb V)=f(\mb X)f(A\mid \mb X)f(Y\mid A,\mb X).
\end{align}
We then consider a one-dimensional parametric submodel of $f(\mb V)$ as $f_\theta(\mb V)$, which is assumed to contain the truth $f(\mb V)$ at $\theta=0$. In the following, we use $\dot{f}_\theta(\cdot)$ to denote the first partial derivative of a function $f_\theta(\cdot)$ with respect to $\theta$, and $s_\theta(\mb V)$ to denote the score function of $\theta$ with respect to the submodel $f_\theta(\mb V)$. In addition, from \eqref{eq:likeli}, the score function under the submodel can be decomposed as
\begin{align*}
s_\theta(\mb V)=s_\theta(\mb X)+s_\theta(A\mid \mb X)+s_\theta(Y\mid A,\mb X),
\end{align*}
where $s_\theta(\mb X)=\partial\log {f_\theta({\mb X})}/\partial\theta$, $s_\theta(A\mid \mb X)=\partial\log f_\theta(A\mid \mb X)/\partial \theta$ and $s_\theta(Y\mid A,\mb X)=\partial\log f_\theta(Y\mid A,\mb X)/\partial \theta$
are the score functions corresponding to the three components of the likelihood. We also denote $s(\cdot) = s_\theta(\cdot)\big\vert _{\theta=0}$, the score function evaluated at the parameter value corresponding to the true model. 

Note that in \eqref{eq:likeli}, we have $f(A \mid \mb X) = e(\mb X)^A \{1 - e(\mb X)\}^{1-A}$. Therefore, whether $e(\mb X)$ is known or not can affect the subsequent derivations. If $e(\mb X)$ is known, then $s_\theta(A \mid \mb X)=0$ for all $\theta$, since it is no longer treated as a random component, contributing nothing to the score function with respect to $\theta$. In this case, 
\begin{align*}
s_\theta(\mb V)=s_\theta(\mb X)+s_\theta(Y\mid A,\mb X). 
\end{align*}
Hirano et al. \cite{hirano2003efficient} derive the EIF of this case in the following proposition. 

\begin{proposition}\label{prop:hinaro}
Suppose the propensity score $e(\mb X)$ is known for any $\mb X$, then for any weight function $\lambda(t)$ with first order derivative 
$\dot{\lambda}(t)=\rmd\lambda(t)/\rmd t$, the EIF of $\gamma_0=\E\{N(\mb V;\phi_0)\}/\E\{D(\mb V;\phi_0)\}$ is given by
\begin{align}\label{eq:EIF-psknown}
    \widetilde\varphi(\mb V;\phi_0)=\dfrac{\lambda\{e(\mb X)\}}{\E\{D(\mb V;\phi_0)\}}\left\{\psi_{\tau}(\mb V)-\gamma_0\right\}, 
\end{align}
where $\psi_\tau(\mb V) = \dfrac{A}{e(\mb X)}\{Y-\mu_1(\mb X)\} - \dfrac{1-A}{1-e(\mb X)}\{Y-\mu_0(\mb X)\}-\tau(\mb X).$
\end{proposition}

If $e(\mb X)$ is unknown, in Section A.1--A.2 of our \textit{Online Supplemental Material}, we derive the corresponding EIF, which has been shown in Theorem \ref{thm:EIF} as
\begin{align}\label{eq:EIF-psunkown}
    \varphi(\mb V;\phi_0)=\dfrac{\lambda\{e(\mb X)\}}{\E\{D(\mb V;\phi_0)\}}\left\{\psi_{\tau}(\mb V)-\gamma_0\right\} + \dfrac{\dot{\lambda}\{e(\mb X)\}}{\E\{D(\mb V;\phi_0)\}}\{\tau(\mb X)-\gamma_0\}\{A-e(\mb X)\}. 
\end{align}

As we can see, compared to \eqref{eq:EIF-psknown}, the second EIF \eqref{eq:EIF-psunkown} has the additional term 
\begin{align*}
    \dfrac{\dot{\lambda}\{e(\mb X)\}}{\E\{D(\mb V;\phi_0)\}}\{\tau(\mb X)-\gamma_0\}\{A-e(\mb X)\} := I_2,
\end{align*}
and in the following we denote the first term 
\begin{align*}
    \dfrac{\lambda\{e(\mb X)\}}{\E\{D(\mb V;\phi_0)\}}\left\{\psi_{\tau}(\mb V)-\gamma_0\right\} := I_1. 
\end{align*}

The two EIFs \eqref{eq:EIF-psknown} and \eqref{eq:EIF-psunkown} motivate their corresponding EIF-based estimators. We calculate the difference of their asymptotic efficiencies as follows:
\begin{align*}
    \E\{\varphi(\mb V;\phi_0)^2\} - \E\{\widetilde\varphi(\mb V;\phi_0)^2\} = \E\{(I_1+I_2)^2\} - \E(I_1^2) = \E(I_2^2) + 2\E(I_1I_2). 
\end{align*}
Furthermore, 
\begin{align*}
    \E(I_1I_2) & = \E\left[\dfrac{\lambda\{e(\mb X)\}\dot{\lambda}\{e(\mb X)\}}{\E\{D(\mb V;\phi_0)\}^2}\left\{\psi_{\tau}(\mb V)-\gamma_0\right\}\{\tau(\mb X)-\gamma_0\}\{A-e(\mb X)\}\right] \\
    & = \E\left[\E\left\{\dfrac{\lambda\{e(\mb X)\}\dot{\lambda}\{e(\mb X)\}}{\E\{D(\mb V;\phi_0)\}^2}\left\{\psi_{\tau}(\mb V)-\gamma_0\right\}\{\tau(\mb X)-\gamma_0\}\{A-e(\mb X)\}~\bigg|~\mb X\right\}\right] \\
    & = \E\left[\dfrac{\lambda\{e(\mb X)\}\dot{\lambda}\{e(\mb X)\}}{\E\{D(\mb V;\phi_0)\}^2}\{\tau(\mb X)-\gamma_0\}\E\left[\left\{\psi_{\tau}(\mb V)-\gamma_0\right\}\{A-e(\mb X)\}\mid\mb X\right]\right].
\end{align*}
Moreover, we examine the inner conditional expectation above:
\begin{align*}
    & \E\left[\left\{\psi_{\tau}(\mb V)-\gamma_0\right\}\{A-e(\mb X)\}\mid\mb X\right] \\
    & = \E\left[\left\{\dfrac{A}{e(\mb X)}\{Y-\mu_1(\mb X)\} - \dfrac{1-A}{1-e(\mb X)}\{Y-\mu_0(\mb X)\}-\tau(\mb X)-\gamma_0\right\}\{A-e(\mb X)\}~\bigg|~\mb X\right] \\
    & = \E\left[\left\{\dfrac{A}{e(\mb X)}\{Y(1)-\mu_1(\mb X)\} - \dfrac{1-A}{1-e(\mb X)}\{Y(0)-\mu_0(\mb X)\}\right\}\{A-e(\mb X)\}~\bigg|~\mb X\right]~~(\text{by Consistency}) \\
    & \qquad - \underbrace{\E[\{\tau(\mb X)+\gamma_0\}\{A-e(\mb X)\}\mid\mb X]}_{=0} \\
    & = \E\left[A\dfrac{A}{e(\mb X)}\{Y(1)-\mu_1(\mb X)\}~\bigg|~\mb X\right] - \E\left[e(\mb X)\left\{\dfrac{A}{e(\mb X)}\{Y(1)-\mu_1(\mb X)\} - \dfrac{1-A}{1-e(\mb X)}\{Y(0)-\mu_0(\mb X)\}\right\}~\bigg|~\mb X\right]\\
    & = \left[\dfrac{e(\mb X)}{e(\mb X)}\{\mu_1(\mb X)-\mu_1(\mb X)\}\right] - \left[e(\mb X)\left\{\dfrac{e(\mb X)}{e(\mb X)}\{\mu_1(\mb X)-\mu_1(\mb X)\} - \dfrac{1-e(\mb X)}{1-e(\mb X)}\{\mu_0(\mb X)-\mu_0(\mb X)\}\right\}\right] \\
    & = 0, 
\end{align*}
by Assumption \ref{assmp:unconfound} and $\E\{Y(a)\mid\mb X\}=\mu_a(\mb X)$ for $a=0,1$. Therefore, $\E(I_1 I_2) = 0$, which implies that the difference in asymptotic variances is
\begin{align*}
\E\{\varphi(\mb V;\phi_0)^2\} - \E\{\widetilde\varphi(\mb V;\phi_0)^2\} = \E(I_2^2) = \E\left[\dfrac{\dot{\lambda}\{e(\mb X)\}^2}{\E\{D(\mb V;\phi_0)\}^2} \{\tau(\mb X) - \gamma_0\}^2 \{A - e(\mb X)\}^2\right],
\end{align*}
which quantifies the amount of efficiency loss resulting from not knowing the true propensity score. 

For a general WATE, the above difference is non-zero. However, a notable exception is the ATE, for which $\dot{\lambda}\{e(\mb X)\} = 0$ because the corresponding weighting function ${\lambda}\{e(\mb X)\} = 1$ is constant. This aligns with the result of Hahn \cite{hahn1998role}, which shows that the propensity score plays an ancillary role in semiparametrically efficient estimation of the ATE, whereas for other estimands, it does not.

\subsection{Additional technical details for Theorem \ref{thm:RDR-DML}}\label{subapp:condRDR-DML}

\begin{table}[H]
	\caption{Conditions on estimated propensity score $\hat e(\mb X)$ for WATE estimation}
	\label{tab:WATEconditions}
	\centering
    \begin{tabular}{cclllll}
		\toprule 
		 &  $\lambda(t)$  & Conditions by Theorem \ref{thm:RDR-EIF} & Conditions by Theorem \ref{thm:RDR-DML} \\
		\midrule 
        ATE &  $1$ & $\hat{e}(\mb X)\in[C_1, C_2]$ & $\hat{e}(\mb X)\in[C_1, C_2]$ \\
        ATT &  $t$ & $\hat{e}(\mb X)\in (0,C_2]$ & $\hat{e}(\mb X)\in[C_1, C_2]$ \\
        ATC &  $1-t$ & $\hat{e}(\mb X)\in[C_1,1)$ & $\hat{e}(\mb X)\in[C_1, C_2]$  \\
        \addlinespace 
        ATEN & $\displaystyle\sum_{k\in\{t,1-t\}}-k\log k$ & $\hat{e}(\mb X)\in[C_1, C_2]$ & $\hat{e}(\mb X)\in[C_1, C_2]$ \\
        \addlinespace 
        ATO & $t(1-t)$  & $\hat{e}(\mb X)\in(0,1)$ & $\hat{e}(\mb X)\in[C_1, C_2]$ \\
        \addlinespace 
        ATB & 	
        $t^{\nu_{1}-1}(1-t)^{\nu_{2}-1}$ & $\nu_1=\nu_2\in\{2\}\cup[3,\infty)$: $\hat{e}(\mb X)\in(0,1)$ & $\nu_1=\nu_2\geq 4$: $ \hat{e}(\mb X)\in(0,1)$ \\
        & $(\nu_1, \nu_2\geq 2)$ & $\nu_1=\nu_2\in(2,3)$: $\hat{e}(\mb X)\in[C_1, C_2]$ & $\nu_1=\nu_2<4$: $\hat{e}(\mb X)\in[C_1, C_2]$ \\
\bottomrule 
\end{tabular}
\begin{tablenotes}
    \item $C_1$ and $C_2$ are some constants with $0<C_1\leq C_2<1$.
\end{tablenotes}
\end{table}

In addition, we have the following more technical writing in mathematics for results in Theorem \ref{thm:RDR-DML}.  
\begin{itemize}
\item If $\delta_n\geq n^{-1/2}$ for all $n\geq 1$, then 
$\hat\gamma^{\text{dml-1}}-\hat\gamma^{\text{dml-2}}=o_p(1)$. 
\item For $d=1,2$,
$$
\sqrt{n}\sigma^{-1}(\hat{\gamma}^{\text{dml-}d}-\gamma_0)=\frac{1}{\sqrt{n}}\sum_{i=1}^{n}\bar{g}(\mb V_i)+o_p(\rho_n)\rightsquigarrow\mc N(0,1),
$$
uniformly over $n$, where the remainder term obeys $\rho_n\preceq\delta_n$ and $\bar{g}(\cdot)=-(\sigma J)^{-1}g(\cdot;\gamma_0,\phi_0)$ is the EIF and the approximate variance is $\sigma^2=J^{-2}\E[L(\mb V;\gamma_0,\phi_{0})^2]$
with 
$J=\E[D(\mb V;\phi_0)]$. 
\item Suppose $\delta_n\geq n^{-[(1-2/q)\wedge(1/2)]}$ for all $n\geq 1$, then the result in the second item above continues to hold if $\sigma^2$ is replaced by $\hat{\sigma}^2$ defined in Theorem \ref{thm:RDR-DML} (Section \ref{sec:DML}). 
\end{itemize}

\clearpage

\begin{center}
    \LARGE \bf Online Supplemental Material
\end{center}

\appendix
\renewcommand\thesection{\Alph{section}}
\numberwithin{equation}{subsection}
\renewcommand\theequation{\Alph{section}.\arabic{subsection}.\arabic{equation}}
\renewcommand\theassumption{\Alph{section}.\arabic{subsection}.\arabic{assumption}}
\renewcommand\thetheorem{\Alph{section}.\arabic{subsection}.\arabic{theorem}}
\renewcommand\thelemma{\Alph{section}.\arabic{subsection}.\arabic{lemma}}
\renewcommand\theremark{\Alph{section}.\arabic{remark}}
\numberwithin{table}{section}
\numberwithin{figure}{section}

\section{Theory of EIF-based Estimator}\label{apx:theory-EIF}
\subsection{Preliminaries and Lemmata}\label{subapx:prelim-EIF}

Let $\mb V=(\mb X,A,Y)$ be the vector of observed variables with factorization of the likelihood function 
\begin{align}\label{eq:likeli-1}
f(\mb V)=f(\mb X)f(A\mid \mb X)f(Y\mid A,\mb X).
\end{align}
To derive the efficient influence function (EIF) of WATE, we consider a one-dimensional parametric submodel of $f(\mb V)$ as $f_\theta(\mb V)$, which is assumed to contain the truth $f(\mb V)$ at $\theta=0$. 
In the following, we use $\dot{f}_\theta(\cdot)$ to denote the first partial derivative of a function $f_\theta(\cdot)$ with respect to $\theta$, and $s_\theta(\mb V)$ to denote the score function of $\theta$ with respect to the submodel $f_\theta(\mb V)$. In addition, from \eqref{eq:likeli-1}, the score function under the submodel can be decomposed as
\begin{align*}
s_\theta(\mb V)=s_\theta(\mb X)+s_\theta(A\mid \mb X)+s_\theta(Y\mid A,\mb X),
\end{align*}
where $s_\theta(\mb X)=\partial\log {f_\theta({\mb X})}/\partial\theta$, $s_\theta(A\mid \mb X)=\partial\log f_\theta(A\mid \mb X)/\partial \theta$ and $s_\theta(Y\mid A,\mb X)=\partial\log f_\theta(Y\mid A,\mb X)/\partial \theta$
are the score functions corresponding to the three components of the likelihood. We also denote $s(\cdot) = s_\theta(\cdot)\big\vert _{\theta=0}$, the score function evaluated at the parameter value corresponding to the true model. 

Based on the semiparametric theory, we can define the model tangent space
\begin{align*}
\Lambda = \mc H_1\oplus \mc H_2\oplus \mc H_3,
\end{align*}
by the direct sum of the following three spaces:
\begin{align*}
	\mc H_1&=\{h(\mb X): \E\{h(\mb X)\}=0\},\\
	\mc H_2&=\{h(A,\mb X): \E\{h(A,\mb X)\mid \mb X\}=0\},\\
	\mc H_3&=\{h(Y,A,\mb X): \E\{h(Y,A,\mb X)\mid A,\mb X\}=0\},
\end{align*}
where $\mc H_1,\mc H_2$ and $\mc H_3$ are orthogonal to each other. The EIF for $\gamma_n$ and $\gamma_d$, denoted by $\varphi_n,\varphi_d\in\Lambda$, must satisfy, respectively,
\begin{align*}
\dot{\gamma}_{n,\theta}\big\vert_{\theta=0} &= \E\{\varphi_n(\mb V)s(\mb V)\}, \\
\dot{\gamma}_{d,\theta}\big\vert_{\theta=0} &= \E\{\varphi_d(\mb V)s(\mb V)\}.
\end{align*}
To derive these EIFs, we calculate $\dot{\gamma}_{n,\theta}\big\vert_{\theta=0}$ and $ \dot{\gamma}_{d,\theta}\big\vert_{\theta=0}$ separately. To simplify the proof, we introduce some lemmata.

\begin{lemma}\label{lem:ratio} Consider a ratio-type parameter $R=N/D$.
	If there exists functions $\varphi_N$ and $\varphi_D$ such that $\dot{N}_{\theta}\big\vert_{\theta=0}=\E\{\varphi_{N}(\mb V)s(\mb V)\}$ and $\dot{D}_{\theta}\big\vert_{\theta=0}=\E\{\varphi_{D}(\mb V)s(\mb V)\}$,
	then $\dot{R}_{\theta}\big\vert_{\theta=0}=\E\{\varphi_{R}(\mb V)s(\mb V)\}$ where
	$$
	\varphi_{R}(\mb V)=\frac{1}{D}\varphi_{N}(\mb V)-\frac{R}{D}\varphi_{D}(\mb V).
 $$
	In particular, if $\varphi_{N}(\mb V)$ and $\varphi_{D}(\mb V)$ are the
	EIFs for $N$ and $D$, respectively, then $\varphi_{R}(\mb V)$ is the EIF for $R$.
\end{lemma}

\begin{proof}
    Let $R_\theta$, $N_\theta$ and $D_\theta$ denote quantities $R,N$ and $D$, respectively, evaluated with respect to the parametric submodel $f_\theta(\mb V)$. By the chain rule, we have
\begin{align*}
	\dot{R}_\theta\big\vert_{\theta=0}&=\frac{\dot{N}_\theta}{D}\bigg\vert_{\theta=0}-R_\theta\frac{\dot{D}_\theta}{D}\bigg\vert_{\theta=0}\\
	&=\frac{1}{D}\E\{\varphi_{N}(\mb V)s(\mb V)\}-\frac{R}{D}\E\{\varphi_D(\mb V)s(\mb V)\}\\
	&=\E\bigg[\bigg\{\frac{1}{D}\varphi_{N}(\mb V)-\frac{R}{D}\varphi_{D}(\mb V)\bigg\}s(\mb V)\bigg]
\end{align*}
When $\varphi_{N}(\mb V)$ and $\varphi_{D}(\mb V)$ are the EIFs for $N$ and $D$, then $\varphi_{N}$ and $\varphi_{D}(\mb V)$ are in the tangent space, and thus $\varphi_{R}$ is also in the tangent space. Therefore, $\varphi_{R}$ is the EIF for $R$.
\end{proof}

\begin{lemma} \label{lem:scoreE} For any $h(\mb V)$ that does
	not depend on $\theta$, $\partial\E_{\theta}\{h(\mb V)\}/\partial\theta\big\vert_{\theta=0}=\E\{h(\mb V)s(\mb V)\}$.
\end{lemma} 
\begin{proof}
    The proof is straightforward if we assume interchangeability of integral and derivative. 
\end{proof}

\begin{lemma}\label{lem:condx} For propensity score $e(\mb X)$ and CATE $\tau(\mb X)$, we have 
\begin{align*}
    \dot{e}(\mb X)\big\vert_{\theta=0} & = \E\left[\{A-e(\mb X)\}s(A\mid \mb X)\mid \mb X\right],\\
    \dot{\tau}_{\theta}(\mb X)\big\vert_{\theta=0} & = \E\left\{ \psi_{\tau}(\mb V)s(Y\mid A,\mb X)\mid \mb X\right\} .
\end{align*}
\end{lemma}
\begin{proof}
First,
\begin{align*}
	\dot e(\mb X)\big\vert_{\theta=0} 
	& = \frac{\partial}{\partial \theta}\E_\theta\{A\mid \mb X\}\big\vert_{\theta=0}\\
	& = \E\{A\times s(A\mid \mb X)\mid \mb X\} \qquad \mbox{by Lemma \ref{lem:scoreE}}\\
	& = \E\left[\{A-e(\mb X)\}s(A\mid \mb X)\mid \mb X\right].
\end{align*}
The last ``$=$'' is by $\E\{e(\mb X)s(A\mid \mb X)\mid \mb X\}=e(\mb X)\E\{s(A\mid \mb X)\mid \mb X\}=0$, since 
$\E\{s(A\mid \mb X)\mid \mb X\}=\dfrac{\partial}{\partial\theta}\E_\theta(1\mid \mb X)\big\vert_{\theta=0} 
= 0 $ by Lemma \ref{lem:scoreE}. In addition, 
\begin{align*}
	\dot{\tau}_\theta(\mb X)\big\vert_{\theta=0} 
	&= \frac{\partial}{\partial \theta}\{\E_\theta(Y\mid A=1,\mb X)-\E_\theta(Y\mid A=0,\mb X)\}\big\vert_{\theta=0}\\
	&= \E\{Ys(Y\mid A=1,\mb X)\mid A=1,\mb X\} - \E\{Ys(Y\mid A=0,\mb X)\mid A=0,\mb X\}\qquad \mbox{by Lemma \ref{lem:scoreE}}\\
	&= \E\left[\left\{\frac{AY}{e(\mb X)}-\frac{(1-A)Y}{1-e(\mb X)}\right\}s(Y\mid A,\mb X)~\bigg\vert~ \mb X\right]\\
	&= \E\{\psi_{\tau}(\mb V)s(Y\mid A,\mb X)\mid \mb X\},
\end{align*}
where the last equality is by $\E\left[\left\{\dfrac{A}{e(\mb X)}\mu_1(\mb X)-\dfrac{1-A}{1-e(\mb X)}\mu_0(\mb X)-\tau(\mb X)\right\}s(Y\mid A,\mb X)~\bigg\vert~ \mb X\right]=0$ as shown below.
\begin{align*}
& \quad\E\left[\left\{\frac{A}{e(\mb X)}\mu_1(\mb X)-\frac{1-A}{1-e(\mb X)}\mu_0(\mb X)-\tau(\mb X)\right\}s(Y\mid A,\mb X)~\bigg\vert~ \mb X\right]\\
	&=\mu_1(\mb X)\E \{ s(Y\mid A=1,\mb X)\mid A=1,\mb X\}-\mu_0(\mb X)\E\{s(Y\mid A=0,\mb X)\mid A=0,\mb X\} \\
 & \quad - \tau(\mb X)e(\mb X) {\E\{s(Y\mid A=1,\mb X)\mid A=1,\mb X\} + \tau(\mb X)\{1-e(\mb X)\}}\E\{s(Y\mid A=0,\mb X)\mid A=0,\mb X\}\\
	&=0,
\end{align*}
because by Lemma \ref{lem:scoreE}, both $\E\{s(Y\mid A=1,\mb X)\mid A=1,\mb X\}=\dfrac{\partial}{\partial\theta}\E(1\mid A=1,\mb X)\big\vert_{\theta=0}$ and $\E\{s(Y\mid A=0,\mb X)\mid A=0,\mb X\}=\dfrac{\partial}{\partial \theta}\E(1\mid A=0,\mb X)\big\vert_{\theta=0}$ are $0$, which completes the proof.
\end{proof}

Next, we introduce some lemma related to the Donsker property of a class of functions. Consider a set of (nuisance) functions by $\mc{G}_{\phi_0}=\{\phi:\Vert\phi-\phi_0\Vert_2<\delta\}$ for some $\delta>0$ and a given function $\phi_0$. Denote $l^\infty(\mc{G}_{\phi_0})$ the collection of all bounded functions $f:\mc{G}_{\phi_0}\mapsto\mathbb{R}^p$.

\begin{lemma}\label{lem:donsker} Assuming Condition \ref{cond:EIFdonsker} hold, we have the following convergence results. 
\begin{align*}
	\sup_{\phi\in l^\infty(\mc{G}_{\phi_0})}\big\Vert \pr_n [N(\mb V;\phi)]-\pr[N(\mb V;\phi)]\big\Vert_2 \to_p 0, \\
	\sup_{\phi\in l^\infty(\mc{G}_{\phi_0})}\big\Vert \pr_n [D(\mb V;\phi)]-\pr[D(\mb V;\phi)] \big\Vert_2 \to_p 0, \\
    n^{-{1}/{2}}(\pr_n-\pr)[N(\mb V;\phi)] \rightsquigarrow Z\in l^\infty(\mc{G}_{\phi_0}), \text{ and } \\
    n^{-{1}/{2}}(\pr_n-\pr)[D(\mb V;\phi)] \rightsquigarrow Z\in l^\infty(\mc{G}_{\phi_0}),
\end{align*}
 as $n\to\infty$, where the limit process $Z=\{Z(\phi):\phi\in\mc{G}_{\phi_0} \}$ is a mean-zero multivariate Gaussian process and the sample paths of $Z$ belong to $\{z:l^\infty(\mc{G}_{\phi_0}):z\mbox{ is a uniformly continuous with respect to }\Vert\cdot\Vert\}$
\end{lemma}

\subsection{EIFs of $\gamma_d$, $\gamma_n$ and $\gamma$ (Proof of Theorem \ref{thm:EIF} in Section \ref{sec:EIF})}\label{subapx:EIFs}
Based on results in Section \ref{subapx:prelim-EIF}, we derive the EIFs of $\gamma_d$, $\gamma_n$ and $\gamma$ as follows. For EIF of $\gamma_d$, we first calculate $\dot{\gamma}_{d,\theta}\big\vert_{\theta=0}$ as
\begin{align*}
	\dot{\gamma}_{d,\theta}(\mb X)\big\vert_{\theta=0}
	&=\frac{\partial}{\partial \theta}\E_\theta[\lambda\{e(\mb X)\}]\big\vert_{\theta=0} \\
	&=\E[\dot{\lambda}\{e(\mb X)\}\dot{e}(\mb X)+\lambda\{e(\mb X)\}s(\mb X)] \qquad \mbox{by Lemma \ref{lem:scoreE} {(Section \ref{subapx:prelim-EIF})}}\\
	&=\E[\dot{\lambda}\{e(\mb X)\}\E[ \{A-e(\mb X)\}s(A\mid \mb X)  \mid \mb X]+\lambda\{e(\mb X)\}s(\mb X)]\qquad \mbox{by Lemma \ref{lem:condx} {(Section \ref{subapx:prelim-EIF})}}\\
	&=\E\left(\dot{\lambda}\{e(\mb X)\}\{A-e(\mb X)\}s(A\mid \mb X) +[\lambda\{e(\mb X)\}-\gamma_d]s(\mb X)\right)\\
	&=\E\{\varphi_d(\mb V)s(\mb V)\},
\end{align*}
where the fourth ``='' follows from $\E\{\gamma_d s(\mb X)\}=\gamma_d\frac{\partial}{\partial \theta}\E\{1\}\big\vert_{\theta=0}=0$ and
the last ``='' follows from 
\begin{align}
\E[\dot{\lambda}\{e(\mb X)\}\{A-e(\mb X)\}s(\mb X)]&=0 \label{eq:prove EIF 1}\\
\E([\lambda\{e(\mb X)\}-\gamma_d]s(Y\mid A,\mb X))&=0  \label{eq:prove EIF 2}\\
\E([\lambda\{e(\mb X)\}-\gamma_d]s(A\mid \mb X))&=0  \label{eq:prove EIF 3}\\
\E[\dot{\lambda}\{e(\mb X)\}\{A-e(\mb X)\}s(Y\mid A,\mb X)]&=0  \label{eq:prove EIF 4}
\end{align}
as shown below. 
\eqref{eq:prove EIF 1}--\eqref{eq:prove EIF 3} hold because
\begin{align*}
	\E[\dot{\lambda}\{e(\mb X)\}\{A-e(\mb X)\}s(\mb X)]=\E[\dot{\lambda}\{e(\mb X)\}s(\mb X) \E\{A-e(\mb X)\mid \mb X\}]&= 0,\\
	\E([\lambda\{e(\mb X)\}-\gamma_d]s(Y\mid A,\mb X))
	=\E([\lambda\{e(\mb X)\}-\gamma_d] \cdot \E \{s(Y\mid A,\mb X)\mid A,\mb X\}) &= 0,\\
	\E([\lambda\{e(\mb X)\}-\gamma_d]s(A\mid \mb X))
	=\E([\lambda\{e(\mb X)\}-\gamma_d]\cdot \E\{s(A\mid \mb X)\mid \mb X\})&=0,
\end{align*}
since $\E\{A-e(\mb X)\mid \mb X\}$, $\E\{s(Y\mid A,\mb X)\mid A,\mb X\}$ and  $\E\{s(A\mid \mb X)\mid \mb X\}$ are all equal to 0.
\eqref{eq:prove EIF 4} also holds because
\begin{align*}
	&\quad\E[\dot{\lambda}\{e(\mb X)\}\{A-e(\mb X)\}s(Y\mid A,\mb X)]\\
	&= \E[\dot{\lambda}\{e(\mb X)\}\E[\{A-e(\mb X)\}s(Y\mid A,\mb X)\mid \mb X]]\\
	&=\E[\dot{\lambda}\{e(\mb X)\}\{1-e(\mb X)\}e(\mb X)]\cdot\left[\E\{s(Y\mid A=1,\mb X)\mid A=1,\mb X\} - \E\{s(Y\mid A=0,\mb X)\mid A=0,\mb X\}\right]\\
	&=0.
\end{align*}
Then, by \eqref{eq:prove EIF 1} and \eqref{eq:prove EIF 2}, we can verify that $\dot{\lambda}\{e(\mb X)\}\{A-e(\mb X)\}\in \mc H_2$ and $\lambda\{e(\mb X)\}-\gamma_d\in \mc H_1$. Therefore, $\varphi_d(\mb V)$ lies in the tangent space and hence is the EIF of $\gamma_d$.

Similarly, to derive the EIF of $\gamma_n$, we first calculate $\dot{\gamma}_{n,\theta}\big\vert_{\theta=0}$. 
\begin{align*}
	\dot{\gamma}_{n,\theta}\big\vert_{\theta=0}
	&=\E[\dot{\lambda}\{e(\mb X)\}\dot{e}(\mb X)\tau(\mb X)]+\E[\lambda\{e(\mb X)\}\dot{\tau}(\mb X)] + \E[\lambda\{e(\mb X)\}\tau(\mb X)s(\mb X)]\\
	&=\E (\dot{\lambda}\{e(\mb X)\}\tau(\mb X) \E[\{A-e(\mb X)\}s(A\mid \mb X)\mid \mb X])+\E(\lambda\{e(\mb X)\}\E[\psi_{\tau}(\mb V)s(Y\mid A,\mb X)\mid \mb X])\\ 
 &\quad +\E[\lambda\{e(\mb X)\}\tau(\mb X)s(\mb X)] \qquad \mbox{by Lemma \ref{lem:condx} {(Section \ref{subapx:prelim-EIF})}}
	\\
	&=\E[\dot{\lambda}\{e(\mb X)\}\tau(\mb X)\{A-e(\mb X)\}s(A\mid \mb X)]+\E[\lambda\{e(\mb X)\} \tau(\mb X)s(Y\mid A,\mb X)]\\
	&\quad +\E\left[\lambda\{e(\mb X)\}\left[\frac{A}{e(\mb X)}\{Y-\mu_1(\mb X)\}-\frac{1-A}{1-e(\mb X)}\{Y-\mu_0(\mb X)\}\right]s(Y\mid A,\mb X)\right]\\
	&\quad +\E[\lambda\{e(\mb X)\}\tau(\mb X)s(\mb X)] - \E\{\gamma_n s(\mb X)\}-\E\{\gamma_ns(Y\mid A,\mb X)\}\\
	&=\E[\dot{\lambda}\{e(\mb X)\}\tau(\mb X)\{A-e(\mb X)\}s(A\mid \mb X)] + \E[\lambda\{e(\mb X)\}\tau(\mb X)-\gamma_n]\{s(\mb X)+s(Y\mid A,\mb X)\}\\
	& \quad+\E\left[\lambda\{e(\mb X)\}\left[\frac{A}{e(\mb X)}\{Y-\mu_1(\mb X)\}-\frac{1-A}{1-e(\mb X)}\{Y-\mu_0(\mb X)\}\right]s(Y\mid A,\mb X)\right]
	\\
	&=\E\{\varphi_n(\mb V)s(\mb V)\},
\end{align*}
where the third ``='' follows from $\E\{s(\mb X)\}=0$ and $\E\{s(Y\mid A,\mb X)\}=0$ and
the last ``='' follows from 
\begin{align}
\E[\dot{\lambda}\{e(\mb X)\}\tau(\mb X)\{A-e(\mb X)\}\{s(Y\mid A,\mb X)+s(A\mid \mb X)\}]&=0 \label{eq:prove EIF 5}\\
\E[\lambda\{e(\mb X)\}\tau(\mb X)-\gamma_n]s(A\mid \mb X)&=0 \label{eq:prove EIF 6}\\
\E\left[\lambda\{e(\mb X)\}\left[\frac{A}{e(\mb X)}\{Y-\mu_1(\mb X)\}-\frac{1-A}{1-e(\mb X)}\{Y-\mu_0(\mb X)\}\right]\{s(A\mid \mb X)+s(\mb X)\}\right]&=0 \label{eq:prove EIF 7}
\end{align}
as shown below.
\eqref{eq:prove EIF 5} holds because
\begin{align*}
	& \E[\dot{\lambda}\{e(\mb X)\}\tau(\mb X)\{A-e(\mb X)\}s(Y\mid A,\mb X)] =  \E[\dot{\lambda}\{e(\mb X)\}\tau(\mb X)\cdot\E\{(A-e(\mb X))s(Y\mid A,\mb X)\mid \mb X\}]\\
	& = \E[\dot{\lambda}\{e(\mb X)\}\tau(\mb X)\cdot
	\{1-e(\mb X)\}e(\mb X)\E\{s(Y\mid A=1,\mb X)\mid A=1,\mb X\}]\\
	& \quad -\E[\dot{\lambda}\{e(\mb X)\}\tau(\mb X)\cdot e(\mb X)\{1-e(\mb X)\}\E\{s(Y\mid A=0,\mb X)\mid A=0,\mb X\}] =0,
\end{align*}
and
\begin{align*}
& \E[\dot{\lambda}\{e(\mb X)\}\tau(\mb X)\{A-e(\mb X)\}s(A\mid \mb X)] = \E(\dot{\lambda}\{e(\mb X)\}\tau(\mb X)\E[\{A-e(\mb X)\}s(A\mid \mb X)\mid \mb X])\\
	& = \E[\dot{\lambda}\{e(\mb X)\}\tau(\mb X)\cdot\{1-e(\mb X)\}e(\mb X)\E\{s(A=1\mid \mb X)\mid A=1,\mb X\}] \\
    & \qquad -\E[\dot{\lambda}\{e(\mb X)\}\tau(\mb X)\cdot e(\mb X)\{1-e(\mb X)\}\E\{s(A=0\mid \mb X)\mid A=0,\mb X\}] =0,
\end{align*}
since $\E\{s(Y\mid A=1,\mb X)\mid A=1,\mb X\}$, $\E\{s(Y\mid A=0,\mb X)\mid A=0,\mb X\}$, $\E\{s(A=1\mid \mb X)\mid A=1,\mb X\}$ and $\E\{s(A=0\mid \mb X)\mid A=0,\mb X\}$ are all equal to 0.
\eqref{eq:prove EIF 6} holds because
\begin{align*}
	\E([\lambda\{e(\mb X)\}\tau(\mb X)-\gamma_n]s(A\mid \mb X))= \E([\lambda\{e(\mb X)\}\tau(\mb X)-\gamma_n]\cdot\E[s(A\mid \mb X)\mid \mb X])=0.
\end{align*}
\eqref{eq:prove EIF 7} holds because
\begin{align*}
	&\E\left(\lambda\{e(\mb X)\}\left[\frac{A}{e(\mb X)}\{Y-\mu_1(\mb X)\}-\frac{1-A}{1-e(\mb X)}\{Y-\mu_0(\mb X)\}\right]s(A\mid \mb X)\right)\\
	& =\E\left(\frac{\lambda\{e(\mb X)\}}{e(\mb X)}\E[A\{Y-\mu_1(\mb X)\}s(A\mid \mb X)\mid \mb X]\right)-
	\E \left(\frac{\lambda\{e(\mb X)\}}{1-e(\mb X)} \E[(1-A)\{Y-\mu_0(\mb X)\}s(A\mid \mb X)\mid \mb X]\right)\\
	&=\E (\lambda\{e(\mb X)\}\E[\{Y-\mu_1(\mb X)\}s(A=1\mid \mb X)\mid A=1, \mb X])-\E (\lambda\{e(\mb X)\}\E[\{Y-\mu_0(\mb X)\}s(A=0\mid \mb X)\mid A=0, \mb X])\\
	&=\E (\lambda\{e(\mb X)\}s(A=1\mid \mb X)\E[\{Y-\mu_1(\mb X)\}\mid A=1, \mb X])-\E(\lambda\{e(\mb X)\}s(A=0\mid \mb X)\E[\{Y-\mu_0(\mb X)\}\mid A=0, \mb X])\\
&=0,
\end{align*}
and
\begin{align*}
	&\E\left(\lambda\{e(\mb X)\}\left[\frac{A}{e(\mb X)}\{Y-\mu_1(\mb X)\}-\frac{1-A}{1-e(\mb X)}\{Y-\mu_0(\mb X)\}\right]s(\mb X)\right)\\
	&=\E\left(\lambda\{e(\mb X)\}s(\mb X)\E\left[\frac{A}{e(\mb X)}\{Y-\mu_1(\mb X)\}-\frac{1-A}{1-e(\mb X)}\{Y-\mu_0(\mb X)\}~\bigg\vert~ \mb X\right]\right)=0.
\end{align*}
By \eqref{eq:prove EIF 5}--\eqref{eq:prove EIF 7}, it can be verified that $\dot{\lambda}\{e(\mb X)\}\tau(\mb X)\{A-e(\mb X)\}\in \mc H_2,
\lambda\{e(\mb X)\}\tau(\mb X)-\gamma_n\in \mc H_1$ and
$\lambda\{e(\mb X)\}\left[\dfrac{A}{e(\mb X)}\{Y-\mu_1(\mb X)\}-\dfrac{1-A}{1-e(\mb X)}\{Y-\mu_0(\mb X)\}\right]\in \mc H_3$. Therefore,  $\varphi_n(\mb V)$ lies in the tangent space, and thus is the EIF for $\gamma_n$. Finally, according to Lemma \ref{lem:ratio} {in Section \ref{subapx:prelim-EIF}}, the EIF for the WATE $\gamma={\gamma_n}/{\gamma_d}$ is calculated by
\begin{align*}
	\varphi(\mb V)
	&=\frac{1}{\gamma_d}\varphi_n(\mb V)-\frac{\gamma}{\gamma_d}\varphi_{d}(\mb V)\\
	&=\frac{\lambda\{e(\mb X)\}\psi_{\tau}(\mb V) + \dot{\lambda}\{e(\mb X)\}\tau(\mb X)\{A-e(\mb X)\}-\gamma_n}{\gamma_d}
	- \frac{\gamma}{\gamma_d}[ \lambda\{e(\mb X)\}+\dot{\lambda}\{e(\mb X)\}\{A-e(\mb X)\}-\gamma_d]\\
	&=\frac{\lambda\{e(\mb X)\}}{\gamma_d}\{\psi_\tau(\mb V)-\gamma\}+\frac{\dot{\lambda}\{e(\mb X)\}}{\gamma_d}\{\tau(\mb X)-\gamma\}\{A-e(\mb X)\},
\end{align*}
which completes the proof.

\subsection{Proof of Theorem \ref{thm:RDR-EIF} {in Section \ref{sec:EIF}}}
First, we show the consistency of the EIF-based estimator, $\hat{\gamma}^{\text{eif}}$. By Slutsky's theorem, it is sufficient to show
\begin{align}
\pr_n[D(\mb V;\hat{\phi})] - \pr_n[D(\mb V;\phi_0)]& =o_p(1),\text{ and } \label{eq:prove 16}\\
\pr_n[N(\mb V;\hat{\phi})] - \pr_n[N(\mb V;\phi_0)]& =o_p(1). \label{eq:prove 15}
\end{align}
To show \eqref{eq:prove 16}, it suffices to show $\Vert\pr_n[D(\mb V;\hat{\phi})] - \pr_n [D(\mb V;\phi_0)]\Vert_2\to 0$ as $n\to\infty$. Now, by some direct decomposition and triangle inequality, we have
\begin{align}
& \Vert\pr_n[D(\mb V;\hat{\phi})] - \pr_n[D(\mb V;\phi_0)]\Vert_2\nonumber\\
& \leq\Vert\pr_n[D(\mb V;\hat{\phi})] - \pr[D(\mb V;\hat{\phi})]\Vert_2
+\Vert\pr_n[D(\mb V;\phi_0)]-\pr[D(\mb V;\phi_0)]\Vert_2
+ \Vert\pr[D(\mb V;\hat{\phi})]- \pr[D(\mb V;\phi_0)]\Vert_2 \nonumber\\
& \leq 2\sup_{\phi\in l^\infty(\mc{G}_{\phi_0})}\Vert\pr_n [D(\mb V;\phi)]-\pr[D(\mb V;\phi)]\Vert_2 + \Vert\pr[D(\mb V;\hat{\phi})]-\pr [D(\mb V;\phi_0)]\Vert_2.
\label{eq:prove 20}
\end{align}
Both terms on the right-hand side of \eqref{eq:prove 20} are $o_p(1)$ by the following proofs. Using Taylor expansion of $\lambda\{\hat{e}(\mb X)\}$ at $e(\mb X)$, we have
\begin{align}
	\lambda\{\hat{e}(\mb X)\}=\lambda\{e(\mb X)\}+\dot{\lambda}\{\tilde{e}(\mb X)\}\{\hat{e}(\mb X)-e(\mb X)\}, \label{eq:Taylorlambda}
\end{align}
with $\tilde{e}(\mb X)$ lies between $\hat{e}(\mb X)$ and $e(\mb X)$. Then we have
\begin{align}
	\pr[D(\mb V;\hat{\phi})-D(\mb V;\phi_0)]
	&=\pr[\lambda\{\hat{e}(\mb X)\}+\dot{\lambda}\{\hat{e}(\mb X)\}\{e(\mb X)-\hat{e}(\mb X)\} -  \lambda\{e(\mb X)\}] \nonumber\\
	&=\pr[\{\dot{\lambda}\{\tilde{e}(\mb X)\}- \dot{\lambda}\{\hat{e}(\mb X)\}\}\{\hat{e}(\mb X)-e(\mb X)\}]\nonumber\\
    & \preceq \Vert\dot{\lambda}\{\hat{e}(\mb X)\}-\dot{\lambda}\{e(\mb X)\}\Vert_2\cdot\Vert \hat{e}(\mb X)- e(\mb X) \Vert_2, \label{eq:prove 28}
\end{align} 
using Conditions \ref{cond:EIFbound}, \ref{cond:EIFsup} and \ref{cond:EIFLipch}.

Hence, further by Condition \ref{cond:EIF-nuisop} in Theorem \ref{thm:RDR-EIF} {(Section \ref{sec:EIF})}, we can get
\begin{align}
\Vert\pr[D(\mb V;\hat{\phi})]-\pr[D(\mb V;\phi_0)]\Vert_2=o_p(1).\label{eq:prove 22}
\end{align}
By Lemma \ref{lem:donsker} {(Section \ref{subapx:prelim-EIF})}, we have
\begin{align}
\sup_{\phi\in l^\infty(\mc{G}_{\phi_0})}\Vert\pr_n[D(\mb V;\phi)]-\pr[D(\mb V;\phi)]\Vert_2\to_p 0 \label{eq:prove 23}
\end{align}
as $n\to\infty$. Plugging \eqref{eq:prove 22} and \eqref{eq:prove 23} into \eqref{eq:prove 20} leads to \eqref{eq:prove 16}.

Similarly, to show \eqref{eq:prove 15}, we need $\Vert\pr_n[N(\mb V;\hat{\phi})] - \pr_n[N(\mb V;\phi_0)]\Vert_2\to 0$ as $n\to\infty$, which is bounded by
\begin{align}
& \quad\Vert\pr_n[N(\mb V;\hat{\phi})]-\pr[N(\mb V;\hat{\phi})]\Vert_2
+\Vert\pr_n[N(\mb V;\phi_0)]-\pr[N(\mb V;\phi_0)]\Vert_2
+\Vert\pr[N(\mb V;\hat{\phi})]- \pr[N(\mb V;\phi_0)]\Vert_2 \nonumber\\
&\leq 2 \sup_{\phi\in l^\infty(\mc{G}_{\phi_0})}\Vert \pr_n[N(\mb V;\phi)]-\pr[N(\mb V;\phi)]\Vert_2 + \Vert\pr[N(\mb V;\hat{\phi})]-\pr[N(\mb V;\phi_0)]\Vert_2. 
\label{eq:prove 21}
\end{align}

Both terms on the right-hand side of \eqref{eq:prove 21} are $o_p(1)$, with proofs below. To show the second term is $o_p(1)$, we first express 

\begin{align*}
&\pr [N(\mb V;\hat{\phi})]-\pr [N(\mb V;\phi_0)]\\
&=\pr\left[\lambda\{\hat{e}(\mb X)\} \frac{e(\mb X)}{\hat{e}(\mb X)}\{\mu_1(\mb X)-\hat{\mu}_1(\mb X)\} -\lambda\{\hat{e}(\mb X)\}\frac{1-e(\mb X)}{1-\hat{e}(\mb X)}\{\mu_0(\mb X)-\hat{\mu}_0(\mb X)\} +\lambda\{\hat{e}(\mb X)\}\{\hat{\mu}_1(\mb X)-\hat{\mu}_0(\mb X)\}\right]\\
&\quad-\pr[\lambda\{e(\mb X)\}\{\mu_1(\mb X)-\mu_0(\mb X)\}] +\pr[\dot{\lambda}\{\hat{e}(\mb X)\}\hat{\tau}(\mb X)\{e(\mb X)-\hat{e}(\mb X)\}]\\
&= \pr\left[\frac{\lambda\{\hat{e}(\mb X)\}}{\hat{e}(\mb X)} \{e(\mb X)-\hat{e}(\mb X)\} \{\mu_1(\mb X)-\hat{\mu}_1(\mb X)\}\right] \nonumber\\
& \quad -\pr\left[\frac{\lambda\{\hat{e}(\mb X)\}}{1-\hat{e}(\mb X)} \{\hat{e}(\mb X)-e(\mb X)\} \{\mu_0(\mb X)-\hat{\mu}_0(\mb X)\}\right]+\pr[\dot{\lambda}\{\hat{e}(\mb X)\}\{\hat{\tau}(\mb X)-\tau(\mb X)\}\{e(\mb X)-\hat{e}(\mb X)\}] \nonumber\\
& \quad +\pr[\tau(\mb X)\{\lambda\{\hat{e}(\mb X)\}-\lambda\{e(\mb X)\} -\dot{\lambda}\{\hat{e}(\mb X)\}\{\hat{e}(\mb X)-e(\mb X)\}\}],
\end{align*}
and we further bound $\Vert \pr[N(\mb V;\hat{\phi})]-\pr[N(\mb V;\phi_0)]\Vert_2 $ by
\begin{align}
&\left\Vert \pr\left[\frac{\lambda\{\hat{e}(\mb X)\}}{\hat{e}(\mb X)} \{e(\mb X)-\hat{e}(\mb X)\} \{\mu_1(\mb X)-\hat{\mu}_1(\mb X)\right]\right\Vert_2 + \left\Vert\pr\left[\frac{\lambda\{\hat{e}(\mb X)\}}{1-\hat{e}(\mb X)} \{\hat{e}\{\mb X\}-e(\mb X)\}\{\mu_0(\mb X)-\hat{\mu}_0(\mb X)\}\right]\right\Vert_2\nonumber\\
& \quad +\Vert \pr[\dot{\lambda}\{\hat{e}(\mb X)\}\{\hat{\tau}(\mb X)-\tau(\mb X)\}\{e(\mb X)-\hat{e}(\mb X)\}] \Vert_2 \nonumber\\
& \quad +\Vert \pr [\tau(\mb X)\{\lambda\{\hat{e}(\mb X)\}-\lambda\{e(\mb X)\} -\dot{\lambda}\{\hat{e}(\mb X)\}\{\hat{e}(\mb X)-e(\mb X)\}\}]\Vert_2 . \label{eq:prove 31}
\end{align}
All four terms in \eqref{eq:prove 31} are $o_p(1)$ as shown below. By triangle inequality and Cauchy-Schwarz inequality, the first three terms in \eqref{eq:prove 31} satisfies
\begin{align}
\left\Vert  \pr\left[\frac{\lambda\{\hat{e}(\mb X)\}}{\hat{e}(\mb X)} \{e(\mb X)-\hat{e}(\mb X)\} \{\mu_1(\mb X)-\hat{\mu}_1(\mb X)\}\right]\right\Vert_2 \leq \sup_{\mb X}\left\vert \frac{\lambda\{\hat{e}(\mb X)\}}{\hat{e}(\mb X)}\right\vert\cdot \Vert e(\mb X)-\hat{e}(\mb X) \Vert_2 \cdot\Vert  \mu_1(\mb X)-\hat{\mu}_1(\mb X) \Vert_2 \label{eq:prove 33},\\
\left\Vert \pr\left[\frac{\lambda\{\hat{e}(\mb X)\}}{1-\hat{e}(\mb X)} \{e(\mb X)-\hat{e}(\mb X)\} \{\mu_0(\mb X)-\hat{\mu}_0(\mb X)\}\right] \right\Vert_2\leq \sup_{\mb X}\left\vert \frac{\lambda\{\hat{e}(\mb X)\}}{1-\hat{e}(\mb X)} \right\vert \cdot\Vert e(\mb X)-\hat{e}(\mb X) \Vert_2\cdot  \Vert  \mu_0(\mb X)-\hat{\mu}_0(\mb X)\Vert_2 \label{eq:prove 34},\\
\Vert\pr[\dot{\lambda}\{\hat{e}(\mb X)\}\{\hat{\tau}(\mb X)-\tau(\mb X)\}\{e(\mb X)-\hat{e}(\mb X)\}]\Vert_2
\leq \sup_{\mb X}\left\vert \dot{\lambda}\{\hat{e}(\mb X)\}\right\vert\cdot \Vert e(\mb X)-\hat{e}(\mb X) \Vert_2 \cdot \Vert \tau(\mb X)-\hat\tau(\mb X)\Vert_2, \label{eq:prove 35}
\end{align}
which thus are bounded by $o_p(1)$ using Condition \ref{cond:EIFsup} and Condition \ref{cond:EIF-nuisop} in Theorem \ref{thm:RDR-EIF} {(Section \ref{sec:EIF})}. 
By Taylor expansion (similar to \eqref{eq:Taylorlambda}), the fourth term in \eqref{eq:prove 31} can be expanded as
\begin{align}
&\Vert \pr [\tau(\mb X)\{\lambda\{\hat{e}(\mb X)\}-\lambda\{e(\mb X)\} -\dot{\lambda}\{\hat{e}(\mb X)\}\{\hat{e}(\mb X)-e(\mb X)\}\}]\Vert_2 \nonumber\\
&= \Vert\pr[\{\dot{\lambda}\{\tilde{e}(\mb X)\} -\dot{\lambda}\{\hat{e}(\mb X)\}\{\hat{e}(\mb X)-e(\mb X)\}\tau(\mb X)\}]\Vert_2 \nonumber\\
&\leq\sup_{\mb X}\left\vert\tau(\mb X)\right\vert\cdot \Vert\dot{\lambda}\{\hat{e}(\mb X)\}-\dot{\lambda}\{e(\mb X)\}\Vert_2 \cdot \Vert \hat e(\mb X)-e(\mb X)\Vert_2, \label{eq:prove 32}
\end{align}
which is bounded by $o_p(1)$ using Condition \ref{cond:EIFsup} and \ref{cond:EIFLipch} and Condition \ref{cond:EIF-nuisop} in Theorem \ref{thm:RDR-EIF} {(Section \ref{sec:EIF})}. 

Plugging \eqref{eq:prove 33}--\eqref{eq:prove 32} into \eqref{eq:prove 31}, and by Conditions \ref{cond:EIFbound} and \ref{cond:EIFLipch}, we obtain
\begin{align}
\Vert\pr[N(\mb V;\hat{\phi})]-\pr[N(\mb V;\phi_0)]\Vert_2\preceq \left\Vert \hat e(\mb X)-{e}(\mb X)\right\Vert_2\left\{\sum_{a=0}^1 \left\Vert\hat\mu_a(\mb X)-{\mu}_a(\mb X)\right\Vert_2 + \Vert\dot{\lambda}\{\hat{e}(\mb X)\}-\dot{\lambda}\{e(\mb X)\}\Vert_2\right\}. \label{eq:prove 26}
\end{align}
Thus by Condition \ref{cond:EIF-nuisop} in Theorem \ref{thm:RDR-EIF}  {(Section \ref{sec:EIF})}, we can show that 
\begin{align}
\Vert\pr[N(\mb V;\hat{\phi})]-\pr[N(\mb V;\phi_0)]\Vert_2=o_p(1). \label{eq:prove 36}
\end{align}
By Lemma \ref{lem:donsker} {(Section \ref{subapx:prelim-EIF})}, we get
\begin{align}
\sup_{\phi\in l^\infty(\mc{G}_{\phi_0})}\Vert\pr_n[N(\mb V;\phi)]-\pr[N(\mb V;\phi)]\Vert_2 \to_p 0 \label{eq:prove 27}
\end{align}
as $n\to\infty$. Plugging \eqref{eq:prove 36}--\eqref{eq:prove 27} into \eqref{eq:prove 21} leads to \eqref{eq:prove 15}. Combining \eqref{eq:prove 16} and \eqref{eq:prove 15}, we finally complete proving the consistency of $\hat{\gamma}^{\text{eif}}$.

Second, we show the asymptotic normal distribution of $\hat{\gamma}^{\text{eif}}$. Consider  Taylor expansion of ${\pr_n[N(\mb V;\hat\phi)]}/{\pr_n[D(\mb V;\hat\phi)]}$ at $(\pr[N(\mb V;\phi_0)],\pr[D(\mb V;\phi_0)])$ (more explicitly, consider expanding ${x}/{y}$ near $(x_0,y_0)$ for $(x,y)$ close to $(x_0,y_0)$) together with \eqref{eq:prove 16} and \eqref{eq:prove 15}, we have
\begin{align*}
	\frac{\pr_n[N(\mb V;\hat{\phi})]}{\pr_n[D(\mb V;\hat{\phi})]}
	&= \frac{\pr[N(\mb V;\phi_0)]}{\pr[D(\mb V;\phi_0)]} + \frac{1}{\pr[D(\mb V;\phi_0)]} \{\pr_n[N(\mb V;\hat{\phi})] - \pr[N(\mb V;\phi_0)]\} \\ 
    & \quad -
	\frac{\pr[N(\mb V;\phi_0)]}{\{\pr[D(\mb V;\phi_0)]\}^2} \{\pr_n[D(\mb V;\hat{\phi})] - \pr[D(\mb V;\phi_0)]\} \\
	&\quad -\left[\frac{2}{\{\pr[D(\mb V;\phi_0)]\}^2}+o_p(1)\right]\cdot\{\pr_n [N(\mb V;\hat{\phi})] - \pr_n[N(\mb V;\phi_0)]\} \{\pr_n[D(\mb V;\hat{\phi})] - \pr_n[ D(\mb V;\phi_0)] \}\\
	&\quad +\left[\frac{2 \pr[N(\mb V;\phi_0)]}{\{\pr[D(\mb V;\phi_0)]\}^3}+o_p(1)\right] \cdot \{\pr_n[D(\mb V;\hat{\phi})] - \pr_n [D(\mb V;\phi_0)]\}^2 \\
	& = \gamma + \pr_n[\varphi(\mb V;\phi_0)] + \frac{\pr_n[N(\mb V;\hat{\phi})]- \pr_n [N(\mb V;\phi_0)]}{\pr[D(\mb V;\phi_0)]} \\ 
    & \quad
	- \frac{\{\pr_n[D(\mb V;\hat{\phi})] - \pr_n[D(\mb V;\phi_0)]\} \pr[N(\mb V;\phi_0)]}{\{\pr[D(\mb V;\phi_0)]\}^2 }\\
	&\quad -\left[\frac{2}{\{\pr[D(\mb V;\phi_0)]\}^2}+o_p(1)\right]\cdot\{\pr_n [N(\mb V;\hat{\phi})]- \pr_n[N(\mb V;\phi_0)]\} \{\pr_n[D(\mb V;\hat{\phi})] - \pr_n [D(\mb V;\phi_0)] \}\\
	&\quad +\left[\frac{2 \pr[N(\mb V;\phi_0)]}{\{\pr[D(\mb V;\phi_0)]\}^3   }+o_p(1)\right] \cdot \{\pr_n[D(\mb V;\hat{\phi})] - \pr_n[D(\mb V;\phi_0)]\}^2. 
\end{align*}
 By Conditions \ref{cond:EIFbound}, \ref{cond:EIFsup} and \ref{cond:EIFLipch}, we only need to show the following: 
\begin{align}
\pr_n[N(\mb V;\hat{\phi})]- \pr_n[N(\mb V;\phi_0)]& =o_p(n^{-{1}/{2}}), \text{ and}  \label{eq:prove 29}\\
\pr_n[D(\mb V;\hat{\phi})] - \pr_n[D(\mb V;\phi_0)]& =o_p(n^{-{1}/{2}}) \label{eq:prove 30}.
\end{align}
To show \eqref{eq:prove 29}, we need 
\begin{align}
\pr[N(\mb V;\hat{\phi})]- \pr[N(\mb V;\phi_0)]&=o_p(n^{-1/2}), \text{ and} \label{eq:prove 11}\\
(\pr_n-\pr)[N(\mb V;\hat{\phi})] - (\pr_n-\pr) [N(\mb V;\phi_0)]&=o_p(n^{-1/2})\label{eq:prove 12}.
\end{align}

By Condition \ref{cond:EIF-nuisop} in Theorem \ref{thm:RDR-EIF} {(Section \ref{sec:EIF})}, Condition \ref{cond:EIFsup}  and the bound in \eqref{eq:prove 32}, we can get \eqref{eq:prove 11}. Then, Lemma \ref{lem:donsker} {(Section \ref{subapx:prelim-EIF})} further implies
\begin{align*}
	n^{1/2}(\pr_n-\pr)[N(\mb V;\phi)]\rightsquigarrow Z\in l^\infty(\mc{G}_{\phi_0})
\end{align*}
as $n\to\infty$. Furthermore, given the fact that $\Vert\hat{\phi}-\phi_0 \Vert_2=o_p(1)$, we have
\begin{align}\label{eq:CLT-N}
	\begin{pmatrix}
		n^{1/2}(\pr_n-\pr)[N(\mb V;\phi)]\\
		\hat{\phi}
	\end{pmatrix}
	\rightsquigarrow 
	\begin{pmatrix}
		Z\\
		\phi_0
	\end{pmatrix}
\end{align}
in $l^\infty(\mc{G}_{\phi_0})\times\mc{G}_{\phi_0}$ as $n\to\infty$.
Consider function $s:l^\infty(\mc{G}_{\phi_0})\times \mc{G}_{\phi_0}\mapsto \mathbb{R}^p$, where $s(z,\phi)=z(\phi)-z(\phi_0)$, which is continuous for all $(z,\phi)$ and $\phi\mapsto z(\phi)$ is continuous. By Lemma \ref{lem:donsker} {(Section \ref{subapx:prelim-EIF})}, all sample paths of $Z$ are continuous on $\mc{G}_{\phi_0}$ and thus $s(z,\phi)$ is continuous for $(Z,\phi)$. By continuous mapping theorem,
\begin{align}\label{eq:CMT-N}
	s(Z,\hat{\phi})=n^{1/2}(\pr_n-\pr)N(\mb V;\hat{\phi})-n^{1/2}(\pr_n-\pr)N(\mb V;\phi_0)\rightsquigarrow s(Z,\phi_0)=0,
\end{align}
which leads to \eqref{eq:prove 12}. Combining \eqref{eq:prove 11} and \eqref{eq:prove 12}, we can verify \eqref{eq:prove 29}. 
 
Next, \eqref{eq:prove 30} can be shown similarly. Therefore, we obtain
\begin{align*}
	\hat{\gamma}^{\text{eif}}= \pr_n[\varphi(\mb V;\phi_0)] + \gamma + o_p(n^{-1/2}),
\end{align*}
implying that $\hat{\gamma}^{\text{eif}}$ is root-$n$ consistent with
\begin{align*}
	n^{1/2}(\hat{\gamma}^{\text{eif}}-\gamma)\rightsquigarrow \mc N(0,\E[\varphi(\mb V;\phi_0)^2]).
\end{align*}

\section{Theory of DML-based Estimators}\label{apx:theory-DML}
\subsection{Preliminaries and Lemmata}\label{subapx:prelim-DML}

\subsubsection{Theory of Gateaux derivatives}\label{subsubapx:Gate}

For convenience, we first introduce the following knowledge of \textit{Gateaux derivative}. For a function set $\tilde{T}=\{\phi-\phi_0:\phi\in \mc{T}\}$, the Gateaux derivative is a mapping, $D_r:\tilde{T}\mapsto \mathbb{R}$, defined by
\begin{align*}
	D_r[\phi-\phi_0]:=\partial_r\{\E[L(\mb V;\gamma_0,\phi+r(\phi-\phi_0))]\}, 
\end{align*}
where $\phi\in \mc{T}$ for all $r\in[0,1)$. We also denote 
\begin{align*}
\partial_\phi\E[L(\mb V;\gamma_0,\phi_0)][\phi-\phi_0]:=D_0[\phi-\phi_0], \quad \mbox{ for }\phi\in \mc{T},
\end{align*}
which is well-defined given that $\mc{T}$ is a convex set, and it holds that for all $r\in[0,1)$ and $\phi\in \mc{T}$,
\begin{align*}
	\phi_0+r(\phi-\phi_0)=(1-r)\phi_0+r\phi\in \mc{T}.
\end{align*}

We first introduce the following two Lemmata, \ref{lem:firstGate} and \ref{lem:secondGate}, for computing the first- and second-order Gateaux derivatives of $\tilde T\in\mc T$ when $\phi=(\tilde e, \tilde\mu_0,\tilde\mu_1)$ chosen in Section \ref{sec:DML}. 

\begin{lemma}[First-order Gateaux]  \label{lem:firstGate}For any $\phi\in \mc{T}$, the Gateaux derivative in the direction $\phi-\phi_0$ is defined as 
\begin{align*}
	\partial_\phi \E[L(\mb V;\gamma_0,\phi_0)][\phi-\phi_0]:=\partial_r \E[L(\mb V;\gamma_0,\phi_r)]\big\vert_{r=0},
\end{align*}
for $r\in[0,1)$. Following the notations in Section \ref{sec:DML}, if we choose $\phi=(\tilde{e},\tilde{\mu}_0,\tilde{\mu}_1)$, $\phi-\phi_0=(\tilde{e}-e,\tilde{\mu}_0-\mu_0,\tilde{\mu}_1-\mu_1)$, and $\phi_r =\phi_0+r(\phi-\phi_0)=(e_r, \mu_{0,r},\mu_{1,r} )$, where $e_r=e+r(\tilde{e}-e)$, $\mu_{0,r}=\mu_0+r(\tilde{\mu}_0-\mu_0)$ and
$\mu_{1,r}=\mu_1+r(\tilde{\mu}_1-\mu_1)$, then we have the first-order Gateaux derivative evaluated at $r=0$ is $\partial_\phi \E[L(\mb V;\gamma_0,\phi_0)][\phi-\phi_0]=0$.
\end{lemma}

\begin{proof}
We consider $L(\mb V;\gamma_0,\phi_r) = \gamma_0D(\mb V;\phi_r) + N(\mb V;\phi_r)$ is linear in $\gamma_0$. 
First, by interchanging the order of derivative and expectation, we can write
\begin{align}\label{eq:firstGate}
\frac{\partial}{\partial r}\E[L(\mb V;\gamma_0,\phi_r)]
=\gamma_0\E\left[\frac{\partial}{\partial r} D(\mb V;\phi_r)\right]-\E\left[\frac{\partial}{\partial r}N(\mb V;\phi_r)\right]. 
\end{align}
We further expand the first term in the right-hand side of \eqref{eq:firstGate}. The inner part of the expectation can be expressed by
\begin{align}
\frac{\partial}{\partial r} D(\mb V;\phi_r) 
&=\frac{\partial}{\partial r} [\lambda\{e_r(\mb X)\}] + \{A-e_r(\mb X)\}\frac{\partial}{\partial r}[\dot{\lambda}\{e_r(\mb X)\}] + \dot{\lambda}\{e_r(\mb X)\}\cdot\frac{\partial}{\partial r}\{A-e_r(\mb X)\}\nonumber
\\
&=\dot{\lambda}\{e_r(\mb X)\}\{\tilde{e}(\mb X)-e(\mb X)\}+\ddot{\lambda}\{e_r(\mb X)\}\{\tilde{e}(\mb X)-e(\mb X)\}\{A-e_r(\mb X)\}\nonumber\\
& \quad -\dot{\lambda}\{e_r(\mb X)\}\{\tilde{e}(\mb X)-e(\mb X)\}\nonumber
\\
&=\ddot{\lambda}\{e_r(\mb X)\}\{\tilde{e}(\mb X)-e(\mb X)\}\{A-e_r(\mb X)\}, \label{eq:firstDer-ga}
\end{align}
hence
\begin{align*}
\E\left[\frac{\partial}{\partial r} D(\mb V;\phi_r)\right]
=\E[\ddot{\lambda}\{e_r(\mb X)\}\{\tilde{e}(\mb X)-e(\mb X)\}\{e(\mb X)-e_r(\mb X)\}]
=-r\cdot\E[\ddot{\lambda}\{e_r(\mb X)\}\{\tilde{e}(\mb X)-e(\mb X)\}^2],
\end{align*}
and so 
\begin{align}
\E\left[\frac{\partial}{\partial r} D(\mb V;\phi_r)\right]\bigg\vert_{r=0}=0.\label{eq:firstDer-ga-value}
\end{align}
Next, since
\begin{align}\label{eq:firstDer-gb}
\frac{\partial}{\partial r} N(\mb V;\phi_r) & =\lambda\{e_r(\mb X)\}\left\{\frac{\partial}{\partial r}\psi_{\tau}(\mb V;\phi_r)\right\}\nonumber  \\
 & \quad + \psi_{\tau}(\mb V;\phi_r)\frac{\partial}{\partial r}\lambda\{e_r(\mb X)\} + \left[\frac{\partial}{\partial r} {\dot{\lambda}}\{e_r(\mb X)\}\right] \{\mu_{1,r}(\mb X)-\mu_{0,r}(\mb X)\}\{A-e_r(\mb X)\} \nonumber\\
& \quad +
\dot{\lambda}\{e_r(\mb X)\}\left[\frac{\partial}{\partial r} \{\mu_{1,r}(\mb X)-\mu_{0,r}(\mb X)\}\right]\{A-e_r(\mb X)\} \nonumber \\
& \quad +\dot{\lambda}\{e_r(\mb X)\}\{\mu_{1,r}(\mb X)-\mu_{0,r}(\mb X)\}\frac{\partial}{\partial r}\{A-e_r(\mb X)\}, 
\end{align}
we decompose the second term in the right-hand side of \eqref{eq:firstGate} as
\begin{align}
\E\left[\frac{\partial}{\partial r} N(\mb V;\phi_r)\right]
&= \mc{I}_1 + \mc{I}_2 + \mc{I}_3 + \mc{I}_4 + \mc{I}_5, \label{eq:decompose-gb}
\end{align}
where 
\begin{align*}
\mc{I}_1 &=\E\left[\lambda\{e_r(\mb X)\}\cdot\frac{\partial}{\partial r}\psi_{\tau}(\mb V;\phi_r)\right], \\
\mc{I}_2 &= \E\left[\psi_{\tau}(\mb V;\phi_r)\cdot\frac{\partial}{\partial r}\lambda\{e_r(\mb X)\}\right],\\
\mc{I}_3  &=\E\left(\left[\frac{\partial}{\partial r} {\dot{\lambda}}\{e_r(\mb X)\}\right]\cdot \{\mu_{1,r}(\mb X)-\mu_{0,r}(\mb X)\}\{A-e_r(\mb X)\}\right),\\
\mc{I}_4  &=
\E\left(\dot{\lambda}\{e_r(\mb X)\}\cdot \left[\frac{\partial}{\partial r} \{\mu_{1,r}(\mb X)-\mu_{0,r}(\mb X)\}\right]\cdot\{A-e_r(\mb X)\}\right),\\
\mc{I}_5 &=\E\left(\dot{\lambda}\{e_r(\mb X)\}\{\mu_{1,r}(\mb X)-\mu_{0,r}(\mb X)\}\cdot \frac{\partial}{\partial r}\{A-e_r(\mb X)\}\right). 
\end{align*}
Note that in $\mc I_1$,
\begin{align}
\frac{\partial}{\partial r}\psi_{\tau}(\mb V;\phi_r)
&= \{Y-\mu_{1,r}(\mb X)\}\cdot\frac{\partial}{\partial r}\left\{\frac{A}{e_r(\mb X)}\right\} -\frac{A}{e_r(\mb X)}\cdot \frac{\partial}{\partial r}\mu_{1,r}(\mb X) -\{Y-\mu_{0,r}(\mb X)\}\cdot\frac{\partial}{\partial r}\left\{\frac{1-A}{1-e_r(\mb X)}\right\}\nonumber\\
& \quad+\frac{1-A}{1-e_r(\mb X)}\cdot \frac{\partial}{\partial r}\mu_{0,r}(\mb X) +\frac{\partial}{\partial r}\{\mu_{1,r}(\mb X)-\mu_{0,r}(\mb X)\}\nonumber\\
&=-\frac{A}{ e_r(\mb X)^2 }\{\tilde{e}(\mb X)-e(\mb X)\}\{Y-\mu_{1,r}(\mb X)\}
-\frac{A}{e_r(\mb X)}\{\tilde{\mu}_1(\mb X)-\mu_1(\mb X)\}\nonumber\\
&\quad -\frac{1-A}{ \{1-e_r(\mb X)\}^2 }\{\tilde{e}(\mb X)-e(\mb X)\}\{Y-\mu_{0,r}(\mb X)\}\nonumber\\
&\quad +\frac{1-A}{1-e_r(\mb X)}\{\tilde{\mu}_0(\mb X)-\mu_0(\mb X)\}+[\{\tilde{\mu}_1(\mb X)-\mu_1(\mb X)\} - \{\tilde{\mu}_0(\mb X)-\mu_0(\mb X)\}] 
\label{eq:I1-tau-psiDer}.
\end{align}

Therefore, 
\begin{align*}
	\mc I_1 & = -r\cdot \E\left[\lambda\{e_r(\mb X)\} \frac{e(\mb X)}{e_r(\mb X)^2}\{\tilde{e}(\mb X)-e(\mb X)\}\{\tilde{\mu}_1(\mb X)-\mu_1(\mb X)\}\right] -
	\E\left[\lambda\{e_r(\mb X)\}\frac{e(\mb X)}{e_r(\mb X)}\{\tilde{\mu}_1(\mb X)-\mu_1(\mb X)\}\right]\\
&\quad +r\cdot\E\left[\lambda\{e_r(\mb X)\}\frac{1-e(\mb X)}{ \{1-e_r(\mb X)\}^2 } \{\tilde{e}(\mb X)-e(\mb X)\}\{\tilde{\mu}_0(\mb X)-\mu_0(\mb X)\}\right] \\
&\quad
+\E\left[\lambda\{e_r(\mb X)\}\frac{1-e(\mb X)}{1-e_r(\mb X)}\{\tilde{\mu}_0(\mb X)-\mu_0(\mb X)\}\right] +\E(\lambda\{e_r(\mb X)\}[\{\tilde{\mu}_1(\mb X)-\mu_1(\mb X)\} - \{\tilde{\mu}_0(\mb X)-\mu_0(\mb X)\}])\\ 
& = -r\cdot \E\bigg[\lambda\{e_r(\mb X)\} \frac{e(\mb X)}{e_r(\mb X)^2}\{\tilde{e}(\mb X)-e(\mb X)\}\{\tilde{\mu}_1(\mb X)-\mu_1(\mb X)\} \\ 
& \quad -\lambda\{e_r(\mb X)\}\frac{1-e(\mb X)}{ \{1-e_r(\mb X)\}^2 } \{\tilde{e}(\mb X)-e(\mb X)\}\{\tilde{\mu}_0(\mb X)-\mu_0(\mb X)\}\bigg] \\
& \quad + \underbrace{\E\left[\lambda\{e_r(\mb X)\}\{\tilde{\mu}_1(\mb X)-\mu_1(\mb X)\}-\lambda\{e_r(\mb X)\}\frac{e(\mb X)}{e_r(\mb X)}\{\tilde{\mu}_1(\mb X)-\mu_1(\mb X)\}\right]}_{ =~ 0\text{ when }r~=~0} \\
& \quad + \underbrace{\E\left[\lambda\{e_r(\mb X)\}\frac{1-e(\mb X)}{ \{1-e_r(\mb X)\}^2 }\{\tilde{\mu}_0(\mb X)-\mu_0(\mb X)\} - \lambda\{e_r(\mb X)\}\{\tilde{\mu}_0(\mb X)-\mu_0(\mb X)\}\right]}_{ = ~0\text{ when }r~=~0}. 
\end{align*}
In addition, it is straightforward to show
\begin{align*}
	\mc I_3 & =-r\cdot \E[\ddot{\lambda}\{e_r(\mb X)\}\{\tilde{e}(\mb X)-e(\mb X)\}^2 \{\mu_{1,r}(\mb X)-\mu_{0,r}(\mb X)\}], \text{ and }\\
 \mc I_4 & = -r \cdot  \E(\dot{\lambda}\{e_r(\mb X)\}\cdot   [\{\tilde{\mu}_{1}(\mb X)-\mu_{1}(\mb X)\} -\{\tilde{\mu}_0(\mb X)-\mu_0(\mb X)\}  ]\cdot\{\tilde{e}(\mb X)-e(\mb X)\} ). 
\end{align*}
Therefore, 
\begin{align*}
\mc I_1\big\vert_{r=0} = \mc I_3\big\vert_{r=0} = \mc I_4\big\vert_{r=0}= 0.
\end{align*}
Moreover,
\begin{align*}
	\mc I_2 & =\E[\dot{\lambda}\{e_r(\mb X)\}\{\tilde{e}(\mb X)-e(\mb X)\}\psi_{\tau}(\mb V;\phi_r)]\\
	& = \E[\dot{\lambda}\{e_r(\mb X)\}\{\tilde{e}(\mb X)-e(\mb X)\}\{\mu_{1,r}(\mb X)-\mu_{0,r}(\mb X)\}] \\
    & \quad - r\cdot \E \left[\dot{\lambda}\{e_r(\mb X)\}\{\tilde{e}(\mb X)-e(\mb X)\}\frac{e(\mb X)}{e_r(\mb X)}\{\tilde{\mu}_1(\mb X)-\mu_{1}(\mb X)\}\right] \\
    & \quad +r\cdot \E \left[\dot{\lambda}\{e_r(\mb X)\}\{\tilde{e}(\mb X)-e(\mb X)\}\frac{1-e(\mb X)}{1-e_r(\mb X)}\{\tilde{\mu}_0(\mb X)-\mu_{0}(\mb X)\}\right],
\end{align*}
therefore 
\begin{align*}
\mc I_2\big\vert_{r=0}
=\E[\dot{\lambda}\{e(\mb X)\}\{\tilde{e}(\mb X)-e(\mb X)\}\{\mu_{1}(\mb X)-\mu_{0}(\mb X)\}]. 
\end{align*}
It is also straightforward to get
\begin{align*}
\mc I_5\big\vert_{r=0}= -\E[\dot{\lambda}\{e(\mb X)\}\{\mu_{1}(\mb X)-\mu_{0}(\mb X)\}\{\tilde{e}(\mb X)-e(\mb X)\}].
\end{align*}
Therefore, we note that $\mc I_5\big\vert_{r=0} + \mc I_2\big\vert_{r=0}=0$, and so 
\begin{align}
\E\left[\frac{\partial }{\partial r}N(\mb V;\phi_r)\right]\bigg\vert_{r=0} = 
\mc I_1\big\vert_{r=0} + \mc I_2\big\vert_{r=0} + \mc I_3\big\vert_{r=0} + \mc I_4\big\vert_{r=0} + \mc I_5\big\vert_{r=0} = 0. \label{eq:decompose-gb-r0value}
\end{align}
Finally, plugging \eqref{eq:firstDer-ga-value} and \eqref{eq:decompose-gb-r0value} in \eqref{eq:firstGate}, we get $\dfrac{\partial}{\partial r} \E[L(\mb V;\gamma_0,\phi_r)]\big\vert_{r=0}
=0$, which completes the proof.
\end{proof}

\begin{lemma}[Second-order Gateaux] \label{lem:secondGate} Following notations in Lemma \ref{lem:firstGate}, for any $\phi=(\tilde{e},\tilde{\mu}_0,\tilde{\mu}_1)\in \mc{T}$, the second-order Gateaux derivative in the direction $\phi-\phi_0=(\tilde{e}-e,\tilde{\mu}_0-\mu_0,\tilde{\mu}_1-\mu_1)$ is given by 
	\begin{align}
	& \partial_r^2 \E[L(\mb V;\gamma_0,\phi_r)]
	= -r\cdot\gamma_0 \E[\lambda^{(3)}\{e_r(\mb X)\}\{\tilde{e}(\mb X)-e(\mb X)\}^3]-\gamma_0\E[\ddot{\lambda}\{e_r(\mb X)\}\{\tilde{e}(\mb X)-e(\mb X)\}^2]  \nonumber\\
	& \quad +2r\cdot \E\left(\dot{\lambda}\{e_r(\mb X)\}\cdot\{\tilde{e}(\mb X)-e(\mb X)\}^2\left[\frac{e(\mb X)}{e_r(\mb X)^2}\{\tilde{\mu}_1(\mb X)-\mu_1(\mb X)\} - \frac{1-e(\mb X)}{\{1-e_r(\mb X)\}^2}\{\tilde{\mu}_0(\mb X)-\mu_0(\mb X)\}\right]\right)\nonumber\\
	& \quad
	+2\E\left(\dot{\lambda}\{e_r(\mb X)\}\cdot\{\tilde{e}(\mb X)-e(\mb X)\}\left[\frac{e(\mb X)}{e_r(\mb X)}\{\tilde{\mu}_1(\mb X)-\mu_1(\mb X)\} - \frac{1-e(\mb X)}{1-e_r(\mb X)}\{\tilde{\mu}_0(\mb X)-\mu_0(\mb X)\}\right]\right) \nonumber\\
	&\quad +r\cdot \E\left(\ddot{\lambda}\{e_r(\mb X)\}\cdot \{\tilde{e}(\mb X)-e(\mb X)\}^2\left[\frac{e(\mb X)}{e_r(\mb X)}\{\tilde{\mu}_1(\mb X)-\mu_1(\mb X)\} -\frac{1-e(\mb X)}{1-e_r(\mb X)}\{\tilde{\mu}_0(\mb X)-\mu_0(\mb X)\}\right]\right)  \nonumber\\
	&\quad +2r\cdot \E\left(\lambda\{e_r(\mb X)\}\cdot\{\tilde{e}(\mb X)-e(\mb X)\}^2\left[\frac{e(\mb X) }{ e_r(\mb X)^3 } \{\tilde{\mu}_1(\mb X)-\mu_1(\mb X)\} - \frac{1-e(\mb X)}{\{1-e_r(\mb X)\}^3 } \{\tilde{\mu}_0-\mu_0(\mb X)\}\right]\right)\nonumber\\
    &\quad
    +2\E\left( \lambda\{e_r(\mb X)\} \cdot\{\tilde{e}(\mb X)-e(\mb X)\}\left[\frac{e(\mb X)}{ e_r(\mb X)^2 } \{\tilde{\mu}_1(\mb X)-\mu_1(\mb X)\} - \frac{1-e(\mb X)}{ \{1-e_r(\mb X)\}^2 }\{\tilde{\mu}_0(\mb X)-\mu_0(\mb X)\}\right]\right)\nonumber\\
    &\quad+\E[ \ddot{\lambda}\{e_r(\mb X)\}\cdot\{\tilde{e}(\mb X)-e(\mb X)\}^2\{\mu_1(\mb X)-\mu_0(\mb X)\}]\nonumber\\
    &\quad+r\cdot \E[\lambda^{(3)}\{e_r(\mb X)\}\cdot\{\tilde{e}(\mb X)-e(\mb X)\}^3\{\mu_{1}(\mb X)-\mu_0(\mb X)\}]\nonumber\\
    &\quad+3r\cdot\E(\ddot{\lambda}\{e_r(\mb X)\}\{\tilde{e}(\mb X)-e(\mb X)\}^2[\{\tilde{\mu}_1(\mb X)-\mu_1(\mb X)\}-\{\tilde{\mu}_0(\mb X)-\mu_0(\mb X)\}]) \nonumber\\
	&\quad +r^2\cdot \E(\lambda^{(3)}\{e_r(\mb X)\}\cdot\{\tilde{e}(\mb X)-e(\mb X)\}^3[\{\tilde{\mu}_{1}(\mb X)-\mu_1(\mb X)\} 
    -\{\tilde{\mu}_{0}(\mb X)-\mu_0(\mb X)\}]). \label{eq:secGate-formula}
	\end{align}
\end{lemma}

\begin{proof}
We consider $L(\mb V;\gamma_0,\phi_r) = \gamma_0D(\mb V;\phi_r) + N(\mb V;\phi_r)$ is linear in $\gamma_0$. By interchanging the differentiation and expectation, the second-order Gateaux derivative can be rewritten as
\begin{align}
\frac{\partial^2}{\partial r^2} \E[L(\mb V;\gamma_0,\phi_r)]
=\gamma_0 \E\left[\frac{\partial^2}{\partial r^2} D(\mb V;\phi_r)\right]-\E\left[\frac{\partial^2 }{\partial r^2}N(\mb V;\phi_r)\right], \label{eq:secondGate}
\end{align}
for $r\in[0,1)$ and $\phi_r =\phi_0+r(\phi-\phi_0)=(e_r, \mu_{0,r},\mu_{1,r})$. First, we further express inner parts of expectations in \eqref{eq:firstDer-ga} as follows. 
\begin{align*}
	\frac{\partial^2}{\partial r^2} D(\mb V;\phi_r)
	&=\lambda^{(3)}\{e_r(\mb X)\}\{\tilde{e}(\mb X)-e(\mb X)\}^2\{A-e_r(\mb X)\}-\ddot{\lambda}\{e_r(\mb X)\}\{\tilde{e}(\mb X)-e(\mb X)\}^2, 
\end{align*}
which yields
\begin{align}
\E\left[\frac{\partial^2}{\partial r^2}
D(\mb V;\phi_r)\right] 
& =\E[\lambda^{(3)}\{e_r(\mb X)\}\{\tilde{e}(\mb X)-e(\mb X)\}^2\{e(\mb X)-e_r(\mb X)\}-\ddot{\lambda}\{e_r(\mb X)\}\{\tilde{e}(\mb X)-e(\mb X)\}^2]\nonumber\\
& = -r\cdot\E[\lambda^{(3)}\{e_r(\mb X)\}\{\tilde{e}(\mb X)-e(\mb X)\}^3]-\E[\ddot{\lambda}\{e_r(\mb X)\}\{\tilde{e}(\mb X)-e(\mb X)\}^2], 
\label{eq:secondGate-ga}
\end{align}
since $e(\mb X)-e_r(\mb X) = -r\cdot\{\tilde e(\mb X)-e(\mb X)\}$. 
In addition, similar to \eqref{eq:decompose-gb}, we decompose
\begin{align*}
\E\left[\frac{\partial^2}{\partial r^2} N(\mb V;\phi_r)\right] = \mc{I}_6 + \mc{I}_7+\mc{I}_8+\mc{I}_9+\mc{I}_{10}+\mc{I}_{11}, 
\end{align*}
where 
\begin{align*}
	\mc{I}_6 &= \E\left[\frac{\partial}{\partial r}\lambda\{e_r(\mb X)\}\cdot \frac{\partial}{\partial r}\psi_{\tau}(\mb V;\phi_r)\right],\\
	\mc{I}_7 &= \E\left[\lambda\{e_r(\mb X)\}\cdot \frac{\partial^2}{\partial r^2}\psi_{\tau}(\mb V;\phi_r)\right],\\
	\mc{I}_8 &= \E\left[\psi_{\tau}(\mb V;\phi_r)\cdot\frac{\partial^2}{\partial r^2}\lambda\{e_r(\mb X)\}\right] +
	\E\left[\frac{\partial }{\partial r}\psi_{\tau}(\mb V;\phi_r)\cdot\frac{\partial}{\partial r}\lambda\{e_r(\mb X)\}\right],\\
	\mc{I}_9 &= \E\left(\left[\frac{\partial^2}{\partial r^2} \dot{\lambda}\{e_r(\mb X)\}\right]\cdot\{\mu_{1,r}(\mb X)-\mu_{0,r}(\mb X)\}\{A-e_r(\mb X)\}\right)\\
	&\quad + \E\left(\left[\frac{\partial}{\partial r} {\dot{\lambda}}\{e_r(\mb X)\}\right]\cdot\left[\frac{\partial}{\partial r} \{\mu_{1,r}(\mb X)-\mu_{0,r}(\mb X)\}\right] \cdot \{A-e_r(\mb X)\}\right)\\
	&\quad +\E\left(\left[\frac{\partial}{\partial r}{\dot{\lambda}}\{e_r(\mb X)\}\right]\cdot \{\mu_{1,r}(\mb X)-\mu_{0,r}(\mb X)\}\cdot \left[\frac{\partial}{\partial r}\{A-e_r(\mb X)\}\right]\right),\\
	\mc{I}_{10} & = \E\left(\left[\frac{\partial}{\partial r}\dot{\lambda}\{e_r(\mb X)\}\right]\cdot \left[\frac{\partial}{\partial r} \{\mu_{1,r}(\mb X)-\mu_{0,r}(\mb X)\}\right]\cdot\{A-e_r(\mb X)\} \right)\\
	&\quad
    +\E\left(\dot{\lambda}\{e_r(\mb X)\}\cdot \left[\frac{\partial^2}{\partial r^2} \{\mu_{1,r}(\mb X)-\mu_{0,r}(\mb X)\}\right]\cdot\{A-e_r(\mb X)\}\right)\\
	&\quad +\E\left(\dot{\lambda}\{e_r(\mb X)\}\cdot\left[\frac{\partial}{\partial r} \{\mu_{1,r}(\mb X)-\mu_{0,r}(\mb X)\}\right]\cdot\left[\frac{\partial}{\partial r}\{A-e_r(\mb X)\}\right]\right),\\
	\mc{I}_{11}&= \E\left(\left[\frac{\partial}{\partial r}\dot{\lambda}\{e_r(\mb X)\}\right]\cdot\{\mu_{1,r}(\mb X)-\mu_{0,r}(\mb X)\}\cdot \left[\frac{\partial}{\partial r}\{A-e_r(\mb X)\}\right]\right)\\
	&\quad
    + \E\left(\dot{\lambda}\{e_r(\mb X)\}\cdot\left[\frac{\partial}{\partial r}\{\mu_{1,r}(\mb X)-\mu_{0,r}(\mb X)\}\right]\cdot \left[\frac{\partial}{\partial r}\{A-e_r(\mb X)\}\right]\right)\\
	& \quad +\E\left(\dot{\lambda}\{e_r(\mb X)\}\{\mu_{1,r}(\mb X)-\mu_{0,r}(\mb X)\}\cdot\left[\frac{\partial^2}{\partial r^2}\{A-e_r(\mb X)\}\right]\right).
\end{align*} 
By \eqref{eq:I1-tau-psiDer}, $\mc{I}_6$ can be rewritten as
\begin{align}
\mc{I}_6 
&= \E \left[\dot{\lambda}\{e_r(\mb X)\}\frac{A}{ e_r(\mb X)^2 }\{\tilde{e}(\mb X)-e(\mb X)\}^2\{Y-\mu_{1,r}(\mb X)\}\right] \\
	&\quad-\E\left[\dot{\lambda}\{e_r(\mb X)\}\frac{A}{e_r(\mb X)}\{\tilde{\mu}_1(\mb X)-\mu_1(\mb X)\}\{\tilde{e}(\mb X)-e(\mb X)\}\right]
\nonumber\\
&\quad -\E\left[\dot{\lambda}\{e_r(\mb X)\}\frac{1-A}{ \{1-e_r(\mb X)\}^2 }\{\tilde{e}(\mb X)-e(\mb X)\}^2\{Y-\mu_{0,r}(\mb X)\}\right] \\
	&\quad+\E\left[\dot{\lambda}\{e_r(\mb X)\}\frac{1-A}{1-e_r(\mb X)}\{\tilde{\mu}_0(\mb X)-\mu_0(\mb X)\}\{\tilde{e}(\mb X)-e(\mb X)\}\right]\nonumber\\
&\quad +\E(\dot{\lambda}\{e_r(\mb X)\}\cdot[\{\tilde{\mu}_1(\mb X)-\mu_1(\mb X)\} - \{\tilde{\mu}_0(\mb X)-\mu_0(\mb X)\}]\cdot\{\tilde{e}(\mb X)-e(\mb X)\}) \nonumber\\
&= -r\cdot\E\left[\dot{\lambda}\{e_r(\mb X)\}\frac{e(\mb X)}{ e_r(\mb X)^2 }\{\tilde{e}(\mb X)-e(\mb X)\}^2\{\tilde\mu_1(\mb X)-\mu_1(\mb X)\}\right] \\
	&\quad-\E\left[\dot{\lambda}\{e_r(\mb X)\}\frac{e(\mb X)}{e_r(\mb X)}\{\tilde{\mu}_1(\mb X)-\mu_1(\mb X)\}\{\tilde{e}(\mb X)-e(\mb X)\}\right]  \nonumber\\
&\quad+r\cdot\E\left[\dot{\lambda}\{e_r(\mb X)\}\frac{1-e(\mb X)}{ \{1-e_r(\mb X)\}^2 }\{\tilde{e}(\mb X)-e(\mb X)\}^2\{\tilde\mu_0(\mb X)-\mu_{0}(\mb X)\}\right]\\
	&\quad+\E\left[\dot{\lambda}\{e_r(\mb X)\}\frac{1-e(\mb X)}{1-e_r(\mb X)}\{\tilde{\mu}_0(\mb X)-\mu_0(\mb X)\}\{\tilde{e}(\mb X)-e(\mb X)\}\right]\nonumber\\
&\quad +\E(\dot{\lambda}\{e_r(\mb X)\}\cdot[\{\tilde{\mu}_1(\mb X)-\mu_1(\mb X)\} - \{\tilde{\mu}_0(\mb X)-\mu_0(\mb X)\}]\cdot\{\tilde{e}(\mb X)-e(\mb X)\}). 
\label{eq:I6-form}
\end{align}
For $\mc I_7$, by \eqref{eq:I1-tau-psiDer} again, we can first get 
\begin{align*}
\frac{\partial^2}{\partial r^2}\psi_r(\mb V;\phi_r)
	&=\frac{2A}{ e_r(\mb X)^3 }\{\tilde{e}(\mb X)-e(\mb X)\}^2 \{Y-\mu_{1,r}(\mb X)\}\\
    & \quad + \frac{2A}{ e_r(\mb X)^2 }\{\tilde{e}(\mb X)-e(\mb X)\} \{\tilde{\mu}_{1}(\mb X)-\mu_1(\mb X)\}\\
	&\quad -\frac{2(1-A)}{ \{1-e_r(\mb X)\}^3 }\{\tilde{e}(\mb X)-e(\mb X)\}^2\{Y-\mu_{0,r}(\mb X)\}\\
    & \quad +\frac{2(1-A)}{\{1-e_r(\mb X)\}^2}\{\tilde{e}(\mb X)-e(\mb X)\}\{\tilde{\mu}_0(\mb X)-\mu_{0}(\mb X)\}, 
\end{align*}
and so 
\begin{align}
\mc{I}_7
&=-2r\cdot\E\left[\lambda\{e_r(\mb X)\}\frac{e(\mb X)}{ e_r(\mb X)^3 }\{\tilde{e}(\mb X)-e(\mb X)\}^2 \{\tilde\mu_1(\mb X)-\mu_{1}(\mb X)\}\right]\\
	&\quad
-2\E\left[\lambda\{e_r(\mb X)\}\frac{e(\mb X)}{ e_r(\mb X)^2 }\{\tilde{e}(\mb X)-e(\mb X)\} \{\tilde{\mu}_{1}(\mb X)-\mu_1(\mb X)\}\right] \nonumber\\
&\quad +2r\cdot\E\left[\lambda\{e_r(\mb X)\}\frac{\{1-e(\mb X)\}}{ \{1-e_r(\mb X)\}^3 }\{\tilde{e}(\mb X)-e(\mb X)\}^2\{\mu_0(\mb X)-\mu_{0,r}(\mb X)\}\right]\\
	&\quad+2\E\left[\lambda\{e_r(\mb X)\}\frac{\{1-e(\mb X)\}}{\{1-e_r(\mb X)\}^2}\{\tilde{e}(\mb X)-e(\mb X)\}\{\tilde{\mu}_0(\mb X)-\mu_{0}(\mb X)\}\right]. \label{eq:I7-form}
\end{align}
Similarly, we can further express $\mc{I}_8$ as
\begin{align}
\mc{I}_8
&= \E\bigg(\left[\frac{A}{e_r(\mb X)}\{Y-\mu_{1,r}(\mb X)\}-\frac{1-A}{1-e_r(\mb X)}\{Y-\mu_{0,r}(\mb X)\}+\{\mu_{1,r}(\mb X)-\mu_{0,r}(\mb X)\}\right] \ddot{\lambda}\{e_r(\mb X)\}\{\tilde{e}(\mb X)-e(\mb X)\}^2\bigg) \nonumber\\
&\quad -\E\left[\frac{A}{ e_r(\mb X)^2 }\{\tilde{e}(\mb X)-e(\mb X)\}^2\{Y-\mu_{1,r}(\mb X)\}\dot{\lambda}\{e_r(\mb X)\}\right] \nonumber \\
& \quad -\E\left[\frac{A}{e_r(\mb X)}\{\tilde{\mu}_1(\mb X)-\mu_1(\mb X)\}\cdot \dot{\lambda}\{e_r(\mb X)\}\{\tilde{e}(\mb X)-e(\mb X)\}\right] \nonumber \\
& \quad -\E\left[\frac{1-A}{ \{1-e_r(\mb X)\}^2 }\{\tilde{e}(\mb X)-e(\mb X)\}^2\{Y-\mu_{0,r}(\mb X)\}\cdot \dot{\lambda}\{e_r(\mb X)\}\right] \nonumber \\ & \quad +\E\left[\frac{1-A}{1-e_r(\mb X)}\{\tilde{\mu}_0(\mb X)-\mu_0(\mb X)\}\cdot \dot{\lambda}\{e_r(\mb X)\}\{\tilde{e}(\mb X)-e(\mb X)\}\right] \nonumber\\
&\quad +\E(\dot{\lambda}\{e_r(\mb X)\}\cdot[\{\tilde{\mu}_1(\mb X)-\mu_1(\mb X)\} - \{\tilde{\mu}_0(\mb X)-\mu_0(\mb X)\}]\cdot\{\tilde{e}(\mb X)-e(\mb X)\})  \nonumber\\
&=-r\cdot\E\left[\frac{e(\mb X)}{e_r(\mb X)}\{\tilde\mu_1(\mb X)-\mu_{1}(\mb X)\} \cdot\ddot{\lambda}\{e_r(\mb X)\}\{\tilde{e}(\mb X)-e(\mb X)\}^2\right] \nonumber\\
	&\quad +r\cdot\E\left[\frac{1-e(\mb X)}{1-e_r(\mb X)}\{\tilde\mu_0(\mb X)-\mu_{0}(\mb X)\}  \cdot\ddot{\lambda}\{e_r(\mb X)\}\{\tilde{e}(\mb X)-e(\mb X)\}^2\right] \nonumber\\
&\quad +r\cdot\E\left[\frac{e(\mb X)}{ e_r(\mb X)^2 }\{\tilde{e}(\mb X)-e(\mb X)\}^2\{\tilde\mu_1(\mb X)-\mu_{1}(\mb X)\}\dot{\lambda}\{e_r(\mb X)\}\right] \nonumber\\
	&\quad -\E\left[\frac{e(\mb X)}{e_r(\mb X)}\{\tilde{\mu}_1(\mb X)-\mu_1(\mb X)\}\cdot \dot{\lambda}\{e_r(\mb X)\}\{\tilde{e}(\mb X)-e(\mb X)\}\right] \nonumber\\
&\quad +r\cdot\E\left[\frac{1-e(\mb X)}{ \{1-e_r(\mb X)\}^2 }\{\tilde{e}(\mb X)-e(\mb X)\}^2\{\tilde\mu_0(\mb X)-\mu_{0}(\mb X)\}\cdot \dot{\lambda}\{e_r(\mb X)\}\right] \nonumber\\
	&\quad +\E\left[\frac{1-e(\mb X)}{1-e_r(\mb X)}\{\tilde{\mu}_0(\mb X)-\mu_0(\mb X)\}\cdot \dot{\lambda}\{e_r(\mb X)\}\{\tilde{e}(\mb X)-e(\mb X)\}\right]\nonumber\\
&\quad +\E(\dot{\lambda}\{e_r(\mb X)\}\cdot[\{\tilde{\mu}_1(\mb X)-\mu_1(\mb X)\} - \{\tilde{\mu}_0(\mb X)-\mu_0(\mb X)\}]\cdot\{\tilde{e}(\mb X)-e(\mb X)\}). \label{eq:I8-form}
\end{align}
Moreover, 
\begin{align}
\mc{I}_9 
&=\E [\lambda^{(3)}\{e_r(\mb X)\}\cdot \{\tilde{e}(\mb X)-e(\mb X)\}^2\cdot \{\mu_{1,r}(\mb X)-\mu_{0,r}(\mb X)\}\{e(\mb X)-e_r(\mb X)\}] \nonumber\\
&\quad +\E (\ddot{\lambda}\{e_r(\mb X)\}\cdot\{\tilde{e}(\mb X)-e(\mb X)\}\cdot [\{\tilde{\mu}_{1}(\mb X)-\mu_{1}(\mb X)\} - \{\tilde{\mu}_{0}(\mb X)-\mu_{0}(\mb X)\}]\cdot \{e(\mb X)-e_r(\mb X)\} ) \nonumber\\
&\quad -\E[{\ddot{\lambda}}\{e_r(\mb X)\}\cdot \{\tilde{e}(\mb X)-e(\mb X)\}^2 \cdot\{\mu_{1,r}(\mb X)-\mu_{0,r}(\mb X)\}]\nonumber \\
&=-r\cdot\E [\lambda^{(3)}\{e_r(\mb X)\}\cdot \{\tilde{e}(\mb X)-e(\mb X)\}^3\cdot \{\mu_{1,r}(\mb X)-\mu_{0,r}(\mb X)\}] \nonumber\\
&\quad -r\cdot\E (\ddot{\lambda}\{e_r(\mb X)\}\cdot \{\tilde{e}(\mb X)-e(\mb X)\}^2\cdot [\{\tilde{\mu}_{1}(\mb X)-\mu_{1}(\mb X)\} \nonumber\\
	&\quad - \{\tilde{\mu}_{0}(\mb X)-\mu_{0}(\mb X)\}]) -\E[{\ddot{\lambda}}\{e_r(\mb X)\}\cdot \{\tilde{e}(\mb X)-e(\mb X)\}^2 \cdot\{\mu_{1,r}(\mb X)-\mu_{0,r}(\mb X)\}]\nonumber \\
&=-r\cdot\E [\lambda^{(3)}\{e_r(\mb X)\}\cdot \{\tilde{e}(\mb X)-e(\mb X)\}^3\cdot \{\mu_{1}(\mb X)-\mu_{0}(\mb X)\}] \nonumber\\
	&\quad - r^2\cdot \E (\lambda^{(3)}\{e_r(\mb X)\}\cdot \{\tilde{e}(\mb X)-e(\mb X)\}^3\cdot[\{\tilde\mu_{1}(\mb X)-\mu_{1}(\mb X)\}-\{\tilde\mu_{0}(\mb X)-\mu_{0}(\mb X)\}]) \nonumber\\
&\quad -r\cdot\E (\ddot{\lambda}\{e_r(\mb X)\}\cdot \{\tilde{e}(\mb X)-e(\mb X)\}^2\cdot [\{\tilde{\mu}_{1}(\mb X)-\mu_{1}(\mb X)\} - \{\tilde{\mu}_{0}(\mb X)-\mu_{0}(\mb X)\}]) \nonumber\\
&\quad -r\cdot\E [\ddot{\lambda}\{e_r(\mb X)\}\cdot \{\tilde{e}(\mb X)-e(\mb X)\}^2\cdot \{\mu_{1}(\mb X)-\mu_{0}(\mb X)\}]\nonumber\\
	&\quad  - r^2\cdot \E (\ddot{\lambda}\{e_r(\mb X)\}\cdot \{\tilde{e}(\mb X)-e(\mb X)\}^2\cdot[\{\tilde\mu_{1}(\mb X)-\mu_{1}(\mb X)\}-\{\tilde\mu_{0}(\mb X)-\mu_{0}(\mb X)\}]). \\
\mc{I}_{10} &=
\E(\ddot{\lambda}\{e_r(\mb X)\}\cdot\{\tilde{e}(\mb X)-e(\mb X)\}\cdot [ \{\tilde{\mu}_1(\mb X)-\mu_1(\mb X)\}-
\{\tilde{\mu}_0(\mb X)-\mu_0(\mb X)\}  ]\cdot\{e(\mb X)-e_r(\mb X)\} ) \nonumber\\
&\quad -\E(\dot{\lambda}\{e_r(\mb X)\}\cdot [ \{\tilde{\mu}_1(\mb X)-\mu_1(\mb X)\}-
\{\tilde{\mu}_0(\mb X)-\mu_0(\mb X)\}  ]\cdot\{\tilde{e}(\mb X)-e(\mb X)\}) \nonumber\\
&=
-r\cdot\E(\ddot{\lambda}\{e_r(\mb X)\}\cdot\{\tilde{e}(\mb X)-e(\mb X)\}^2\cdot[\{\tilde{\mu}_1(\mb X)-\mu_1(\mb X)\}-
\{\tilde{\mu}_0(\mb X)-\mu_0(\mb X)\}]) \nonumber\\
&\quad -\E(\dot{\lambda}\{e_r(\mb X)\}\cdot [ \{\tilde{\mu}_1(\mb X)-\mu_1(\mb X)\}-
\{\tilde{\mu}_0(\mb X)-\mu_0(\mb X)\}  ]\cdot\{\tilde{e}(\mb X)-e(\mb X)\}). \\
\mc{I}_{11}
&=
-\E[\ddot{\lambda}\{e_r(\mb X)\}\cdot\{\mu_{1,r}(\mb X)-\mu_{0,r}(\mb X)\}\cdot \{\tilde{e}(\mb X)-e(\mb X)\}^2] \nonumber\\
	&\quad - \E(\dot{\lambda}\{e_r(\mb X)\}\cdot [ \{\tilde{\mu}_1(\mb X)-\mu_1(\mb X)\} -   \{\tilde{\mu}_0(\mb X)-\mu_0(\mb X)\}   ]  \cdot \{\tilde{e}(\mb X)-e(\mb X)\})\nonumber\\
&=
-r\cdot\E [\ddot{\lambda}\{e_r(\mb X)\}\cdot \{\tilde{e}(\mb X)-e(\mb X)\}^2\cdot \{\mu_{1}(\mb X)-\mu_{0}(\mb X)\}] \nonumber\\
	&\quad - r^2\cdot \E (\ddot{\lambda}\{e_r(\mb X)\}\cdot \{\tilde{e}(\mb X)-e(\mb X)\}^2\cdot[\{\tilde\mu_{1}(\mb X)-\mu_{1}(\mb X)\}-\{\tilde\mu_{0}(\mb X)-\mu_{0}(\mb X)\}])\nonumber\\
&\quad-\E(\dot{\lambda}\{e_r(\mb X)\}\cdot [\{\tilde{\mu}_1(\mb X)-\mu_1(\mb X)\} - \{\tilde{\mu}_0(\mb X)-\mu_0(\mb X)\}]\cdot \{\tilde{e}(\mb X)-e(\mb X)\}). \label{eq:I9-11-form}
\end{align}
Plugging \eqref{eq:secondGate-ga}--\eqref{eq:I9-11-form} in \eqref{eq:secondGate} gives us \eqref{eq:secGate-formula}, which completes the proof.
\end{proof}

\subsubsection{Assumptions and two useful lemmata}\label{subsubapx:twolem-2018paper}

In this section, we introduce two lemmata by Chernozhukov et al. \cite{chernozhukov2018double}, which summarized the conditions for asymptotic normality of DML estimators. First, the following assumptions are needed, following the notations of Gateaux derivatives in Section \ref{subsubapx:Gate}. Suppose $C_1\geq C_0>0$ are some finite constants, $\{\delta_n\}_{n\geq 1}$ and $\{\Delta_n\}_{n\geq 1}$ are some sequences of positive sequence converging to 0 with $\delta_n\geq n^{-1/2}$. 
 
\begin{assumption}[Assumption 3.1 in Chernozhukov et al. \cite{chernozhukov2018double}]\label{assp:Assp3.1-2018paper} For $n\geq 3$, the following conditions hold.
	\begin{enumerate}[label=(\alph*)]
		\item The true parameter $\gamma_0$ obeys 
		$\E[L(\mb V;\gamma_0,\phi_0)]=0$. 
		\item The score function $g$ is linear in the form of $L(\mb V;\gamma,\phi)=\gamma D(\mb V;\phi) - N(\mb V;\phi)$.
		\item The mapping $\phi \mapsto \E[L(\mb V;\gamma,\phi)]$ is twice continuously Gateaux-differentiable. 
		\item The score $g$ obeys the Neyman orthogonality, or more generally, the Neyman $\lambda_n$ near-orthogonality condition at $(\theta_0,\phi_0)$ with respect to the nuisance realization set $\mc{T}_n\subset \mc{T}$ for 
		\begin{align*}
			\lambda_n:= \sup_{\phi\in\mc{T}_n}\big\vert	\partial_\phi \E[L(\mb V;\gamma_0,\phi_0)] [\phi-\phi_0]\big\vert\leq \delta_n n^{-1/2}.
		\end{align*} 
		\item The identification condition holds, that is, $J_0:=\E[D(\mb V;\phi_0)]$ is between $C_0$ and $C_1$. 
	\end{enumerate}
\end{assumption}

\begin{assumption}[Assumption 3.2 in Chernozhukov et al. \cite{chernozhukov2018double}]\label{assp:Assp3.2-2018paper}  For $n\geq 3$ and $q>2$, the following conditions hold.
	\begin{enumerate}[label=(\alph*)]
		\item Given a random subset $I$ of $n$ of size $n/K$, the nuisance
		parameter estimator $\hat{\phi}=\hat{\phi}((\mb V_i)_{i\in I^c})$ belongs to the realization set $\mc{T}_n$ with probability at least $1-\Delta_n$, where $\mc{T}_n$ contains $\phi_0$ and is constrained by the next conditions.
		\item The moment conditions hold:
		\begin{align*}
			m_n&:=\sup_{\phi\in\mc{T}_n}\left(\E[\big\vert L(\mb V;\gamma_0,\phi)\big\vert^q]\right)^{1/q}\leq C_1,\\
			m_n^\prime&:=\sup_{\phi\in\mc T_n}\left(\E[\big\vert D(\mb V;\phi)\big\vert^q]\right)^{1/q}\leq C_1. 
		\end{align*}
		\item The following conditions on the statistical rates $r_n$, $r^\prime_n$ and $\lambda_n^\prime$ hold:
		\begin{align*}
			r_n&:=\sup_{\phi\in\mc{T}_n}\big\vert\E[ D(\mb V;\phi)]-\E[D(\mb V;\phi_0)]\big\vert\leq \delta_n,  \\
			r_n^\prime&:=\sup_{\phi\in\mc{T}_n}\left(\E\big\vert L(\mb V;\gamma_0,\phi)-L(\mb V;\gamma_0,\phi_0)\big\vert^2\right)^{1/2}\leq \delta_n, \\
			\lambda^{\prime}_n &:= \sup_{r\in(0,1),\phi\in\mc{T}_n}
			\big\vert \partial_r^2\E[L(\mb V;\gamma_0,\phi_0+r(\phi-\phi_0))]\big\vert \leq \delta_n/\sqrt{n}.
			\\
		\end{align*}
		\item 
		The variance of the score $g$ is non-degenerate: all eigenvalues of the matrix 
			$\E[L(\mb V;\gamma_0,\phi_0)^{\otimes2}]$ are bounded from below by $C_0$. 
	\end{enumerate}
\end{assumption}

\begin{lemma}[Theorem 3.1 in Chernozhukov et al. \cite{chernozhukov2018double}]\label{thm:prop-DML-2018paper} For $n\geq 3$, suppose that Assumptions \ref{assp:Assp3.1-2018paper} and \ref{assp:Assp3.2-2018paper} hold. In addition, suppose that $\delta_n\geq n^{-1/2}$ for all $n\geq 1$. Then the DML-1 and DML-2 estimators $\hat{\gamma}^{\text{dml}-d}$ ($d=1,2$)  constructed in Algorithm \ref{alg:DMLs}, are first order equivalent and obey 
	\begin{align*}
		\sqrt{n}\sigma^{-1}(\hat{\gamma}^{\text{dml}-d}-\gamma_0)=\frac{1}{\sqrt{n}}\sum_{i=1}^{n} \bar{g}(\mb V_i)+O_p(\rho_n)\rightsquigarrow \mc N(0,1),
	\end{align*}
	uniformly for any $n$, where the size of the remainder term obeys 
	\begin{align*}
		\rho_n:=n^{-1/2}+r_n+r^\prime_n + n^{1/2}\lambda_n+n^{1/2}\lambda_n^\prime \lesssim \delta_n,
	\end{align*} where
$\bar{g}(\cdot):=-\sigma^{-1}J_0^{-1}L(\cdot;\gamma_0,\phi_0)$ is the influence function and the asymptotic variance is
	\begin{align*}
    \sigma^2:=J_0^{-1}\E[L(\mb V;\gamma_0,\phi_0)^{\otimes2}] J_0^{-1},
	\end{align*} 
	with $J_0=\E[D(\mb V;\phi_0)]$.
\end{lemma}

\begin{lemma}[Theorem 3.2 in Chernozhukov et al. \cite{chernozhukov2018double}]\label{thm:prop2-DML-2018paper} For $n\geq 3$, suppose \ref{assp:Assp3.1-2018paper} and \ref{assp:Assp3.2-2018paper} hold, and suppose that $\delta_n\geq n^{-[(1-2/q)\wedge 1/2]}$ for all $n\geq 1$. Consider the following estimator of the asymptotic variance of $\sqrt{n}(\hat{\gamma}^{\text{dml}-d}-\gamma_0):$
	\begin{align*}
		\hat{\sigma}^2=\hat{J}_0^{-1}\frac{1}{K}\sum_{k=1}^{K}\E_{n,k}[g^2(\mb V;\hat{\gamma}^{\text{dml}-d},\hat{\phi}_{k})]\hat{J}_0^{-1},
	\end{align*}
	where 
	\begin{align*}
		\hat{J}_0 = \frac{1}{K}\sum_{k=1}^K\E_{n,k}[D(\mb V;\hat{\phi}_{k})],
	\end{align*}
	and $d=1,2$. This estimator concentrates around the true variance matrix $\sigma^2$,
	\begin{align*}
		\hat{\sigma}^2=\sigma^2+O_p(\varrho_n), \qquad \varrho_n:=n^{-[(1-2/q)\wedge 1/2]} + \lambda_n +\lambda_n^\prime \lesssim \delta_n.
	\end{align*}
	Moreover, $\hat{\sigma}^2$ can replace $\sigma^2$ in the statement of Theorem 3.1 with the size of the remainder term updated as $\rho_n=n^{-[(1-2/q)\wedge 1/2]}+r_n+r_n^\prime+n^{1/2}\lambda_n + n^{1/2}\lambda_n^\prime$.
\end{lemma}

\subsection{Proof of Theorem \ref{thm:RDR-DML} {in Section \ref{sec:DML}}}

To prove Theorem \ref{thm:RDR-DML}, based on discussions in Sections \ref{subsubapx:Gate}--\ref{subsubapx:twolem-2018paper}, we only need to show that under all Conditions specified in Theorem \ref{thm:RDR-DML}, Assumptions \ref{assp:Assp3.1-2018paper} and \ref{assp:Assp3.2-2018paper} hold (and so Lemmata \ref{thm:prop-DML-2018paper} and \ref{thm:prop2-DML-2018paper} hold). 

\subsubsection{Proof of Assumption \ref{assp:Assp3.1-2018paper}}
We first verify the Neyman orthogonality condition. Given $\gamma_0= \dfrac{\E[\lambda\{e(\mb X)\}\tau(\mb X)]}{\E [\lambda\{e(\mb X)\}]}$,
 we have $\E[L(\mb V;\gamma_0,\phi_0)]=0$ and thus (a) in Assumption \ref{assp:Assp3.1-2018paper} naturally holds. Also, by Lemma \ref{lem:firstGate} {(Section \ref{subapx:prelim-DML})}, for any $\phi=(\tilde{e},\tilde{\mu}_0,\tilde{\mu}_1)\in \mc{T}_n$, the Gateaux derivative in the direction $\phi-\phi_0=(\tilde{e}-e,\tilde{\mu}_0-\mu_0,\tilde{\mu}_1-\mu_1)$ is given by
\begin{align*}
	\partial_\phi \E[L(\mb V;\theta_0,\phi_0)][\phi-\phi_0]=0.
\end{align*}
This verifies (d) in Assumption \ref{assp:Assp3.1-2018paper} with $\lambda_n=0$.

Note that $J_0=\E[D(\mb V;\phi)]=\E[\lambda\{e(\mb X)\}]$, so by Condition \ref{cond:DMLlambda} in Theorem \ref{thm:RDR-DML}, (e) in Assumption \ref{assp:Assp3.1-2018paper} naturally holds.
 
Note that (b) and (c) in Assumption \ref{assp:Assp3.1-2018paper} hold trivially. Therefore, the whole Assumption \ref{assp:Assp3.1-2018paper} holds.

\subsubsection{Proof of (b) in Assumption \ref{assp:Assp3.2-2018paper}}

We first prove the moment condition of $m_n^\prime$. Consider Taylor expansion of $\lambda\{\tilde{e}(\mb X)\}$ at $e(\mb X)$: 
\begin{align}
\lambda\{\tilde{e}(\mb X)\}&=\lambda\{e(\mb X)\}+\dot{\lambda}\{\bar{e}(\mb X)\}\{\tilde{e}(\mb X)-e(\mb X)\}, \label{eq:lambdaTaylor}
\end{align}
with $\bar{e}(\mb X)$ is between $\tilde{e}(\mb X)$ and $e(\mb X)$. Plugging \eqref{eq:lambdaTaylor} into $m_n^\prime$, and by Minkowski inequality when $q\geq 1$,
\begin{align*}
	\left(\E\left[\left\vert D(\mb V;\phi)\right\vert^q\right]\right)^{1/q}
	&=\left(\E\left\vert\lambda\{e(\mb X)\}+\{\tilde{e}(\mb X)-e(\mb X)\}[\dot{\lambda}\{\bar{e}(\mb X)\}-\dot{\lambda}\{\tilde{e}(\mb X)\}] +\dot{\lambda}\{\tilde{e}(\mb X)\}\{A-e(\mb X)\}\right\vert^q \right)^{1/q}
    \\
	&\leq\left(\E\left\vert\lambda\{e(\mb X)\}\right\vert^q\right)^{1/q} + \left(\E\left[\left\vert e(\mb X)-\tilde{e}(\mb X)\right\vert^q\cdot\left\vert \dot{\lambda}\{\bar{e}(\mb X)\}-\dot{\lambda}\{\tilde{e}(\mb X)\} \right\vert^q\right]\right)^{1/q} \nonumber\\
    & \qquad + 2\sup_\mb X\left\vert\dot{\lambda}\{\tilde{e}(\mb X)\}\right\vert \\
    &\leq\Vert\tilde{e}(\mb X)-e(\mb X)\Vert_q\cdot\Vert\dot{\lambda}\{\bar{e}(\mb X)\} -\dot{\lambda}\{\tilde{e}(\mb X)\}\Vert_q + 3\sup_\mb X\left\vert{\lambda}\{{e}(\mb X)\}\vee\dot{\lambda}\{\tilde{e}(\mb X)\}\right\vert. 
\end{align*}
Condition \ref{cond:DMLbound} implies $\sup_{\mb X}\left\vert \dot\lambda\{e(\mb X)\} \vee \dot{\lambda}\{\tilde{e}(\mb X)\} \right\vert$ has an upper bound, with Condition \ref{cond:DMLconverge} (Theorem \ref{thm:RDR-DML}), we find that there exists $C_1$ such that $m_n^\prime\leq C_1$. Notice that for ATE, this bound is just $3$. 

Second, we prove the moment condition of $m_n$. By Minkowski inequality, 
\begin{align*}
\left(\E\left[\left\vert L(\mb V;\gamma_0,\phi)\right\vert^q\right]\right)^{1/q}
	& = \left(\E\left[\left\vert \gamma_0 D(\mb V;\phi)-N(\mb V;\phi)\right\vert^q\right]\right)^{1/q}\\
	& \leq \left(\E\left[\left\{\left\vert \gamma_0 D(\mb V;\phi)\right\vert +\left\vert N(\mb V;\phi)\right\vert \right\}^q\right]\right)^{1/q} \\
	& \leq\left\vert\gamma_0\right\vert\cdot \left\{\E \left\vert D(\mb V;\phi)\right\vert^q\right\}^{1/q} + \left\{\E\left\vert N(\mb V;\phi) \right\vert^q \right \}^{1/q}.
\end{align*}
As we showed that $\left\{\E\left[\left\vert D(\mb V;\phi)\right\vert^q\right]\right\}^{1/q}$ is upper bounded, we only need to show the second term $\left\{\E\left[\left\vert N(\mb V;\phi)\right\vert^q\right]\right\}^{1/q}$ is bounded as follows. 
\begin{align}
\left\{\E\left[\left\vert N(\mb V;\phi)\right\vert^q\right]\right\}^{1/q}
	&=\bigg(\E\bigg\vert\lambda\{\tilde{e}(\mb X)\} \left[\frac{A}{\tilde{e}(\mb X)}\{Y-\tilde{\mu}_1(\mb X)\} -\frac{1-A}{1-\tilde{e}(\mb X)}\{Y-\tilde{\mu}_0(\mb X)\}+\{\tilde{\mu}_1(\mb X)-\tilde{\mu}_0(\mb X)\}\right] \\
    & \qquad + \dot{\lambda}\{\tilde{e}(\mb X)\}\{\tilde{\mu}_1(\mb X)-\tilde{\mu}_0(\mb X)\}\{A-\tilde{e}(\mb X)\}\bigg\vert^q\bigg)^{1/q}  \nonumber\\
	& \leq \mc{I}_{12}+\mc{I}_{13}. \label{eq:I12+13}
\end{align}
where we denote 
\begin{align*}
	\mc{I}_{12}& = \left(\E\left\vert\lambda\{\tilde{e}(\mb X)\} \left[\frac{A}{\tilde{e}(\mb X)}\{Y-\tilde{\mu}_1(\mb X)\} - \frac{1-A}{1-\tilde{e}(\mb X)}\{Y-\tilde{\mu}_0(\mb X)\}\right]\right\vert^q\right)^{1/q},\text{ and } \\
	\mc{I}_{13}& =\left(\E\left\vert\lambda\{\tilde{e}(\mb X)\}\{\tilde{\mu}_1(\mb X)-\tilde{\mu}_0(\mb X)\}+ \dot{\lambda}\{\tilde{e}(\mb X)\}\{\tilde{\mu}_1(\mb X)-\tilde{\mu}_0(\mb X)\}\{A-\tilde{e}(\mb X)\}\right\vert^q\right)^{1/q}. 
\end{align*}
$\mc{I}_{12}$ can be rewritten, by Minkowski inequality, as
\begin{align*}
	\mc{I}_{12}
	& = \bigg(\E\bigg\{\lambda\{\tilde{e}(\mb X)\}^q \cdot \bigg[\frac{A}{\tilde{e}(\mb X)}\{Y-\mu_1(\mb X)\} \\
	&\quad\quad\quad\quad
	+\frac{A}{\tilde{e}(\mb X)}\{\mu_1(\mb X)-\tilde{\mu}_1(\mb X)\} - \frac{1-A}{1-\tilde{e}(\mb X)}\{Y-\mu_0(\mb X)\}- \frac{1-A}{1-\tilde{e}(\mb X)}\{\mu_0(\mb X)-\tilde{\mu}_0(\mb X)\}\bigg]^q\bigg\}\bigg)^{1/q}\\
	& \leq \left(\E \left[\lambda\{\tilde{e}(\mb X)\} \frac{A}{\tilde{e}(\mb X)}\{Y-\mu_1(\mb X)\}\right]^q\right)^{1/q}
	+\left(\E\left[\lambda\{\tilde{e}(\mb X)\}\frac{1-A}{1-\tilde{e}(\mb X)}\{Y-\mu_0(\mb X)\}\right]^q \right)^{1/q}\\
	&\quad + \left(\E \left[\lambda\{\tilde{e}(\mb X)\} \frac{A}{\tilde{e}(\mb X)}\{\mu_1(\mb X)-\tilde{\mu}_1(\mb X)\}\right]^q\right)^{1/q} + 
	\left(\E\left[\lambda\{\tilde{e}(\mb X)\} \frac{1-A}{1-\tilde{e}(\mb X)} \{\mu_0(\mb X)-\tilde{\mu}_0(\mb X)\}\right]^q\right)^{1/q}\\
	&\leq  \sup_{\mb X}\left\vert \frac{\lambda\{\tilde{e}(\mb X)\}}{\tilde{e}(\mb X)}\vee \frac{\lambda\{\tilde{e}(\mb X)\}}{1-\tilde{e}(\mb X)} \right\vert \cdot \left(\Vert A\{Y-\mu_1(\mb X)\}\Vert_q+\Vert (1-A)\{Y-\mu_0(\mb X)\}\Vert_q + \sum_{a=0}^1 \Vert \mu_a(\mb X)-\tilde{\mu}_a(\mb X)\Vert_q\right). 
\end{align*}
Plugging the Taylor expansion \eqref{eq:lambdaTaylor} into $\mc{I}_{13}$, by Minkowski inequality, we can rewrite $\mc{I}_{13}$ as
\begin{align*}
\mc{I}_{13}	& = \big(\E\big\vert\lambda\{e(\mb X)\}\cdot\{\tilde{\mu}_1(\mb X)-\tilde{\mu}_0(\mb X)\}+
	\{\tilde{\mu}_1(\mb X)-\tilde{\mu}_0(\mb X)\}\cdot\tilde{e}(\mb X) \cdot[\dot{\lambda}\{\bar{e}(\mb X)\} - \dot{\lambda}\{\tilde{e}(\mb X)\}] \\
	&  \quad - \{\tilde{\mu}_1(\mb X)-\tilde{\mu}_0(\mb X)\}\cdot e(\mb X)\cdot[\dot{\lambda}\{\bar{e}(\mb X)\}-\dot{\lambda}\{\tilde{e}(\mb X)\}] + \{\tilde{\mu}_1(\mb X)-\tilde{\mu}_0(\mb X)\}\cdot\dot{\lambda}\{\tilde{e}(\mb X)\}\cdot\{A-e(\mb X)\}\big\vert^q\big)^{1/q} 
	\\
	&\leq\left(\E\big\vert\lambda\{e(\mb X)\}\cdot[\{\tilde{\mu}_1(\mb X)-\tilde{\mu}_0(\mb X)\}-\{\mu_1(\mb X)-\mu_0(\mb X)\}]\big\vert^q\right)^{1/q} \\ 
	&\quad 
    +\left(\E \big\vert\{\tilde{\mu}_1(\mb X)-\tilde{\mu}_0(\mb X)\}\cdot e(\mb X)\cdot[\dot{\lambda}\{\bar{e}(\mb X)\}-\dot{\lambda}\{\tilde{e}(\mb X)\}] \big\vert^q\right)^{1/q} \\
    & \quad +\left(\E \big\vert\{\tilde{\mu}_1(\mb X)-\tilde{\mu}_0(\mb X)\}\cdot\dot{\lambda}\{\tilde{e}(\mb X)\}\cdot\{A-e(\mb X)\}\big\vert^q\right)^{1/q} \\
	&\leq  \sup_{\mb X}\left\vert \lambda\{e(\mb X)\}\right\vert\cdot \sum_{a=0}^1  \Vert\tilde{\mu}_a(\mb X)-\mu_a(\mb X)\Vert_q
	+\sup_{\mb X}\left\vert \lambda\{e(\mb X)\} \right\vert\cdot\Vert\tau(\mb X)\Vert_q \\
    & \quad
    + \sup_{\mb X}\left\vert \tilde{e}(\mb X)\right\vert\cdot \left(\sum_{a=0}^1  \Vert \tilde{\mu}_a(\mb X)-\mu_a(\mb X) \Vert_q + \Vert \tau(\mb X) \Vert_q\right) \cdot\Vert \dot{\lambda}\{\bar{e}(\mb X)\}-\dot{\lambda}\{\tilde{e}(\mb X)\Vert_q
    \\
	& \quad +\sup_{\mb X}\left\vert \dot{\lambda}\{\tilde{e}(\mb X)\} \right\vert\cdot\sup_{\mb X}\left\vert A-e(\mb X) \right\vert\cdot  \left(\sum_{a=0}^1  \Vert \tilde{\mu}_a(\mb X)-\mu_a(\mb X) \Vert_q + \Vert \tau(\mb X) \Vert_q\right) \\
	& \leq 
	\left( \Vert \dot{\lambda}\{\bar{e}(\mb X)\}-\dot{\lambda}\{\tilde{e}(\mb X)\Vert_q 
	+ 2\sup_{\mb X}\left\vert {\lambda}\{e(\mb X)\}\vee \dot{\lambda}\{\tilde{e}(\mb X)\} \right\vert\right) \cdot \left(\sum_{a=0}^1  \Vert\tilde{\mu}_a(\mb X)-\mu_a(\mb X)\Vert_q+ \Vert \tau(\mb X) \Vert_q\right). 
\end{align*}
Condition \ref{cond:DMLbound} implies $\sup_{\mb X}\left\vert {\lambda}\{e(\mb X)\}\vee \dot{\lambda}\{\tilde{e}(\mb X)\} \right\vert$ has an upper bound. According to Condition \ref{cond:DMLconverge} (Theorem \ref{thm:RDR-DML}), Conditions \ref{cond:DMLgamma} and \ref{cond:DML-U} and inequality \eqref{eq:I12+13}, there is a constant $C_2$ such that
$\mc{I}_{12}+\mc{I}_{13}\leq C_2$. 
Therefore, it is easy to know, if we choose $C_1'=C_1\vee C_2$, we have $m_n\leq C_1'$ and $m_n'\leq C_1'$, and so (b) in  Assumption \ref{assp:Assp3.2-2018paper} holds.

\subsubsection{Proof of (c) in Assumption \ref{assp:Assp3.2-2018paper}}
We first bound $r_n$. By Taylor expansion,
\begin{align*}
\left\vert \E[ D(\mb V;\phi)]-\E[D(\mb V;\phi_0)]\right\vert
	= \left\vert \E [\lambda\{{e}(\mb X)\}+ \dot{\lambda}\{\bar{e}(\mb X)\}\{e(\mb X)-\tilde{e}(\mb X)\}-\lambda\{e(\mb X)\}]\right\vert 
    \preceq \Vert\tilde{e}(\mb X)-e(\mb X) \Vert_2,
\end{align*}
by Conditions \ref{cond:DMLbound} and \ref{cond:DMLLips}. 

Next, for $r_n'$, we want to show
\begin{align}
\left(\E\left\vert L(\mb V;\gamma_0,\phi)-L(\mb V;\gamma_0,\phi_0)\right\vert ^2\right)^{1/2} 
& = \Vert  \gamma_0\{D(\mb V;\phi)-D(\mb V;\phi_0)\}-\{N(\mb V;\phi)-N(\mb V;\phi_0)\}  \Vert_2 \nonumber \\
& \leq \left\vert\gamma_0\right\vert \cdot \mc{I}_{14}+\mc{I}_{15}, \label{eq:r_n_prime}
\end{align}
where 
\begin{align*}
	\mc{I}_{14} = \Vert  D(\mb V;\phi)-D(\mb V;\phi_0) \Vert_2 \text{ and }
	\mc{I}_{15} = \Vert N(\mb V;\phi)-N(\mb V;\phi_0) \Vert_2.
\end{align*}
For $\mc{I}_{14}$, we bound it by
\begin{align}
\mc{I}_{14}
&= \Vert \lambda\{\tilde{e}(\mb X)\}-\lambda\{e(\mb X)\} - \dot{\lambda}\{\tilde{e}(\mb X)\}\{A-\tilde{e}(\mb X)\}+\dot{\lambda}\{e(\mb X)\}\{A-e(\mb X)\} \Vert_2 \nonumber  \\
& = \Vert \dot\lambda\{\bar{e}(\mb X)\}\{\tilde e(\mb X)-e(\mb X)\} - \dot{\lambda}\{\tilde{e}(\mb X)\}\{A-\tilde{e}(\mb X)\}+\dot{\lambda}\{e(\mb X)\}\{A-e(\mb X)\} \Vert_2 \nonumber \\
& = \Vert \dot\lambda\{\bar{e}(\mb X)\}\{\tilde e(\mb X)-e(\mb X)\} - A\{\dot{\lambda}\{\tilde{e}(\mb X)\}-\dot{\lambda}\{e(\mb X)\}\} + \dot{\lambda}\{\tilde{e}(\mb X)\}\tilde e(\mb X) - \dot{\lambda}\{e(\mb X)\}e(\mb X) \Vert_2\nonumber   \\
& = \Vert \dot\lambda\{\bar{e}(\mb X)\}\{\tilde e(\mb X)-e(\mb X)\} - A\ddot\lambda\{\bar e(\mb X)\}\{\tilde e(\mb X)-e(\mb X)\} + \dot{\lambda}\{\tilde{e}(\mb X)\}\{\tilde e(\mb X)-e(\mb X)\}\nonumber \\
& \quad + e(\mb X)\{\dot{\lambda}\{\tilde e(\mb X)\}-\dot{\lambda}\{e(\mb X)\}\} \Vert_2\nonumber \\
& \preceq \Vert \tilde e(\mb X)-e(\mb X)\Vert_2.
 \label{eq:I14}
\end{align}
For $\mc{I}_{15}$, we first decompose it as
\begin{align*}
	\mc{I}_{15}
	& =\bigg\Vert  \lambda\{\tilde{e}(\mb X)\}\left[\frac{A}{\tilde{e}(\mb X)}\{Y-\tilde{\mu}_1(\mb X)\} - \frac{1-A}{1-\tilde{e}(\mb X)}\{Y-\tilde{\mu}_0(\mb X)\}+\{\tilde{\mu}_1(\mb X)-\tilde{\mu}_0(\mb X)\}\right]\\
	& \qquad-\lambda\{e(\mb X)\}\left[\frac{A}{e(\mb X)}\{Y-\mu_1(\mb X)\} - \frac{1-A}{1-e(\mb X)}\{Y-\mu_0(\mb X)\}+\{\mu_1(\mb X)-\mu_0(\mb X)\}\right]  \\
	& \qquad+\dot{\lambda}\{\tilde{e}(\mb X)\}\{\tilde{\mu}_1(\mb X)-\tilde{\mu}_0(\mb X)\}\{A-\tilde{e}(\mb X)\} -\dot{\lambda}\{e(\mb X)\}\{\mu_1(\mb X)-\mu_0(\mb X)\}\{A-e(\mb X)\}\bigg\Vert_2 \\
	& \leq \mc{I}_{16} +\mc{I}_{17} + \mc{I}_{18} +\mc{I}_{19}+\mc{I}_{20},
\end{align*}
where we define
\begin{align*}
    \mc{I}_{16}&= \left\Vert \left[ \frac{\lambda\{\tilde{e}(\mb X)\}}{\tilde{e}(\mb X)} - \frac{\lambda\{e(\mb X)\}}{e(\mb X)} \right] A\{Y-\mu_1(\mb X)\}  \right\Vert_2,\\
	\mc{I}_{17}&= \left\Vert  \left[ \frac{\lambda\{\tilde{e}(\mb X)\}}{1-\tilde{e}(\mb X)} - \frac{\lambda\{e(\mb X)\}}{1-e(\mb X)} \right] (1-A)\{Y-\mu_0(\mb X)\}  \right\Vert_2,\\
	\mc{I}_{18} & = \left\Vert \frac{\lambda\{\tilde{e}(\mb X)\}}{\tilde{e}(\mb X)} A \{ \tilde{\mu}_1(\mb X)-\mu_1(\mb X) \} \right\Vert_2+\left\Vert \frac{\lambda\{\tilde{e}(\mb X)\}}{1-\tilde{e}(\mb X)} (1-A) \{ \tilde{\mu}_0(\mb X)-\mu_0(\mb X) \} \right\Vert_2,\\
	\mc{I}_{19} &= \Vert \lambda\{\tilde{e}(\mb X)\}\{\tilde{\mu}_1(\mb X)-\tilde{\mu}_0(\mb X)\}
	-\lambda\{e(\mb X)\}\{\mu_1(\mb X)-\mu_0(\mb X)\} \Vert_2, \text{ and }\\
	\mc{I}_{20} &= \Vert \dot{\lambda}\{\tilde{e}(\mb X)\}\{\tilde{\mu}_1(\mb X)-\tilde{\mu}_0(\mb X)\}\{A-\tilde{e}(\mb X)\}  
	-\dot{\lambda}\{e(\mb X)\}\{\mu_1(\mb X)-\mu_0(\mb X)\}\{A-e(\mb X)\} \Vert_2.
\end{align*}
We can bound $\mc{I}_{16}$ by
\begin{align}
\mc{I}_{16} & 
\leq \sup_{\mb X}\left\vert\frac{{\lambda\{\tilde{e}(\mb X)\}} / {\tilde{e}(\mb X)}-{\lambda\{e(\mb X)\}} / {e(\mb X)} }{\tilde{e}(\mb X)-e(\mb X)}\right\vert 
\cdot \left(\E [A\{Y-\mu_1(\mb X)\}^2\{\tilde{e}(\mb X)-e(\mb X)\}^2]\right)^{1/2} \nonumber\\
&\leq \sup_{\mb X}\left\vert \frac{{\lambda\{\tilde{e}(\mb X)\}} / {\tilde{e}(\mb X)}-{\lambda\{e(\mb X)\}} / {e(\mb X)} }{\tilde{e}(\mb X)-e(\mb X)}  \right\vert\cdot 
\sup_{\mb X}\left(\E[A\{Y-\mu_1(\mb X)\}^2\mid \mb X]\right)^{1/2}\cdot \Vert \tilde{e}(\mb X)-e(\mb X) \Vert_2 \nonumber\\
&\preceq  \Vert \tilde{e}(\mb X)-e(\mb X) \Vert_2, \label{eq:I16}
\end{align}
by Conditions \ref{cond:DML-U} and \ref{cond:DMLLips} in Theorem \ref{thm:RDR-DML} {(Section \ref{sec:DML})}. Similarly, we can get the upper bound of $\mc{I}_{17}$, $\mc{I}_{18}$ and $\mc{I}_{19}$ using Conditions \ref{cond:DMLbound}, \ref{cond:DMLLips}, \ref{cond:DML-U} and \ref{cond:DMLgamma}, i.e., 

\begin{align}
\mc{I}_{17} 
&\leq\sup_{\mb X}\left\vert \frac{{\lambda\{\tilde{e}(\mb X)\}}/\{1-\tilde{e}(\mb X)\} - {\lambda\{e(\mb X)\}} / \{1-e(\mb X)\}}{\tilde{e}(\mb X)-e(\mb X)} \right\vert \nonumber\\
& \qquad \times \sup_{\mb X}\left(\E[(1-A)\{Y-\mu_0(\mb X)\}^2\mid \mb X]\right)^{1/2}\cdot \Vert \tilde{e}(\mb X)-e(\mb X) \Vert_2 \nonumber \\ 
& \preceq\Vert \tilde{e}(\mb X)-e(\mb X) \Vert_2, \label{eq:I17}  \\ 
\mc{I}_{18}
&\leq \sup_{\mb X}\left\vert \frac{\lambda\{\tilde{e}(\mb X)\}}{\tilde{e}(\mb X)}\vee\frac{\lambda\{\tilde{e}(\mb X)\}}{1-\tilde{e}(\mb X)}\right\vert\cdot\sum_{a=0}^1  \Vert \tilde{\mu}_a(\mb X)-\mu_a(\mb X)\Vert_2 \preceq \sum_{a=0}^1  \Vert \tilde{\mu}_a(\mb X)-\mu_a(\mb X)\Vert_2, \text{  and  }
  \label{eq:I18} \\
\mc{I}_{19} 
& \leq \Vert\lambda\{\tilde{e}(\mb X)\}\cdot[\{\tilde{\mu}_1(\mb X)-\tilde{\mu}_0(\mb X)\}-\{\mu_1(\mb X)-\mu_0(\mb X)\}]\Vert_2 + \Vert [\lambda\{\tilde{e}(\mb X)\}-\lambda\{e(\mb X)\}] \cdot\{\mu_1(\mb X)-\mu_0(\mb X)\}\Vert_2  \nonumber\\
& \leq  \sup_{\mb X}\left\vert \lambda\{\tilde{e}(\mb X)\} \right\vert\cdot \sum_{a=0}^1 \Vert \tilde{\mu}_a(\mb X)-\mu_a(\mb X) \Vert_2 + \sup_{\mb X}\left\vert \tau(\mb X)\right\vert\cdot \sup_{\mb X}\left\vert \dot\lambda\{\bar e(\mb X)\}\right\vert\cdot \Vert \tilde{e}(\mb X)-e(\mb X) \Vert_2  \nonumber\\
&\preceq \sum_{a=0}^1 \Vert \tilde{\mu}_a(\mb X)-\mu_a(\mb X) \Vert_2+ \Vert \tilde{e}(\mb X)-e(\mb X) \Vert_2. \label{eq:I19}
\end{align}

Similarly, we bound $\mc{I}_{20}$ by
\begin{align}
\mc{I}_{20} 
& = \Vert A\cdot[\dot{\lambda}\{\tilde{e}(\mb X)\}\cdot\{\tilde{\mu}_1(\mb X)-\tilde{\mu}_0(\mb X)\}-\dot{\lambda}\{e(\mb X)\}\cdot\{\mu_1(\mb X)-\mu_0(\mb X)\}]\nonumber\\
& \quad
-[\dot{\lambda}\{\tilde{e}(\mb X)\}\cdot\{\tilde{\mu}_1(\mb X)-\tilde{\mu}_0(\mb X)\}\cdot\tilde{e}(\mb X)-\dot{\lambda}\{e(\mb X)\}\cdot\{\mu_1(\mb X)-\mu_0(\mb X)\}\cdot e(\mb X)] \Vert_2 \nonumber\\  
&\leq \Vert  \dot{\lambda}\{\tilde{e}(\mb X)\}\{\tilde{\mu}_1(\mb X)-\tilde{\mu}_0(\mb X)\}-\dot{\lambda}\{e(\mb X)\}\{\mu_1(\mb X)-\mu_0(\mb X)\}\Vert_2 \nonumber\\
& \quad + \Vert \dot{\lambda}\{\tilde{e}(\mb X)\} \cdot \{\tilde{\mu}_1(\mb X)-\tilde{\mu}_0(\mb X)\}\cdot \tilde{e}(\mb X)-  \dot{\lambda}\{e(\mb X)\} \cdot \{\mu_1(\mb X)-\mu_0(\mb X)\}\cdot e(\mb X) \Vert_2  \nonumber\\ 
&\leq \Vert\dot{\lambda}\{\tilde{e}(\mb X)\}\cdot [\{\tilde{\mu}_1(\mb X)-\tilde{\mu}_0(\mb X)\}-\{\mu_1(\mb X)-\mu_0(\mb X)\}] \Vert_2 + \Vert [\dot{\lambda}\{\tilde{e}(\mb X)\}-\dot{\lambda}\{e(\mb X)\}] \cdot \{\mu_1(\mb X)-\mu_0(\mb X)\} \Vert_2 \nonumber\\
&\quad + \Vert \dot{\lambda}\{\tilde{e}(\mb X)\} \cdot \tilde{e}(\mb X) \cdot[\{\tilde{\mu}_1(\mb X)-\tilde{\mu}_0(\mb X)\}-\{\mu_1(\mb X)-\mu_0(\mb X)\}]\Vert_2 \nonumber \\
& \qquad +\Vert
[\dot{\lambda}\{e(\mb X)\} e(\mb X)-\dot{\lambda}\{\tilde{e}(\mb X)\}\tilde{e}(\mb X)] \cdot \tau(\mb X)\Vert_2  \nonumber\\
& \leq
2\sup_\mb X\left\vert \dot{\lambda}\{\tilde{e}(\mb X)\}\vee \dot{\lambda}\{{e}(\mb X)\} \right\vert\cdot\sum_{a=0}^1  \Vert \tilde{\mu}_a(\mb X)-\mu_a(\mb X)\Vert_2 +
\sup_{\mb X}\left\vert \tau(\mb X)\right\vert \cdot \Vert \dot\lambda\{\tilde{e}(\mb X)\}-\dot\lambda\{e(\mb X)\} \Vert_2  \nonumber\\
&\quad+ \Vert \tau(\mb X) e(\mb X)\cdot[\dot{\lambda}\{e(\mb X)\} -\dot{\lambda}\{\tilde{e}(\mb X)\}] \Vert_2 +
\Vert \tau(\mb X) \dot{\lambda}\{\tilde{e}(\mb X)\}\cdot\{\tilde{e}(\mb X)-e(\mb X)\} \Vert_2  \nonumber\\
& \preceq \sum_{a=0}^1 \Vert\tilde{\mu}_a(\mb X)-\mu_a(\mb X) \Vert_2 + \Vert \tilde{e}(\mb X)-e(\mb X) \Vert_2 + \Vert \dot\lambda\{\tilde{e}(\mb X)\}-\dot\lambda\{e(\mb X)\} \Vert_2 \nonumber \\
& \preceq \sum_{a=0}^1 \Vert\tilde{\mu}_a(\mb X)-\mu_a(\mb X) \Vert_2 + \Vert \tilde{e}(\mb X)-e(\mb X) \Vert_2. 
\label{eq:I20}
\end{align}
From \eqref{eq:I16}--\eqref{eq:I20}, 
\begin{align}
	\mc{I}_{15} \preceq \sum_{a=0}^1 \Vert\tilde{\mu}_a(\mb X)-\mu_a(\mb X) \Vert_2 + \Vert  \tilde{e}(\mb X)-e(\mb X) \Vert_2.  \label{eq:I15-upbound}
\end{align}
Furthermore, by Conditions \ref{cond:DML-U} and \ref{cond:DMLgamma} as well as plugging \eqref{eq:I14} and \eqref{eq:I15-upbound} into \ref{eq:r_n_prime}, we get an upper bound on $r_n^\prime$ as $r_n^\prime \preceq\left\vert \gamma_0\right\vert \cdot \Vert \tilde{e}(\mb X)-e(\mb X) \Vert_2 +\sum_{a=0}^1 \Vert \tilde{\mu}_a(\mb X)-\mu_a(\mb X)\Vert_2.$ 

Finally, we bound $\lambda^\prime_n$ Assumption \ref{assp:Assp3.2-2018paper}(c) by Lemma \ref{lem:secondGate} and Cauchy-Schwarz inequality:

\begin{align*}
	& \left\vert \gamma_0 \E\left[\frac{\partial^2 }{\partial^2}D(\mb V;\phi)\right] - \E\left[\frac{\partial^2 }{\partial^2}N(\mb V;\phi)\right] \right\vert \\
 & \leq 
	2r\cdot \sup_{\mb X}\left\vert e(\mb X) \right\vert\cdot\sup_{\mb X}\left\vert \tilde{e}(\mb X)-e(\mb X)\right\vert \cdot \sup_\mb X\left\vert \frac{\dot{\lambda}\{e_r(\mb X)\}}{e_r(\mb X)^2} \right\vert \cdot 
	\Vert \tilde{e}(\mb X)-e(\mb X)\Vert_2\cdot  \Vert \tilde{\mu}_1(\mb X)-\mu_1(\mb X)\Vert_2\\
	&\quad + 2r \cdot \sup_{\mb X}\left\vert 1-e(\mb X) \right\vert\cdot\sup_{\mb X}\left\vert \tilde{e}(\mb X)-e(\mb X)\right\vert \cdot \sup_\mb X\left\vert \frac{\dot{\lambda}\{e_r(\mb X)\}}{\{1-e_r(\mb X)\}^2} \right\vert \cdot \Vert \tilde{e}(\mb X)-e(\mb X) \Vert_2 \cdot \Vert \tilde{\mu}_0(\mb X)-\mu_0(\mb X) \Vert_2 \\
	& \quad + 2\cdot \sup_{\mb X}\left\vert 1-e(\mb X) \right\vert\cdot\sup_\mb X \left\vert\frac{\dot{\lambda}\{e_r(\mb X)\}}{1-e_r(\mb X)}  \right\vert \cdot \Vert \tilde{e}(\mb X)-e(\mb X)\Vert_2 \cdot \Vert\tilde{\mu}_0(\mb X)-\mu_0(\mb X)\Vert_2  \\
	&\quad+2 \cdot \sup_{\mb X}\left\vert e(\mb X) \right\vert \cdot\sup_\mb X \left\vert  \frac{\dot{\lambda}\{e_r(\mb X)\}}{e_r(\mb X)} \right\vert   \cdot  \Vert \tilde{e}(\mb X)-e(\mb X)\}\Vert_2\cdot   \Vert \tilde{\mu}_1(\mb X)-\mu_1(\mb X)\Vert_2  \\
	&\quad+3r\cdot \sup_{\mb X}\left\vert \tilde{e}(\mb X)-e(\mb X) \right\vert\cdot \sup_\mb X \left\vert \ddot{\lambda}\{e_r(\mb X)\} \right\vert   \cdot  \Vert \tilde{e}(\mb X)-e(\mb X)\Vert_2\cdot  \Vert\tilde{\mu}_1(\mb X)-\mu_1(\mb X)\Vert_2\\
	&\quad+3r\cdot \sup_{\mb X}\left\vert \tilde{e}(\mb X)-e(\mb X) \right\vert\cdot \sup_\mb X\left\vert \ddot{\lambda}\{e_r(\mb X)\} \right\vert\cdot   \Vert\tilde{e}(\mb X)-e(\mb X)\Vert_2 \cdot \Vert \tilde{\mu}_0(\mb X)-\mu_0(\mb X)\Vert_2 \\
	&\quad+r \cdot \sup_{\mb X}\left\vert e(\mb X) \right\vert\cdot\sup_{\mb X}\left\vert \tilde{e}(\mb X)-e(\mb X)\right\vert \cdot \sup_\mb X\left\vert\frac{\ddot{\lambda}\{e_r(\mb X)\}}{e_r(\mb X)} \right\vert\cdot \Vert \tilde{e}(\mb X)-e(\mb X)\Vert_2   \cdot \Vert\tilde{\mu}_1(\mb X)-\mu_1(\mb X)\Vert_2 \\
	&\quad+r\cdot \sup_{\mb X}\left\vert 1-e(\mb X) \right\vert\cdot\sup_{\mb X}\left\vert \tilde{e}(\mb X)-e(\mb X)\right\vert \cdot  \sup_\mb X \left\vert \frac{\ddot{\lambda}\{e_r(\mb X)\}}{1-e_r(\mb X)} \right\vert \cdot \big\Vert  \tilde{e}(\mb X)-e(\mb X)\big\Vert_2  \cdot \big\Vert \tilde{\mu}_0(\mb X)-\mu_0(\mb X) \big\Vert_2\\
	&\quad+2r \cdot \sup_{\mb X}\left\vert e(\mb X) \right\vert\cdot\sup_{\mb X}\left\vert \tilde{e}(\mb X)-e(\mb X)\right\vert \cdot \sup_\mb X \left\vert \frac{\lambda\{e_r(\mb X)\} }{ e_r(\mb X)^3 } \right\vert   \cdot  \Vert \tilde{e}(\mb X)-e(\mb X)\Vert_2 \cdot \Vert\tilde{\mu}_1(\mb X)-\mu_1(\mb X) \Vert_2\\
	&\quad+2\cdot \sup_{\mb X}\left\vert e(\mb X) \right\vert\cdot   \sup_\mb X \left\vert  \frac{\lambda\{e_r(\mb X)\} }{ e_r(\mb X)^2 } \right\vert   \cdot \Vert \tilde{e}(\mb X)-e(\mb X)\Vert_2\cdot \Vert \tilde{\mu}_1(\mb X)-\mu_1(\mb X)\Vert_2\\
	&\quad+2r \cdot \sup_{\mb X}\left\vert 1-e(\mb X) \right\vert\cdot\sup_{\mb X}\left\vert \tilde{e}(\mb X)-e(\mb X)\right\vert \cdot   \sup_\mb X \left\vert \frac{\lambda\{e_r(\mb X)\}}{ \{1-e_r(\mb X)\}^3 }  \right\vert \cdot \Vert \tilde{e}(\mb X)-e(\mb X)\Vert_2  \cdot \Vert \tilde{\mu}_0-\mu_0(\mb X)\Vert_2\\
	&\quad+2\cdot \sup_{\mb X}\left\vert 1-e(\mb X) \right\vert \cdot \sup_\mb X \left\vert \frac{\lambda\{e_r(\mb X)\} }{ \{1-e_r(\mb X)\}^2 } \right\vert\cdot \Vert \tilde{e}(\mb X)-e(\mb X)\Vert_2 \cdot \Vert \tilde{\mu}_0(\mb X)-\mu_0(\mb X)\Vert_2 \\
	& \quad+r^2\cdot \sup_{\mb X}\left\vert \tilde{e}(\mb X)-e(\mb X)\right\vert^2 \cdot  \sup_\mb X \left\vert \lambda^{(3)}\{e_r(\mb X)\} \right\vert  \cdot  \Vert \tilde{e}(\mb X)-e(\mb X)\Vert_2 \cdot \Vert \tilde{\mu}_{1}(\mb X)-\mu_1(\mb X)\Vert_2\\
	&\quad+r^2\cdot \sup_{\mb X}\left\vert \tilde{e}(\mb X)-e(\mb X)\right\vert^2\cdot   \sup_\mb X \left\vert \lambda^{(3)}\{e_r(\mb X)\}\right\vert  \cdot   \Vert \tilde{e}(\mb X)-e(\mb X)\Vert_2 \cdot \Vert\tilde{\mu}_{0}(\mb X)-\mu_0(\mb X) \Vert_2\\
	& \quad+\bigg\vert  -\gamma_0\cdot \E\left[r\cdot \lambda^{(3)}\{e_r(\mb X)\}\{\tilde{e}(\mb X)-e(\mb X)\}^3+\ddot{\lambda}\{e_r(\mb X)\}\{\tilde{e}(\mb X)-e(\mb X)\}^2\right]  \\
	& \qquad + \E\left(\left[r\cdot \lambda^{(3)}\{e_r(\mb X)\}\{\tilde{e}(\mb X)-e(\mb X)\}^3 +\ddot{\lambda}\{e_r(\mb X)\}\{\tilde{e}(\mb X)-e(\mb X)\}^2\right] \tau(\mb X)\right)\bigg\vert,
\end{align*}
 which gives that
\begin{align*}
    & \left\vert \gamma_0 \E\left[\frac{\partial^2 }{\partial^2}D(\mb V;\phi)\right] - \E\left[\frac{\partial^2 }{\partial^2}N(\mb V;\phi)\right] \right\vert \\
	& \preceq \Vert \tilde{e}(\mb X)-e(\mb X)\Vert_2\cdot \sum_{a=0}^1 \Vert\mu_{a}(\mb X)-\mu_a(\mb X) \Vert_2 \\
   & \quad + \sup_{r\in[0,1]}\left\{\Vert\lambda^{(3)}\{e_r(\mb X)\}\Vert_2\cdot\Vert\tilde{e}(\mb X)-e(\mb X)\Vert_2 + \Vert\ddot\lambda\{e_r(\mb X)\}\Vert_2\right\}\Vert\tilde{e}(\mb X)-e(\mb X)\Vert_2^2,
 \end{align*}
where the last three inequalities come from Conditions \ref{cond:DMLbound}, \ref{cond:DMLLips}, \ref{cond:DMLgamma} and \ref{cond:DMLconverge}. Furthermore, by Condition \ref{cond:DMLconverge}, the inequality above gives the upper bound to $\lambda_n^\prime$ in (c) in Assumption \ref{assp:Assp3.2-2018paper}. Hence we finish its verification. 

To summary, that we have shown that
\begin{align*}
\lambda_n &=0,\\
r_n &\preceq  \Vert \tilde{e}(\mb X)-e(\mb X) \Vert_2 \\
r_n^\prime &\preceq \Vert \tilde{e}(\mb X)-e(\mb X) \Vert_2 +\sum_{a=0}^1  \Vert \tilde{\mu}_a(\mb X)-\mu_a(\mb X)\Vert_2, \quad\text{ and } \\
\lambda_n^\prime &\preceq \Vert \tilde{e}(\mb X)-e(\mb X)\Vert_2\cdot \sum_{a=0}^1 \Vert\mu_{a}(\mb X)-\mu_a(\mb X) \Vert_2 \\
& \quad + \sup_{r\in[0,1]}\left\{\Vert\lambda^{(3)}\{e_r(\mb X)\}\Vert_2\cdot\Vert\tilde{e}(\mb X)-e(\mb X)\Vert_2 + \Vert\ddot\lambda\{e_r(\mb X)\}\Vert_2\right\}\Vert\tilde{e}(\mb X)-e(\mb X)\Vert_2^2.
\end{align*}
Hence by Condition \ref{cond:DMLconverge}, there must be a sequence $\{\delta_n^\prime\}_{n\geq 1}$ such that 
\begin{align*}
\lambda_n =0,\qquad
r_n\leq   \delta_n^\prime,\qquad 
r_n^\prime \leq \delta_n^\prime,\quad\text{ and }\quad
\lambda_n^\prime \leq   \delta_n^\prime/\sqrt{n}.
\end{align*}
We can replace the sequence $\{\delta_n\}_{n\geq 1}$ in Assumptions \ref{assp:Assp3.1-2018paper} and \ref{assp:Assp3.2-2018paper} by $M\{\delta_n^\prime\vee N^{-[(1-{2}/{q})\wedge{1}/{2}]}\}$ for a large constant $M$. We now completed the proof of (c) in Assumption \ref{assp:Assp3.1-2018paper}. 

In all the proofs above, if we target ATE, ATT or ATC, since $\dot\lambda\{e(\mb X)\}\equiv 0$, we can see that the proofs can be simplified substantially and the final bound of $\lambda_n'$ can be simplified to the standard product-type error term of propensity score and outcome models. 

\subsubsection{Proof of (d) in Assumption \ref{assp:Assp3.2-2018paper}}
Finally, we verify (d) in Assumption \ref{assp:Assp3.2-2018paper}. We decompose the variance as follows. 
\begin{align}
\E[L(\mb V;\gamma_0,\phi_0)^2]
&=\gamma_0^2 \cdot \mc{I}_{21} -2\gamma_0\cdot \mc{I}_{22} +\mc{I}_{23} , \label{eq:Eg2}
\end{align}
where 
\begin{align*}
	\mc{I}_{21} &=  \E [D(\mb V;\phi_0)^2],\\
	\mc{I}_{22} &= \E [D(\mb V;\phi_0)\cdot N(\mb V;\phi_0)],\\
	\mc{I}_{23} &= \E [N(\mb V;\phi_0)^2].
\end{align*}
First, we can simplify $\mc{I}_{21}$ as
\begin{align}
\mc{I}_{21}
&=\E[\lambda\{e(\mb X)\}^2] +\E(\dot{\lambda}\{e(\mb X)\}^2\cdot\{A^2-2Ae(\mb X)+e(\mb X)^2\}) +
2\E [\lambda\{e(\mb X)\}\cdot\dot{\lambda}\{e(\mb X)\}\{A-e(\mb X)\}] \nonumber\\
&= \E[\lambda\{e(\mb X)\}^2] +\E(\dot{\lambda}\{e(\mb X)\}^2\cdot e(\mb X)\{1-e(\mb X)\}). \label{eq:I21-Eg2}
\end{align}
$\mc{I}_{22}$ can be simplified as
\begin{align}
\mc{I}_{22}
&=\E ([\lambda\{e(\mb X)\}+\dot{\lambda}\{e(\mb X)\}\{A-e(\mb X)\}]\cdot [ \lambda\{e(\mb X)\}\psi_{\tau}(\mb X;\phi_0)+\dot{\lambda}\{e(\mb X)\}\tau(\mb X)\{A-e(\mb X)\}  ]) \nonumber\\
&=\E\left(\lambda\{e(\mb X)\}^2\cdot \left[ \frac{A}{e(\mb X)}\{Y-\mu_1(\mb X)\}-\frac{ 1-A  }{1-e(\mb X)}\{Y-\mu_0(\mb X)\}+\tau(\mb X) \right]  \right) \nonumber\\
&  \quad + \E [\lambda\{e(\mb X)\}\dot{\lambda}\{e(\mb X)\}\cdot\tau(\mb X)\{A-e(\mb X)\}]+\E(\dot{\lambda}\{e(\mb X)\}^2\cdot\tau(\mb X)\{A^2+e(\mb X)^2-2Ae(\mb X)\}) \nonumber\\
& \quad+\E \left(\dot{\lambda}\{e(\mb X)\}\lambda\{e(\mb X)\}\{A-e(\mb X)\}\cdot \left[ \frac{A}{e(\mb X)}\{Y-\mu_1(\mb X)\}-\frac{1-A }{1-e(\mb X)}\{Y-\mu_0(\mb X)\}+\tau(\mb X)\right]\right) \nonumber\\
&=\E[\lambda\{e(\mb X)\}^2\tau(\mb X)]+\E (\dot{\lambda}\{e(\mb X)\}^2\tau(\mb X)\{1-e(\mb X)\}e(\mb X)). \label{eq:I22-Eg2}
\end{align}
We further decompose $\mc{I}_{23}$ as
\begin{align*}
	\mc{I}_{23}&=\mc{I}_{24} + \mc{I}_{25} + \mc{I}_{26},
\end{align*}
where 
\begin{align*}
	\mc{I}_{24}  &= \E ([\dot{\lambda}\{e(\mb X)\}\tau(\mb X)\{A-e(\mb X)\}]^2)= \E [\dot{\lambda}\{e(\mb X)\}^2\tau(\mb X)^2 \cdot e(\mb X)\{1-e(\mb X)\}],\\
	\mc{I}_{25}  &= 2\E[\lambda\{e(\mb X)\} \psi_{\tau}(\mb X) \dot{\lambda}\{e(\mb X)\}\cdot \tau(\mb X)\{A-e(\mb X)\} ]= 0,\\
	\mc{I}_{26}  &= \E([\lambda\{e(\mb X)\}\psi_{\tau}(\mb X)]^2).
\end{align*}
Furthermore, $\mc{I}_{26}$ can be simplified as
\begin{align*}
	\mc{I}_{26} 
	& = \E\left[\lambda\{e(\mb X)\}^2\frac{A}{e(\mb X)^2}\{Y-\mu_1(\mb X)\}^2\right] +
	\E \left[\lambda\{e(\mb X)\}^2\frac{1-A}{\{1-e(\mb X)\}^2}\{Y-\mu_0(\mb X)\}^2\right]\\
	&   \quad + \E([\lambda\{e(\mb X)\}\tau(\mb X)]^2) 
	+ 2\E \left(\lambda\{e(\mb X)\}\tau(\mb X)\cdot \left[\frac{A}{e(\mb X)}\{Y-\mu_1(\mb X)\}  - \frac{1-A}{1-e(\mb X)}\{Y-\mu_0(\mb X)\}\right]\right)
	\\
	&= \E\left(\left[\frac{\lambda\{e(\mb X)\}}{e(\mb X)}\right]^2 A\{Y-\mu_1(\mb X)\}^2\right) \\
	& \quad +
	\E\left(\left[\frac{\lambda\{e(\mb X)\}}{1-e(\mb X)}\right]^2 (1-A)\{Y-\mu_0(\mb X)\}^2\right) + \E ([\lambda\{e(\mb X)\}\tau(\mb X)]^2), 
\end{align*}
so that we can obtain
\begin{align}
\mc{I}_{23}
&=
\E\left(\left[\frac{\lambda\{e(\mb X)\}}{e(\mb X)}\right]^2 A\{Y-\mu_1(\mb X)\}^2\right)+
	\E\left(\left[\frac{\lambda\{e(\mb X)\}}{1-e(\mb X)}\right]^2 (1-A)\{Y-\mu_0(\mb X)\}^2\right) \nonumber\\
&  \quad + \E ([\lambda\{e(\mb X)\}\tau(\mb X)]^2)+ \E [\dot{\lambda}\{e(\mb X)\}^2\tau(\mb X)^2 \cdot e(\mb X)\{1-e(\mb X)\}].  \label{eq:I23-Eg2}
\end{align}
Plugging \eqref{eq:I21-Eg2}--\eqref{eq:I23-Eg2} into \eqref{eq:Eg2} gives
\begin{align*}
	\E [L(\mb V;\gamma_0,\phi_0)^2]
	& = \E (\gamma_0^2 \lambda\{e(\mb X)\}^2 - 2\gamma_0\lambda\{e(\mb X)\}^2\tau(\mb X)+[\lambda\{e(\mb X)\}\tau(\mb X)]^2)\\
	&  \quad + \E(e(\mb X)\{1-e(\mb X)\}\cdot [\gamma_0^2\dot{\lambda}\{e(\mb X)\}^2-2\gamma_0\dot{\lambda}\{e(\mb X)\}^2\tau(\mb X)+[\dot{\lambda}\{e(\mb X)\}\tau(\mb X)]^2])\\
	&  \quad +  \E\left(\left[\frac{\lambda\{e(\mb X)\}}{e(\mb X)}\right]^2 A\{Y-\mu_1(\mb X)\}^2\right)+
	\E\left(\left[\frac{\lambda\{e(\mb X)\}}{1-e(\mb X)}\right]^2 (1-A)\{Y-\mu_0(\mb X)\}^2\right) \\
	&=\E([\gamma_0-\tau(\mb X)]^2\cdot [\lambda\{e(\mb X)\}^2 +  \dot{\lambda}\{e(\mb X)\}^2\cdot e(\mb X)\{1-e(\mb X)\}])\\
	&    \quad + 
	\E\left(\left[\frac{\lambda\{e(\mb X)\}}{e(\mb X)}\right]^2 A\{Y-\mu_1(\mb X)\}^2\right)+
	\E\left(\left[\frac{\lambda\{e(\mb X)\}}{1-e(\mb X)}\right]^2 (1-A)\{Y-\mu_0(\mb X)\}^2\right).
\end{align*}
Therefore, it can be seen that there exists $C_0>0$, such that $\E [L(\mb V;\gamma_0,\phi_0)^2]\geq C_0>0$, which verifies (d) in Assumption \ref{assp:Assp3.2-2018paper}. 

Hence, we completed all verifications of Assumptions \ref{assp:Assp3.1-2018paper} and \ref{assp:Assp3.2-2018paper}, and so the proof of Theorem \ref{thm:RDR-DML} {in Section \ref{sec:DML}} (the RDR property of DML estimators) is completed. 

\section{Complete Simulation Results}\label{app:fullSimu}

\subsection{Propensity score distributions}\label{subapp:PSdist}

\begin{figure}[H]
    \centering
    \includegraphics[width=\linewidth]{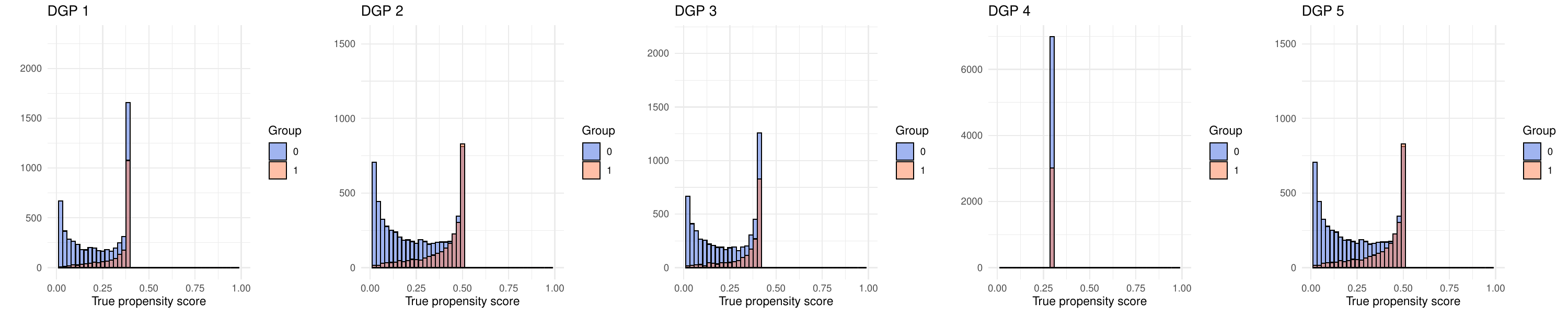}
    \caption{True propensity score distributions. Each panel is by generating a random sample ($n=10000)$.}
    \label{fig:PS-true}
\end{figure}

\begin{figure}[H]
    \centering
    \includegraphics[width=\linewidth]{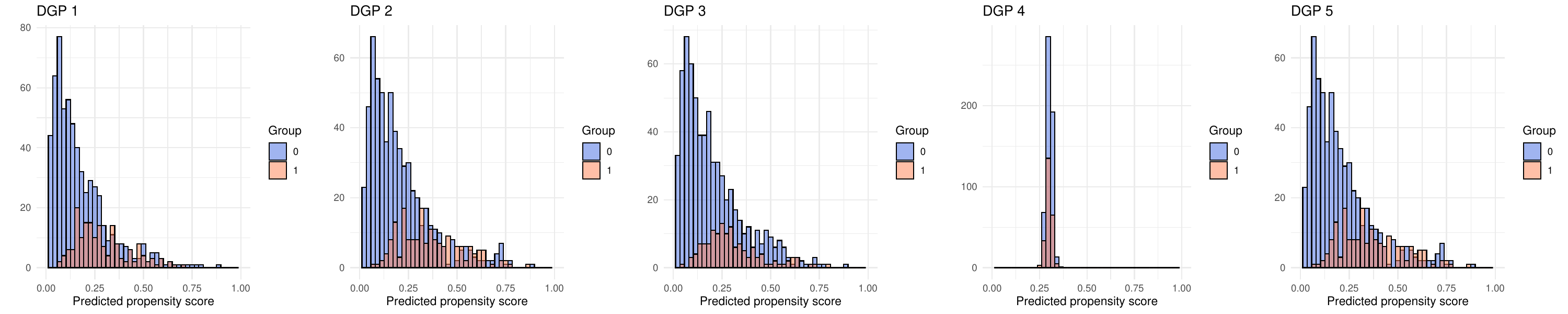}
    \caption{Predicted propensity score distributions by treatment group using simple GLM (logistic regression). Each panel is by the prediction set of a random sample ($n=4000)$, where 80\% of the data are used for the training set. }
    \label{fig:PS-glm}
\end{figure}

Comparing Figure \ref{fig:PS-true} to Figures \ref{fig:PS-glm} and \ref{fig:PS-ensb}, we observe that, at first glance, the true and fitted propensity score distributions for DGPs 1, 2, 3, and 5 differ notably. This is expected, as the propensity score generating functions in these DGPs restrict the true propensity score values to be less than 0.5. Consequently, the true propensity score distributions exhibit heavy tails concentrated near a boundary point  below 0.5. When fitting propensity score models using the observed data, it becomes challenging to closely approximate the true propensity scores for observations located near this boundary. This mismatch highlights the inherent difficulty in estimating propensity scores accurately when the true distribution is highly skewed or bounded away from the central region. However, away from the boundary region, Figures \ref{fig:PS-glm} and \ref{fig:PS-ensb} show that the propensity score distributions are fitted fairly well. 

\begin{figure}[H]
    \centering
    \includegraphics[width=\linewidth]{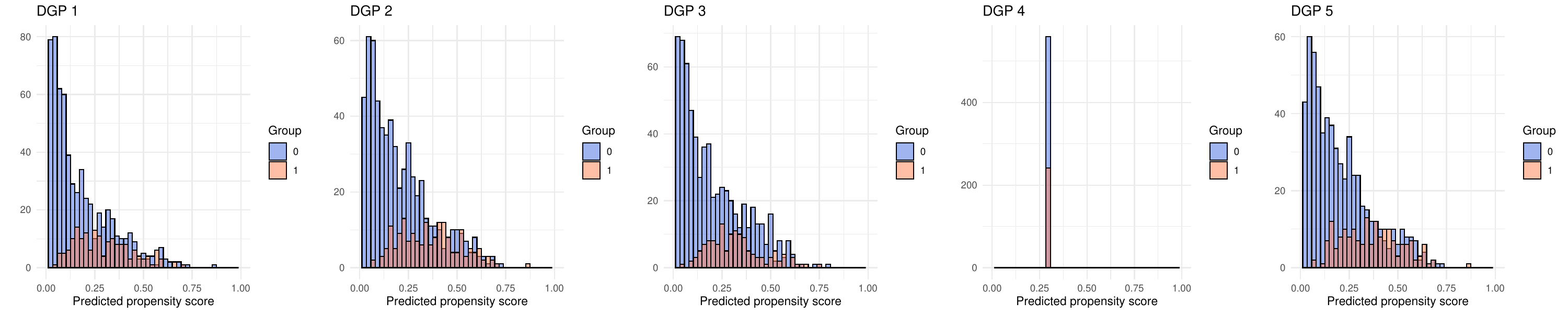}
    \caption{Predicted propensity score distributions by treatment group using methods ensemble learning. Each panel is by the prediction set of a random sample ($n=4000)$, where 80\% of the data are used for the training set.}
    \label{fig:PS-ensb}
\end{figure}

\subsection{Covariate balance by propensity score weights}\label{subapp:covbal}

In this section, we report the covariate balance results for all DGPs and estimands considered in our simulation study. Figures \ref{fig:covbal-glm} and \ref{fig:covbal-ensb} display the absolute standardized mean differences (ASMDs) under all scenarios, using a simple GLM and ensemble learning (based on SuperLearner methods in our simulation), respectively, for training and predicting the propensity scores.

\begin{figure}[H]
    \centering
    \includegraphics[width=\linewidth]{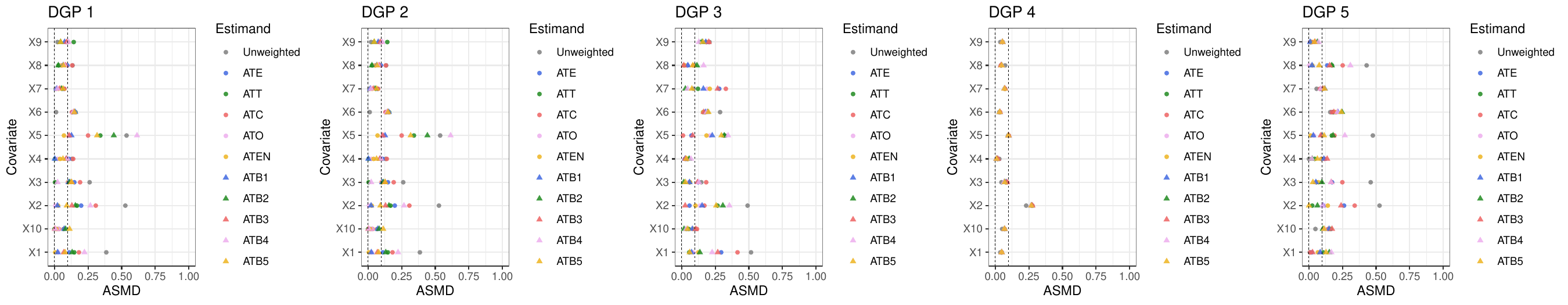}
    \caption{Love plots for checking covariate balance by absolute  standardized mean difference (ASMD). The propensity scores for constructing the balancing weights in each panel is by using a simple GLM model, and a sample-splitting is applied on a random sample $(n=4000)$ with the 80\% of the data are used for the training set. }
    \label{fig:covbal-glm}
\end{figure}

\begin{figure}[H]
    \centering
    \includegraphics[width=\linewidth]{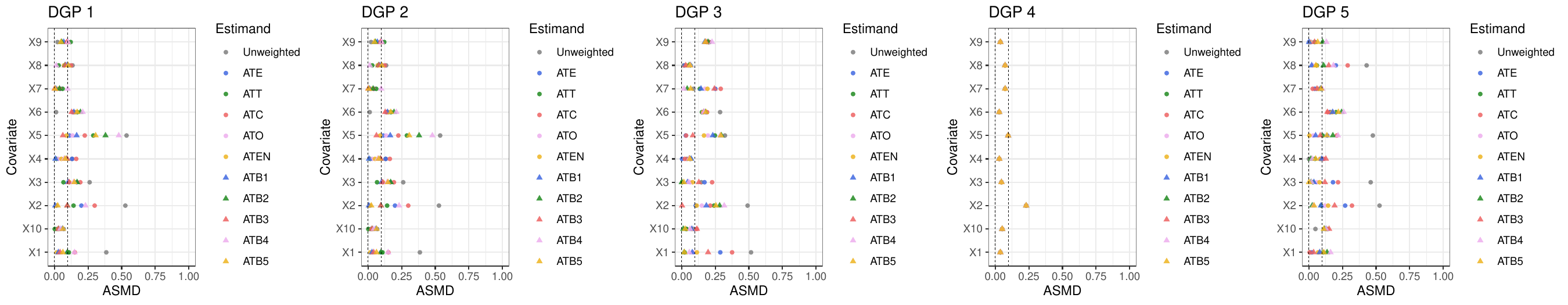}
    \caption{Love plots for checking covariate balance by absolute standardized mean difference (ASMD). The propensity scores for constructing the balancing weights in each panel is by using methods ensemble learning, and a sample-splitting is applied on a random sample $(n=4000)$ with the 80\% of the data are used for the training set. }
    \label{fig:covbal-ensb}
\end{figure}

\subsection{Results by using simple GLMs for nuisance function models}\label{subapp:simpleGLM}

\begin{figure}[H]
    \centering
    \includegraphics[width=\linewidth]{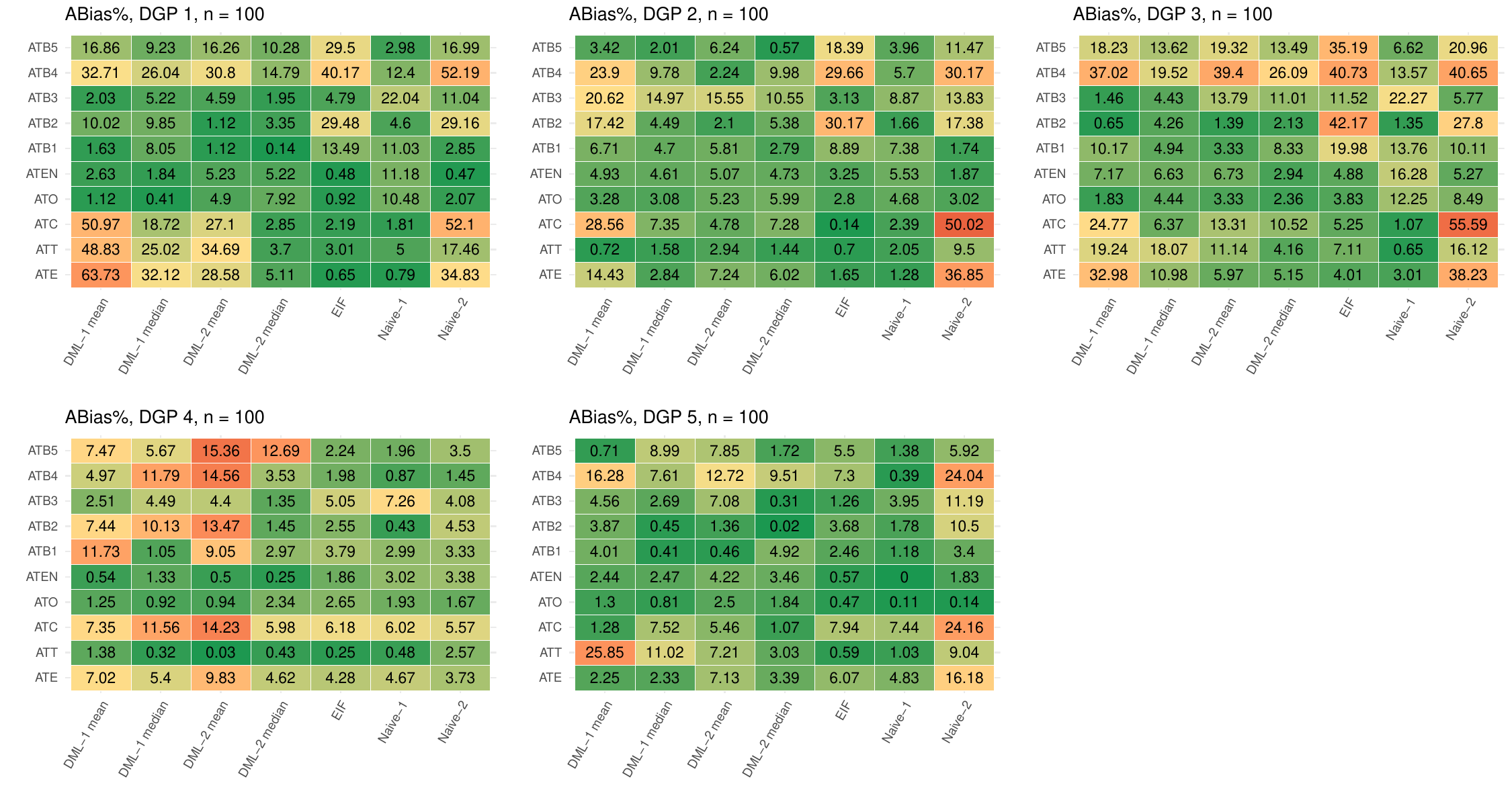}
    \caption{Simulation results of ARBias\% under $n=100$ and simple GLM models. }
\end{figure}

\begin{figure}[H]
    \centering
    \includegraphics[width=\linewidth]{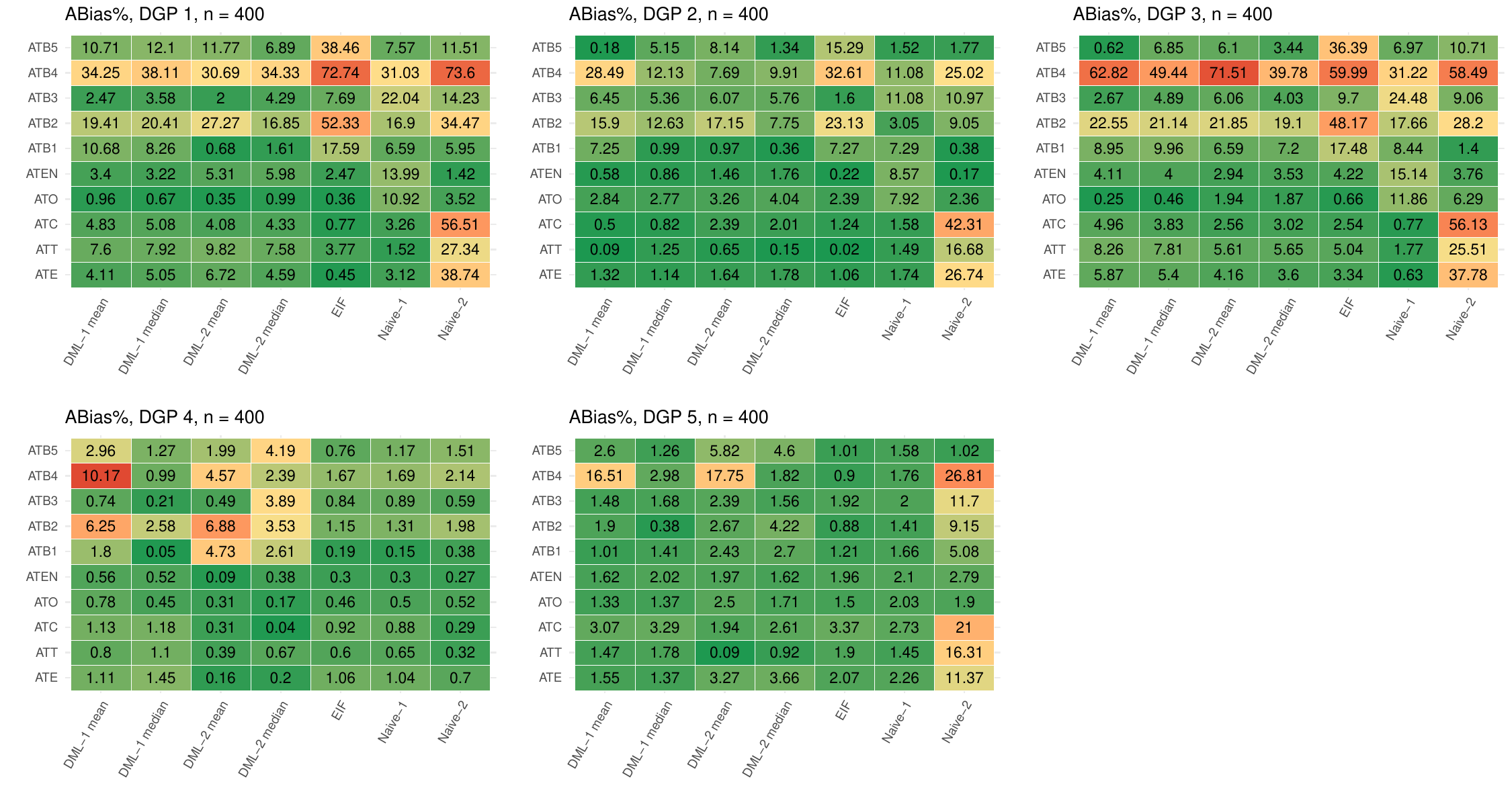}
    \caption{Simulation results of ARBias\% under $n=400$ and simple GLM models. }
\end{figure}

\begin{figure}[H]
    \centering
    \includegraphics[width=\linewidth]{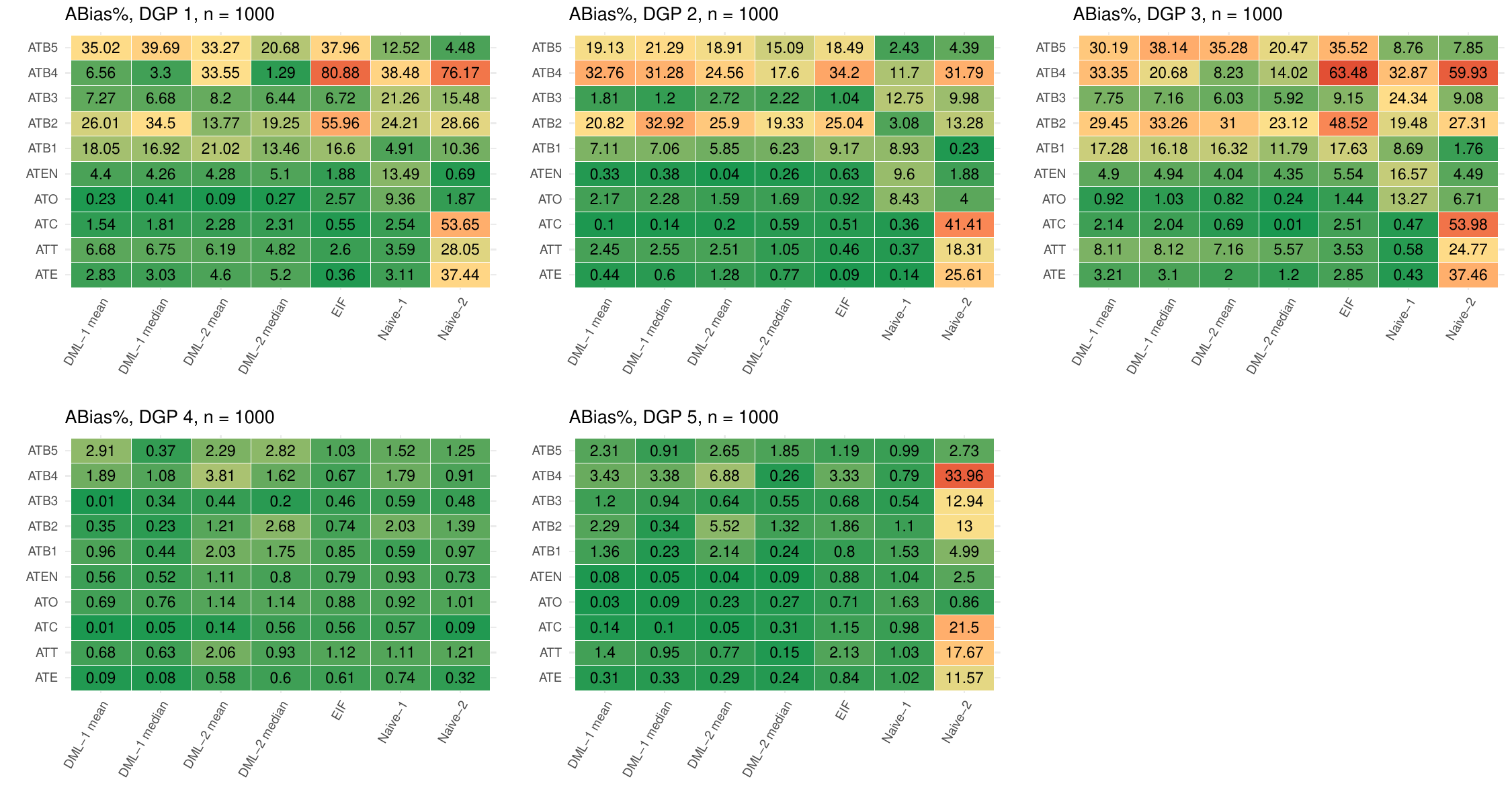}
    \caption{Simulation results of ARBias\% under $n=1000$ and simple GLM models. }
\end{figure}

\begin{figure}[H]
    \centering
    \includegraphics[width=\linewidth]{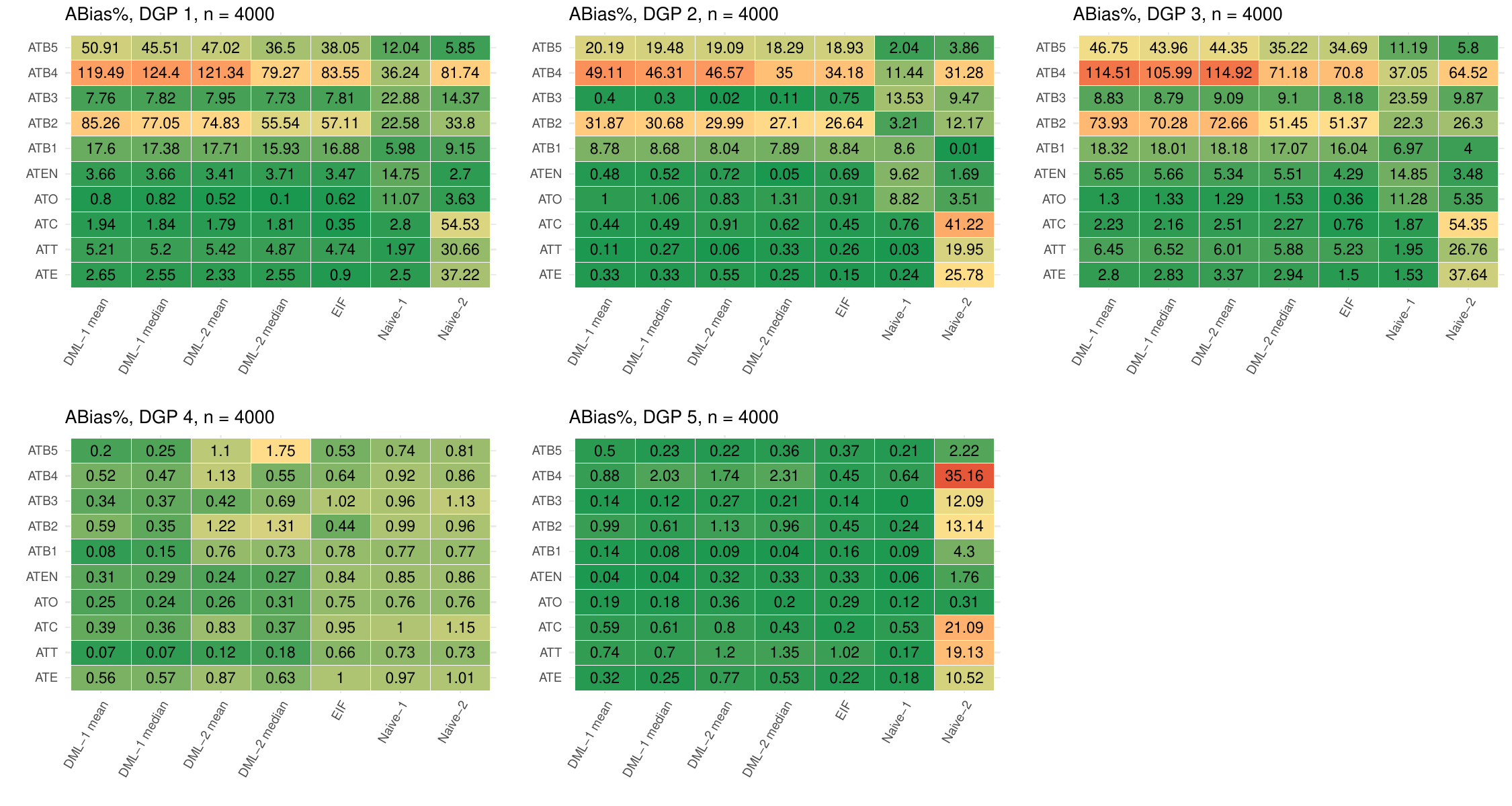}
    \caption{Simulation results of ARBias\% under $n=4000$ and simple GLM models. }
\end{figure}

\begin{figure}[H]
    \centering
    \includegraphics[width=\linewidth]{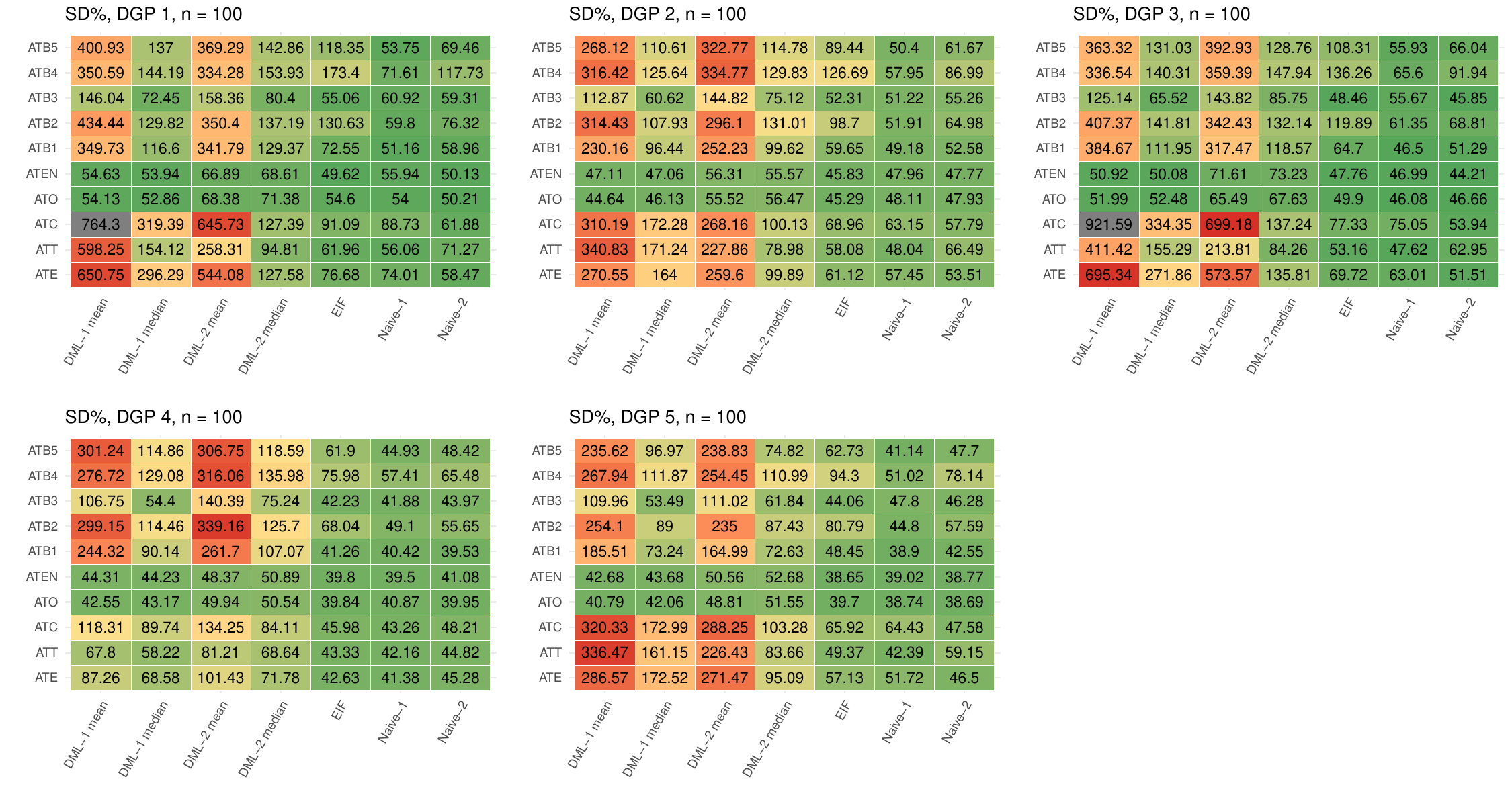}
    \caption{Simulation results of SD\% under $n=100$ and simple GLM models. }
\end{figure}

\begin{figure}[H]
    \centering
    \includegraphics[width=\linewidth]{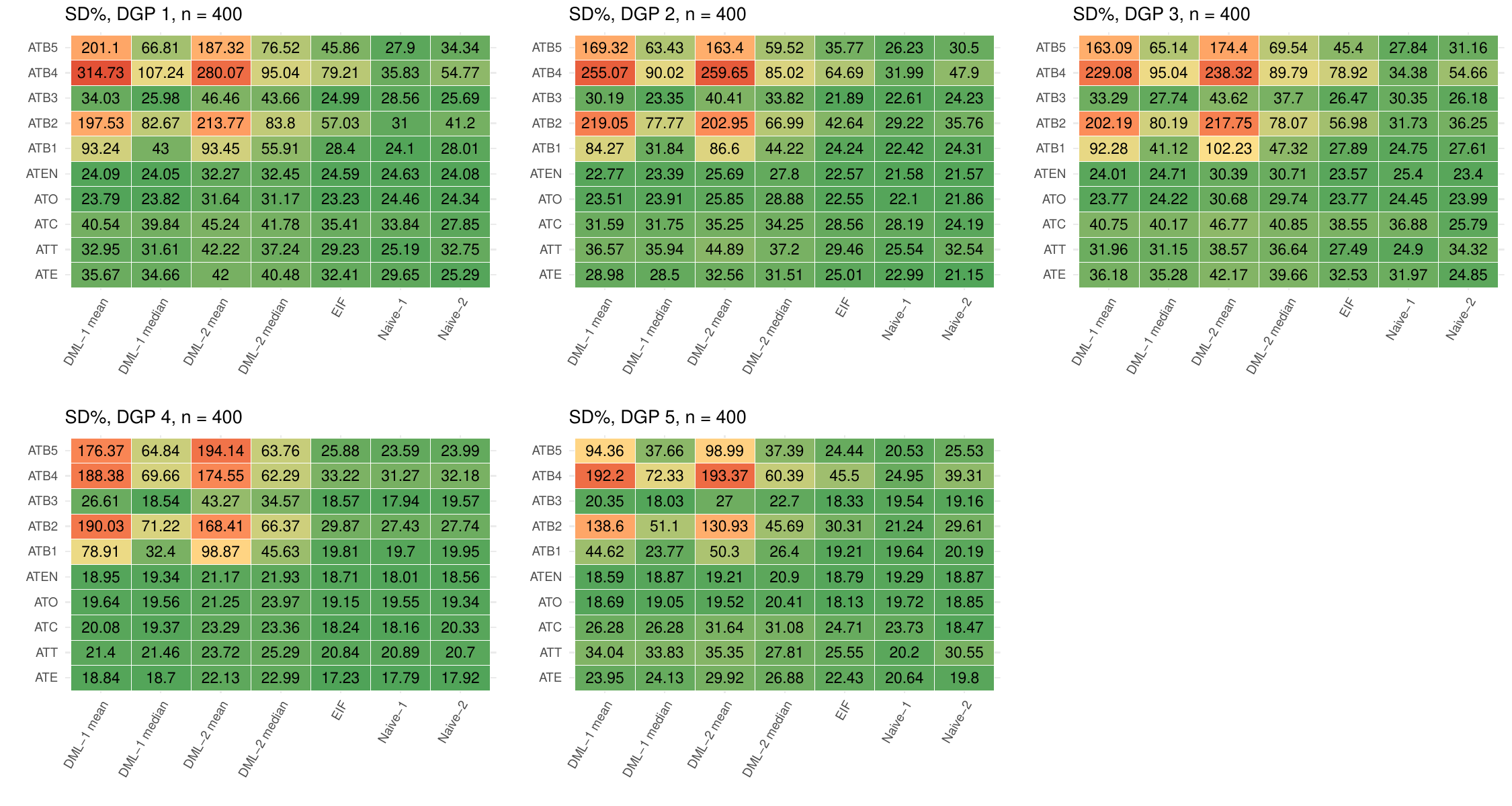}
    \caption{Simulation results of SD\% under $n=400$ and simple GLM models. }
\end{figure}

\begin{figure}[H]
    \centering
    \includegraphics[width=\linewidth]{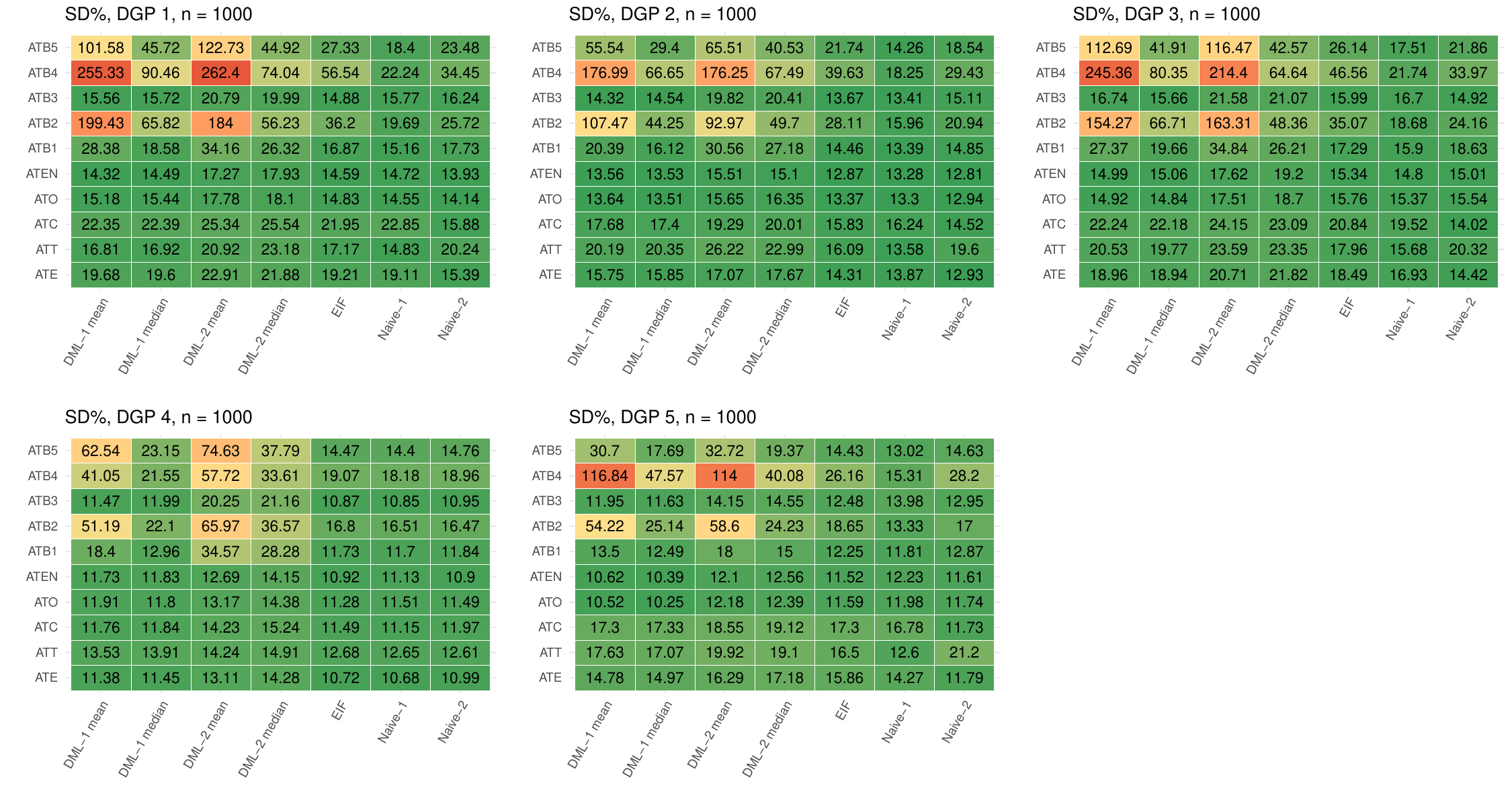}
    \caption{Simulation results of SD\% under $n=1000$ and simple GLM models. }
\end{figure}

\begin{figure}[H]
    \centering
    \includegraphics[width=\linewidth]{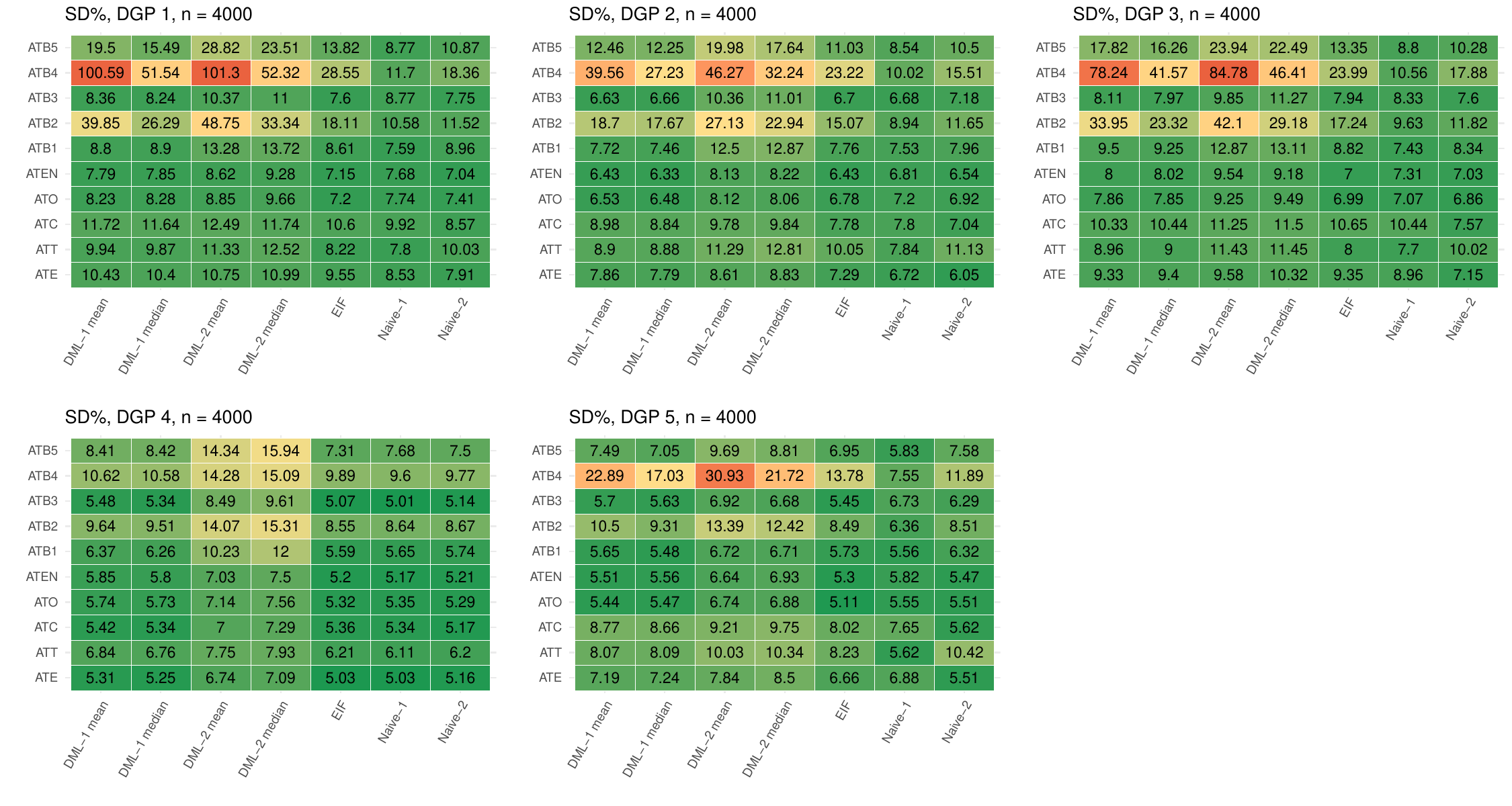}
    \caption{Simulation results of SD\% under $n=4000$ and simple GLM models. }
\end{figure}

\begin{figure}[H]
    \centering
    \includegraphics[width=\linewidth]{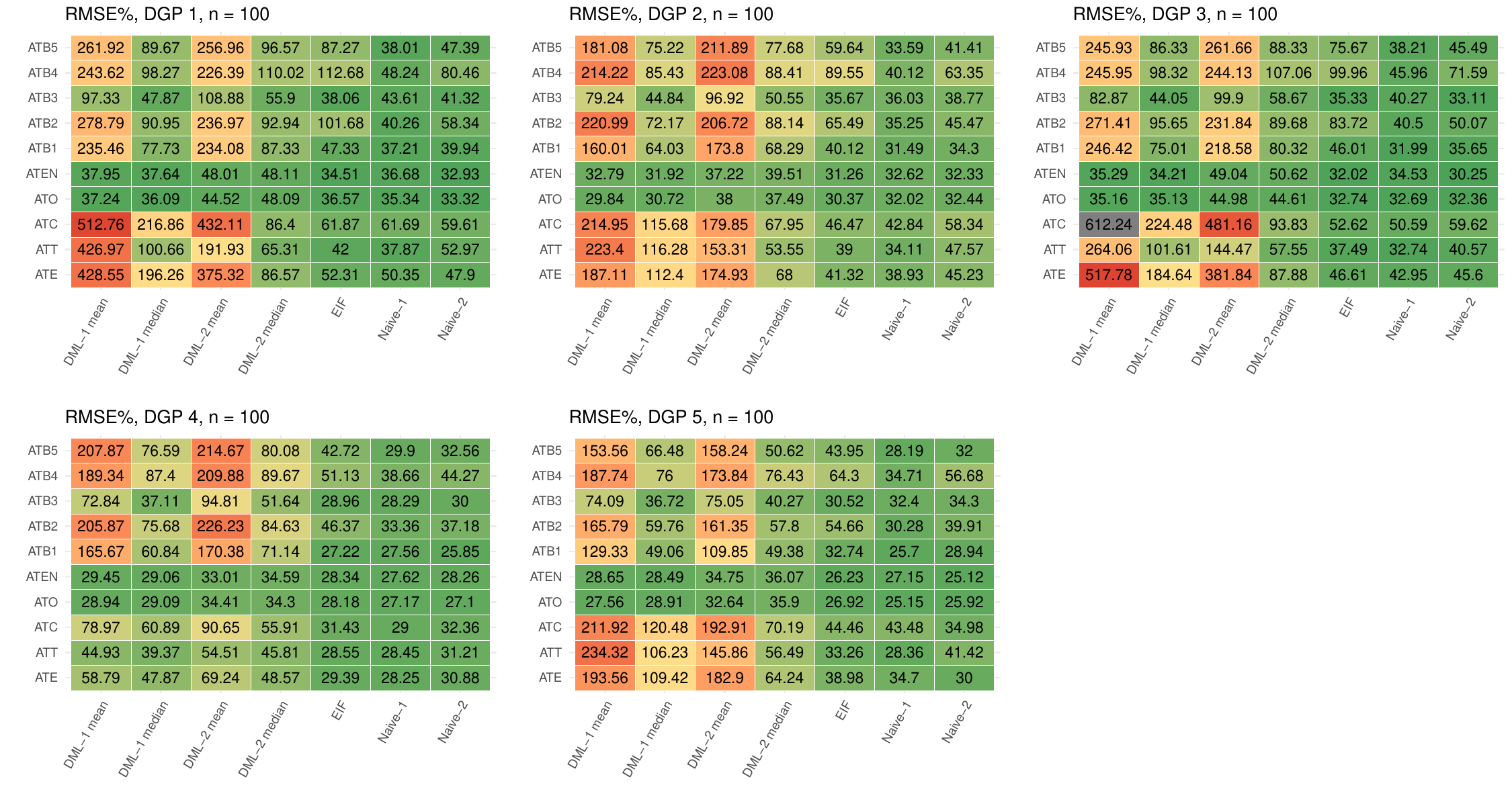}
    \caption{Simulation results of RMSE\% under $n=100$ and simple GLM models. }
\end{figure}

\begin{figure}[H]
    \centering
    \includegraphics[width=\linewidth]{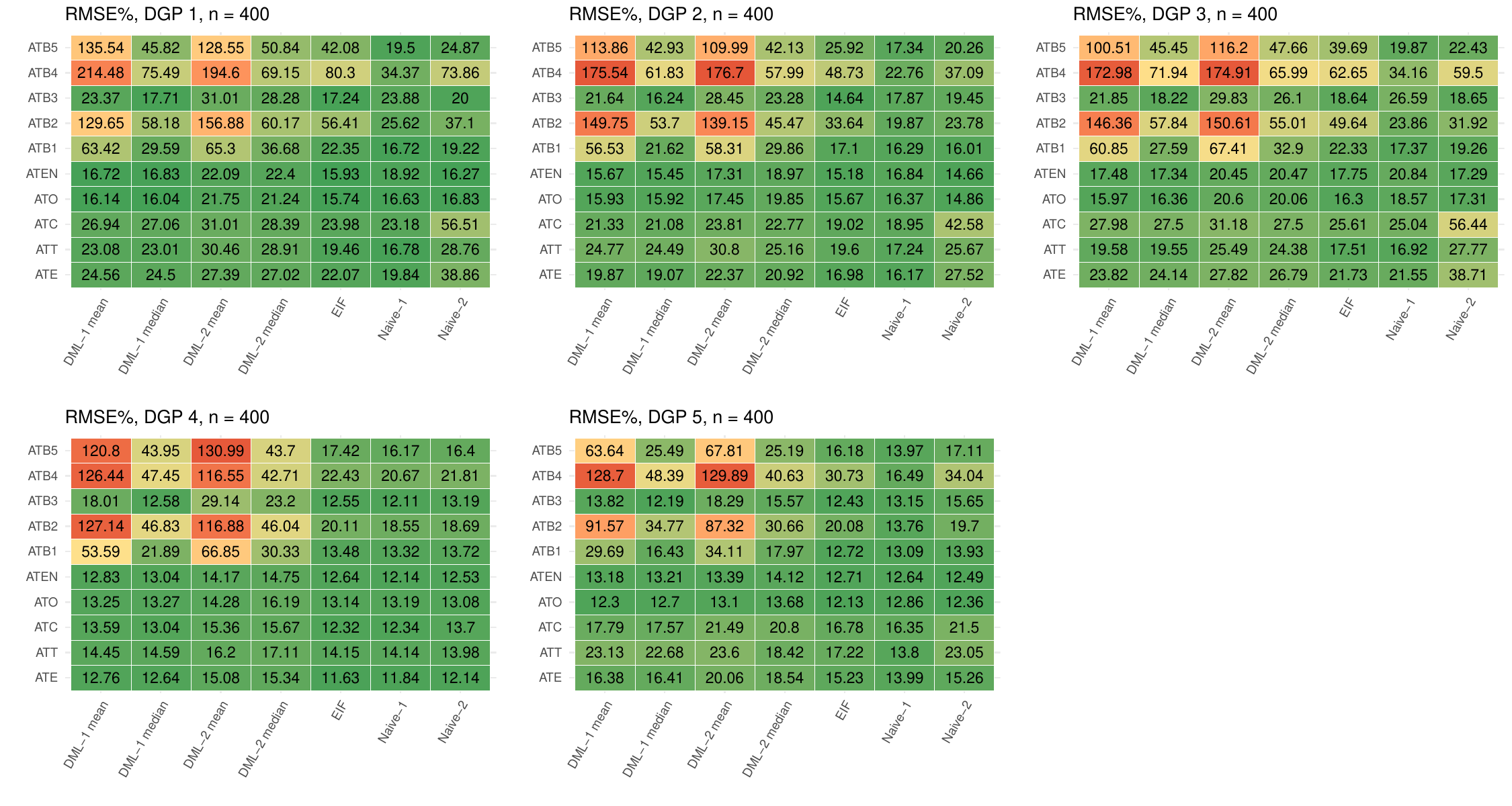}
    \caption{Simulation results of RMSE\% under $n=400$ and simple GLM models. }
\end{figure}

\begin{figure}[H]
    \centering
    \includegraphics[width=\linewidth]{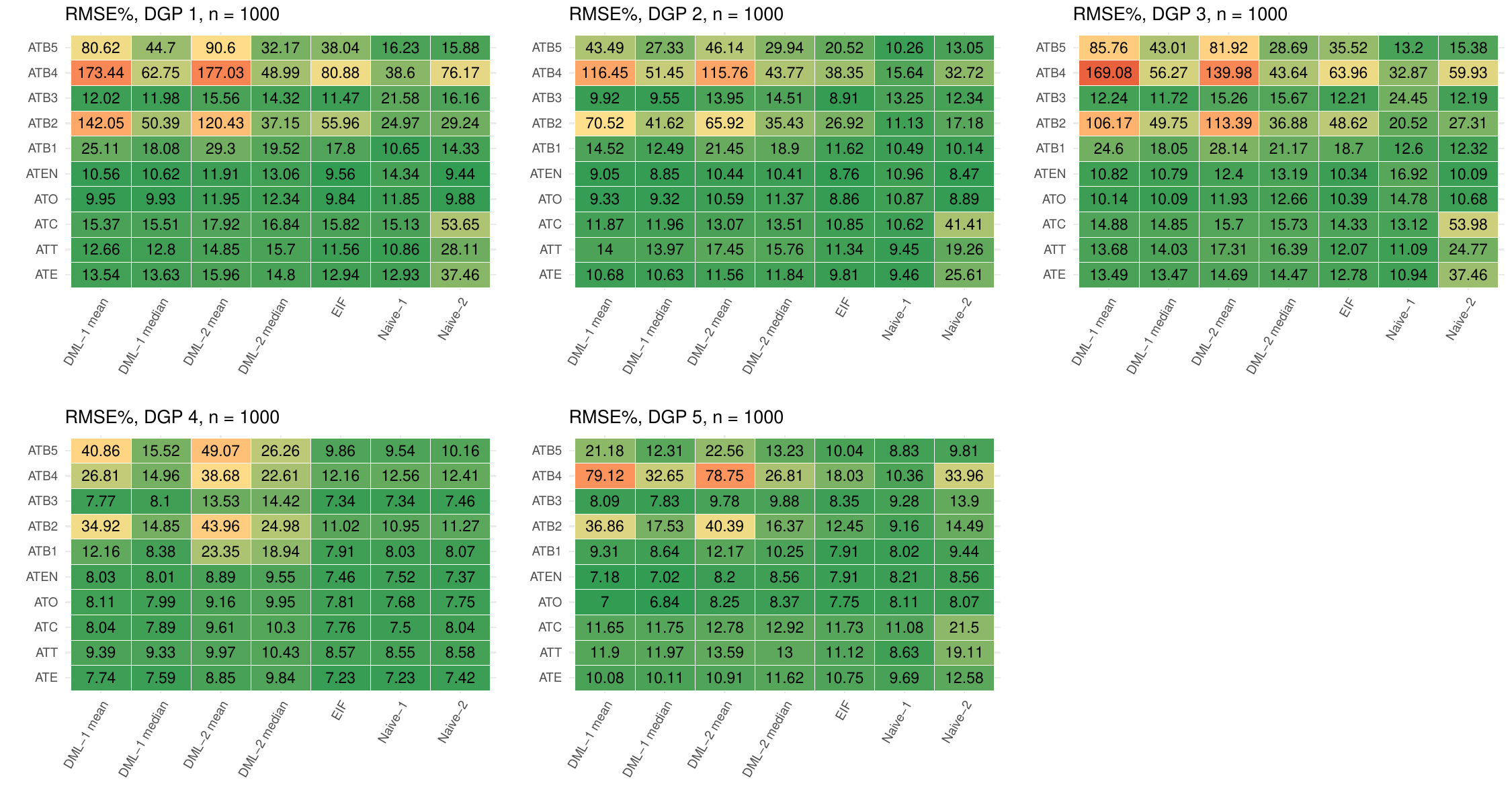}
    \caption{Simulation results of RMSE\% under $n=1000$ and simple GLM models. }
\end{figure}

\begin{figure}[H]
    \centering
    \includegraphics[width=\linewidth]{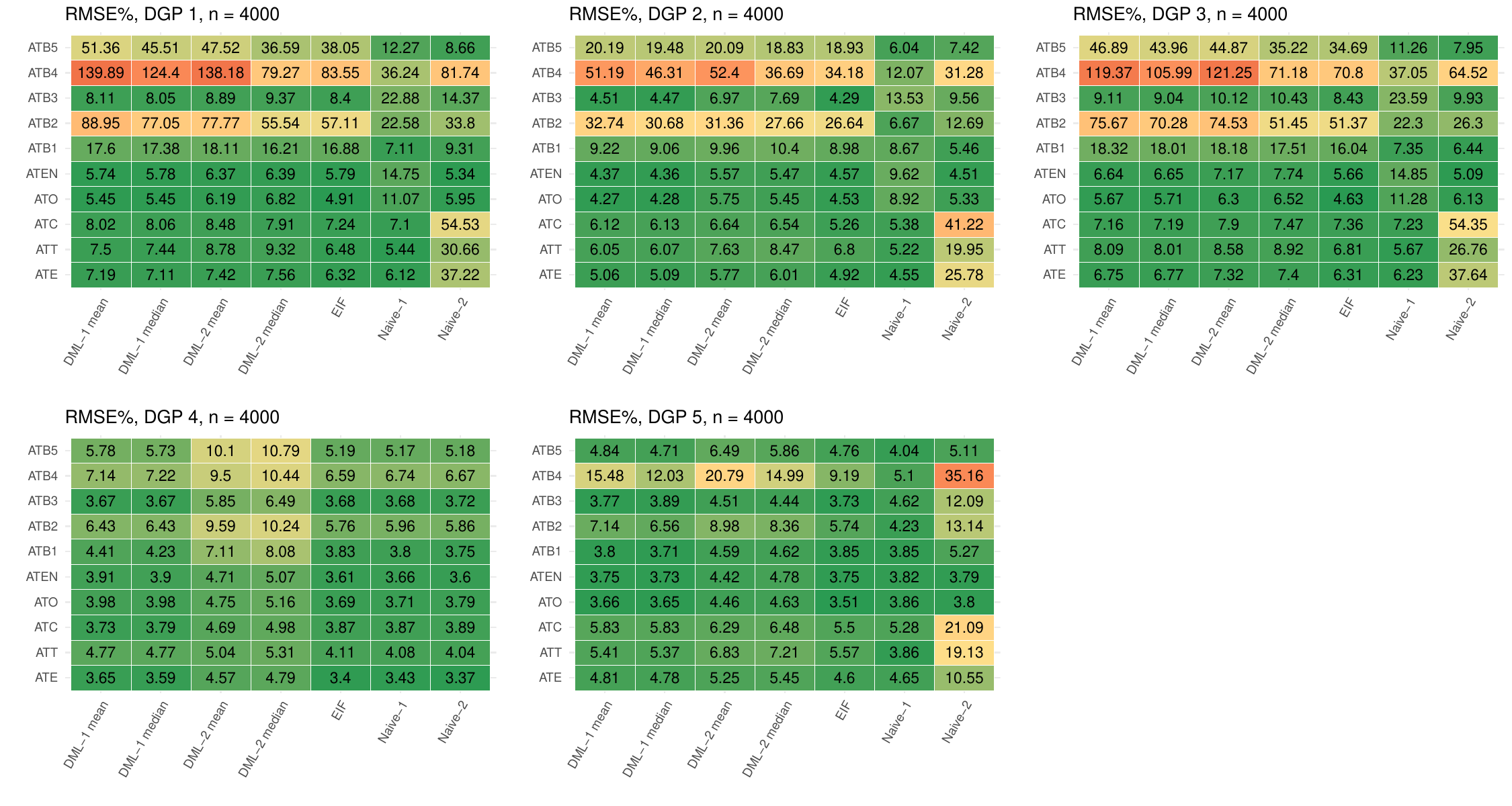}
    \caption{Simulation results of RMSE\% under $n=4000$ and simple GLM models. }
\end{figure}

\begin{figure}[H]
    \centering
    \includegraphics[width=\linewidth]{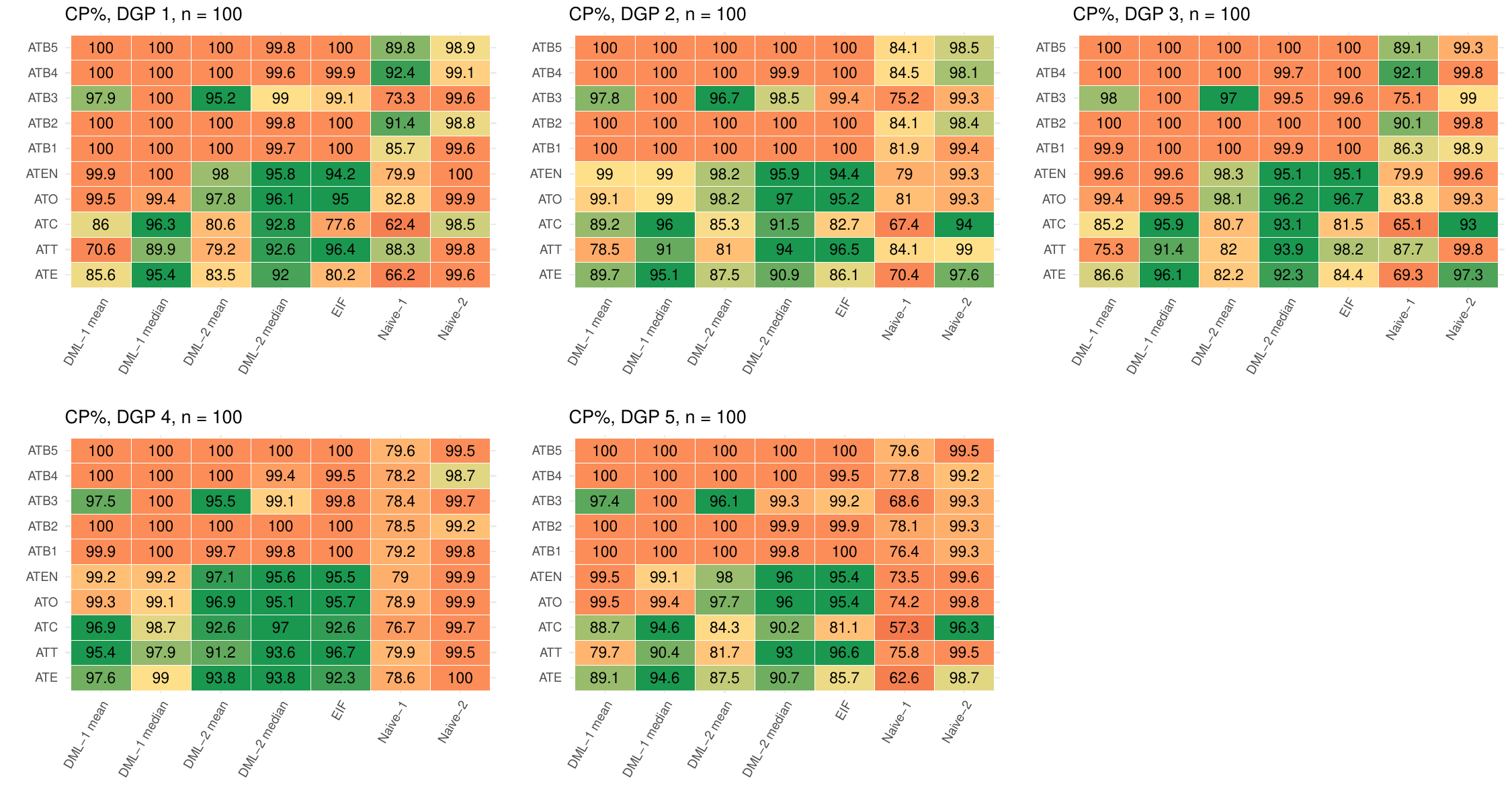}
    \caption{Simulation results of CP\% under $n=100$ and simple GLM models. }
\end{figure}

\begin{figure}[H]
    \centering
    \includegraphics[width=\linewidth]{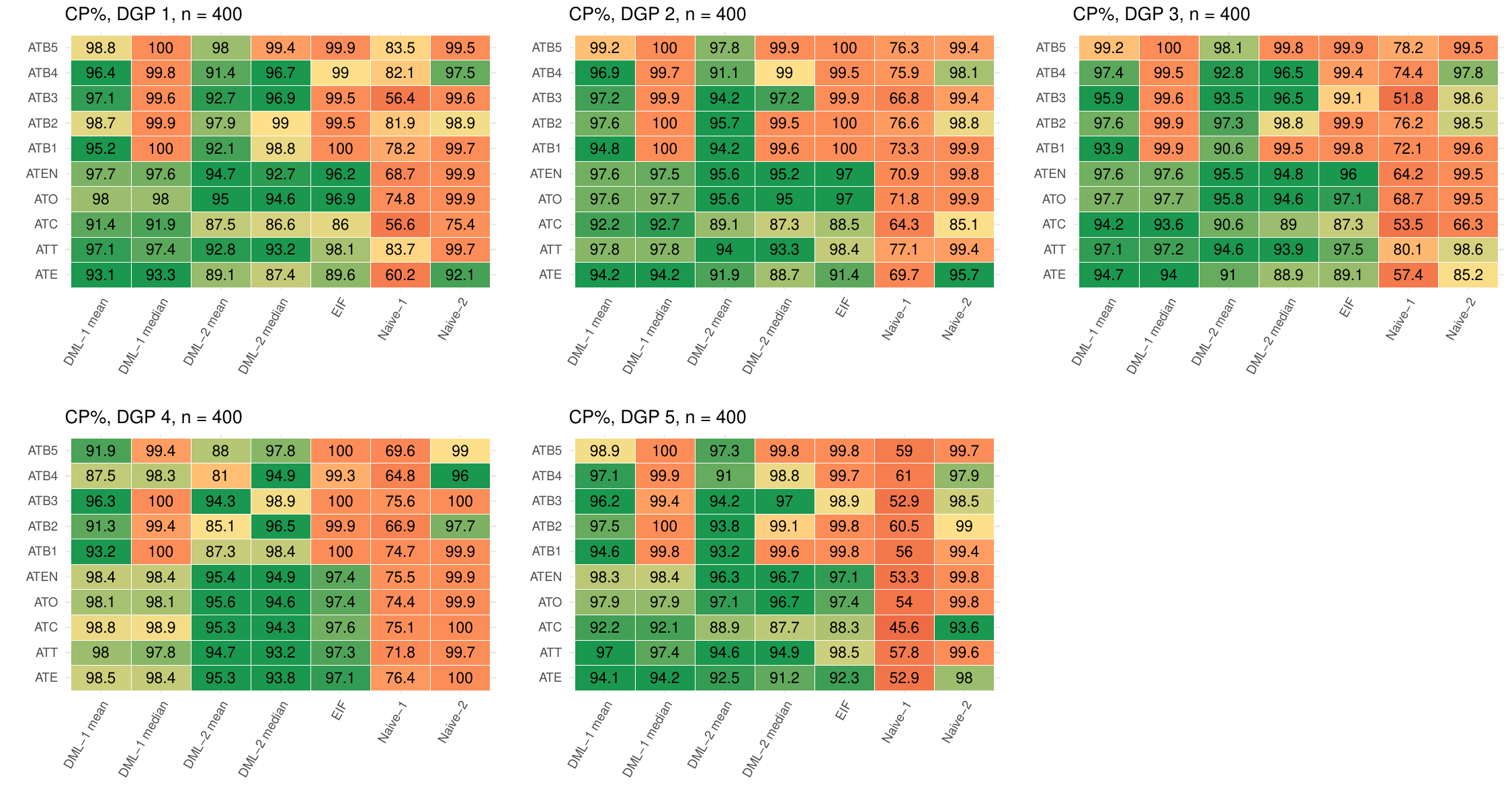}
    \caption{Simulation results of CP\% under $n=400$ and simple GLM models. }
\end{figure}

\begin{figure}[H]
    \centering
    \includegraphics[width=\linewidth]{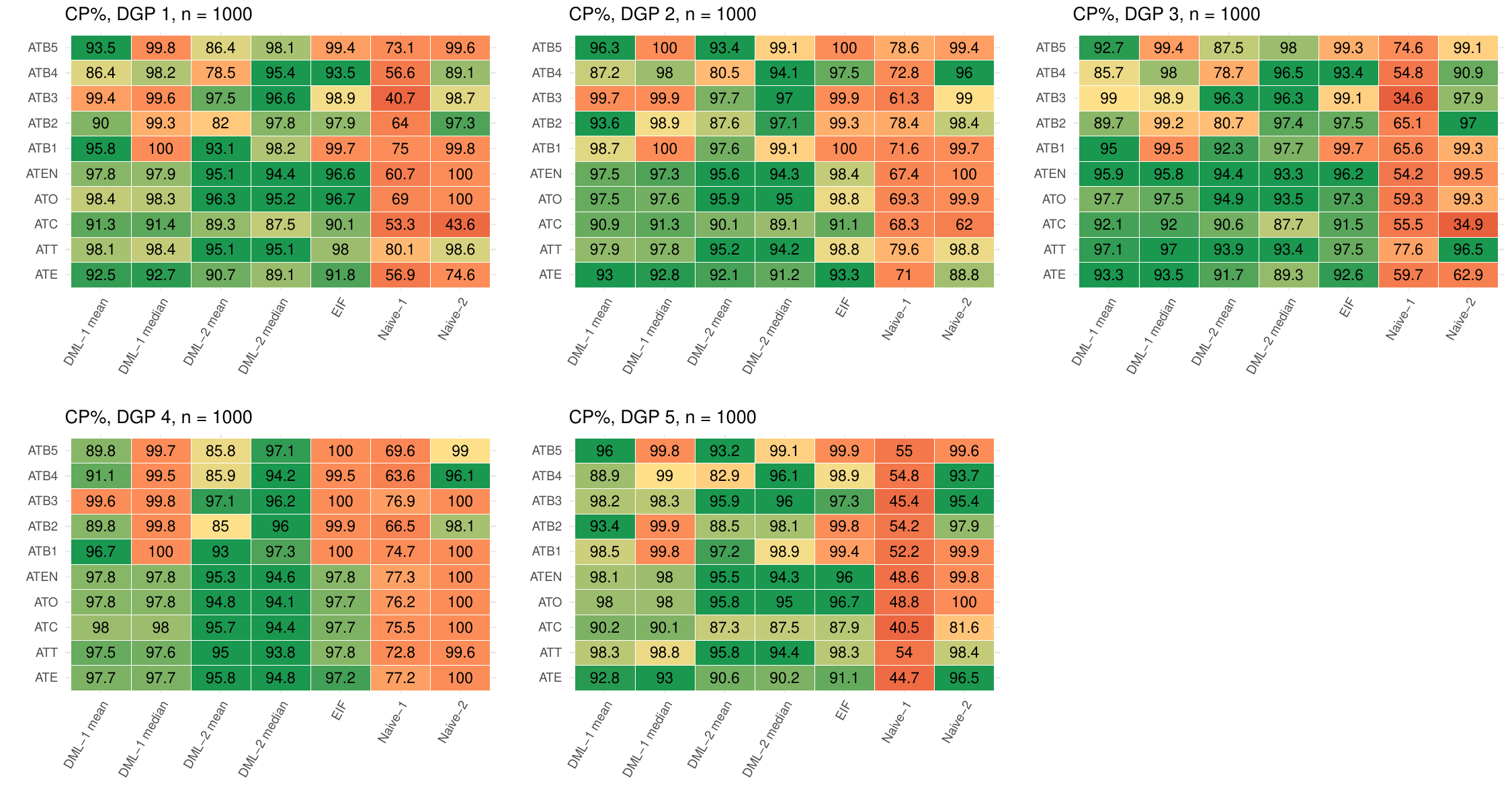}
    \caption{Simulation results of CP\% under $n=1000$ and simple GLM models. }
\end{figure}

\begin{figure}[H]
    \centering
    \includegraphics[width=\linewidth]{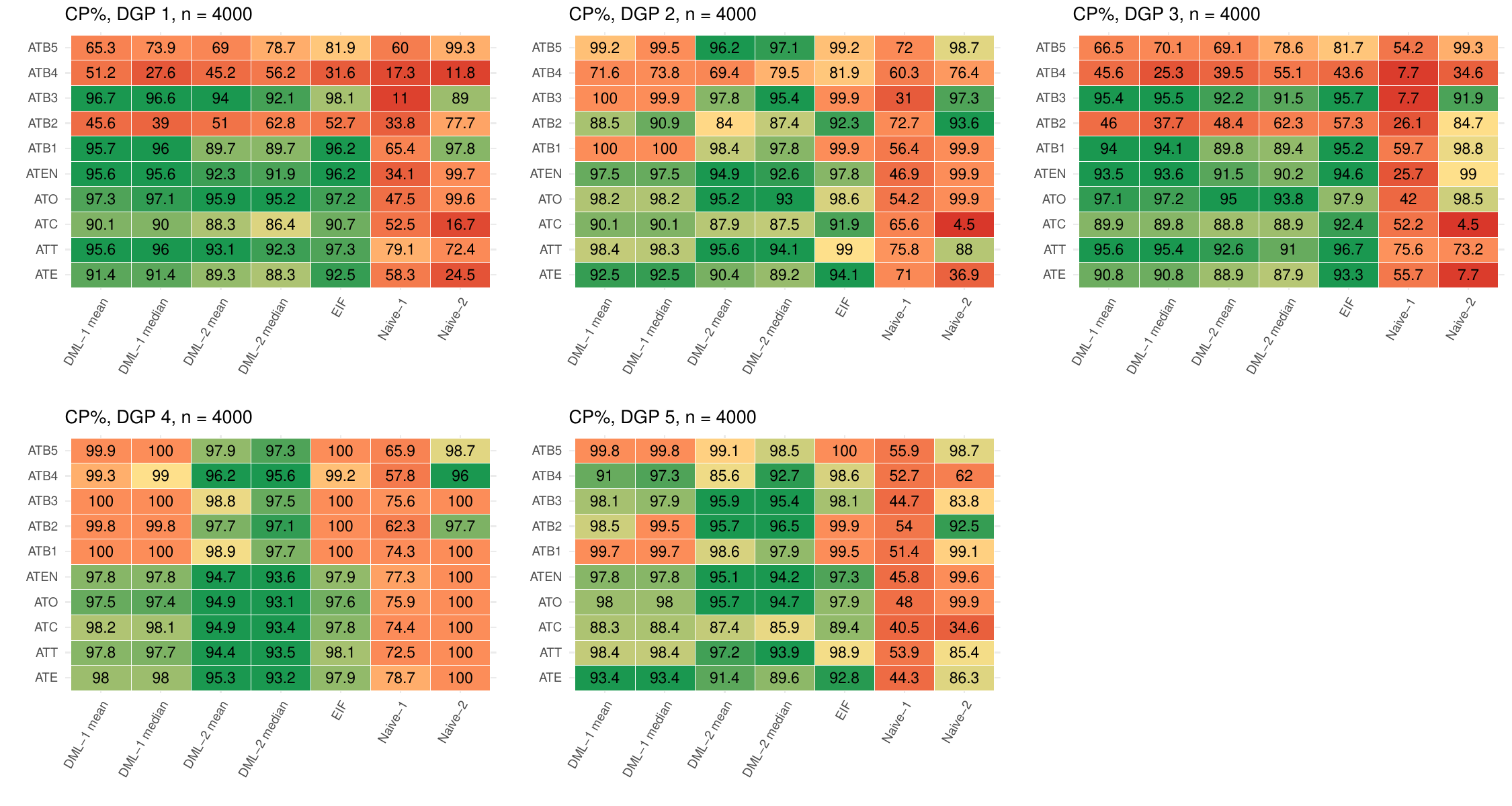}
    \caption{Simulation results of CP\% under $n=4000$ and simple GLM models. }
\end{figure}

\clearpage

\subsection{Results by using methods ensemble (parametric and ML models) for nuisance functions}\label{subapp:methodEnsb}

\begin{figure}[H]
    \centering
    \includegraphics[width=\linewidth]{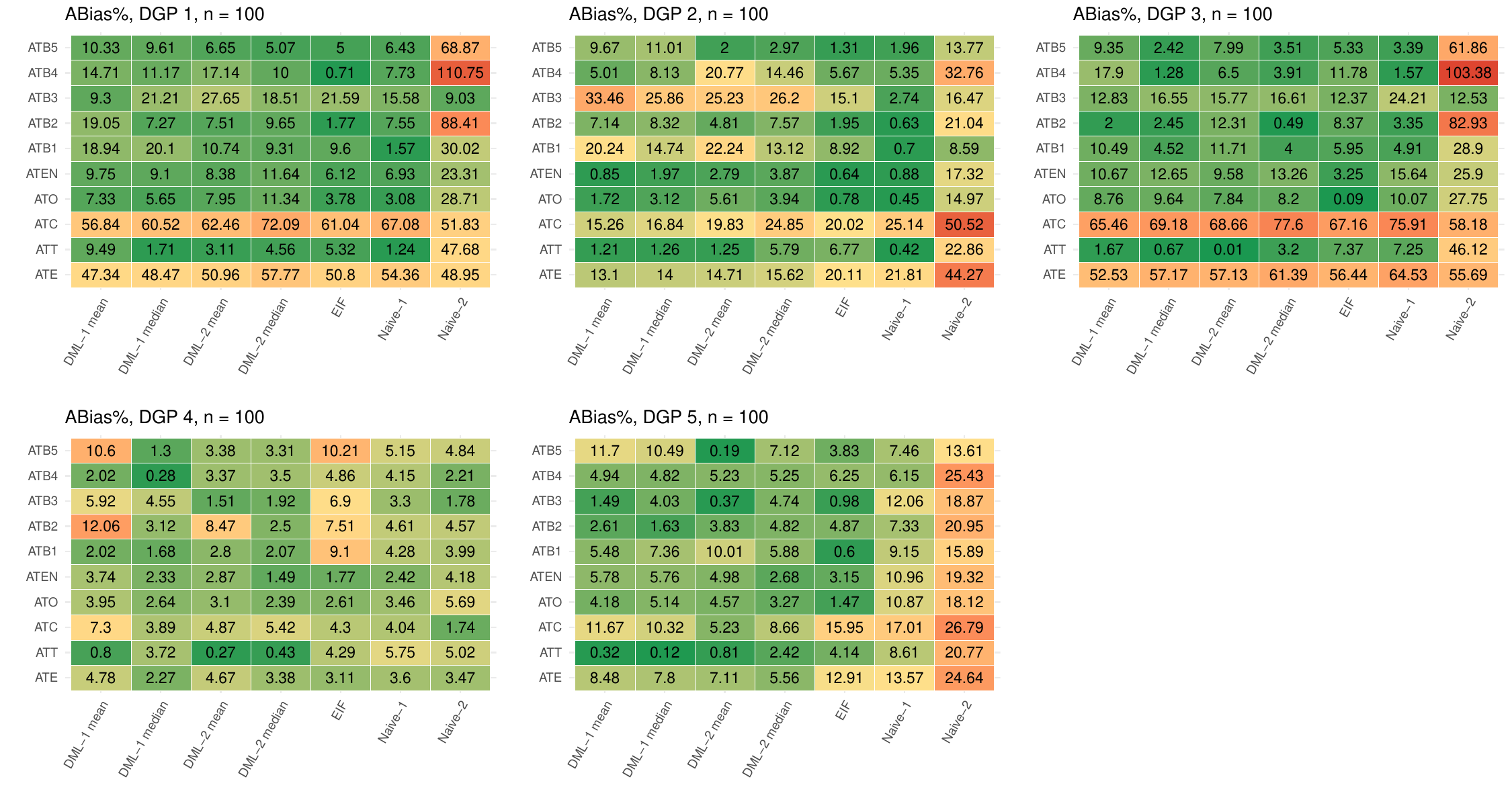}
    \caption{Simulation results of ARBias\% under $n=100$ and methods ensemble. }
\end{figure}

\begin{figure}[H]
    \centering
    \includegraphics[width=\linewidth]{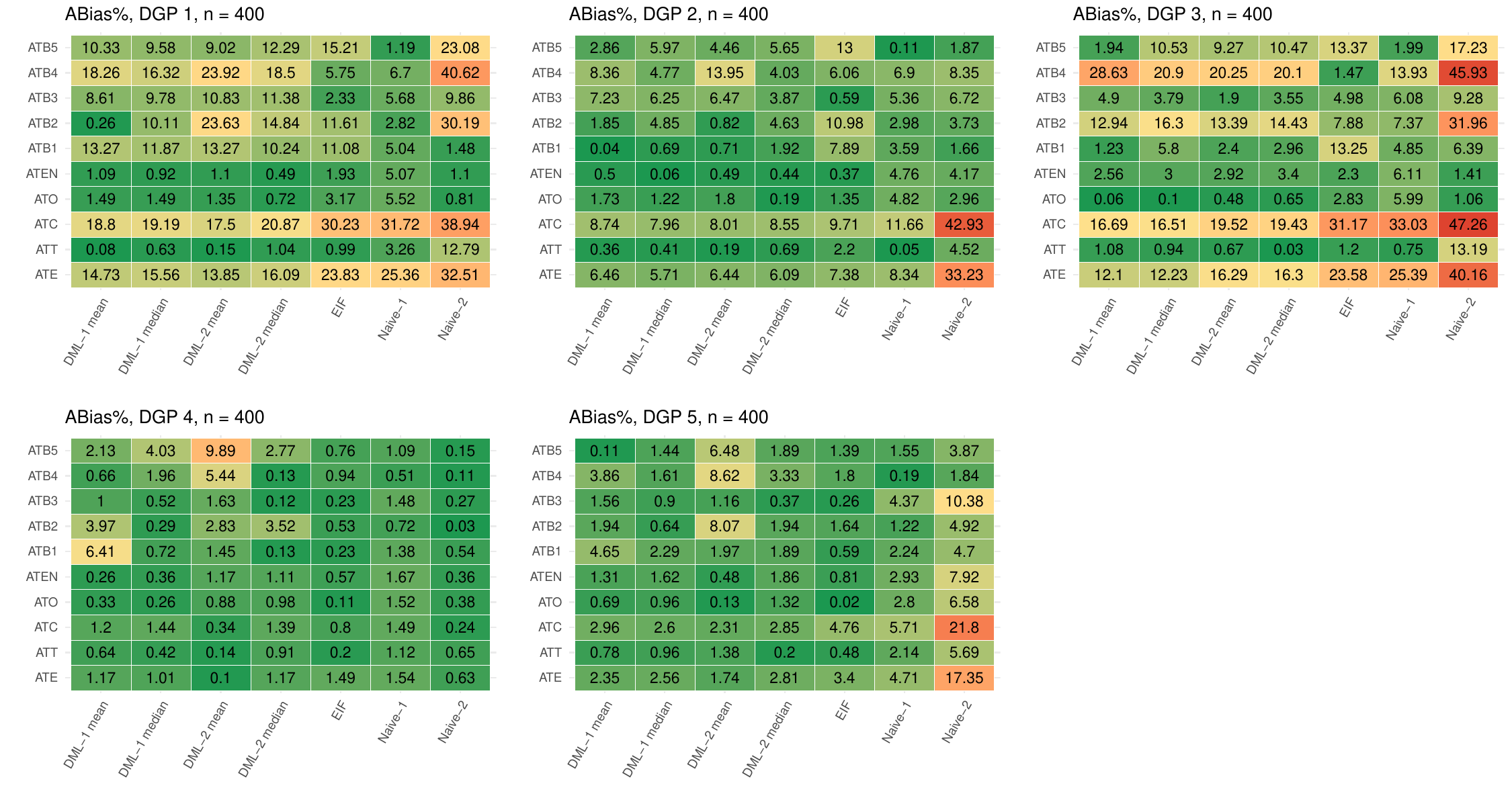}
    \caption{Simulation results of ARBias\% under $n=400$ and methods ensemble. }
\end{figure}

\begin{figure}[H]
    \centering
    \includegraphics[width=\linewidth]{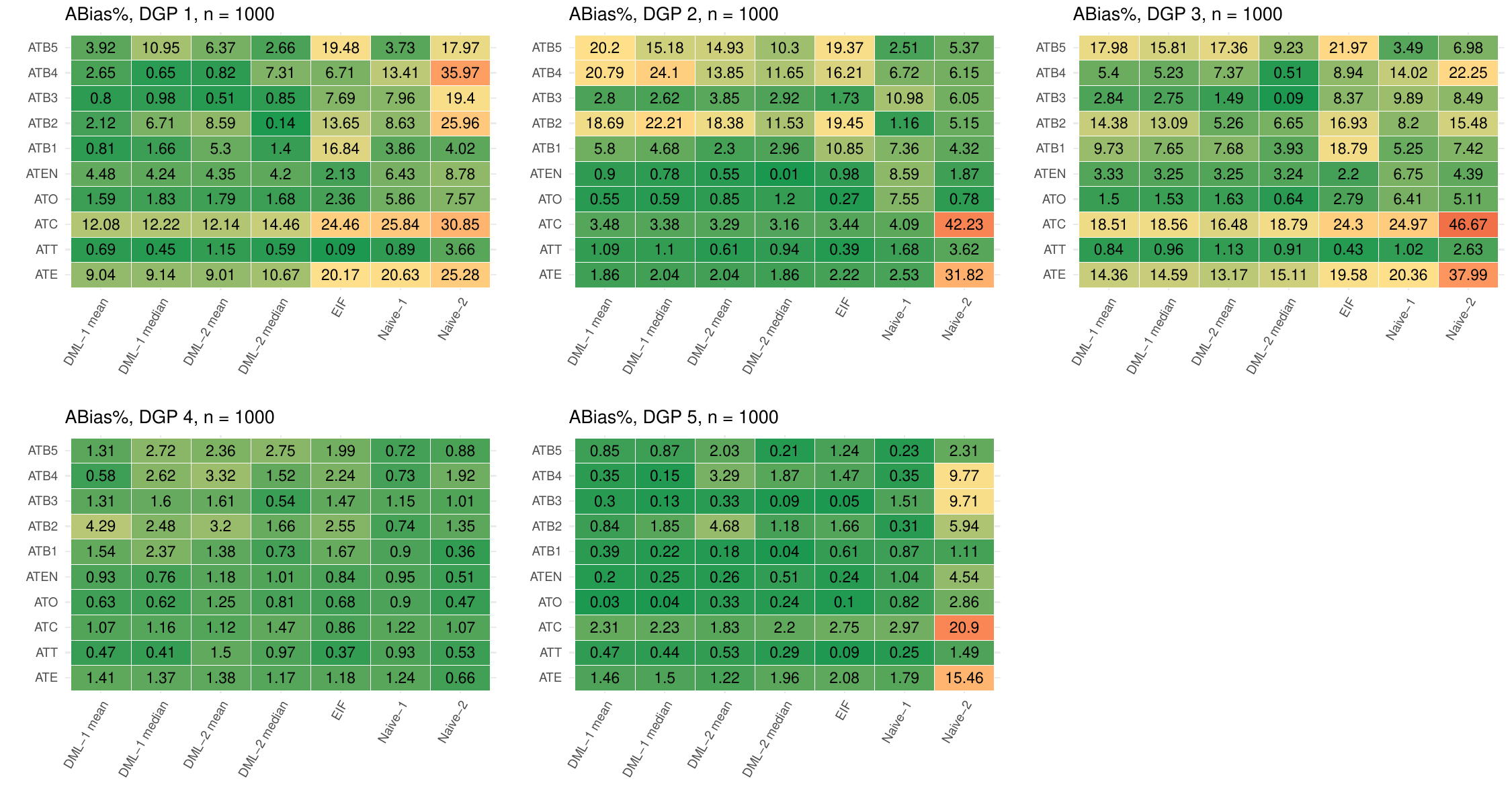}
    \caption{Simulation results of ARBias\% under $n=1000$ and methods ensemble. }
\end{figure}

\begin{figure}[H]
    \centering
    \includegraphics[width=\linewidth]{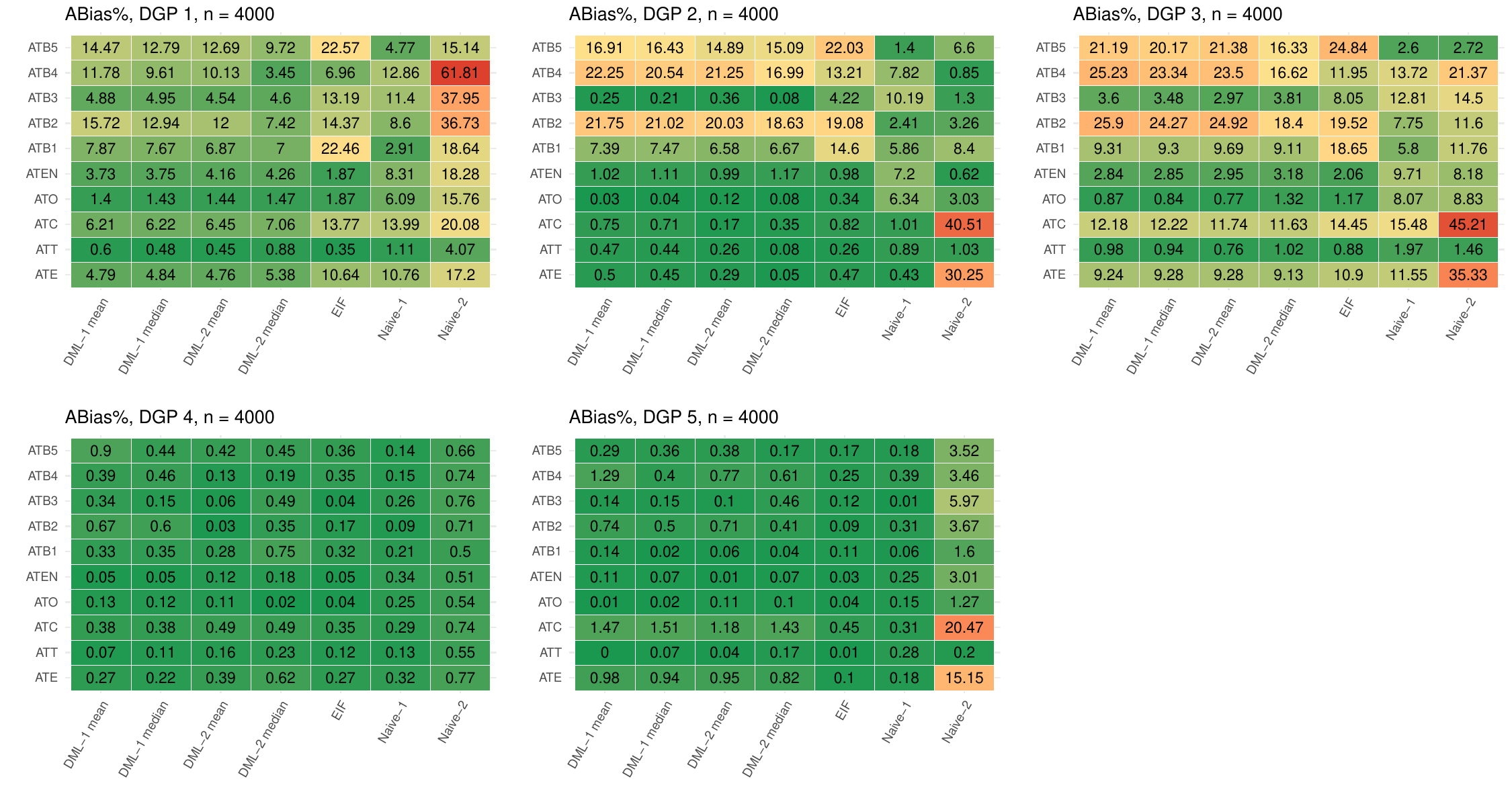}
    \caption{Simulation results of ARBias\% under $n=4000$ and methods ensemble. }
\end{figure}

\begin{figure}[H]
    \centering
    \includegraphics[width=\linewidth]{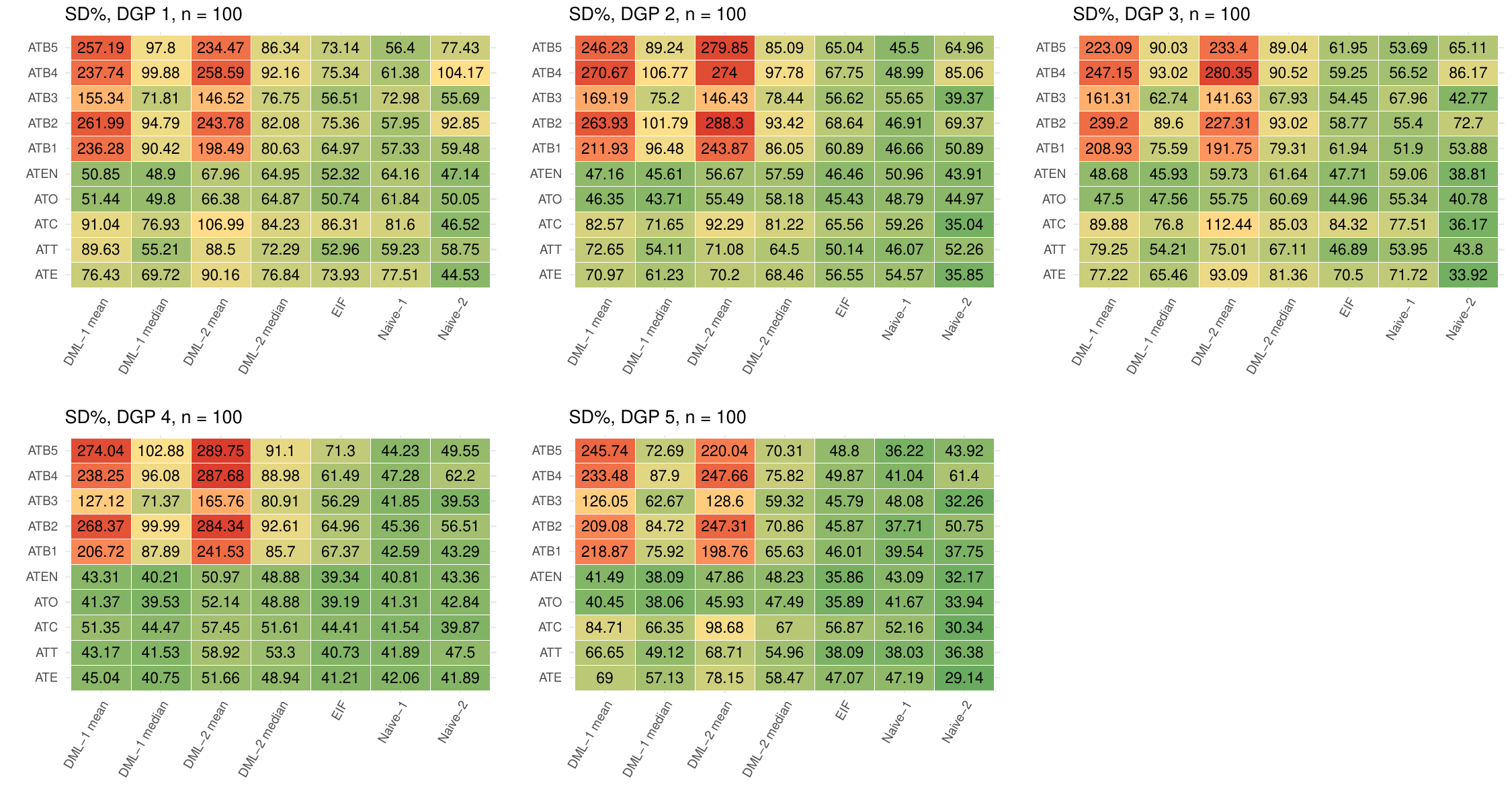}
    \caption{Simulation results of SD\% under $n=100$ and methods ensemble. }
\end{figure}

\begin{figure}[H]
    \centering
    \includegraphics[width=\linewidth]{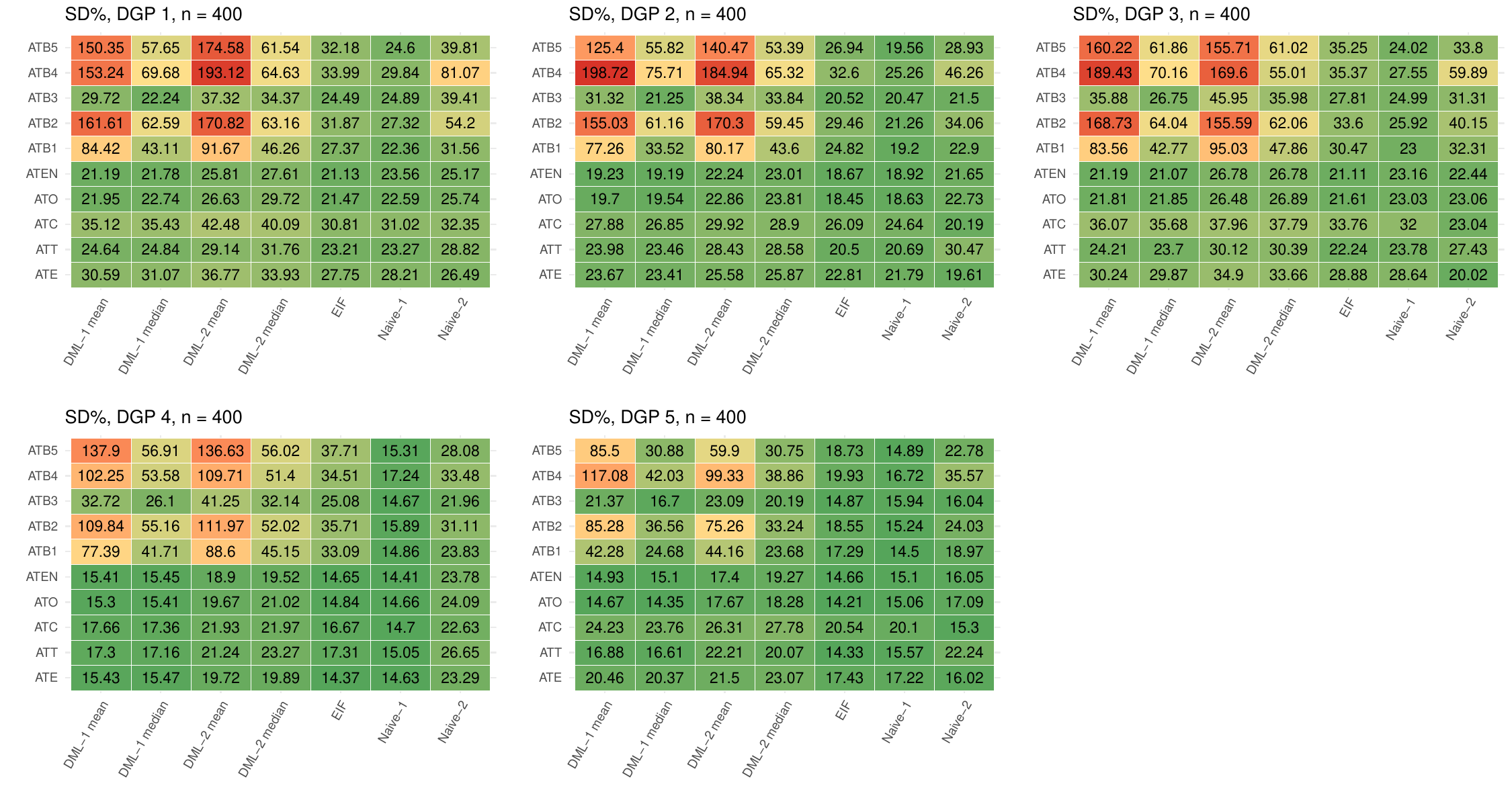}
    \caption{Simulation results of SD\% under $n=400$ and methods ensemble. }
\end{figure}

\begin{figure}[H]
    \centering
    \includegraphics[width=\linewidth]{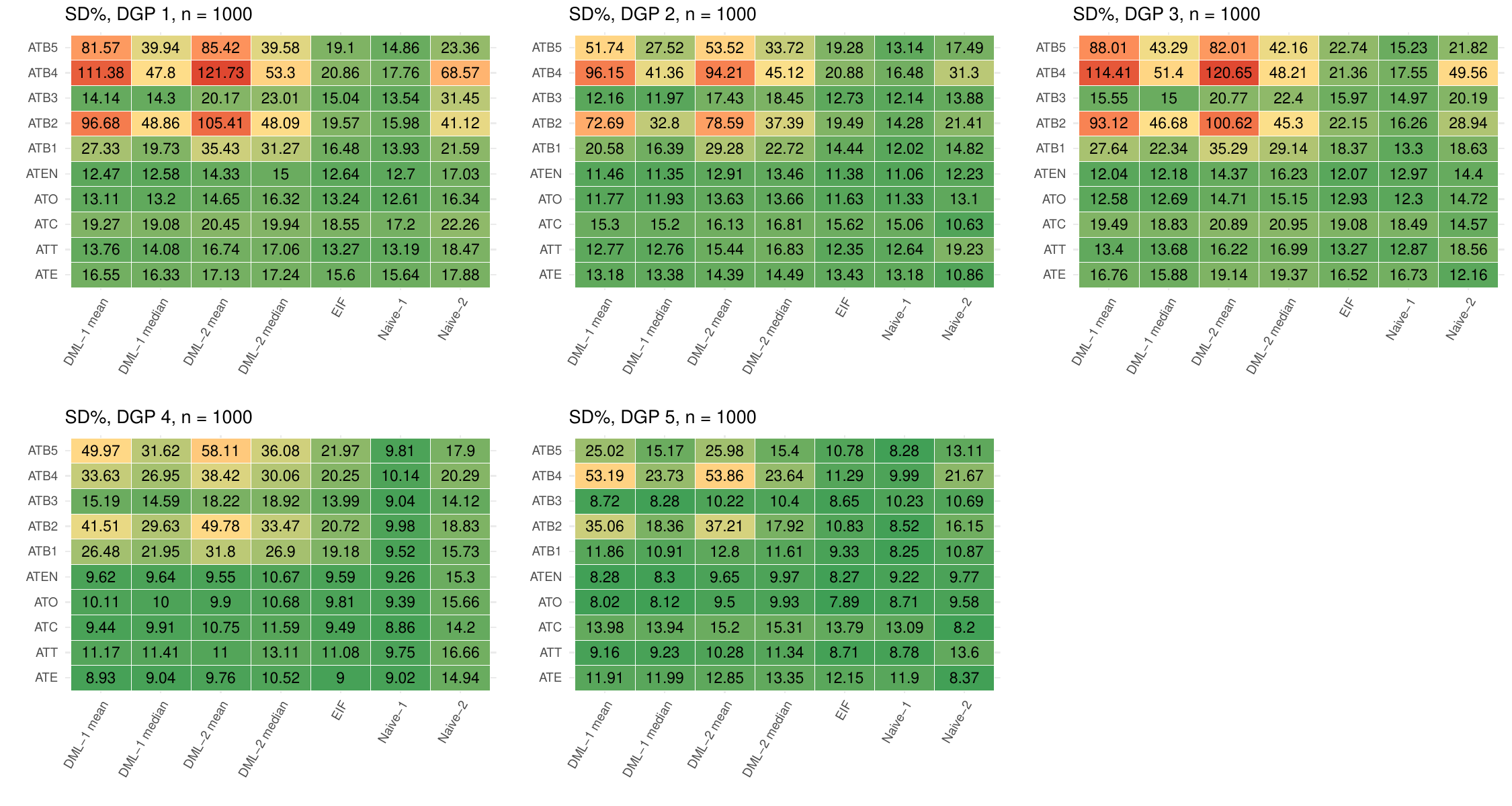}
    \caption{Simulation results of SD\% under $n=1000$ and methods ensemble. }
\end{figure}

\begin{figure}[H]
    \centering
    \includegraphics[width=\linewidth]{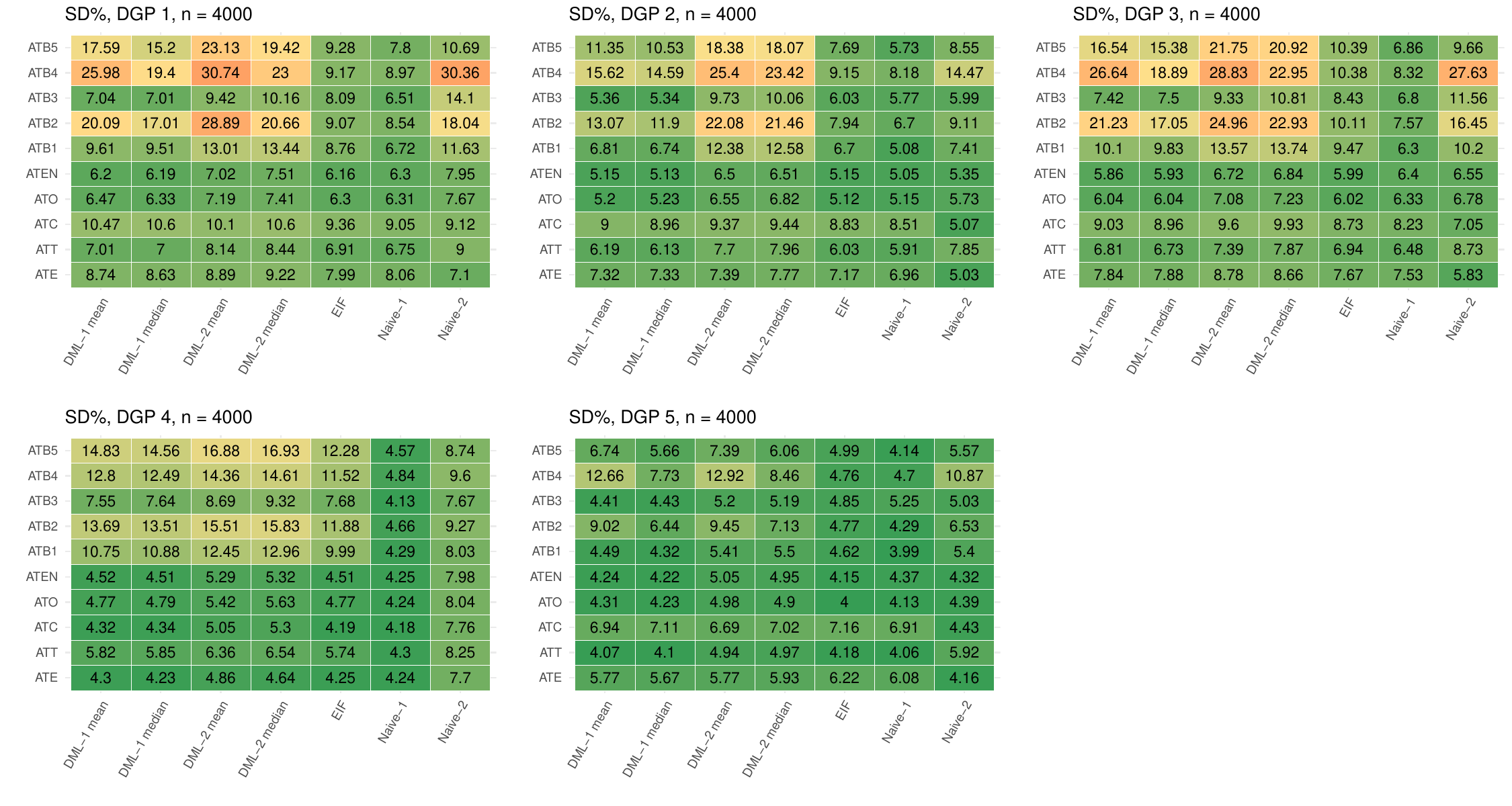}
    \caption{Simulation results of SD\% under $n=4000$ and methods ensemble. }
\end{figure}

\begin{figure}[H]
    \centering
    \includegraphics[width=\linewidth]{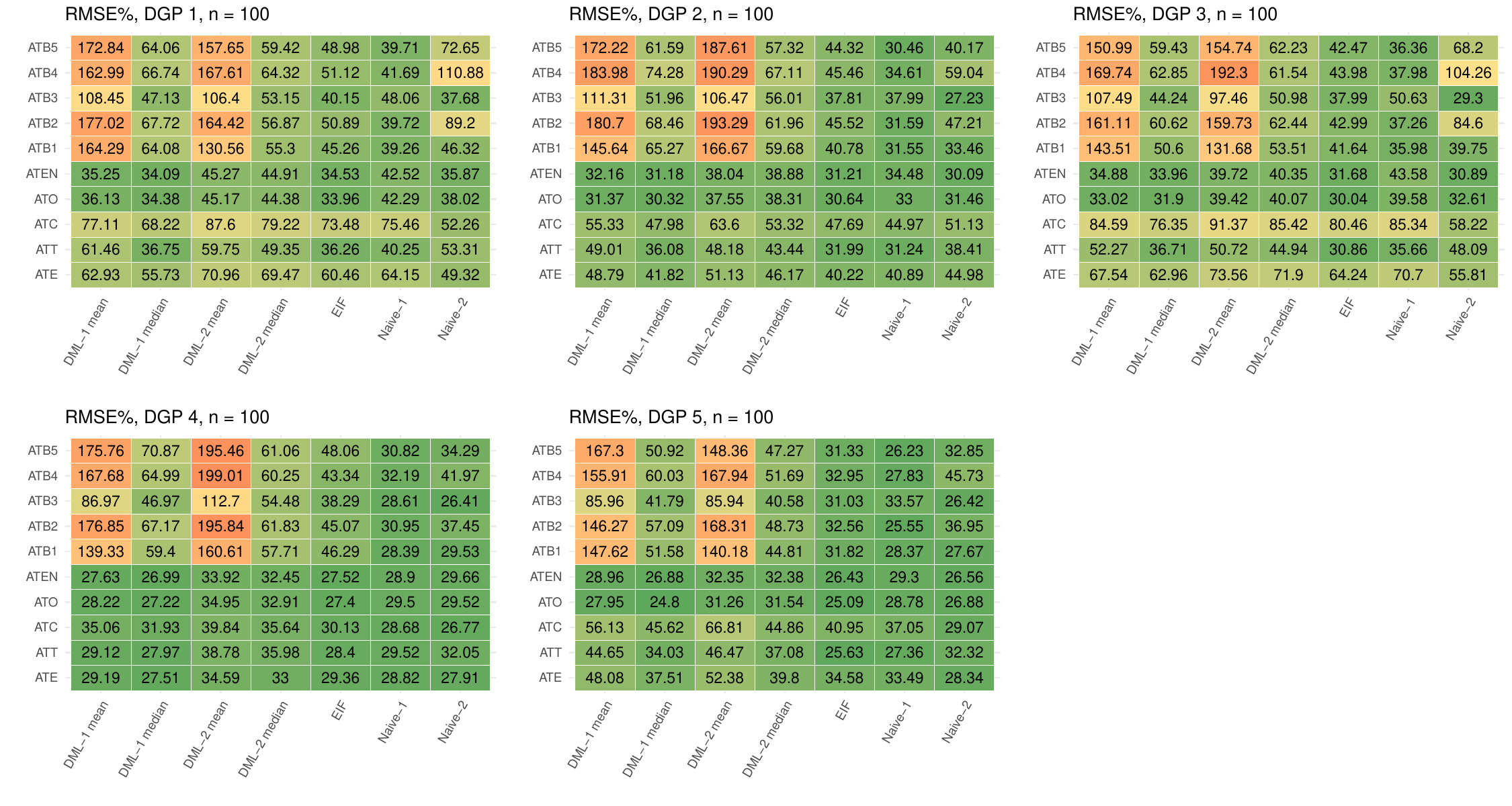}
    \caption{Simulation results of RMSE\% under $n=100$ and methods ensemble. }
\end{figure}

\begin{figure}[H]
    \centering
    \includegraphics[width=\linewidth]{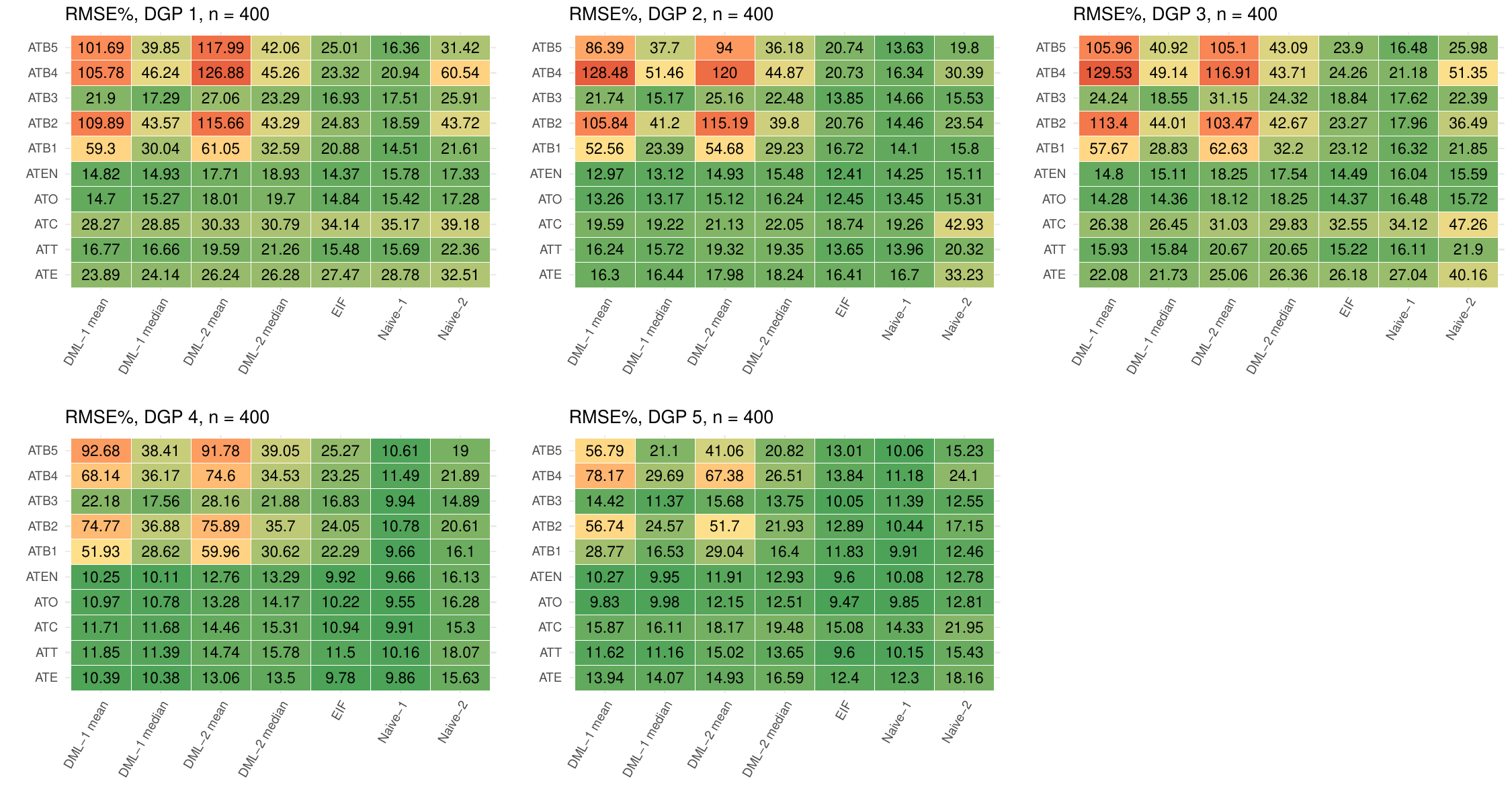}
    \caption{Simulation results of RMSE\% under $n=400$ and methods ensemble. }
\end{figure}

\begin{figure}[H]
    \centering
    \includegraphics[width=\linewidth]{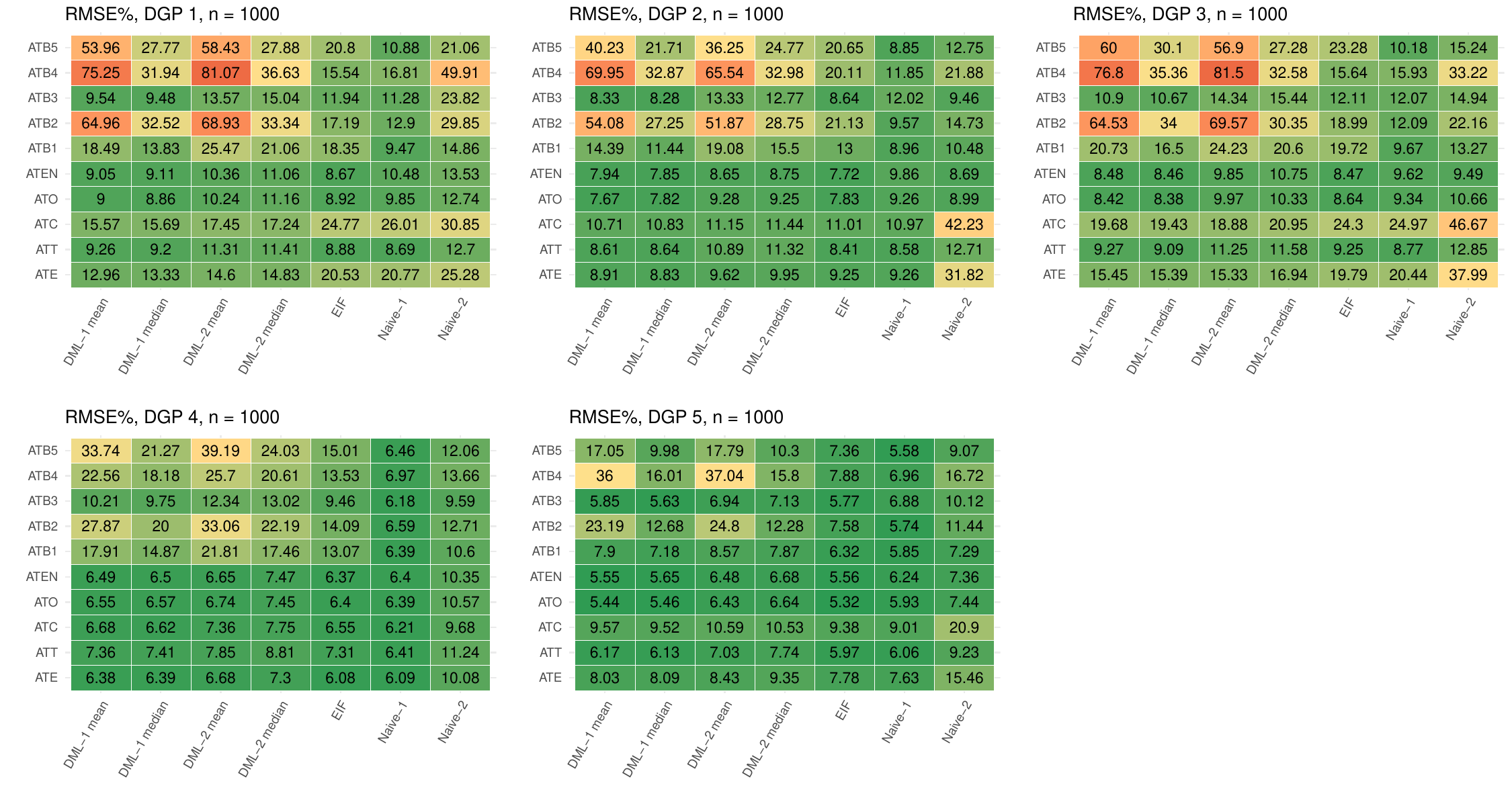}
    \caption{Simulation results of RMSE\% under $n=1000$ and methods ensemble. }
\end{figure}

\begin{figure}[H]
    \centering
    \includegraphics[width=\linewidth]{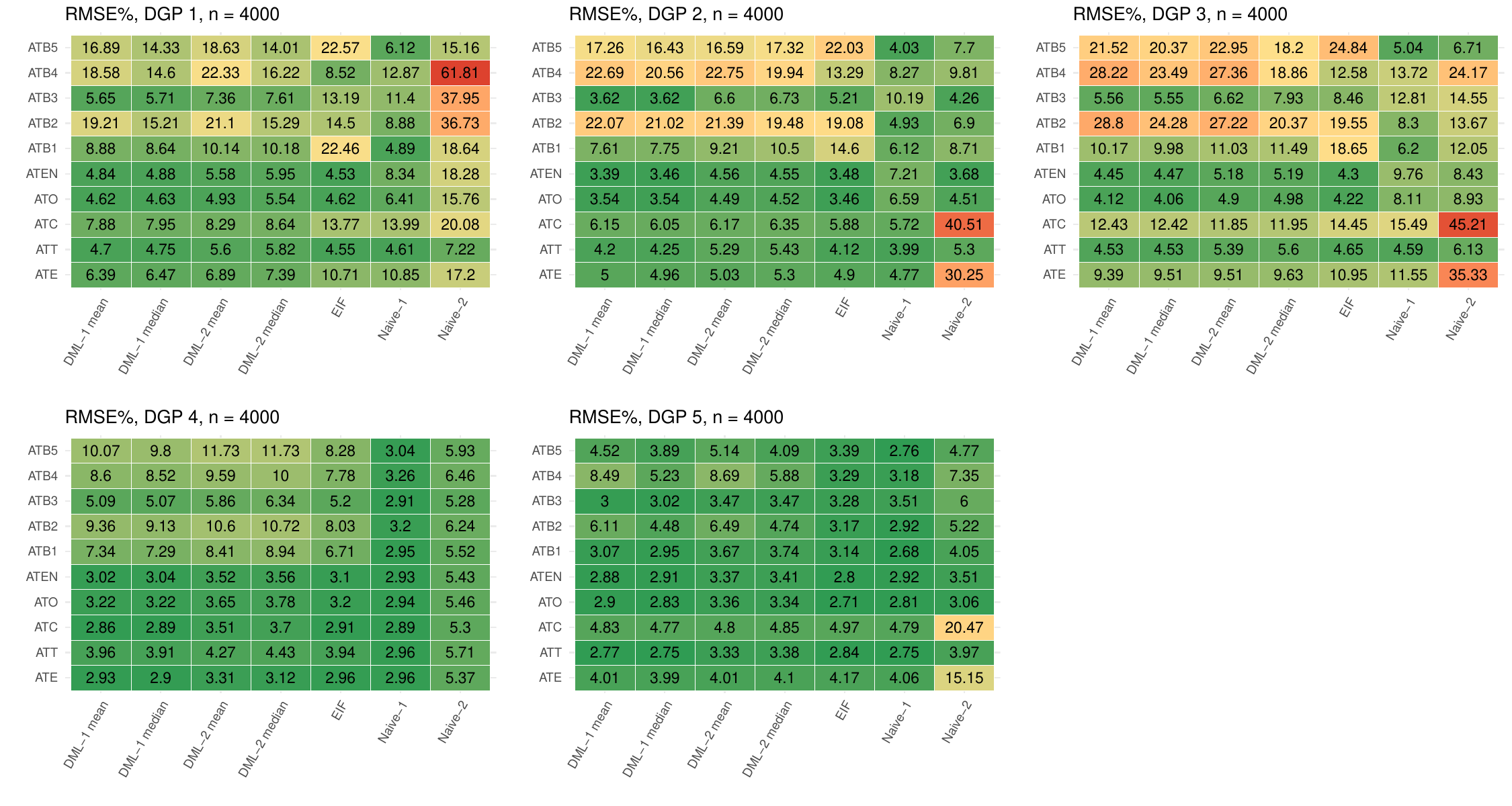}
    \caption{Simulation results of RMSE\% under $n=4000$ and methods ensemble. }
\end{figure}

\begin{figure}[H]
    \centering
    \includegraphics[width=\linewidth]{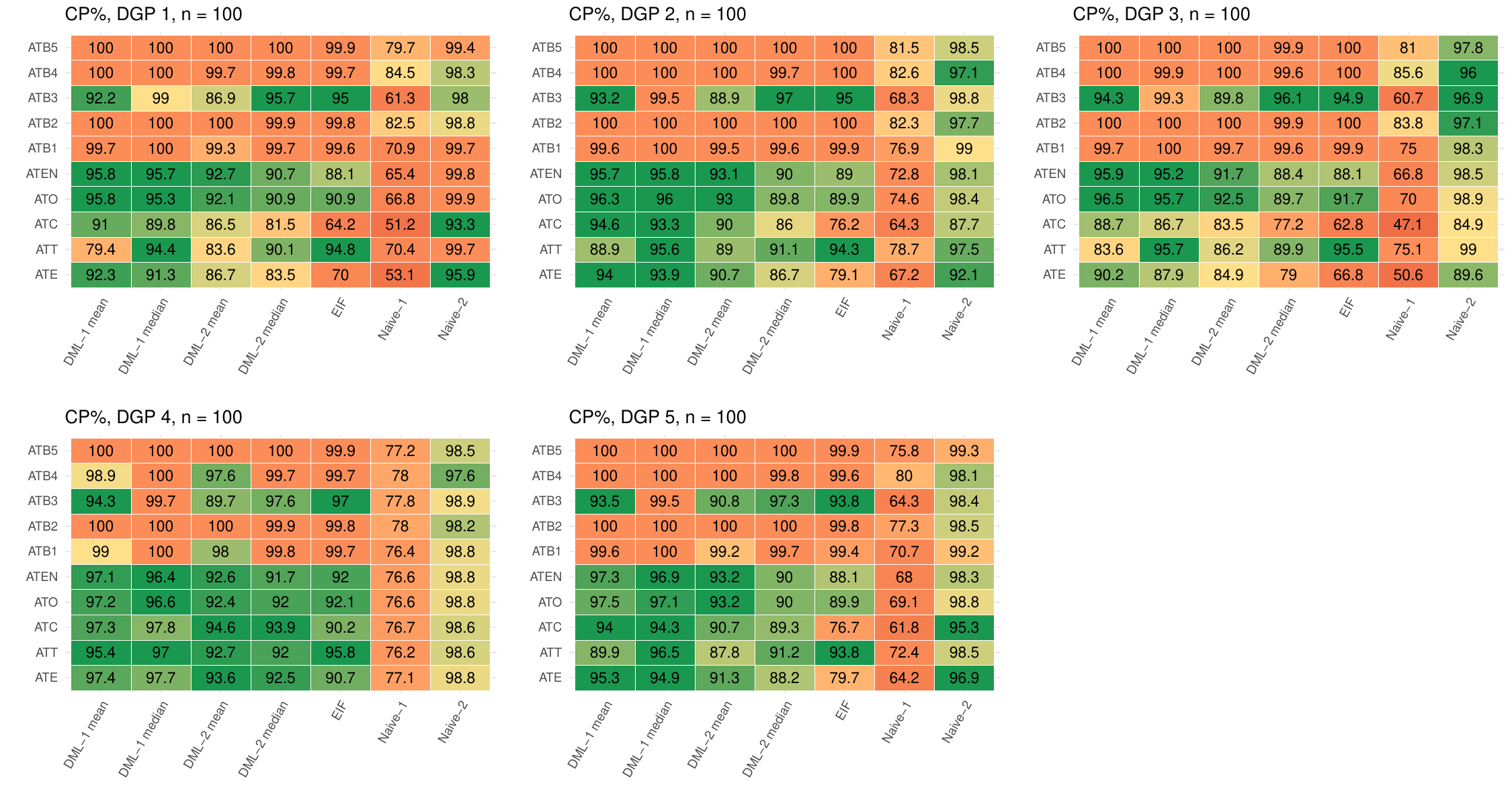}
    \caption{Simulation results of CP\% under $n=100$ and methods ensemble. }
\end{figure}

\begin{figure}[H]
    \centering
    \includegraphics[width=\linewidth]{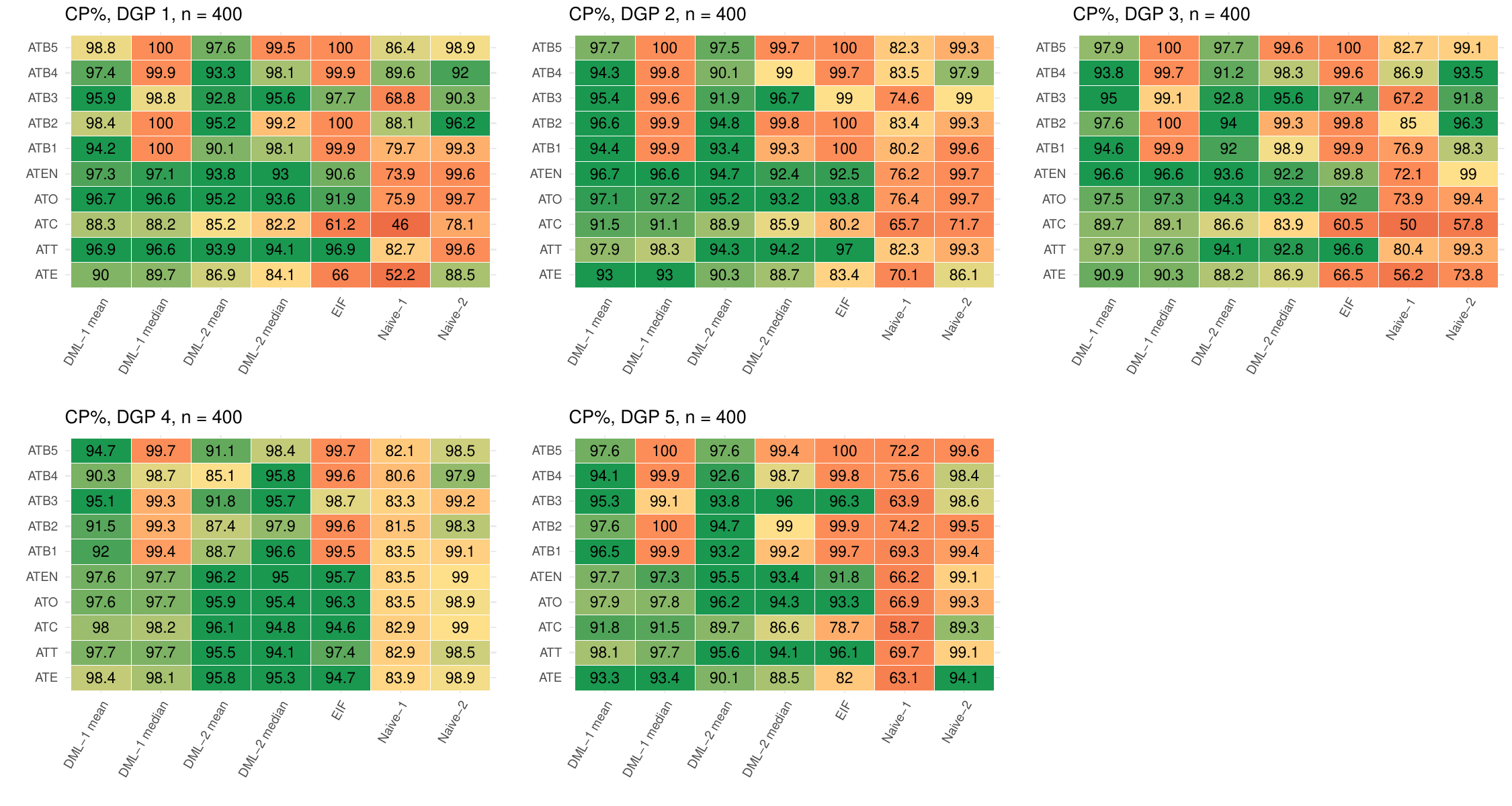}
    \caption{Simulation results of CP\% under $n=400$ and methods ensemble. }
\end{figure}

\begin{figure}[H]
    \centering
    \includegraphics[width=\linewidth]{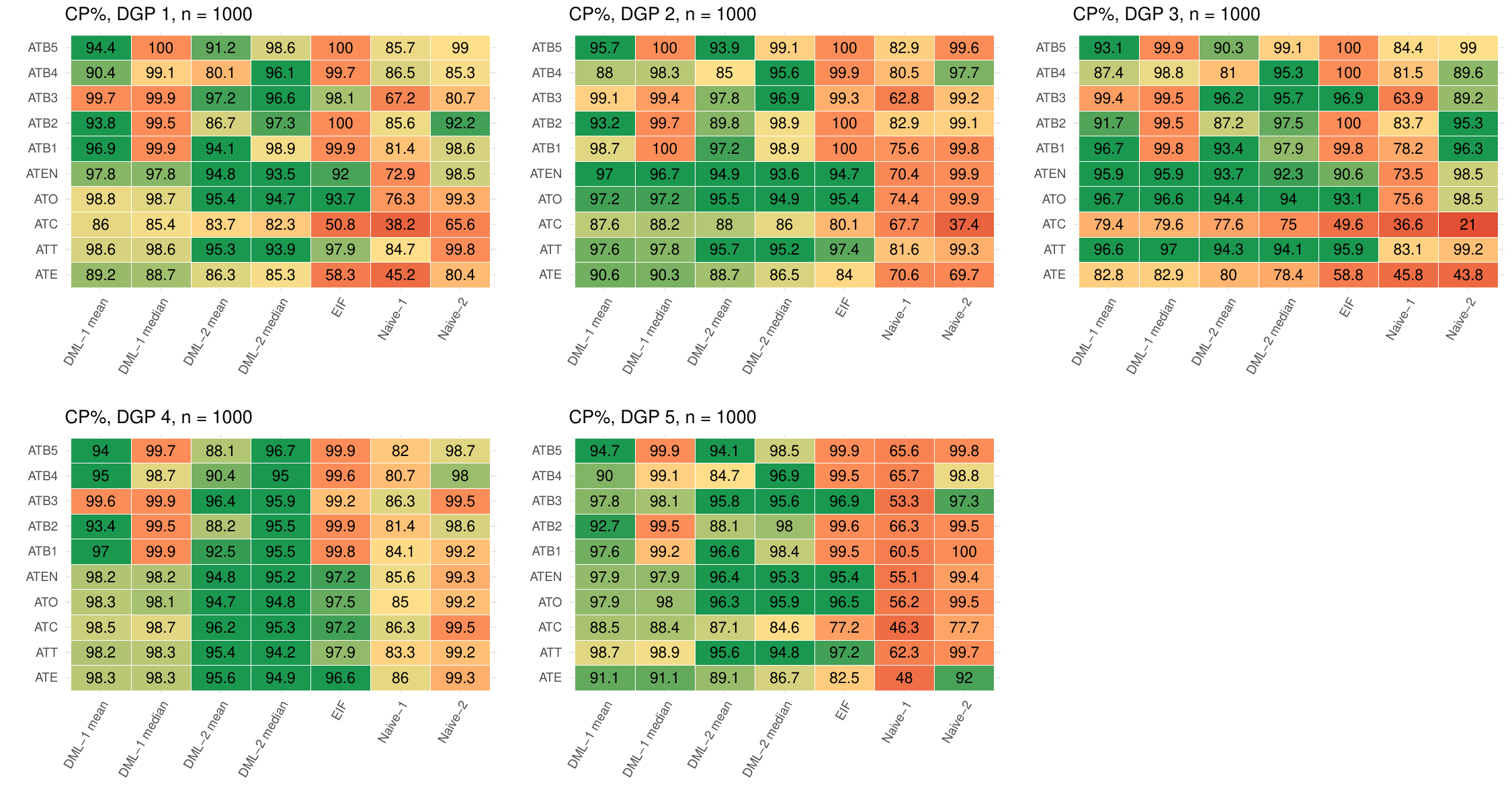}
    \caption{Simulation results of CP\% under $n=1000$ and methods ensemble. }
\end{figure}

\begin{figure}[H]
    \centering
    \includegraphics[width=\linewidth]{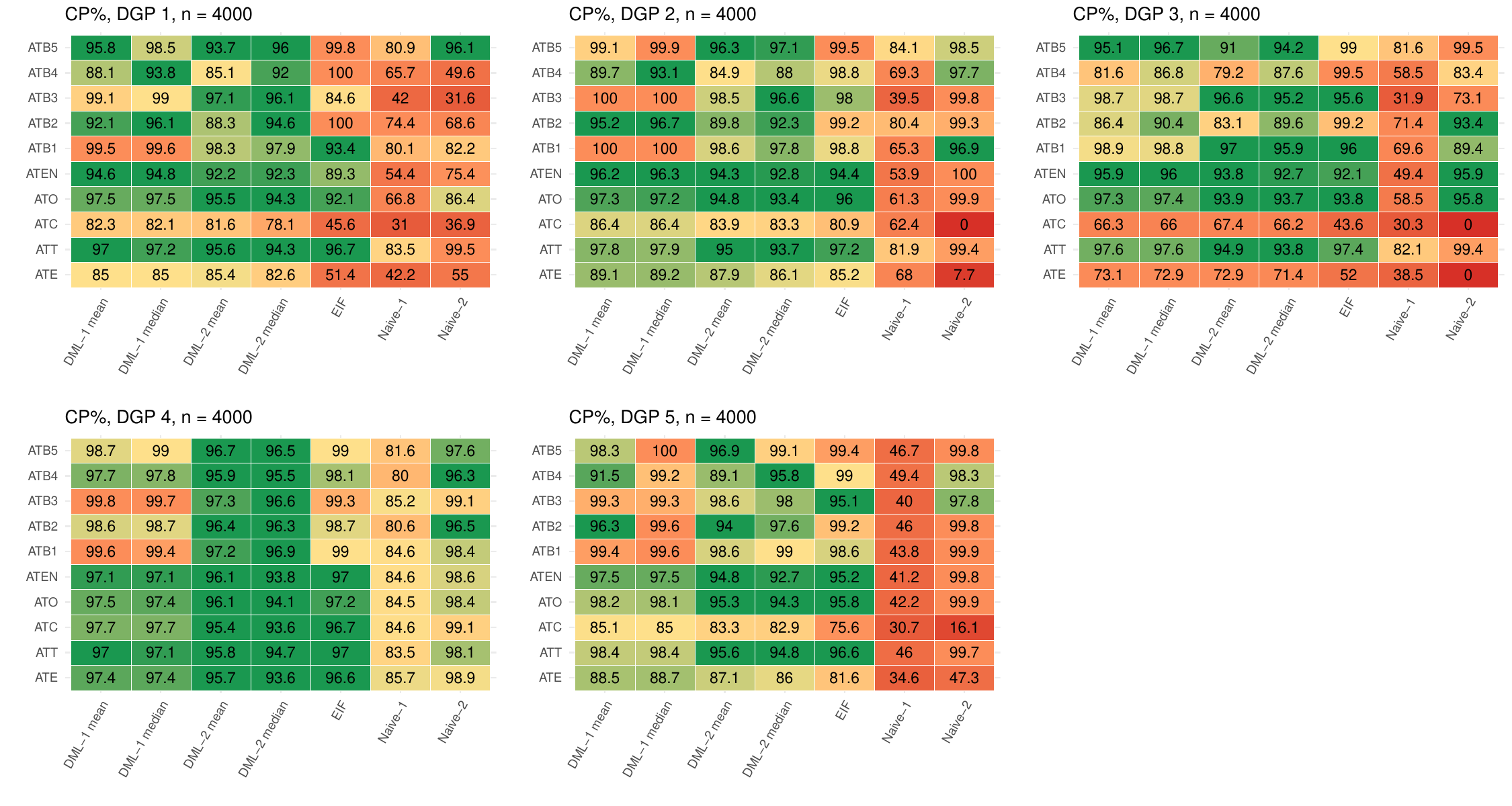}
    \caption{Simulation results of CP\% under $n=4000$ and methods ensemble. }
\end{figure}

\end{document}